\documentclass{lmcs} 
\pdfoutput=1

\usepackage{lastpage}
\lmcsdoi{17}{3}{17}
\lmcsheading{}{\pageref{LastPage}}{}{}%
{Jun.~11,~2020}{Aug.~10,~2021}{}

\usepackage[utf8]{inputenc}


\usepackage{packages}

\usepackage[noend]{algpseudocode}

\usepackage{macros}

\newif\ifLongVersionWithAppendix\LongVersionWithAppendixtrue

\begin{document}

 
\title[Axiomatisation for Quantifier-Free Separation Logic]{\texorpdfstring{A Complete Axiomatisation for\\ Quantifier-Free Separation Logic}{A Complete Axiomatisation for Quantifier-Free Separation Logic}}
\titlecomment{{\lsuper*}This is the long version of the first part of~\cite{Demri&Lozes&Mansutti20}.}

\author[S.~Demri]{St\'ephane Demri\rsuper{a}}	

\author[E.~Lozes]{{\'E}tienne Lozes\rsuper{b}}	

\author[A.~Mansutti]{Alessio Mansutti\rsuper{a}}	

\address{\lsuper{a}LSV, CNRS, ENS Paris-Saclay - 4, avenue des Sciences - 91190 Gif-sur-Yvette}	
\email{demri@lsv.fr, mansutti@lsv.fr}  

\address{\lsuper{b}I3S - Les Algorithmes - B\^atiment Euclide B - 2000, route des Lucioles - 06900 Sophia Antipolis}	
\email{etienne.lozes@i3s.unice.fr}  





\begin{abstract}
We present the first complete axiomatisation for quantifier-free separation logic.
The logic is equipped with the standard concrete heaplet semantics and the proof system
has no external  feature such as nominals/labels.
It is not possible to rely completely on proof systems for Boolean BI as
the concrete semantics needs to be taken into account.
Therefore, we present the first internal Hilbert-style axiomatisation for quantifier-free separation logic.
The calculus is divided in three parts: the axiomatisation
of core formulae where Boolean combinations of core formulae capture the expressivity of the whole logic,
axioms and inference rules to simulate a bottom-up elimination of separating connectives, and finally
structural axioms and inference rules from propositional calculus and Boolean BI with the magic wand.

\end{abstract}

\maketitle

\section{Introduction}\label{section:introduction}
\subsection*{The virtue of axiomatising program logics}
Designing a Hilbert-style axiomatisation for your favourite logic is usually quite challenging.
This does not  lead necessarily to optimal decision procedures, but the completeness proof usually provides essential
insights to better understand the logic at hand. That is why many logics related to program verification
have been axiomatised, often requiring non-trivial completeness proofs.
By way of example, there are axiomatisations for
the linear-time $\mu$-calculus~\cite{Kaivola95,Doumane17},
the modal $\mu$-calculus~\cite{Walukiewicz00}
or for
the alternating-time temporal logic ATL~\cite{Goranko&vanDrimmelen06},
the full computation tree logic CTL$^{*}$~\cite{Reynolds01},
for probabilistic extensions of $\mu$-calculus~\cite{Larsen&Mardare&Xue16}
or for a coalgebraic generalisation~\cite{Schroder&Venema18}.
Concerning the separation logics that extend Hoare-Floyd logic
to verify programs with mutable data structures
(see e.g.~\cite{OHearn&Pym99,Reynolds02,Ishtiaq&OHearn01,OHearn12,Pym&Spring&OHearn18}),
a Hilbert-style axiomatisation of Boolean BI has been introduced in~\cite{Galmiche&Larchey06},
but remained at the abstract level of Boolean BI. More recently, HyBBI~\cite{Brotherston&Villard14},
a hybrid version of Boolean BI has been introduced in order
to axiomatise various classes of abstract separation logics; HyBBI
naturally considers  classes of abstract models (typically preordered partial monoids) but it does not fit
exactly the heaplet semantics of separation logics. Furthermore, the addition of nominals (in the sense
of hybrid modal logics, see e.g.~\cite{Areces&Blackburn&Marx01}) extends substantially the object language. Other  frameworks
to axiomatise classes of abstract separation logics can be found
in~\cite{Docherty&Pym18,Docherty19} and in~\cite{Houetal18}, respectively with labelled tableaux calculi
and with sequent-style proof systems.
%
\subsection*{Our motivations}
Since the birth of separation logics, there has been  a lot of interest in the study of 
decidability and computational  complexity issues,
see e.g.~\cite{Calcagno&Yang&OHearn01,Brochenin&Demri&Lozes09,Bozga&Iosif&Perarnau10,Cooketal11,Demrietal17,Brotherston&Kanovich18,DemriLM18,Mansutti18,Mansutti20},
and comparatively less attention to the design of proof systems, and even less with the puristic approach
that consists in discarding any external feature such as nominals or labels in the calculi.
The well-known advantages of such an approach include an exhaustive understanding of the expressive power of
the logic and discarding the use of any external artifact referring to semantical objects.
For instance, a tableaux calculus with labels for quantifier-free separation logic is designed in~\cite{Galmiche&Mery10},
whereas Hilbert-style calculi for abstract separation logics with nominals are defined in~\cite{Brotherston&Villard14}.
Similarly, display calculi for bunched logics are provided in~\cite{Brotherston12} but such calculi extend
Gentzen-style proof systems by allowing new structural connectives, which provides an elegant means to simulate labels.
In this paper, we advocate a puristic approach and aim at designing a Hilbert-style proof system 
for quantifier-free 
separation logic \slSW (which includes the separating conjunction $\separate$ and implication $\magicwand$, as well as all 
Boolean connectives) and more generally for other separation logics,  while remaining within the very logical
language (see the second part of~\cite{Demri&Lozes&Mansutti20}).%
\footnote{We aim at defining {\em internal} calculi according to the terminology from
the  Workshop on External and Internal Calculi for Non-Classical Logics, FLOC'18, Oxford,
\url{http://weic2018.loria.fr}.}
Consequently, in this work, we only focus on axiomatising separation logics,
and  we have no claim for  practical applications in the field of program verification with separation logics.
Aiming at internal calculi is  a non-trivial
task as the general frameworks for abstract separation logics make use of labels, see e.g.~\cite{Docherty&Pym18,Houetal18}.
We cannot rely on label-free calculi for BI, see e.g.~\cite{Pym02,Galmiche&Larchey06}, as separation logics are usually understood
as Boolean BI interpreted on models of heap memory and therefore require calculi that cannot abstract as much as
it is the case for Boolean BI. Finally, there are many translations from separation logics into
logics or theories, 
see e.g.~\cite{Calcagno&Gardner&Hague05,Piskac&Wies&Zufferey13,Brochenin&Demri&Lozes12,Reynoldsetal16}.
However, completeness cannot in general be inherited by sublogics as the proof system should
only use the sublogic and therefore the axiomatisation of sublogics may lead
to different methods.
A more  detailed discussion
about the related work can be found in Section~\ref{section:related-work}.
%
\subsection*{Our contribution}
We propose a modular axiomatisation of quantifier-free separation logic, starting with a complete
axiomatisation of a Boolean algebra of core formulae, and incrementally adding support for the
spatial connectives: the separating conjunction and the separating implication (a.k.a.~the magic wand).
The same approach could be followed for other fragments of separation logic, as we did in the conference
version of this paper~\cite{Demri&Lozes&Mansutti20} (see also a similar approach in~\cite{Demri&Fervari&Mansutti19}). 
Thus, our approach can be considered with the broader perspective of a generic method for axiomatising 
separation logics. Let us be a bit more precise.

In Section~\ref{section:PSL}, we present the first Hilbert-style proof system
for \slSW
that uses axiom schemas and rules involving only formulae of this logic.
We mainly introduce our approach and present the notations
that are used throughout the paper. 
Each formula of \slSW is  equivalent to a Boolean combination of {\em core formulae}: 
simple formulae of the logic expressing elementary properties
about the models~\cite{Lozes04}.
Though core formulae (also called {\em test formulae}) have been handy in several occasions for
establishing complexity results for separation logics, 
see e.g.~\cite{Brochenin&Demri&Lozes09,DemriLM18,Mansutti18,Echenim&Iosif&Peltier19},
in the paper, these formulae are instrumental for the axiomatisation.
Indeed, the axiomatisation of \slSW is designed starting from 
an axiomatisation of Boolean combinations of core formulae (introduced in Section~\ref{subsection:axiomCoreFormulae}),
and adding axioms and rules that allow to syntactically transform every formula of \slSW into such Boolean combinations.
This transformation is introduced in Section~\ref{section:starelimination} and in Section~\ref{section:magicwandelimination}:
the former section shows how to eliminate the separating conjunction~$\separate$, 
whereas the latter one treat the separating implication~$\magicwand$. 
Schematically, for a valid formula~$\aformula$, we conclude $\vdash \aformula$  from
$\vdash \aformula'$ and $\vdash \aformula' \Leftrightarrow \aformula$, where
$\aformula'$ is a Boolean combination of core formulae.
Our methodology leads to a calculus that is divided in three parts: (1) the axiomatisation
of Boolean combinations of core formulae,
(2) axioms and inference rules to simulate a bottom-up elimination of the separating conjunction, 
and (3) axioms and inference rules to simulate a bottom-up elimination of the magic wand.
Such an approach that consists in first axiomatising a syntactic fragment of the whole logic (in our case, the core formulae), 
is best described in~\cite{Doumane17}
(see also~\cite{Walukiewicz00,vanBenthem2011ldii,WangC13,Luck18,Demri&Fervari&Mansutti19}).
Section~\ref{section:related-work} compares works from the literature with our contribution, either for
separation logics (abstract versions, fragments, etc.) or for knowledge logics for which the axiomatisation has been performed
by using a reduction to a strict syntactic fragment though expressively complete. 

This paper is the complete version of the first part of~\cite{Demri&Lozes&Mansutti20} dedicated
to quantifier-free separation logic \slSW. The complete version of the second part of~\cite{Demri&Lozes&Mansutti20} 
dedicated  to the new  separation logic \intervalSL is too long to be included in the present document.  
A technical appendix contains syntactic derivations omitted from the body of the paper.

\section{Preliminaries}\label{section:preliminaries}
\subsection{Quantifier-free separation logic}
We present the quantifier-free separation logic \slSW, that includes standard features such as
the separating conjunction $\separate$, the separating implication $\magicwand$
and closure under Boolean connectives.
Let $\PVAR = \set{\avariable, \avariablebis, \ldots}$ be a countably infinite set of \defstyle{program variables}.
The formulae $\aformula$ of \slSW and its atomic formulae
$\aatomicformula$ are built from  the grammars below where $\avariable, \avariablebis \in \PVAR$. 
\begin{nscenter}
$
\aatomicformula ::= \avariable = \avariablebis \ \mid \
                    \avariable \Ipto \avariablebis \ \mid \
                    \emptyconstant
  \qquad\qquad
\aformula ::= \aatomicformula \ \mid \  \neg \aformula \ \mid \ \aformula \wedge \aformula
                 \ \mid \ \aformula \separate \aformula  \ \mid \ \aformula \magicwand \aformula.
$
\end{nscenter}
The connectives $\Rightarrow$, $\Leftrightarrow$ and $\vee$ are defined as usually.
In the heaplet semantics, the formulae of \slSW are interpreted on  \defstyle{memory states} that are pairs
$\pair{\astore}{\aheap}$  where
${\astore: \PVAR \rightarrow \LOC}$ is
a variable valuation (the \defstyle{store}) from the set of program variables to a
countably infinite set of \defstyle{locations} $\LOC = \set{\alocation_0,\alocation_1, \alocation_2, \ldots}$, whereas
$\aheap: \LOC \to_{\fin} \LOC$ is a partial function with finite domain (the \defstyle{heap}).
We write $\domain{\aheap}$ to denote its domain and $\range{\aheap}$ to denote its range.
A \defstyle{memory cell} of $\aheap$ is understood as a pair of locations $\pair{\alocation}{\alocation'}$
such that $\alocation \in \domain{\aheap}$ and $\alocation'  = \aheap(\alocation)$.
As usual, the heaps $\aheap_1$ and $\aheap_2$ are said to be \defstyle{disjoint}, written $\aheap_1 \hdisjoint \aheap_2$,
if ${\domain{\aheap_1} \cap \domain{\aheap_2} = \emptyset}$;  when this holds, we write $\aheap_1 + \aheap_2$ to denote the heap
corresponding to the disjoint union of the graphs of $\aheap_1$ and $\aheap_2$, hence $\domain{\aheap_1 + \aheap_2} = \domain{\aheap_1} \uplus \domain{\aheap_2}$.
When the domains of $\aheap_1$ and $\aheap_2$  are not disjoint, the composition $\aheap_1 + \aheap_2$ is not defined.
Moreover, we write $\aheap' \sqsubseteq \aheap$ to denote that $\domain{\aheap'} \subseteq \domain{\aheap}$ and for all locations
$\alocation \in \domain{\aheap'}$, we have $\aheap'(\alocation) = \aheap(\alocation)$.
If $\aheap' \sqsubseteq \aheap$ then $\aheap'$ is said to be a \defstyle{subheap} of $\aheap$.
The satisfaction relation~$\models$ is defined as follows
(we omit standard clauses for the Boolean connectives $\neg$ and $\wedge$):
\begin{nscenter}
\begin{tabular}[t]{lll}
  $\pair{\astore}{\aheap} \models \avariable = \avariablebis$ & $\equivdef$ & $\astore(\avariable) = \astore(\avariablebis)$,  \\[2pt]
  $\pair{\astore}{\aheap} \models \emp$ & $\equivdef$ & $\domain{\aheap} = \emptyset$, \\[2pt]
  
  $\pair{\astore}{\aheap} \models \avariable {\Ipto} \avariablebis$ & $\equivdef$ & 
  $\astore(\avariable)\in\domain{\aheap}$ and $\aheap(\astore(\avariable)) = \astore(\avariablebis)$, \\[2pt]
  
  $\pair{\astore}{\aheap} \models \aformula_1 \separate \aformula_2$ & $\equivdef$ & there are 
  $\aheap_1,\aheap_2$ such that 
  $\aheap_1 \hdisjoint \aheap_2$, $(\aheap_1 + \aheap_2) = \aheap$,  \\
  & &  $\pair{\astore}{\aheap_1}  \models \aformula_1$ and  $\pair{\astore}{\aheap_2}  \models \aformula_2$,\\[2pt]
  
  $\pair{\astore}{\aheap} \models \aformula_1 \magicwand \aformula_2$ & $\equivdef$ & for all $\aheap_1$ such that
  $\aheap_1 \hdisjoint \aheap$ and $\pair{\astore}{\aheap_1} \models \aformula_1$, \\
  & &  we have $\pair{\astore}{\aheap + \aheap_1} \models \aformula_2$.
\end{tabular}
\end{nscenter}
\cut{
\begin{nscenter}
\begin{tabular}[t]{cc}
  $\pair{\astore}{\aheap} \models \avariable = \avariablebis$ $\equivdef$ $\astore(\avariable) = \astore(\avariablebis)$ &
  $\pair{\astore}{\aheap} \models \emp$ $\equivdef$ $\domain{\aheap} = \emptyset$ \\
  \multicolumn{2}{l}{$\pair{\astore}{\aheap} \models \avariable {\Ipto} \avariablebis$ $\equivdef$
  $\astore(\avariable)\in\domain{\aheap}$ and $\aheap(\astore(\avariable)) = \astore(\avariablebis)$} \\
  \multicolumn{2}{l}{$\pair{\astore}{\aheap} \models \aformula_1 \separate \aformula_2$ $\equivdef$ there are 
  $\exists\aheap_1,\aheap_2$ such that 
  $\aheap_1 \hdisjoint \aheap_2$, $(\aheap_1 + \aheap_2) = \aheap$, 
  $\pair{\astore}{\aheap_1}  \models \aformula_1$ and  $\pair{\astore}{\aheap_2}  \models \aformula_2$}\\
  \multicolumn{2}{l}{$\pair{\astore}{\aheap} \models \aformula_1 \magicwand \aformula_2$ $\equivdef$ $\forall\aheap_1.$
  ($\aheap_1 \hdisjoint \aheap$ and $\pair{\astore}{\aheap_1} \models \aformula_1$)
  implies $\pair{\astore}{\aheap + \aheap_1} \models \aformula_2$.}
\end{tabular}
\end{nscenter}
}
We denote with $\bot$ the contradiction $\avariable \neq \avariable$, and with $\top$ its negation $\neg\bot$.
The septraction operator $\septraction$ (kind of dual of $\magicwand$), defined
by $\aformula\septraction\aformulabis \egdef \neg(\aformula\magicwand\neg\aformulabis)$,
has the following semantics:
\begin{nscenter}
$\pair{\astore}{\aheap}\models \aformula\septraction\aformulabis$
$\equivdef$ there is a heap $\aheap'$ such that 
$\aheap \hdisjoint \aheap'$, $\pair{\astore}{\aheap'}\models\aformula$, and 
$\pair{\astore}{\aheap + \aheap'}\models\aformulabis$.
\end{nscenter}
We adopt the standard precedence between classical connectives, and extend it for the connectives of separation logic as follows:~$\set{\lnot} > \{\land,\lor,\separate\} > \set{\implies,\magicwand,\septraction} > \{\iff\}$. 
Notice that the separating conjunction~$\separate$ has a higher precedence than the separating implication~$\magicwand$, and it has the same precedence as the (classical) conjunction~$\land$. For instance, $\aformula \separate \aformulabis \implies \aformulater$ and $\lnot \aformula \magicwand \aformulabis \separate \aformulabis$ 
stand for $(\aformula \separate \aformulabis) \implies \aformulater$ and $(\lnot \aformula) \magicwand (\aformulabis \separate \aformulabis)$, respectively.

A formula $\aformula$ is \defstyle{valid} if $\pair{\astore}{\aheap}\models\aformula$ for all 
memory states $\pair{\astore}{\aheap}$ (and we write $\models \aformula$).
For a complete description of separation logic, see e.g.~\cite{Reynolds02}.
Given a set of formulae $\Gamma$, we write $\Gamma \models \aformula$ (semantical entailment)
whenever $\pair{\astore}{\aheap} \models \aformula$ holds
for every memory state $\pair{\astore}{\aheap}$ satisfying every formula in $\Gamma$.

It is worth noting that quantifier-free $\slSW$ axiomatised in the paper admits a \pspace-complete validity problem, 
see e.g.~\cite{Calcagno&Yang&OHearn01}, and should not be confused with propositional separation logic with the stack-heap models
 shown undecidable in~\cite[Corollary 5.1]{Brotherston&Kanovich14}
(see also~\cite[Section 4]{DemriDeters15bis}), in which there are propositional variables interpreted by sets of
memory states.

\subsection{Core formulae}\label{subsection:core-formulae}
We introduce the following well-known shortcuts, that play an important role in the sequel. 
Let $\avariable \in \PVAR$ and $\inbound \in \Nat$.%
\begin{center}
  \bgroup 
  \def\arraystretch{1.5}
  \begin{tabular}{lll}
    Shortcut: & Definition: & Semantics:\\[2pt]
    \hline
    $\alloc{\avariable}$ 
      $\egdef$&\!%
      $(\avariable \Ipto \avariable) \magicwand \false$
      &
      $\pair{\astore}{\aheap} \models \alloc{\avariable}$ iff $\astore(\avariable) \in \domain{\aheap}$ 
      \\[4pt]
    $\size \geq \inbound$
    $\egdef$&\hspace{-0.3cm}%
    $\begin{cases}
      \true &\text{if}~ \inbound = 0\\
      \lnot \emp &\text{if}~\inbound = 1\\
      \lnot \emp \separate \size \geq \inbound{-}1 &\text{otherwise}
    \end{cases}$ 
    &
    $\pair{\astore}{\aheap} \models \size \geq \inbound$ iff $\card{\domain{\aheap}} \geq \inbound$
  \end{tabular}
  \egroup
  \vspace{4pt}
\end{center}
We use $\size {=} \inbound$ as a shorthand for $\size {\geq }\inbound \land \lnot \size {\geq} \inbound {+} 1$.
We also write $\card{\aset}$ to denote the cardinality of the set $\aset$.

The \defstyle{core formulae} are expressions of the form
$\avariable = \avariablebis$,
$\alloc{\avariable}$,
$\avariable \Ipto \avariablebis$ and
$\size \geq \inbound$,
where $\avariable,\avariablebis \in \PVAR$ and $\inbound \in \Nat$.
As we can see, the core formulae are simple \slSW
formulae. 
It is well-known, see e.g.~\cite{Yang01,Lozes04bis}, that these formulae capture essential properties of the memory states.
In particular, every formula of \slSW is logically equivalent to a Boolean combination of
core formulae~\cite{Lozes04bis}.

As a simple but crucial insight, since the core formulae are formulae of \slSW, we can 
freely use them to help us defining the proof system for \slSW, and preventing us from going outside the original language. 
Having this in mind, the resulting proof system is Hilbert-style and completely internal 
(the formal definition of these types of systems is recalled below).

Given~$\asetvar \subseteq_\fin \PVAR$  and $\bound \in \Nat$,
we define~$\coreformulae{\asetvar}{\bound}$ as the set 
$$\{ \avariable = \avariablebis,\ \alloc{\avariable},\ \avariable \Ipto \avariablebis,\ \size \geq 
\inbound \mid \avariable,\avariablebis \in \asetvar,\ \inbound \in \interval{0}{\bound}\}.
$$ 
$\boolcomb{\coreformulae{\asetvar}{\bound}}$ is defined as the set of Boolean 
combinations of formulae from~$\coreformulae{\asetvar}{\bound}$, whereas
$\conjcomb{\coreformulae{\asetvar}{\bound}}$ is the set of conjunctions of literals built upon~$\coreformulae{\asetvar}{\bound}$.
As usual, a \defstyle{literal} is understood as  a core formula or its negation.
Let~${\aformula = \aliteral_1 \land \dots \land \aliteral_n} \in \conjcomb{\coreformulae{\asetvar}{\bound}}$ be a conjunction of literals $\aliteral_1,\dots,\aliteral_n$.
We write $\literals{\aformula}$ for $\{\aliteral_1,\dots,\aliteral_n\}$.
In forthcoming developments, we are interested in the maximum $\inbound$ (if any) of formulae of the form $\size \geq \inbound$ occurring positively in a conjunction of literals, if any.
For this reason, we write $\maxsize{\aformula}$ for ${\max(\{\inbound \in \Nat \mid \size \geq \inbound \in \literals{\aformula}\} \cup \{0\})}$.
For instance,
given $\aformula ={\alloc{\avariable} \land \size \geq 2 \land \lnot \size \geq 4}$,
we have
$\literals{\aformula} = \{\alloc{\avariable}, \size \geq 2, \lnot \size \geq 4\}$, 
and $\maxsize{\aformula} = 2$.
Given two conjunctions of literals $\aformula \in \conjcomb{\coreformulae{\asetvar}{\bound_1}}$ and $\aformulabis \in \conjcomb{\coreformulae{\asetvar}{\bound_2}}$,
$\aformulabis \inside \aformula$ stands for $\literals{\aformulabis} \subseteq \literals{\aformula}$.
Finally, we introduce a few more shortcuts and we write 

\begin{itemize}[before=\vspace{4pt},after=\vspace{1pt}]
  \setlength{\itemsep}{5pt}
  \begin{minipage}[t]{0.5\linewidth}
  \item $\aformulater \inside \orliterals{\aformula}{\aformulabis}$ \,for \,%
  ``$\aformulater \inside \aformula$ or $\aformulater \inside \aformulabis$'',
  \item $\orliterals{\aformula}{\aformulabis} \inside \aformulater$ \,for \,%
  ``$\aformula \inside \aformulater$ or $\aformulabis \inside \aformulater$''.
  \end{minipage}%
  \begin{minipage}[t]{0.5\linewidth}
  \item $\aformulater \inside \andliterals{\aformula}{\aformulabis}$ \,for \,%
  ``$\aformulater \inside \aformula$ and $\aformulater \inside \aformulabis$'', 
  \end{minipage}
\end{itemize}
Given a finite set of formulae $\Gamma = \{\aformula_1,\dots,\aformula_n\}$, 
we write $\bigwedge \Gamma$ as a shorthand for $\aformula_1 \land \dots \land \aformula_n$.
Similarly, 
$\bigseparate \Gamma$ stands for $\aformula_1 \separate \dots \separate \aformula_n$.
It is important to notice that, similarly to the classical conjunction, the separating conjunction $\separate$ is associative and commutative 
(see the axioms~\ref{starAx:FS:Assoc} and~\ref{starAx:FS:Commute} in Figure~\ref{figure-full-proof-system}), and therefore 
the semantics of $\bigseparate \Gamma$ is uniquely defined, 
regardless of the choice of ordering for $\aformula_1 , \dots, \aformula_n$.

\subsection{Hilbert-style proof systems}
A \defstyle{Hilbert-style proof system} $\aproofsystem$ is defined as a set of 
tuples
 $((\aformulaschema_1,\dots,\aformulaschema_n),\aformulaschemabis)$ with $n\geq 0$,
where $\aformulaschema_1,\dots,\aformulaschema_n,\aformulaschemabis$ are \defstyle{formula schemata} (a.k.a 
\defstyle{axiom schemata}).
When $n \geq 1$, $((\aformulaschema_1,\dots,\aformulaschema_n),\aformulaschemabis)$ is called an \defstyle{inference rule},
otherwise it is an \defstyle{axiom}. As usual, formula schemata generalise the notion of formulae by allowing metavariables
for formulae (typically $\aformula, \aformulabis, \aformulater$), for program variables (typically $\avariable, \avariablebis, \avariableter$)
or for any type of syntactic objects in formulae, depending on the context.
 The set of formulae \defstyle{derivable} from
$\aproofsystem$ is the least set $S$ such that for all
 $((\aformulaschema_1,\dots,\aformulaschema_n),\aformulaschemabis)\in\aproofsystem$ and for all
substitutions $\asubstitution$, if
$\aformulaschema_1\asubstitution,\dots,\aformulaschema_n\asubstitution\in S$
then $\aformulaschemabis\asubstitution\in S$.
We write $\pfentails{\aproofsystem}{\aformula}$ if $\aformula$ is
derivable from~$\aproofsystem$.
A proof system $\aproofsystem$
is \defstyle{sound} if all derivable formulae are valid. $\aproofsystem$
is \defstyle{complete} if
all valid formulae are derivable.
We say that $\aproofsystem$ is \defstyle{adequate} whenever it is both sound and complete.
Lastly, $\aproofsystem$ is \defstyle{strongly complete}
whenever for all sets of formulae $\Gamma$ and formulae $\aformula$,
we have $\Gamma \models \aformula$ (semantical entailment) if and only if
$\pfentails{\aproofsystem\cup \Gamma}{\aformula}$.

Interestingly enough, there is no strongly complete proof system
for \slSW, as strong completeness implies compactness
and separation logic is not compact.
Indeed, the set ${\set{\size\geq \inbound \mid \inbound\in\Nat}}$ is unsatisfiable, as heaps have finite domains, but all finite subsets of it
are satisfiable.
Even for the weaker notion of completeness, deriving an Hilbert-style axiomatisation for \slSW remains challenging.
Indeed, the satisfiability problem  for \slSW reduces to its validity problem,
making \slSW an unusual logic from a proof-theoretical point of view.
Let us develop a bit further this point.

Let $\aformula$ be a formula built over program variables in $\asetvar \subseteq_{\fin} \PVAR$, and let $\approx$ be an equivalence relation on $\asetvar$.
  The formula
  $\aformulabis_{\approx} \egdef (\emp \land \bigwedge_{\avariable \approx \avariablebis} \avariable = \avariablebis \land  \bigwedge_{\substack{\avariable \not\approx \avariablebis}} \avariable \neq \avariablebis )\implies (\aformula \septraction \true)$
  can be shown to be valid iff
  for every store $\astore$ agreeing on $\approx$, there is a heap $\aheap$ such that $\pair{\astore}{\aheap} \models \aformula$.
  It is known
  that for  all stores $\astore,\astore'$ agreeing on $\approx$, and every heap $\aheap$, 
  the memory states $\pair{\astore}{\aheap}$ and $\pair{\astore'}{\aheap}$
  satisfy the same set of formulae having variables from $\asetvar$.
  Since the antecedent of $\aformulabis_{\approx}$ is satisfiable, we conclude that
   $\aformulabis_{\approx}$ is valid iff there are a store $\astore$ agreeing on $\approx$ and a heap $\aheap$ such that 
   $\pair{\astore}{\aheap} \models \aformula$.
   To check whether $\aformula$ is satisfiable, it is sufficient to find an equivalence relation $\approx$ on $\asetvar$ such that
   $\aformulabis_{\approx}$ is valid. As the number of equivalence relations on $\asetvar$ is finite, we obtain a Turing reduction 
   from satisfiability to validity.
Consequently, it is not possible to define sound and complete axiom systems for any extension of 
\slSW admitting an undecidable validity problem
(as long as there is a reduction from satisfiability to validity, as above).
A good example is
the logic $\seplogic{\separate,\magicwand,\ls}$~\cite{Demri&Lozes&Mansutti18bis} (extension of \slSW with the well-known list-segment predicate $\ls$); see also the first-order
separation logic in~\cite{Brochenin&Demri&Lozes12}.
Indeed, to obtain a sound and complete axiom system, the validity problem has to be recursively enumerable (r.e.).
However,  this would imply that the satisfiability problem is also r.e.. As
a formula
$\aformula$ is not valid if and only if
$\lnot \aformula$ is satisfiable, we then conclude that the set of valid formulae is recursive, 
hence decidable, a contradiction.

\cut{
It is worth also noting that quantifier-free $\slSW$ axiomatised in the paper admits a \pspace-complete validity problem, see e.g.~\cite{Calcagno&Yang&OHearn01},
and should not be confused with propositional separation logic with the stack-heap models shown undecidable in~\cite[Corollary 5.1]{Brotherston&Kanovich14}
(see also~\cite[Section 4]{DemriDeters15bis}),
in which there are propositional variables interpreted by sets of
memory states.
}
\section{Hilbert-style proof system for \texorpdfstring{\slSW}{propositional SL}}\label{section:PSL}
\begin{figure}

\begin{footnotesize}

\fbox{
\begin{minipage}{0.97\linewidth}
\begin{enumerate}[align=left,leftmargin=*]
  \setlength\itemsep{4pt}
\begin{minipage}{0.51\linewidth}
\item[\axlab{A^\corepedix}{coreAx:FS:EqRef}] $\avariable = \avariable$
\item[\axlab{A^\corepedix}{coreAx:FS:EqSub}] $\aformula \land \avariable = \avariablebis \implies \aformula\completesubstitute{\avariable}{\avariablebis}$
\end{minipage}\hfill\
\begin{minipage}{0.50\linewidth}
\item[\axlab{A^\corepedix}{coreAx:FS:PointAlloc}] $\avariable \Ipto \avariablebis \implies \alloc{\avariable}$
\item[\axlab{A^\corepedix}{coreAx:FS:PointInj}] $\avariable \Ipto \avariablebis \land \avariable \Ipto \avariableter \implies \avariablebis = \avariableter$
\end{minipage}
\end{enumerate}

\vspace{0.1cm}
\hrule
\vspace{0.2cm}

\noindent
\begin{enumerate}[align=left,leftmargin=*]
  \setlength\itemsep{4pt}
\addtocounter{enumi}{6}
\begin{minipage}[t]{0.51\linewidth}
\item[\axlab{A^\separate}{starAx:FS:Commute}] $(\aformula \separate \aformulabis) \iff (\aformulabis \separate \aformula)$
\item[\axlab{A^\separate}{starAx:FS:Assoc}] $(\aformula \separate \aformulabis) \separate \aformulater \iff \aformula \separate (\aformulabis \separate \aformulater)$
\addtocounter{enumi}{2}
\item[\axlab{A^\separate}{starAx:FS:Emp}] $\aformula \iff \aformula \separate \emp$
\addtocounter{enumi}{1}
\item[\axlab{A^\separate}{starAx:FS:DoubleAlloc}] $\alloc{\avariable} \separate \alloc{\avariable} \iff\bot$
\item[\axlab{A^\separate}{starAx:FS:MonoCore}] $\aelement {\separate} \true \implies \aelement \!\assuming{\aelement \in \{
\lnot\emp, \avariable = \avariablebis, \avariable \neq \avariablebis, \avariable \Ipto \avariablebis\}}$
\item[\axlab{A^\separate}{starAx:FS:AllocNeg}] $\lnot\alloc{\avariable} \separate \lnot\alloc{\avariable} \implies \lnot\alloc{\avariable}$

\end{minipage}\hfill\
\begin{minipage}[t]{0.50\linewidth}
\item[\axlab{A^\separate}{starAx:FS:PointsNeg}] \!$(\alloc{\avariable} \land \lnot \avariable \Ipto \avariablebis) \separate \true \implies \lnot \avariable \Ipto \avariablebis$
\item[\axlab{A^\separate}{starAx:FS:AllocSizeOne}] \!$
\alloc{\avariable} \implies
(\alloc{\avariable} \land \size = 1) \separate \true$
\item[\axlab{A^\separate}{starAx:FS:SizeOne}] \!$\lnot \emp \implies \size = 1 \separate \true$
\item[\axlab{A^\separate}{starAx:FS:SizeNeg}] \!$\lnot \size \,{\geq}\,\inbound_1 \,{\separate} \lnot \size \,{\geq}\, \inbound_2 \implies\! \lnot \size \,{\geq}\, \inbound_1{+}\inbound_2{\dotminus}1$
\item[\axlab{A^\separate}{starAx:FS:SizeTwo}] \!$\alloc{\avariable}\wedge\alloc{\avariablebis}\wedge \avariable\neq\avariablebis\implies \size\geq2$
\end{minipage}
\end{enumerate}

\vspace{0.1cm}
\hrule
\vspace{0.2cm}

\noindent
\begin{enumerate}[align=left,leftmargin=*]
  \setlength\itemsep{4pt}
\addtocounter{enumi}{20}
\begin{minipage}{0.56\linewidth}
\item[\axlab{A^{\magicwand}}{wandAx:FS:Size}]$(\size = 1 \land \bigwedge_{\avariable \in \asetvar}\lnot \alloc{\avariable}) \septraction \true\!\assuming{\asetvar \subseteq_{\fin} \PVAR}$
\item[\axlab{A^{\magicwand}}{wandAx:FS:PointsTo}] $\lnot \alloc{\avariable} \implies (\avariable \Ipto \avariablebis \land \size = 1 \septraction \true)$
\end{minipage}%
\item[\axlab{A^{\magicwand}}{wandAx:FS:Alloc}] $\lnot \alloc{\avariable} \implies ((\alloc{\avariable} \land \size = 1 \land \bigwedge_{\avariablebis \in \asetvar}\lnot \avariable \Ipto \avariablebis ) \septraction \true) \assuming{\asetvar \subseteq_{\fin} \PVAR}$
\end{enumerate}

\vspace{0.1cm}
\hrule
\vspace{0.2cm}

\noindent
\begin{center}
  \hfill
  \rulelab{\textbf{$\separate$-Intro}}{rule:FS:starinference}
  $\inference{\aformula \implies \aformulater}{\aformula \separate \aformulabis \implies \aformulater \separate \aformulabis }{}$
  \hfill
  \rulelab{\textbf{$\separate$-Adj}}{rule:FS:staradj}
  $\inference{\aformula \separate \aformulabis \implies \aformulater}{\aformula \implies (\aformulabis \magicwand \aformulater)}{}$
  \hfill 
  \rulelab{\textbf{$\magicwand$-Adj}}{rule:FS:magicwandadj}
  $\inference{\aformula \implies (\aformulabis \magicwand \aformulater)}{\aformula \separate \aformulabis \implies \aformulater}{}$
  \hfill\,

  \vspace{0.3cm}
  (axioms and modus ponens from propositional calculus are omitted)
\end{center}
\end{minipage}
}
\end{footnotesize}

\caption{The proof system $\magicwandsys$.}
\label{figure-full-proof-system}
\end{figure}

In Figure~\ref{figure-full-proof-system}, we present the proof system~$\magicwandsys$
that shall be shown to be sound and complete for quantifier-free separation logic  \slSW.
$\magicwandsys$  and all the subsequent fragments of $\magicwandsys$ contain
the axiom schemata and modus ponens for the propositional calculus (we omit these rules
in the presentation). In the axioms~\ref{starAx:FS:MonoCore},~\ref{wandAx:FS:Size} and~\ref{wandAx:FS:Alloc},
the notation $\aformula\!\!\assuming{\mathcal{B}}$
refers to the axiom schema $\aformula$ assuming that the Boolean condition $\mathcal{B}$ holds.
We highlight the fact that, in these three axioms, $\mathcal{B}$ is a simple syntactical condition.
In the axiom~\ref{starAx:FS:SizeNeg}, $a \dotminus b$, where $a,b \in \Nat$, stands for $\max(0,a-b)$.

Though the full proof system $\magicwandsys$ is presented quite early in the paper,
its final design remains the outcome of a refined analysis on principles
behind \slSW tautologies. 
Fortunately, we do not start from scratch as the calculus must contain
the axioms and rules from the Hilbert-style proof system for Boolean BI~\cite{Galmiche&Larchey06}.
At first glance the system $\magicwandsys$ may seem quite arbitrary,
but the role of the different axioms shall become clearer during the paper.
In designing the system, we tried to define axioms that are as simple as possible,
which helps highlighting the most fundamental properties of \slSW.
Note that we have not formally proved that our proof system $\magicwandsys$
is minimal (though we have tried our best to have a small amount of small axioms). 
Such an investigation would be out of the scope of the paper, mainly for lack of space.
The standard way to proceed would be to design models different from memory states
and to establish that all axioms but one are valid (which would prove that this axiom is
needed when all the other axioms are present). 


We insist: the core formulae in~$\magicwandsys$ should be understood as mere abbreviations, 
which makes all the axioms in Figure~\ref{figure-full-proof-system} belong to the original language of \slSW.
In order to show the completeness of~$\magicwandsys$,
we first establish the completeness for subsystems of~$\magicwandsys$,
with respect to syntactical fragments of \slSW. In particular, we consider 
\begin{itemize}
  \item $\coresys$: an adequate proof system for the propositional logic of core formulae (see Figure~\ref{figure-proof-system-core-formulae}),
  \item $\coresys(\separate)$: an extension of $\coresys$ that is adequate for the logic $\slSA$, 
  i.e.~the logic obtained from $\slSW$ by removing the separating implication $\magicwand$ at the price of adding the formula $\alloc{\avariable}$ (see Figure~\ref{figure-proof-system-star}).
  \item The full $\magicwandsys$, which can be seen as an extension of $\coresys(\separate)$ that allows to reason about 
  the separating implication (see Figure~\ref{figure-proof-system-magicwand}).
\end{itemize}

For the completeness of $\coresys$ and $\coresys(\separate)$, we add intermediate axioms that reveal to be useless when the full proof system~$\magicwandsys$ is considered, as they become derivable. By convention,
the axioms whose name is of the form $A^?_i$ are axioms that remain in $\magicwandsys$
(see Figure~\ref{figure-full-proof-system})
whereas those named $I^?_i$ are intermediate axioms that are instrumental for the proof of completeness of a subsystem among $\coresys$ and $\coresys(\separate)$
(and therefore none of them occur in  Figure~\ref{figure-full-proof-system}).
The numbering of the axioms in Figure~\ref{figure-full-proof-system}
is not consecutive, as intermediate axioms shall be placed within the holes.
It is worth noting that the axiom~\ref{starAx:FS:DoubleAlloc} had an intermediate status in~\cite{Demri&Lozes&Mansutti20}
but we realised that actually this axiom does need to be considered as a first-class axiom in the proof system
$\magicwandsys$.

The choice of introducing $\coresys$ and $\coresys(\separate)$ naturally follows from the 
main steps required for the completeness of $\magicwandsys$.
In particular, the main ``task'' of $\coresys(\separate)$ is to produce a bottom-up elimination of the separating conjunction $\separate$, at the price of introducing Boolean combinations of core formulae, which can be proved valid thanks to $\coresys$. Similarly, the axioms and rules added to $\coresys(\separate)$ to define $\magicwandsys$ are dedicated to perform  a bottom-up elimination of the separating implication.
A merit of this methodology is that only the completeness of the calculus $\coresys$
is proved using the standard countermodel method.
The additional steps required to prove the completeness of $\coresys(\separate)$ and $\magicwandsys$ are (almost) completely syntactical. For instance, 
to show the completeness of $\coresys(\separate)$, we consider arbitrary
Boolean combinations of core formulae $\aformula$ and $\aformulabis$, 
and exhibiting a Boolean combination of core formulae $\aformulater$ such that $\aformula \separate \aformulabis \iff \aformulater$ is valid. We show that this validity can be \emph{syntactically} proved within~$\coresys(\separate)$, and then rely on the fact that $\coresys$ is complete for Boolean combination of core formulae to deduce that $\coresys(\separate)$ is complete for $\slSA$.

\cut{
The  \defstyle{core formulae}  are expressions of the form
$\avariable = \avariablebis$,
$\alloc{\avariable}$,
$\avariable \Ipto \avariablebis$ and
$\size \geq \inbound$,
where $\avariable,\avariablebis \in \PVAR$ and $\inbound \in \Nat$.
As previously shown, these formulae are from \slSW and are used in the axiom system as abbreviations.
Given
$\asetvar\subseteq_\fin \PVAR$  and $\bound \in \Nat$,
we define $\coreformulae{\asetvar}{\bound}$ as
the set $\{ \avariable = \avariablebis,\ \alloc{\avariable},\ \avariable \Ipto \avariablebis,\ \size \geq \inbound \mid \avariable,\avariablebis \in \asetvar,\ \inbound \in \interval{0}{\bound}\}$.
$\boolcomb{\coreformulae{\asetvar}{\bound}}$ is the set of Boolean combinations of formulae from $\coreformulae{\asetvar}{\bound}$, whereas
$\conjcomb{\coreformulae{\asetvar}{\bound}}$ is the set of conjunctions of literals built upon $\coreformulae{\asetvar}{\bound}$
(a literal being a core formula or its negation).
Given $\aformula = \aliteral_1 \land \dots \land \aliteral_n \in \conjcomb{\coreformulae{\asetvar}{\bound}}$, every $\aliteral_i$ being a literal, $\literals{\aformula} \egdef \{\aliteral_1,\dots,\aliteral_n\}$.
$\aformulabis \inside \aformula$ stands for $\literals{\aformulabis} \subseteq \literals{\aformula}$.
We write $\aformulater \inside \orliterals{\aformula}{\aformulabis}$,
$\orliterals{\aformula}{\aformulabis} \inside \aformulater$ and
$\aformulater \inside \andliterals{\aformula}{\aformulabis}$
for
``$\aformulater \inside \aformula$ or $\aformulater \inside \aformulabis$'',
``$\aformula \inside \aformulater$ or $\aformulabis \inside \aformulater$'', and
``$\aformulater \inside \aformula$ and $\aformulater \inside \aformulabis$'', respectively.\\[2pt]
}

Along the paper, we shall have the opportunity to explain the intuition between the axioms and rules.
Below, we provide a few hints. The axioms~\ref{coreAx:FS:EqRef}--~\ref{coreAx:FS:PointInj} deal with the core formulae
and are quite immediate to grasp. More interestingly, whereas the axioms~\ref{starAx:FS:Commute}--\ref{starAx:FS:Emp}
are quite general about separating conjunction and are inherited from Boolean BI, the axioms~\ref{starAx:FS:MonoCore}--\ref{starAx:FS:SizeTwo}
state how separating conjunction behaves with the core formulae. As for Boolean combinations of core formulae
involved in the axioms~\ref{coreAx:FS:EqRef}--~\ref{coreAx:FS:PointInj}, these axioms~\ref{starAx:FS:MonoCore}--\ref{starAx:FS:SizeTwo}
are also not difficult
to understand.
Besides, the inference rules~\ref{rule:FS:staradj} and~\ref{rule:FS:magicwandadj} simply
reflect that separating conjunction and separating implication are adjoint operators,
and are taken from Boolean BI, see e.g.~\cite{Galmiche&Larchey06}.
The axioms~\ref{wandAx:FS:Size}--\ref{wandAx:FS:Alloc} dedicated to the interaction between the separating
implication and core formulae are expressed with the help of the septraction operator $\septraction$ to ease the understanding
but as well-known, septraction is defined with the help of the separating implication and Boolean negation.
For instance, the axiom~\ref{wandAx:FS:PointsTo} states that it is always possible to add some one-memory-cell
heap $\aheap'$ to some heap $\aheap$ while none of the variables from a finite set~$\asetvar$ is allocated
in  $\aheap'$. This natural property in our framework would not hold in general if $\LOC$ were
not infinite. Obviously, the septraction $\septraction$ is also understood as an abbreviation.

As a sanity check, we show that the proof system $\magicwandsys$ is sound with respect
to \slSW. The proof does not pose any specific difficulty (as usual with most soundness proofs)
but this is the opportunity for the reader to further get familiar with the axioms and rules
from $\magicwandsys$.

\begin{lem}\label{lemma:magicwandPSLvalid}
$\magicwandsys$ is sound.
\end{lem}

\noindent The validity of the axioms~\ref{coreAx:FS:EqRef},~\ref{coreAx:FS:EqSub},~\ref{coreAx:FS:PointAlloc} and~\ref{coreAx:FS:PointInj} is 
straightforward. Moreover, the validity of the axioms~\ref{starAx:FS:Commute},~\ref{starAx:FS:Assoc} and~\ref{starAx:FS:Emp}
and the three inference rules~(\ref{rule:FS:starinference},~\ref{rule:FS:staradj} and~\mbox{\ref{rule:FS:magicwandadj}}) is inherited from Boolean BI
(see~\cite{Brotherston&Villard14} and~\cite[Section 2]{Galmiche&Larchey06}).
Below, we show the validity of the remaining axioms, thus proving~Lemma~\ref{lemma:magicwandPSLvalid}.

\begin{proof}[Validity of the axiom~{\rm\ref{starAx:FS:DoubleAlloc}}]
Let us show that $(\alloc{\avariable} \separate \alloc{\avariable})$ is not satisfiable.
{\em Ad~absurdum}, suppose there is a memory state $\pair{\astore}{\aheap}$ such that
$\pair{\astore}{\aheap} \models (\alloc{\avariable} \separate \alloc{\avariable})$. 
By definition of $\models$,  there are 
  $\aheap_1,\aheap_2$ such that 
  $\aheap_1 \bot \aheap_2$, $(\aheap_1 + \aheap_2) = \aheap$,  
  $\pair{\astore}{\aheap_1}  \models \alloc{\avariable}$ and  $\pair{\astore}{\aheap_2}  \models \alloc{\avariable}$.
Thus, $\astore(\avariable) \in \domain{\aheap_1}$ and $\astore(\avariable) \in \domain{\aheap_2}$,
which leads to a contradiction with $\aheap_1 \bot \aheap_2$.
\end{proof}
\begin{proof}[Validity of the axiom~{\rm\ref{starAx:FS:MonoCore}}]
The proof of the validity of every instantiation of \ref{starAx:FS:MonoCore} is similar (and quite easy), 
therefore we show just the case with $\avariable \Ipto \avariablebis \separate \true \implies \avariable \Ipto \avariablebis$.
Suppose $\pair{\astore}{\aheap} \models \avariable \Ipto \avariablebis \separate \true$. Then, there is a subheap $\aheap_1 \sqsubseteq \aheap$ such that $\pair{\astore}{\aheap_1} \models \avariable \Ipto \avariablebis$.
Hence, $\aheap_1(\astore(\avariable)) = \astore(\avariablebis)$.
As $\aheap_1 \sqsubseteq \aheap$, we obtain  $\aheap(\astore(\avariable)) = \astore(\avariablebis)$, 
which implies $\pair{\astore}{\aheap} \models \avariable \Ipto \avariablebis$.
\end{proof}
\begin{proof}[Validity of the axiom~{\rm \ref{starAx:FS:AllocNeg}}]
Suppose $\pair{\astore}{\aheap} \models \lnot \alloc{\avariable} \separate \lnot \alloc{\avariable}$. Then, there are
two disjoint heaps $\aheap_1,\aheap_2$ such that $\aheap = \aheap_1 + \aheap_2$, $\pair{\astore}{\aheap_1} \models \lnot \alloc{\avariable}$ and $\pair{\astore}{\aheap_2} \models \lnot \alloc{\avariable}$.
Then $\astore(\avariable) \not \in \domain{\aheap_1}$ and $\astore(\avariable) \not \in \domain{\aheap_2}$.
Since $\aheap = \aheap_1 + \aheap_2$, $\domain{\aheap} = \domain{\aheap_1}\cup\domain{\aheap_2}$ and therefore
$\astore(\avariable) \not \in \domain{\aheap}$.
We conclude that $\pair{\astore}{\aheap} \models \lnot \alloc{\avariable}$.
\end{proof}
\begin{proof}[Validity of the axiom~{\rm\ref{starAx:FS:PointsNeg}}]
Suppose $\pair{\astore}{\aheap} \models (\alloc{\avariable} \land \lnot \avariable \Ipto \avariablebis) \separate \true$.
Then there is a subheap $\aheap_1 \sqsubseteq \aheap$ such that $\pair{\astore}{\aheap_1} \models \alloc{\avariable} \land \lnot \avariable \Ipto \avariablebis$.
Hence, $\astore(\avariable) \in \domain{\aheap_1}$ and $\aheap_1(\astore(\avariable)) \neq \astore(\avariablebis)$.
As $\aheap_1 \sqsubseteq \aheap$, we obtain  $\astore(\avariable) \in \domain{\aheap}$ and $\aheap(\astore(\avariable)) \neq \astore(\avariablebis)$ which by definition implies $\pair{\astore}{\aheap} \models \lnot \avariable \Ipto \avariablebis$.
\end{proof}
\begin{proof}[Validity of the axiom~{\rm\ref{starAx:FS:AllocSizeOne}}]
Suppose $\pair{\astore}{\aheap} \models \alloc{\avariable}$.
Let $\aheap_1 \egdef \{\astore(\avariable)\pto\aheap(\astore(\avariable))\}$
As~$\astore(\avariable) \in \domain{\aheap}$, $\aheap_1 \sqsubseteq \aheap$ and
$\pair{\astore}{\aheap_1} \models \alloc{\avariable} \land \size = 1$.
We define $\aheap_2$ as the unique heap such that $\aheap_2 + \aheap_1 = \aheap$.
As $\pair{\astore}{\aheap_2} \models \true$,
we have $\pair{\astore}{\aheap} \models (\alloc{\avariable} \land \size = 1) \separate \true$.
\end{proof}
\noindent The proof of~axiom~\ref{starAx:FS:SizeOne} is similar to the one of~\ref{starAx:FS:AllocSizeOne}, and hence omitted herein.

\begin{proof}[Validity of the axiom~{\rm \ref{starAx:FS:SizeNeg}}]
Suppose $\pair{\astore}{\aheap} \models \lnot \size \geq \inbound_1 \separate \lnot \size \geq \inbound_2$, where~${\inbound_1,\inbound_2\geq 0}$. 
Since $\lnot \size \geq 0$ is not satisfiable, this implies that necessarily $\inbound_1,\inbound_2\geq 1$.
Hence, the axiom~\ref{starAx:FS:SizeNeg} is trivially valid  when $\inbound_1 = 0$ or $\inbound_2 = 0$.
In the sequel, $\inbound_1,\inbound_2\geq 1$.
Then, there are heaps $\aheap_1,\aheap_2$ such that $\aheap_1 \disjoint \aheap_2$, $\aheap_1 + \aheap_2 = \aheap$, 
$\pair{\astore}{\aheap_1} \models \lnot \size \geq \inbound_1$ and ${\pair{\astore}{\aheap_2} \models \lnot \size \geq \inbound_2}$.
By definition, $\card{\domain{\aheap_1}} \leq \inbound_1-1$ and $\card{\domain{\aheap_2}} \leq \inbound_2-1$.
Since $\domain{\aheap} = \domain{\aheap_1}\cup\domain{\aheap_2}$, we obtain $\card{\domain{\aheap}} \leq \inbound_1+\inbound_2-2$, which implies
$\pair{\astore}{\aheap} \models \lnot \size \geq \inbound_1+\inbound_2 \dotminus 1$.
\end{proof}
\begin{proof}[Validity of the axiom~{\rm\ref{starAx:FS:SizeTwo}}]
Suppose $\pair{\astore}{\aheap} \models \alloc{\avariable} \land \alloc{\avariablebis} \land \avariable \neq \avariablebis$. 
By definition,
$\astore(\avariable) \neq \astore(\avariablebis)$, and $\astore(\avariable), \astore(\avariablebis) \in \domain{\aheap}$. Hence, $\card{\domain{\aheap}} \geq 2$, and $\pair{\astore}{\aheap} \models \size \geq 2$.
\end{proof}
\begin{proof}[Validity of the axiom~{\rm\ref{wandAx:FS:Size}}]
Let $\asetvar\subseteq_\fin\PVAR$ and  $\pair{\astore}{\aheap}$ be a memory state.
Let $\aheap_1$ be a heap of size one such that $\aheap_1(\alocation) = \alocation$ for some 
$\alocation \not \in \domain{\aheap}\cup\astore(\asetvar)$. We write $\astore(\asetvar)$ to denote
the set $\set{\astore(\avariable) \mid \avariable \in \asetvar}$. 
Trivially $\pair{\astore}{\aheap_1}\models \size = 1 \land \bigwedge_{\avariable \in \asetvar} \lnot \alloc{\avariable}$.
Moreover $\aheap_1 \hdisjoint \aheap$ holds, hence $\aheap_1+\aheap_2$ is defined and $\pair{\astore}{\aheap+\aheap_1} \models \true$.
Then, $\pair{\astore}{\aheap} \models (\size = 1 \land \bigwedge_{\avariable \in \asetvar} \lnot \alloc{\avariable}) \septraction \true$.
\end{proof}
\begin{proof}[Validity of the axiom~{\rm\ref{wandAx:FS:PointsTo}}]
Suppose $\pair{\astore}{\aheap} \models \lnot \alloc{\avariable}$.
Let $\aheap_1$ be the heap of size one such that $\aheap_1(\astore(\avariable)) = \astore(\avariablebis)$.
Trivially, $\pair{\astore}{\aheap_1} \models \avariable \Ipto \avariablebis \land \size = 1$. 
Moreover, as $\astore(\avariable) \not\in\domain{\aheap}$, $\aheap_1 \hdisjoint \aheap$ holds. Therefore, $\aheap_1+\aheap$ 
is defined, and $\pair{\astore}{\aheap+\aheap_1} \models \true$.
Then, $\pair{\astore}{\aheap} \models (\avariable \Ipto \avariablebis \land \size=1) \septraction \true$.
\end{proof}
\begin{proof}[Validity of the axiom~{\rm\ref{wandAx:FS:Alloc}}]
Suppose $\pair{\astore}{\aheap} \models \lnot \alloc{\avariable}$.
Let $\asetvar \subseteq_\fin \PVAR$ and $\aheap_1  \egdef \{\astore(\avariable)\pto \alocation\}$, where $\alocation \not \in \astore(\asetvar)$.
Hence, ${\pair{\astore}{\aheap_1} \models \alloc{\avariable} \land \size = 1} \land \bigwedge_{\avariablebis \in \asetvar} \lnot \avariable \Ipto \avariablebis$.
Since $\astore(\avariable) \not\in\domain{\aheap}$,
$\aheap_1 \disjoint \aheap$.
Therefore, the heap $\aheap+\aheap_1$ is defined and $\pair{\astore}{\aheap+\aheap_1} \models \true$.
Then, $\pair{\astore}{\aheap} \models (\alloc{\avariable} \land \size = 1 \land \bigwedge_{\avariablebis \in \asetvar} \lnot \avariable \Ipto \avariablebis) \septraction \true$.
\end{proof}

\begin{figure}
  \begin{syntproof}
  1 & \emp \implies \lnot \size \geq 1
  & \mbox{\ref{axiom:dubneg} and def. of $\size \geq 1$}
  \\
  2 & \alloc{\avariable} \land \size = 1 \implies \lnot \size \geq 2
  & \mbox{\ref{axiom:andelim}}
  \\
  3 & \emp \separate (\alloc{\avariable} \land \size = 1 )\implies \lnot \size \geq 1 \separate \lnot \size \geq 2
  & \mbox{\ref{rule:starintroLR}, 1, 2}
  \\
  4 & \lnot \size \geq 1 \separate \lnot \size \geq 2 \implies \lnot \size \geq 2
  & \mbox{\ref{starAx:FS:SizeNeg}}
  \\
  5 & \emp \separate (\alloc{\avariable} \land \size = 1) \implies \lnot \size \geq 2
  & \mbox{\ref{rule:imptr}, 3, 4}
  \\
  6 & \emp \implies \big(\alloc{\avariable} \land \size = 1 \magicwand \lnot \size \geq 2\big)
  & \mbox{\ref{rule:FS:staradj}, 5}
  \end{syntproof}
  \caption{A proof of \ $\emp \implies \big((\alloc{\avariable} \land \size = 1) \magicwand \lnot \size \geq 2\big)$.}\label{figure:aproof}
\end{figure}

\begin{exa}\label{exa:aproof}
To further familiarise with the axioms and the rules of~$\coresys(\separate,\magicwand)$, in Figure~\ref{figure:aproof},
we present a proof of
$\emp \implies \big(\alloc{\avariable} \land \size = 1 \magicwand \lnot \size \geq 2\big)$.
In the proof, a line ``$j\, \mid\, \aformulater \ \ A, i_1,\dots,i_k$''
states that $\aformulater$ is a theorem denoted by the index $j$ and derivable by the axiom or the rule $A$.
If $A$ is a rule, the indices $i_{1},\dots,i_{k} < j$ denote the theorems used as premises in order to derive $\aformulater$.
When a formula is obtained as a propositional tautology or by propositional reasoning from other
formulae, we may write ``PC'' (standing for short `Propositional Calculus'). 
Similarly, we provide any useful piece of information
justifying the derivation, such as ``Ind. hypothesis'', ``See \dots'' or ``Previously derived''.
In the example,
we use the rule~\ref{rule:FS:staradj}, which together with the rule~\ref{rule:FS:magicwandadj} states that the connectives
$\separate$ and $\magicwand$ are adjoint operators,
as well as the axiom~\ref{starAx:FS:SizeNeg}, stating that $\card{\domain{\aheap}} \leq \inbound_1 {+} \inbound_2$ holds whenever a heap $\aheap$ can be split into two subheaps whose domains have less than $\inbound_1{+}1$ and $\inbound_2{+}1$ elements, respectively.
We also use the following theorems and rules:

\begin{nscenter}
  \noindent
  \scalebox{0.9}{
  \lemmalab{\textbf{($\land$Er)}}{axiom:andelim}\,
  $
  \aformulabis \land \aformula \implies \aformula
  $}
  \hfill
    \scalebox{0.9}{
    \lemmalab{\textbf{($\lnot\lnot$I)}}{axiom:dubneg}\,
    $
    \aformula \implies \lnot\lnot\aformula
    $}%
  \hfill
  \scalebox{0.9}{
  \rulelab{\textbf{$\implies$-Tr}}{rule:imptr}
  $
  \inference{\aformula \implies \aformulater\quad\aformulater \implies \aformulabis}{\aformula \implies \aformulabis}
  $}%
  \hfill
  \scalebox{0.9}{
  \rulelab{\textbf{$\separate$-Ilr}}{rule:starintroLR}
  $
  \inference{\aformula \implies \aformula'\quad \aformulabis \implies \aformulabis'}{\aformula \separate \aformulabis \implies \aformula' \separate \aformulabis'}{}
  $}%
  \hfill\,
\end{nscenter}
\end{exa}

\noindent The first two theorems and the first rule are derivable by pure propositional reasoning. By way of example, we show that the inference rule~\ref{rule:starintroLR} is admissible.

\vspace{-12pt}

\noindent
\begin{minipage}[t]{0.43\linewidth}
\begin{syntproof}
1 & \aformula \implies \aformula'
& \mbox{Hypothesis} \\
2 & \aformulabis \implies \aformulabis'
& \mbox{Hypothesis} \\
3 & \aformula \separate \aformulabis \implies \aformula' \separate \aformulabis
& \mbox{\ref{rule:FS:starinference}, 1}\\
4 & \aformulabis \separate \aformula' \implies \aformulabis' \separate \aformula'
& \mbox{\ref{rule:FS:starinference}, 2} \\
\end{syntproof}
\end{minipage}%
\hfill
\begin{minipage}[t]{0.53\linewidth}
  \begin{syntproof}
  5 & \aformula' \separate \aformulabis \implies  \aformulabis \separate \aformula'
  & \mbox{\ref{starAx:FS:Commute}} \\
  6 & \aformulabis' \separate \aformula' \implies  \aformula' \separate \aformulabis'
  & \mbox{\ref{starAx:FS:Commute}} \\
  7 & \aformula \separate \aformulabis \implies \aformulabis \separate \aformula'
    & \mbox{\ref{rule:imptr}, 3, 5}\\
  8 &  \aformula \separate \aformulabis \implies \aformula' \separate \aformulabis'
  & \mbox{\ref{rule:imptr} twice, 7, 4, 6}
  \end{syntproof}
  \end{minipage}%

\vspace{2pt}

\begin{rem}
Note that an alternative proof of theorem 5 in Figure~\ref{figure:aproof} consists in applying \ref{rule:imptr} to theorem~2 and
  $\emp \separate \big(\alloc{\avariable}\land\size{=}1\big) \implies \alloc{\avariable}\land\size {=} 1$, which holds by
the axioms~\ref{starAx:FS:Emp}~and~\ref{starAx:FS:Commute}.
\end{rem}

\begin{exa}
In Figure~\ref{fig:anotherproof}, we develop the proof of $\emp \implies (\alloc{\avariable} \land \size = 1
\magicwand \size = 1)$ as a more complete example.
We use the following theorems and rules:
\begin{center}
  \noindent
  \scalebox{0.8}{
  \lemmalab{\textbf{($\magicwand\land$-DistrL)}}{axiom:magicsep}\,
  $
  (\aformula \magicwand \aformulabis) \land (\aformula \magicwand \aformulater) \implies (\aformula \magicwand \aformulabis \land \aformulater)
  $}%
  \hfill
  \scalebox{0.8}{
  \lemmalab{\textbf{($\land\true$IL)}}{axiom:andtrue}\,
  $
  \aformula \implies \true \land \aformula
  $}%
  \hfill
  \scalebox{0.8}{
  \rulelab{\textbf{$\land$-InfL}}{rule:andinf}
  $
  \inference{\aformula \implies \aformulater}{\aformula \land \aformulabis \implies \aformulater \land \aformulabis}
  $}%
\end{center}

\noindent The rightmost axiom and the only rule are derivable by propositional reasoning. 
We show the admissibility of the axiom
\ref{axiom:magicsep}.

\begin{syntproof}
1 & (\aformula \septraction \lnot \aformulabis \vee \lnot \aformulater)
\implies (\aformula \septraction \lnot \aformulabis) \vee (\aformula \septraction \lnot \aformulater)
& \mbox{\ref{mwAx:OrR}, Lemma~\ref{lemma:septractionadmissible}} \\
2 & \lnot (\aformula \magicwand \lnot (\lnot \aformulabis \vee \lnot \aformulater))
\implies \lnot (\aformula \magicwand \lnot \lnot \aformulabis) \vee \lnot (\aformula \magicwand \lnot \lnot \aformulater)
& \mbox{Def. $\septraction$, 1} \\
3 & \lnot (\aformula \magicwand \aformulabis \wedge \aformulater)
\implies \lnot (\aformula \magicwand \aformulabis) \vee \lnot (\aformula \magicwand  \aformulater)
& \mbox{Replacement of equivalents, 2} \\
4 & (\aformula \magicwand \aformulabis) \wedge (\aformula \magicwand  \aformulater)
\implies (\aformula \magicwand \aformulabis \wedge \aformulater)
& \mbox{PC, 3}
\end{syntproof}
\end{exa}

\begin{figure}
\begin{syntproof}
  1 & \true \separate (\alloc{\avariable} \land \size = 1) \implies (\alloc{\avariable} \land \size = 1) \separate \true
  & \mbox{\ref{starAx:FS:Commute}}
  \\
  2 & \alloc{\avariable} \land \size = 1 \implies \size \geq 1
  & \mbox{\ref{axiom:andelim}}
  \\
  3 & (\alloc{\avariable} \land \size = 1) \separate \true \implies \size \geq 1 \separate \true
  & \mbox{\ref{rule:FS:starinference}, 2}
  \\
  4 & \size \geq 1 \separate \true \implies \size \geq 1
  & \mbox{\ref{starAx:FS:MonoCore} ($\size \geq 1 \egdef \lnot \emp$)}
  \\
  5 & \true \separate (\alloc{\avariable} \land \size = 1) \implies \size \geq 1
  & \mbox{\ref{rule:imptr} twice, 1, 3, 4}
  \\
  6 & \true \implies (\alloc{\avariable} \land \size = 1 \magicwand \size \geq 1)
  & \mbox{\ref{rule:FS:staradj}, 5}
  \\
  7 & \emp \implies (\alloc{\avariable} \land \size = 1 \magicwand \lnot \size \geq 2)
  & \mbox{\small{See Example~\ref{exa:aproof}}}
  \\
  8 & (\alloc{\avariable} \land \size = 1 \magicwand \lnot \size \geq 2) \implies \\
    & \quad \true \land (\alloc{\avariable} \land \size = 1 \magicwand \lnot \size \geq 2)
  & \mbox{\ref{axiom:andtrue}}
  \\
  9 & \true \land (\alloc{\avariable} \land \size = 1 \magicwand \lnot \size \geq 2) \implies\\
    & \quad \big((\alloc{\avariable} \land \size = 1 \magicwand \size \geq 1) \land\\
    & \quad (\alloc{\avariable} \land \size = 1 \magicwand \lnot \size \geq 2)\big)
  & \mbox{\ref{rule:andinf}, 6}
  \\
  10 & \big((\alloc{\avariable} \land \size = 1 \magicwand \size \geq 1) \land\\
    & \quad (\alloc{\avariable} \land \size = 1 \magicwand \lnot \size \geq 2)\big) \implies\\
    & \quad (\alloc{\avariable} \land \size = 1 \magicwand \size = 1)
  & \mbox{\ref{axiom:magicsep} + Def. $\size$}
  \\
  11 & (\alloc{\avariable} \land \size = 1 \magicwand \lnot \size \geq 2) \implies\\
    & \quad (\alloc{\avariable} \land \size = 1 \magicwand \size = 1)
  & \mbox{\ref{rule:imptr} twice, 8, 9, 10}
  \\
  12 & \emp \implies (\alloc{\avariable} \land \size = 1 \magicwand \size = 1)
  & \mbox{\ref{rule:imptr}, 7, 11}
\end{syntproof}

\footnotesize{
(recall that $\size = \inbound$ is a shortcut for $\size \geq \inbound \land \lnot \size \geq \inbound{+}1$)}

\caption{\label{fig:anotherproof}A proof of
$\emp \implies (\alloc{\avariable} \land \size = 1 \magicwand \size = 1)$.
}

\end{figure}

\vspace{-10pt}

\subsection*{Main ingredients of the method.}
Before showing completeness of $\magicwandsys$, let us recall the key ingredients of the method we follow,
not only to provide
a vade mecum for axiomatising other separation logics (which, in the second part of~\cite{Demri&Lozes&Mansutti20}, we illustrate on the newly
introduced logic  \intervalSL), but also to identify the
essential features  and where variations are still possible.
The Hilbert-style axiomatisation of \slSW shall culminate
with Theorem~\ref{theo:PSLcompleteAx} that states the adequateness of the proof system $\magicwandsys$.

In order to axiomatise \slSW internally, as already emphasised several times, the core formulae  play an essential
role. The main properties of these formulae is that their Boolean combinations capture the full logic \slSW~\cite{Lozes04bis}
and all the core formulae can be expressed in \slSW.
Generally speaking, our axiom system naturally leads to a form of constructive completeness, as
advocated in~\cite{Doumane17,Luck18}: the axiomatisation provides
proof-theoretical means to transform any formula into an equivalent Boolean combination
of core formulae, and it contains also a part dedicated to the derivation of valid Boolean combinations
of core formulae (understood as a syntactical fragment of \slSW).
What is specific to each logic is the design of the set of core formulae and in the case of \slSW,
this was already known since~\cite{Lozes04bis}.

Derivations in the proof system $\magicwandsys$ shall simulate the bottom-up elimination of
separating connectives (see forthcoming Lemmata~\ref{lemma:starPSLelim} and~\ref{lemma:magicwandPSLelim})
when the arguments are two Boolean combinations of core formulae.
To do so, $\magicwandsys$ contains axiom schemas that perform such an elimination in multiple ``small-step''
derivations, e.g.\ by deriving a single $\alloc{\avariable}$ predicate from $\alloc{\avariable}\separate\true$
(with forthcoming intermediate axiom~\ref{starAx:StarAlloc}).
Alternatively, it would have been possible to include ``big-step'' axiom schemas that, given the two Boolean combinations of
core formulae, derive the equivalent formula in one single derivation step (see e.g.~\cite{Echenim&Iosif&Peltier19}).
The main difference is that small-step axioms provide
a simpler understanding of the key properties of the logic.

\section{A simple calculus for the core formulae}\label{subsection:axiomCoreFormulae}
To axiomatise \slSW, we start by introducing the proof
system $\coresys$  dedicated to Boolean combinations of core formulae, see Figure~\ref{figure-proof-system-core-formulae}.
As explained earlier, it also contains the axiom schemata and modus ponens for the propositional calculus.
Moreover, the axioms whose name is of the form $A^{\corepedix}_i$ are axioms that remain in the global system for \slSW,
whereas those named $I^{\corepedix}_i$ are intermediate axioms that are removed when considering the axioms dealing with
the separating connectives. As explained before,
the intermediate axioms are handy to establish results about the axiomatisation of Boolean combinations of core
formulae but are not needed when all the axioms and rules of $\magicwandsys$ are considered. 

In the axiom~\ref{coreAx:EqSub},
$\aformula\completesubstitute{\avariable}{\avariablebis}$ stands for
the formula obtained from $\aformula$ by replacing with the variable $\avariable$ every occurrence of $\avariablebis$.
Let $\pair{\astore}{\aheap}$ be a memory state.
The axioms state that $=$ is an equivalence relation (first two axioms),
$\aheap(\astore(\avariable)) = \astore(\avariablebis)$ implies $\astore(\avariable) \in \domain{\aheap}$ (axiom~\ref{coreAx:PointAlloc}) and
that $\aheap$ is a (partial) function (axiom~\ref{coreAx:PointInj}).
Furthermore, there are two intermediate axioms about size formulae: \ref{coreAx:Size} states that if $\domain{\aheap}$ has at least $\inbound{+}1$ elements,
then it has at least $\inbound$ elements,
whereas \ref{coreAx:AllocSize} states instead that if there are $\inbound$ distinct memory cells corresponding to program variables, then indeed $\domain{\aheap} \geq \inbound$.
It is easy to check that $\coresys$ is sound (see also Lemma~\ref{lemma:magicwandPSLvalid}).
In order to establish its completeness with respect to 
Boolean combinations of core formulae,
we first show that
$\coresys$ is
complete for a subclass of Boolean combinations of core formulae, namely for \defstyle{core types} defined below.
Then, we show that every formula in $\boolcomb{\coreformulae{\asetvar}{\bound}}$ is provably equivalent to a
disjunction of core types (Lemma~\ref{prop:corePSLone}).

\begin{figure}
  \label{figure-proof-system-core-formulae}
  \begin{footnotesize}

  \fbox{
    \begin{minipage}{0.95\linewidth}
      \vspace{2pt}
    \begin{enumerate}[align=left,leftmargin=*]
      \setlength\itemsep{4pt}
      \begin{minipage}{0.44\linewidth}
      \item[\axlab{A^\corepedix}{coreAx:EqRef}] $\avariable = \avariable$
      \item[\axlab{A^\corepedix}{coreAx:EqSub}] $\aformula \land \avariable = \avariablebis \implies \aformula\completesubstitute{\avariable}{\avariablebis}$
      \item[\axlab{A^\corepedix}{coreAx:PointAlloc}] $\avariable \Ipto \avariablebis \implies \alloc{\avariable}$
      \end{minipage}%
      \begin{minipage}{0.56\linewidth}
      \item[\axlab{A^\corepedix}{coreAx:PointInj}] $\avariable \Ipto \avariablebis \land \avariable \Ipto \avariableter \implies \avariablebis = \avariableter$
      \item[\axlab{I^\corepedix}{coreAx:Size}] $\size \geq \inbound{+}1 \implies \size \geq \inbound$
      \item[\axlab{I^\corepedix}{coreAx:AllocSize}] $\bigwedge_{\avariable \in \asetvar}(\alloc{\avariable} \land
      \bigwedge_{\avariablebis \in \asetvar \setminus \{\avariable\}} \avariable \neq \avariablebis) \implies \size \geq \card{\asetvar}$
      \end{minipage}
    \end{enumerate}
    \vspace{2pt}
    \end{minipage}
  }

  \end{footnotesize}
  \caption{Proof system $\coresys$ for Boolean combinations of core formulae.}
  \label{figure}
  \end{figure}

\subsection*{Introduction to core types.}
Let $\asetvar {\subseteq_{\fin}} \PVAR$ and $\bound \in \Nat^+$.
We write $\coretype{\asetvar}{\bound}$ to denote the set of \defstyle{core types} defined by
$$
\scaledformulasubset{
  \aformula \in \conjcomb{\coreformulae{\asetvar}{\bound}}
}{
  \bmat[
    {\rm for \ all} \ \aformulabis \in  \coreformulae{\asetvar}{\bound},\  \orliterals{\aformulabis}{\lnot \aformulabis} \inside \aformula, \text{ and } (\aformulabis \land \lnot \aformulabis) \not\inside \aformula
  ]
}{0.9}{1}.
$$
Note that if $\aformula \in \coretype{\asetvar}{\bound}$, then $\aformula$ is a conjunction such that
for every  $\aformulabis \in \coreformulae{\asetvar}{\bound}$, there is exactly one literal in $\aformula$ built
upon $\aformulabis$.

\cut{
\subsection*{Introduction to core types.}
Let $\asetvar {\subseteq_{\fin}} \PVAR$, $\bound \in \Nat^+$ and $\widehat{\bound} {=} \bound{+}\card{\asetvar}$.
We write $\coretype{\asetvar}{\bound}$ to denote the set of \defstyle{core types} defined by
$$
\scaledformulasubset{
  \aformula \in \conjcomb{\coreformulae{\asetvar}{\widehat{\bound}}}
}{
  \bmat[
    \forall\aformulabis{\in} \coreformulae{\asetvar}{\widehat{\bound}},\  \orliterals{\aformulabis}{\lnot \aformulabis} \inside \aformula, \text{ and } (\aformulabis \land \lnot \aformulabis) \not\inside \aformula
  ]
}{0.9}{1}.
$$
Note that if $\aformula \in \coretype{\asetvar}{\bound}$, then $\aformula$ is a conjunction such that
for every  $\aformulabis \in \coreformulae{\asetvar}{\widehat{\bound}}$, there is exactly one literal in $\aformula$ built
upon $\aformulabis$.
}

\begin{lem}[Refutational completeness]
\label{lemma:corePSLtwo}
Let $\aformula \in \coretype{\asetvar}{\bound}$, where $\bound \geq \card{\asetvar}$.
The formula $\neg \aformula$ is valid if and only if $\vdash_{\coresys} \neg \aformula$.
\end{lem}

\begin{proof}
We show that $\aformula$ is unsatisfiable if and only if $\vdash_{\coresys}\lnot \aformula$.
The ``only if'' part follows from the soundness of $\coresys$, so we prove the ``if'' part.
Let $\aformula\in \coretype{\asetvar}{\bound}$ be such that $\not\vdash_{\coresys}\aformula\Rightarrow \bot$, and let us prove that
$\aformula$ is satisfiable.
By the axioms~\ref{coreAx:EqRef}
and~\ref{coreAx:EqSub},
there is an equivalence relation $\approx$ on $\asetvar$ such that
$\avariable\approx\avariablebis$ iff $\avariable=\avariablebis$ occurs positively
in $\aformula$. We write $[\avariable]$ to denote the equivalence class of $\avariable$
with respect to $\approx$.
By the axioms~$\ref{coreAx:EqSub}$ and~$\ref{coreAx:PointInj}$,
there is a partial map $f: (\asetvar/\approx)\to(\asetvar/\approx)$ on equivalence classes such
that $\avariable\Ipto \avariablebis$ occurs positively iff $f([x])$ is defined and $f([x])=[y]$.
Let $D=\{[\avariable]\mid \alloc{\avariable}\mbox{ occurs positively in }\aformula\}$.
By the axiom $\ref{coreAx:PointAlloc}$, $\mathsf{dom}(f)\subseteq D$.
Let $n = \maxsize{\aformula}$. We recall that, by definition of $\maxsize{.}$, 
$n$ is the greatest $\beta$ such that $\size\geq\beta$ occurs positively in $\aformula$
(or zero if there are none).

Let us show that $n \geq \card{D}$. {\em Ad absurdum}, suppose that $n < \card{D}$. 
From the axiom~\ref{coreAx:AllocSize}, $\vdash_{\coresys}\aformula \Rightarrow \size \geq \card{D}$
and by definition of $n$ and the fact that $\bound \geq \card{\asetvar} \geq \card{D}$, 
$\vdash_{\coresys}\aformula \Rightarrow \size \geq n$ and 
$\vdash_{\coresys}\aformula \Rightarrow \neg (\size \geq (n+1))$ since both $\size \geq n$ and
$(\size \geq (n+1))$ (possibly negated) occur in $\aformula$ as $\bound \geq  \card{\asetvar}$.
By using the axiom~\ref{coreAx:Size} and propositional reasoning, 
we can get that $\vdash_{\coresys}\aformula \Rightarrow \neg (\size \geq \card{D})$ since 
$\vdash_{\coresys}\aformula \Rightarrow \neg (\size \geq (n+1))$,
which leads to a contradiction. Consequently,~$n\geq \card{D}$.

Let $\alocation_0,\alocation_1,\dots,\alocation_n\in\LOC$ be $n+1$ distinct locations, and let
us fix an enumeration $C_1,\dots,C_{\card{D}}$ on the equivalence classes of $\approx$.
Let $(\astore,\aheap)$ be defined by
\begin{itemize}
\item $\astore(\avariable) \egdef \alocation_i$ if $[x]$ is the $i$th equivalence class $C_i$,
\item $\aheap(\alocation_i) \egdef \alocation_j$ if $0<i\leq \card{D}$ and the $i$th equivalence class is mapped to the $j$th
one by $f$,
\item $\aheap(\alocation_i) \egdef \alocation_0$ if either $0<i\leq \card{D}$ and the $i$th equivalence class is not in the domain of $f$, or $i> \card{D}$.
\end{itemize}
Then, by construction, $\pair{\astore}{\aheap}$ satisfies all positive literals of the form $\avariable=\avariablebis$ or 
$\avariable\Ipto\avariablebis$ or $\alloc{\avariable}$
that occur positively in $\aformula$, and all negative literals that  occur  in $\aformula$. It also
satisfies $\size\geq n$, falsifies $\size\geq n+1$ (assuming $n+1 \leq \bound$), and by the axiom~\ref{coreAx:Size}, it satisfies all size literals
in $\aformula$.
\end{proof}

By classical reasoning, one can show that every $\aformula \in \boolcomb{\coreformulae{\asetvar}{\bound}}$
is provably equivalent to a disjunction of core types. Together with Lemma~\ref{lemma:corePSLtwo}, this implies that $\coresys$ is adequate with respect to the propositional logic of core formulae.

To prove forthcoming Theorem~\ref{theo:corePSLcompl}, let us first establish the following simple lemma.

\begin{lem}[Core Types Lemma]
\label{prop:corePSLone}
Let $\aformula \in \boolcomb{\coreformulae{\asetvar}{\bound}}$. There is a disjunction
$\aformulabis=\aformulabis_1\vee\ldots\vee\aformulabis_n$ with
$\aformulabis_i\in\coretype{\asetvar}{\max(\card{\asetvar},\bound)}$ for all $i$
such that $\prove_{\coresys} \aformula \iff \aformulabis$.
\end{lem}

\begin{proof}
Let $\aformulabis_1\vee\ldots\vee\aformulabis_n$ be a formula in disjunctive normal form logically equivalent to $\aformula$.
If $\aformulabis_i$ is not a core type in $\coretype{\asetvar}{\max(\card{\asetvar},\bound)}$, 
there is a core formula
$\aformulater\in\coreformulae{\asetvar}{\max(\card{\asetvar},\bound)}$
that occurs neither positively nor negatively in $\aformulabis_i$.
Replacing $\aformulabis_i$ with $(\aformulabis_i\wedge\aformulater)\vee(\aformulabis_i\wedge\neg\aformulater)$,
 and repeating this for all missing core formulae and for all~$i$, we obtain a
disjunction of core types of the expected form. Since all equivalences follow from pure propositional reasoning, the equivalence between $\aformula$ and the obtained formula can be proved
in $\coresys$.
\end{proof}

\begin{thm}[Adequacy]
\label{theo:corePSLcompl}
A Boolean combination of core formulae $\aformula$ is valid iff $\prove_{\coresys} \aformula$.
\end{thm}

\begin{proof} 
Let $\aformula$ be a Boolean combination of core formulae in $\coretype{\asetvar}{\bound}$ for some $\asetvar$ and $\bound$.
As all the axioms are valid (Lemma~\ref{lemma:magicwandPSLvalid}),
$\prove_{\coresys} \aformula$ implies that $\aformula$ is
valid.
Let us assume that $\aformula$ is valid, and let us prove that $\prove_{\coresys}\aformula$.
By Lemma~\ref{prop:corePSLone}, there is a disjunction $\aformulabis=\aformulabis_1\vee\ldots\vee\aformulabis_n$ of core types
in $\coretype{\asetvar}{\max(\card{\asetvar},\bound)}$
such that $\prove_{\coresys}(\neg \aformula)\Leftrightarrow\aformulabis$.
As $\aformula$ is valid, the formulae
$\neg\aformula$, $\aformulabis$ and all the $\aformulabis_i$'s are unsatisfiable.
By Lemma~\ref{lemma:corePSLtwo}, $\prove_{\coresys}\aformulabis_i\Rightarrow\bot$, for all $i$.
By propositional reasoning, $\prove_{\coresys} \aformula$.
\end{proof}

\section{Axiomatisation for \slSA}\label{section:starelimination}

We write \slSA to denote the fragment of \slSW in which the separating implication is removed
at the price of adding the atomic formulae of the form $\alloc{\avariable}$.
We define an Hilbert-style axiomatisation for \slSA,
obtained by enriching $\coresys$ with axioms and one inference rule
that handle  the separating conjunction~$\separate$, leading to the proof system~$\coresys(\separate)$.
Fundamentally, as we work now within \slSA, the core formula $\size \geq \inbound$ can be encoded in the logic.
According to its definition,
given in~Section~\ref{subsection:core-formulae},
we see $\size \geq 0$ as $\true$, $\size \geq 1$ as $\lnot \emp$ and $\size \geq \inbound{+}2$ as $\lnot \emp \separate \size \geq \beta{+}1$.

\begin{figure}
    \begin{footnotesize}
    \fbox{
    \begin{minipage}{0.95\linewidth}
        \vspace{3pt}
        \noindent
        \begin{enumerate}[align=left,leftmargin=*]
            \setlength\itemsep{4pt}%
        \addtocounter{enumi}{6}%
        \begin{minipage}[t]{0.44\linewidth}%
        \vspace{2pt}
        \item[\axlab{A^\separate}{starAx:Commute}] $(\aformula \separate \aformulabis) \iff (\aformulabis \separate \aformula)$
        \item[\axlab{A^\separate}{starAx:Assoc}] $(\aformula \separate \aformulabis) \separate \aformulater \iff \aformula \separate (\aformulabis \separate \aformulater)$
        \item[\axlab{I^\separate}{starAx:DistrOr}] $(\aformula \lor \aformulabis) \separate \aformulater \implies
        (\aformula \separate \aformulater) \lor (\aformulabis \separate \aformulater)$
        \item[\axlab{I^\separate}{starAx:False}] $(\bot \separate \aformula) \iff \bot$
        \item[\axlab{A^\separate}{starAx:Emp}] $\aformula \iff \aformula \separate \emp$
        \item[\axlab{I^\separate}{starAx:StarAlloc}] $\alloc{\avariable} \separate \true \implies \alloc{\avariable}$
        \item[\axlab{A^\separate}{starAx:DoubleAlloc}] $(\alloc{\avariable} \separate \alloc{\avariable})\iff\bot$\\
        \item[\rulelab{\textbf{$\separate$-Intro}}{rule:starinference}]
        $\inference{\aformula \implies \aformulater}{\aformula \separate \aformulabis \implies \aformulater \separate \aformulabis }{}$
        \vspace{3pt}
        \end{minipage}
        \begin{minipage}[t]{0.56\linewidth}
        \vspace{2pt}
        \item[\axlab{A^\separate}{starAx:MonoCore}] $\aelement \separate \true \implies \aelement \assuming{\aelement \in \{ \lnot\emp,\, \avariable = \avariablebis,\, \avariable \neq \avariablebis,\, \avariable \Ipto \avariablebis\}}$
        \item[\axlab{A^\separate}{starAx:AllocNeg}] $\lnot\alloc{\avariable} \separate \lnot\alloc{\avariable} \implies \lnot\alloc{\avariable}$
        \item[\axlab{A^\separate}{starAx:PointsNeg}] $(\alloc{\avariable} \land \lnot \avariable \Ipto \avariablebis) \separate \true \implies \lnot \avariable \Ipto \avariablebis$
        \item[\axlab{A^\separate}{starAx:AllocSizeOne}] $
        \alloc{\avariable} \implies
        (\alloc{\avariable} \land \size = 1) \separate \true$
        \item[\axlab{A^\separate}{starAx:SizeOne}] $\lnot \emp \implies \size = 1 \separate \true$
        \item[\axlab{A^\separate}{starAx:SizeNeg}] $\lnot \size \geq \inbound_1 \separate \lnot \size \geq \inbound_2 \implies \lnot \size \geq \inbound_1{+}\inbound_2{\dotminus}1$
        \item[\axlab{A^\separate}{starAx:SizeTwo}] $\alloc{\avariable}\wedge\alloc{\avariablebis}\wedge \avariable\neq\avariablebis\implies \size\geq2$\\ 
        \item[] 
        \quad\ ($a \dotminus b = \max(0,a-b)$) 
        \end{minipage}
        \end{enumerate}
    \end{minipage}
    }
    \end{footnotesize}
    
    \caption{Additional axioms and rule for $\coresys(\separate)$.}
    \label{figure-proof-system-star}
    \end{figure}

The axioms and the rule added to $\coresys$ in order to define $\coresys(\separate)$ are presented in Figure~\ref{figure-proof-system-star}.
Their soundness has been proved in~Lemma~\ref{lemma:magicwandPSLvalid},
with the exception of the three intermediate axioms
\ref{starAx:DistrOr},
\ref{starAx:False}
and~\ref{starAx:StarAlloc},
which are used for the completeness of $\coresys(\separate)$ with respect to \slSA, 
but are discharged from the proof system for $\slSW$ (Figure~\ref{figure-full-proof-system}), as they become derivable (Lemma~\ref{lemma:admissible-axioms-2}).

\begin{lem}{$\coresys(\separate)$ is sound.}
    \label{lemma:axiomatisation-star-sound}
\end{lem}
\begin{proof}
The axioms \ref{starAx:DistrOr} 
and~\ref{starAx:False} are inherited from Boolean BI (see~\cite{Brotherston&Villard14} and~\cite[Section 2]{Galmiche&Larchey06}). The soundness of~\ref{starAx:StarAlloc} is straightforward.
Indeed, suppose $\pair{\astore}{\aheap} \models \alloc{\avariable} \separate \true$. 
So, there is $\aheap' \subheap \aheap$ such that $\pair{\astore}{\aheap'} \models \alloc{\avariable}$. By definition of $\alloc{\avariable}$, $\astore(\avariable) \in \domain{\aheap'}$.
By $\aheap' \subheap \aheap$, $\astore(\avariable) \in \domain{\aheap}$. We conclude that
$\pair{\astore}{\aheap} \models \alloc{\avariable}$.
\end{proof}

Let us look further at the axioms in Figure~\ref{figure-proof-system-star}.
The axioms deal with the commutative monoid properties of $\pair{\separate}{\emp}$ and its distributivity
over $\vee$ (as for Boolean BI, see e.g.~\cite{Galmiche&Larchey06}).
The rule~\ref{rule:starinference},
sometimes called ``frame rule'' by analogy with the
rule of the same name in program logic,
states that logical equivalence is a congruence
for $\separate$.
$\coresys(\separate)$ is 
designed with the idea of being as simple as possible. On one side, this helps understanding the key ingredients of \slSA. On the other side, this makes the proof of completeness of $\coresys(\separate)$ more challenging. To work towards this proof while familiarising with the new axioms, we first show a set of 
intermediate theorems (see Appendix~\ref{appendix-separate-auxiliary-stuff}). 

\begin{lem}
    \label{lemma:separate-auxiliary-stuff}
    The following rules and axioms are admissible in $\coresys(\separate)$:
    \begin{enumerate}[align=left,leftmargin=*]
    \setlength\itemsep{3pt}
    \item[\indexedlab{I^{\separate}}{\ref{lemma:separate-auxiliary-stuff}}{starAx:auxilary-1}]
        $\avariable \sim \avariablebis \land (\aformula \separate \aformulabis) \implies (\aformula \land \avariable \sim \avariablebis) \separate \aformulabis$, 
        where~$\sim$ stands for $=$ or $\neq$.
    \item[\indexedlab{I^{\separate}}{\ref{lemma:separate-auxiliary-stuff}}{starAx:auxilary-2}]
        $\avariable = \avariablebis \land ((\aformula \land \alloc{\avariable}) \separate \aformulabis)
        \implies (\aformula \land \alloc{\avariablebis}) \separate \aformulabis$.
    \item[\indexedlab{I^{\separate}}{\ref{lemma:separate-auxiliary-stuff}}{starAx:auxilary-3}]
        $(\aformula\land \alloc{\avariable}) \separate \aformulabis \implies \aformula \separate (\aformulabis \land \lnot \alloc{\avariable})$. 
    \item[\indexedlab{I^{\separate}}{\ref{lemma:separate-auxiliary-stuff}}{starAx:auxilary-4}]
        $\lnot \alloc{\avariable} \land (\aformula \separate \aformulabis) \implies 
        (\aformula \land \lnot \alloc{\avariable}) \separate \aformulabis$.
    \item[\indexedlab{I^{\separate}}{\ref{lemma:separate-auxiliary-stuff}}{starAx:auxilary-4bis}]
        $\alloc{\avariable} \land (\aformula \separate (\lnot \alloc{\avariable} \land \aformulabis)) \implies (\aformula \land \alloc{\avariable}) \separate (\lnot \alloc{\avariable} \land \aformulabis)$
    \item[\indexedlab{I^{\separate}}{\ref{lemma:separate-auxiliary-stuff}}{starAx:auxilary-5}]
        $\avariable \Ipto \avariablebis \land ((\aformula \land \alloc{\avariable}) \separate \aformulabis)
        \implies (\aformula \land \avariable \Ipto \avariablebis) \separate \aformulabis$.
    \item[\indexedlab{I^{\separate}}{\ref{lemma:separate-auxiliary-stuff}}{starAx:auxilary-6}]
        $\lnot \avariable \Ipto \avariablebis \land (\aformula \separate \aformulabis) \implies 
        (\aformula \land \lnot \avariable \Ipto \avariablebis) \separate \aformulabis$.
    \end{enumerate} 
\end{lem}


In~$\coresys(\separate)$, the axioms~\ref{coreAx:Size} and~\ref{coreAx:AllocSize}
of $\coresys$ are superfluous and 
can be removed. Indeed, notice that both axioms do not appear in the proof system~$\magicwandsys$  given in Figure~\ref{figure-full-proof-system}.

\begin{lem}\label{lemma:admissible-axioms-1}
The axioms~\ref{coreAx:Size} and~\ref{coreAx:AllocSize} are derivable in $\coresys(\separate)$.
\end{lem}

\begin{proof}[Derivability of~{\rm\ref{coreAx:Size}}.]
The proof is by induction on $\inbound$.
\begin{description}
    \item[base case: $\inbound = 0$] 
    The instance of the axiom$~\ref{coreAx:Size}$ with $\inbound = 0$ amounts to
    derive the formula ${\size \geq 1 \implies \size \geq 0}$. By definition $\size \geq 1 = \neg \emp$
    and $\size \geq 0 = \top$, and therefore, by propositional reasoning,
    $\vdash_{\coresys(\separate)} \size \geq 1 \implies \size \geq 0$.
    \item[induction step: $\inbound > 0$]
        By induction hypothesis, assume $\vdash_{\coresys(\separate)} \size \geq \inbound \implies \size \geq \inbound-1$.
        The formula $\size \geq \inbound+1 \implies \size \geq \inbound$ is derived as follows:
        \begin{syntproof}
            1 & \size \geq \inbound \implies \size \geq \inbound-1
            & \mbox{Induction hypothesis} \\
            2 &  (\size \geq \inbound) \separate \lnot \emp  \implies (\size \geq \inbound-1) \separate \lnot \emp
            & \mbox{\ref{rule:starinference}, 1}
            \\
            3 &  \size \geq \inbound +1   \implies \size \geq \inbound
            & \mbox{2, def.~of $\size$}\hfill\qedhere
        \end{syntproof}
\end{description}
\end{proof}

\noindent
Before proving the validity of~\ref{coreAx:AllocSize}, we derive the intermediate theorem below.
Let $\asetvar \subseteq_{\fin} \PVAR$.
\begin{enumerate}[align=left, leftmargin=*, topsep=5pt]
    \item[\quad\indexedlab{I^{\magicwand}}{\ref{lemma:admissible-axioms-1}}{starAx:DualView}] \
    $\bigwedge_{\avariable \in \asetvar}(\alloc{\avariable} \land
 \bigwedge_{\avariablebis \in \asetvar \setminus \{\avariable\}} \avariable \neq \avariablebis) \implies (\bigseparate_{\avariable \in \asetvar} (\alloc{\avariable} \land \size = 1)) \separate \true$.
\end{enumerate}

\begin{proof}[Derivability of~{\rm\ref{starAx:DualView}}.]
    The proof is by induction on the size of $\asetvar$. 
    We distinguish two base cases, 
    for $\card{\asetvar} = 1$ and $\card{\asetvar} = 0$.
    \begin{description}
        \item[base case: $\card{\asetvar} = 1$]
            In this case,~\ref{starAx:DualView} is exactly \ref{starAx:AllocSizeOne}.
        \item[base case: $\card{\asetvar} = 0$]
            In this case, \ref{starAx:DualView} is $\true \implies \true \separate \true$.
            \begin{syntproof}
                1   & \emp \implies \true
                    & \mbox{PC} \\
                2   & \true \implies \true \separate \emp 
                    & \mbox{\ref{starAx:Emp}} \\
                3   & \true \separate \emp \implies \emp \separate \true 
                    & \mbox{\ref{starAx:Commute}}\\
                4   & \emp \separate \true \implies \true \separate \true 
                    & \mbox{\ref{rule:starinference}, 1}\\
                5   & \true \implies \true \separate \true 
                    & \mbox{\ref{rule:imptr}, 2, 3, 4}
            \end{syntproof}
        \item[induction step: $\card{\asetvar} \geq 2$]
            Let $\avariableter \in \asetvar$. By induction hypothesis,
            $$
                \textstyle\vdash_{\coresys(\separate)} \bigwedge_{\avariablefour \in \asetvar \setminus \{\avariableter\}}(\alloc{\avariablefour} \land
                \bigwedge_{\avariablefifth \in \asetvar \setminus \{\avariablefour,\avariableter\}} \avariablefour \neq \avariablefifth) \implies (\bigseparate_{\avariablefour \in \asetvar \setminus \{\avariableter\}} (\alloc{\avariablefour} \land \size = 1)) \separate \true.
            $$
            We write $\aformulater$ for the premise $\bigwedge_{\avariablefour \in \asetvar \setminus \{\avariableter\}}(\alloc{\avariablefour} \land
            \bigwedge_{\avariablefifth \in \asetvar \setminus \{\avariablefour,\avariableter\}} \avariablefour \neq \avariablefifth)$ above.
            Below, we aim for a proof of
            $$
                \textstyle\vdash_{\coresys(\separate)} 
                \bigwedge_{\avariable \in \asetvar}(\alloc{\avariable} \land
                \bigwedge_{\avariablebis \in \asetvar \setminus \{\avariable\}} \avariable \neq \avariablebis)
                \implies (\alloc{\avariableter} \land \size = 1) \separate \aformulater.
            $$
            In this way, the provability of~\ref{starAx:DualView}
            follows directly by induction hypothesis together with~\ref{starAx:Commute} 
            and~\ref{rule:starinference}. We have 
            \begin{syntproof}
                1 & \bigwedge_{\avariable \in \asetvar}(\alloc{\avariable} \land
                \bigwedge_{\avariablebis \in \asetvar \setminus \{\avariable\}} \avariable \neq \avariablebis) \implies (\alloc{\avariableter} \land \size = 1) \separate \true 
                & \mbox{\ref{starAx:AllocSizeOne} and PC}\\
                2 & \true \implies \aformulater \lor \lnot \aformulater 
                & \mbox{PC}\\
                3 & (\alloc{\avariableter} \land \size = 1) \separate \true 
                \implies (\alloc{\avariableter} \land \size = 1) \separate (\aformulater \lor \lnot \aformulater)
                & \mbox{\ref{rule:starinference},~\ref{starAx:Commute}, 2}\\
                4 & (\alloc{\avariableter} \land \size = 1) \separate (\aformulater \lor \lnot \aformulater) \implies\\[-2pt]
                &\quad 
                    ((\alloc{\avariableter} \land \size = 1) \separate \aformulater)
                    \lor 
                    ((\alloc{\avariableter} \land \size = 1) \separate \lnot \aformulater)
                & \mbox{\ref{starAx:Commute} and \ref{starAx:DistrOr}}\\
                5 & \bigwedge_{\avariable \in \asetvar}(\alloc{\avariable} \land
                \bigwedge_{\avariablebis \in \asetvar \setminus \{\avariable\}} \avariable \neq \avariablebis) \implies\\[-2pt]
                &\quad 
                    ((\alloc{\avariableter} \land \size = 1) \separate \aformulater)
                    \lor 
                    ((\alloc{\avariableter} \land \size = 1) \separate \lnot \aformulater)
                & \mbox{\ref{rule:imptr} 1, 3, 4}
            \end{syntproof}
            \noindent By propositional reasoning, $\lnot \aformulater$ is equivalent to 
            $\bigvee_{\avariablefour \in \asetvar \setminus \{\avariableter\}}(\lnot\alloc{\avariablefour} \lor
            \bigvee_{\avariablefifth \in \asetvar \setminus \{\avariablefour,\avariableter\}} \avariablefour = \avariablefifth)$.
            Due to the complexity of this formula, we proceed now rather informally, but our arguments entail the existence of a proper derivation.
            We aim at showing that
            \begin{equation}
                \vdash_{\coresys(\separate)} \bigwedge_{\avariable \in \asetvar}(\alloc{\avariable} \land
                \bigwedge_{\avariablebis \in \asetvar \setminus \{\avariable\}} \avariable \neq \avariablebis) \land ((\alloc{\avariableter} \land \size = 1) \separate \lnot \aformulater) \implies \bottom.
            \tag{$\dagger$}\label{parlabel-dagger}
            \end{equation}
            By propositional calculus and~\ref{starAx:DistrOr}, we can distribute conjunctions and separating conjunctions over disjunctions. 
            We derive:
            $$
                \vdash_{\coresys(\separate)} {\bigwedge_{\avariable \in \asetvar}(\alloc{\avariable} \land
                \bigwedge_{\avariablebis \in \asetvar \setminus \{\avariable\}} \avariable \neq \avariablebis)} \land ((\alloc{\avariableter} \land \size = 1) \separate \lnot \aformulater) \implies \aformulafour' \lor \aformulafour'',
            $$
             where $\aformulafour'$ and $\aformulafour''$ are defined, respectively, as%
                $$   
                \begin{aligned}
                    &\qquad
                    \textstyle\bigvee_{\avariablefour \in \asetvar \setminus \{\avariableter\}}
                    \!\left( 
                    \textstyle\bigwedge_{\avariable \in \asetvar}(\alloc{\avariable} \land
                    \bigwedge_{\avariablebis \in \asetvar \setminus \{\avariable\}} \avariable \neq \avariablebis) 
                     \,{\land}\,
                    ((\alloc{\avariableter} \land \size = 1)\,{\separate}\,\lnot \alloc{\avariablefour})\!\right),\\
                    &\qquad \textstyle\bigvee_{\substack{\avariablefour \in \asetvar \setminus \{\avariableter\}\\ \avariablefifth \in \asetvar \setminus \{\avariableter,\avariablefour\}}}
                    \!\Big(  
                    \textstyle\bigwedge_{\avariable \in \asetvar}(\alloc{\avariable} \land
                    \bigwedge_{\avariablebis \in \asetvar \setminus \{\avariable\}} \avariable \neq \avariablebis) 
                    \land 
                    ((\alloc{\avariableter} \land \size = 1) \separate \avariablefour = \avariablefifth)
                    \Big).
                    & 
                \end{aligned}
                $$
            In order to deduce~(\ref{parlabel-dagger}) it is sufficient to prove, in~$\coresys(\separate)$, 
            that every disjunct
            of $\aformulafour'$ and $\aformulafour''$ implies $\false$. 
            Clearly, if $\aformulafour'$ and $\aformulafour''$ do not have any disjunct, 
            i.e. when~$\asetvar \setminus \{\avariableter\}$ is empty,
            then the formula is propositionally equivalent to $\false$, which allows us to conclude~(\ref{parlabel-dagger}).
            Otherwise, let us consider each disjunct in $\aformulafour'$ and $\aformulafour''$ (separately), and prove their inconsistency.
            \begin{description}
                \item[case: $\aformulafour'$]
                    Let $\avariablefour \in \asetvar \setminus \{\avariableter\}$.
                    We show the inconsistency of
                    $$
                        \overline{\aformulafour} \egdef \textstyle\bigwedge_{\avariable \in \asetvar}(\alloc{\avariable} \land
                        \bigwedge_{\avariablebis \in \asetvar \setminus \{\avariable\}} \avariable \neq \avariablebis) 
                        \,{\land}\,
                        ((\alloc{\avariableter} \land \size = 1)\,{\separate}\,\lnot \alloc{\avariablefour}).
                    $$

                    \vspace{-10pt}

                    \begin{syntproof}
                        6   & \textstyle\bigwedge_{\avariable \in \asetvar}(\alloc{\avariable} \land
                        \bigwedge_{\avariablebis \in \asetvar \setminus \{\avariable\}} \avariable \neq \avariablebis) \implies \alloc{\avariablefour} \land \avariablefour \neq \avariableter 
                            & \mbox{PC}\\
                        7   & \overline{\aformulafour} \implies \alloc{\avariablefour} \land \avariablefour \neq \avariableter \land 
                        ((\alloc{\avariableter} \land \size = 1)\,{\separate}\,\lnot \alloc{\avariablefour}) 
                            & \mbox{PC}\\
                        8   & \alloc{\avariablefour} \land 
                        ((\alloc{\avariableter} \land \size = 1)\,{\separate}\,\lnot \alloc{\avariablefour}) 
                        \implies\\[-2pt]
                            &\quad ((\alloc{\avariableter} \land \size = 1 \land \alloc{\avariablefour})\,{\separate}\,\lnot \alloc{\avariablefour})
                            & \mbox{\ref{starAx:auxilary-4bis}}\\
                        9   & \avariablefour \neq \avariableter \land ((\alloc{\avariableter} \land \size = 1 \land \alloc{\avariablefour})\,{\separate}\,\lnot \alloc{\avariablefour}) \implies\\[-2pt]
                        &\quad ((\alloc{\avariableter} \land \size = 1 \land \alloc{\avariablefour} \land \avariablefour \neq \avariableter)\,{\separate}\,\lnot \alloc{\avariablefour})
                            &\mbox{\ref{starAx:auxilary-1}}\\
                        10  & \alloc{\avariableter} \land \alloc{\avariablefour} \land \avariablefour \neq \avariableter \implies \size \geq 2 
                            & \mbox{\ref{starAx:SizeTwo}}\\
                        11  & \size = 1 \implies \lnot \size \geq 2
                            & \mbox{PC}\\
                        12  & \alloc{\avariableter} \land \size = 1 \land \alloc{\avariablefour} \land \avariablefour \neq \avariableter \implies \bottom 
                            & \mbox{\ref{rule:imptr}, PC, 10, 11}\\
                        13  & \overline{\aformulafour} \implies (\alloc{\avariableter} \land \size = 1 \land \alloc{\avariablefour} \land \avariablefour \neq \avariableter) \separate \lnot \alloc{\avariablefour}
                            & \mbox{PC, 7, 8, 9}\\
                        14  & (\alloc{\avariableter} \land \size = 1 \land \alloc{\avariablefour} \land \avariablefour \neq \avariableter) \separate \lnot \alloc{\avariablefour} \implies {\false\! \separate \lnot \alloc{\avariablefour}}
                            & \mbox{\ref{rule:starinference}, 12}\\
                        15  & {\false \separate \lnot \alloc{\avariablefour}} \implies \false 
                            & \mbox{\ref{starAx:False}, 14}\\
                        16  & \overline{\aformulafour} \implies \false 
                            & \mbox{PC, 13, 15}
                    \end{syntproof}
                    \noindent Since $\overline{\aformulafour}$ is an arbitrary disjunct appearing in $\aformulafour'$, we conclude that $\vdash_{\coresys(\separate)} \aformulafour' \implies \false$.
                    \item[case: $\aformulafour''$] Let $\avariablefour \in \asetvar \setminus \{\avariableter\}$ and $\avariablefifth \in \asetvar \setminus \{\avariableter,\avariablefour\}$.
                    Notice that if $\avariablefour$ or $\avariablefifth$ do not exist, then $\aformulafour''$ is defined as $\false$ and so the proof is complete.
                    Otherwise, we show the inconsistency of 
                        $$
                            \widehat{\aformulafour} \egdef 
                            \textstyle\bigwedge_{\avariable \in \asetvar}(\alloc{\avariable} \land
                            \bigwedge_{\avariablebis \in \asetvar \setminus \{\avariable\}} \avariable \neq \avariablebis) 
                            \land 
                            ((\alloc{\avariableter} \land \size = 1) \separate \avariablefour = \avariablefifth).
                        $$
                    \begin{syntproof}
                        17   & 
                        \textstyle\bigwedge_{\avariable \in \asetvar}(\alloc{\avariable} \land
                        \bigwedge_{\avariablebis \in \asetvar \setminus \{\avariable\}} \avariable \neq \avariablebis)  \implies \avariablefour \neq \avariablefifth 
                            & \mbox{PC}\\
                        18   & \alloc{\avariableter} \land \size = 1 \implies \true 
                            & \mbox{PC}\\
                        19   & (\alloc{\avariableter} \land \size = 1) \separate \avariablefour = \avariablefifth \implies \avariablefour = \avariablefifth \separate \true
                            & \mbox{\ref{rule:starinference}, 18, \ref{starAx:Commute}}\\
                        20   & \avariablefour = \avariablefifth \separate \true \implies \avariablefour = \avariablefifth 
                            & \mbox{\ref{starAx:MonoCore}}\\
                        21   & ((\alloc{\avariableter} \land \size = 1) \separate \avariablefour = \avariablefifth) \implies \avariablefour = \avariablefifth   
                            & \mbox{\ref{rule:imptr}, 19, 20}\\
                        22   & \widehat{\aformulafour} \implies \false 
                            & \mbox{PC, 17, 21}
                    \end{syntproof}
                    \noindent Since $\widehat{\aformulafour}$ is an arbitrary disjunct appearing in $\aformulafour''$, we conclude that $\vdash_{\coresys(\separate)} \aformulafour'' \implies \false$.
            \end{description}
            From $\vdash_{\coresys(\separate)} \aformulafour' \implies \false$ and $\vdash_{\coresys(\separate)} \aformulafour'' \implies \false$ we conclude that 
            (\ref{parlabel-dagger}) holds.
            From the theorem 5 derived in this proof, this allows us to conclude that 
                $$
                    \textstyle\vdash_{\coresys(\separate)} 
                    \bigwedge_{\avariable \in \asetvar}(\alloc{\avariable} \land
                    \bigwedge_{\avariablebis \in \asetvar \setminus \{\avariable\}} \avariable \neq \avariablebis)
                    \implies (\alloc{\avariableter} \land \size = 1) \separate \aformulater,
                $$
            which concludes the proof, as explained at the beginning of the induction step. 
            \qedhere
    \end{description}
\end{proof}

\noindent We complete the proof of Lemma~\ref{lemma:admissible-axioms-1} by showing a derivation of~\ref{coreAx:AllocSize}.

\begin{proof}[Derivability of~{\rm\ref{coreAx:AllocSize}}.]
    Let $\asetvar \subseteq_{\fin} \PVAR$.
    If $\asetvar=\emptyset$, then the instance of the axiom~\ref{coreAx:AllocSize} becomes $\top\implies \size\geq 0$,
    which, by definition of~$\size \geq 0$,
    is syntactically equivalent to~$\top\implies \top$ and hence valid by propositional reasoning.
    Below, assume $\asetvar \neq \emptyset$ and fix $\avariableter \in \asetvar$.
    \begin{syntproof}
        1  & \bigwedge_{\avariable \in \asetvar}(\alloc{\avariable} \land
        \bigwedge_{\avariablebis \in \asetvar \setminus \{\avariable\}} \avariable \neq \avariablebis) \implies\\
            & \quad (\bigseparate_{\avariable \in \asetvar} (\alloc{\avariable} \land \size = 1)) \separate \true
            & \mbox{\ref{starAx:DualView}}\\
        2   & \alloc{\avariable} \land \size = 1 \implies \size \geq 1
            & \mbox{PC, def. of $\size = 1$}\\
        3   & (\bigseparate_{\avariable \in \asetvar} (\alloc{\avariable} \land \size = 1)) \separate \true
            \implies (\bigseparate_{\avariable \in \asetvar} \size \geq 1) \separate \true
            & \mbox{multiple applications of}\\[-3pt] 
            && \mbox{\ref{rule:starinference}, 2, \ref{starAx:Commute} and \ref{rule:imptr}}\\
        4   & (\bigseparate_{\avariable \in \asetvar} \size \geq 1) \separate \true
            \implies 
            (\size \geq 1 \separate \true) \separate (\bigseparate_{\avariable \in \asetvar \setminus \{\avariableter\}} \size \geq 1) 
            & \mbox{\ref{starAx:Commute}, \ref{starAx:Assoc}, def.~of~$\avariableter$}\\
        5   & \size \geq 1 \separate \true \implies \size \geq 1 
            & \mbox{\ref{starAx:MonoCore}, def. of~$\size \geq 1$}\\
        6   & (\size \geq 1 \separate \true) \separate (\bigseparate_{\avariable \in \asetvar \setminus \{\avariableter\}} \size \geq 1) \implies (\bigseparate_{\avariable \in \asetvar } \size \geq 1)
            & \mbox{\ref{rule:starinference}}\\
        7   & (\bigseparate_{\avariable \in \asetvar} \size \geq 1) \implies 
            \size \geq \card{\asetvar}
            & \mbox{\ref{starAx:Assoc}, def.~of~$\size \geq \card{\asetvar}$}\\
        8   & \bigwedge_{\avariable \in \asetvar}(\alloc{\avariable} \land
        \bigwedge_{\avariablebis \in \asetvar \setminus \{\avariable\}} \avariable \neq \avariablebis) \implies \size \geq \card{\asetvar}
            & \mbox{\ref{rule:imptr}, 1, 3, 4, 6, 7}\hfill\qedhere
    \end{syntproof}
\end{proof}

\begin{figure}
    \scalebox{1.2}{
      $
      \begin{aligned}
      &\bigwedge \formulasubset{\avariable \sim \avariablebis \inside \orliterals{\aformula}{\aformulabis}}{\bmat[\sim \in \{=,\neq\}]}
      & \land &
      \bigwedge \aformulasubset{\alloc{\avariable}\inside \orliterals{\aformula}{\aformulabis}}
      \\ \land &
      \bigwedge \aformulasubset{\lnot\alloc{\avariable} \inside\andliterals{\aformula}{\aformulabis}}
      & \land &
      \bigwedge \formulasubset{\lnot \avariable \Ipto \avariablebis}{\bmat[\alloc{\avariable}\land\lnot\avariable\Ipto\avariablebis\inside \orliterals{\aformula}{\aformulabis}]}
      \\ \land &
      \bigwedge \formulasubset{\avariable \neq \avariable}{\bmat[\alloc{\avariable}\inside\andliterals{\aformula}{\aformulabis}]}
      & \land & \bigwedge \formulasubset{\size\geq\inbound_1{+}\inbound_2}{\bmat[\size\geq\inbound_1\inside\aformula\\\size\geq\inbound_2\inside\aformulabis]}
      \\ \land &
      \bigwedge \aformulasubset{\avariable \Ipto \avariablebis \inside \orliterals{\aformula}{\aformulabis}}
      &\land &
      \bigwedge \formulasubset{\lnot\size\geq\inbound_1{+}\inbound_2{\dotminus}1}{\bmat[\lnot\size\geq\inbound_1\inside\aformula\\\lnot\size\geq\inbound_2\inside\aformulabis]}
      \end{aligned}
      $
    }

    \caption{The formula $\boxstar{\aformula}{\aformulabis}$.}
    \label{figure:boxstar}
    \end{figure}

From now on, we understand $\coresys(\separate)$ as the proof system obtained from $\coresys$ by adding all schemata from
Figure~\ref{figure-proof-system-star} but by removing~\ref{coreAx:Size} and~\ref{coreAx:AllocSize}.
We show that $\starsys$ enjoys the $\separate$ elimination 
property when the argument formulae are core types. That is, given two satisfiable core types $\aformula$ and $\aformulabis$, in $\coretype{\asetvar}{\bound}$, we show that the formula $\aformula \separate \aformulabis$ is provably equivalent to
the formula
$\boxstar{\aformula}{\aformulabis}$ in $\conjcomb{\coreformulae{\asetvar}{2\bound}}$, defined in~Figure~\ref{figure:boxstar}.

\begin{lem}\label{lemma:starPSLelim-sat}
 Let $\asetvar \subseteq_{\fin} \PVAR$ and $\bound \geq \card{\asetvar}$.
 If $\aformula$ and  $\aformulabis$  are two satisfiable core types in 
 $\coretype{\asetvar}{\bound}$, then
$\prove_{\starsys} \aformula \separate \aformulabis \iff \boxstar{\aformula}{\aformulabis}$.
\end{lem}

\noindent The equivalence $\aformula \separate \aformulabis \iff \boxstar{\aformula}{\aformulabis}$ 
is reminiscent to the one in~\cite[Lemma 3]{Echenim&Iosif&Peltier19} that is proved semantically.
In a way, because $\starsys$ will reveal to be complete, the restriction of
\cite[Lemma 3]{Echenim&Iosif&Peltier19} to \slSA can be
replayed completely syntactically within $\starsys{}$.

\begin{proof}[Structure of the proof of~Lemma~\ref{lemma:starPSLelim-sat}]
\renewcommand{\qedsymbol}{}
Before presenting the technical developments, let us explain the structure of the whole proof of~Lemma~\ref{lemma:starPSLelim-sat}, which might help
to follow the different steps. In order to show that 
$\prove_{\starsys} \aformula \separate \aformulabis \iff \boxstar{\aformula}{\aformulabis}$, 
we start showing that $\prove_{\starsys}\aformula \separate \aformulabis\implies \boxstar{\aformula}{\aformulabis}$.
This can be done rather mechanically since for every literal 
$\aliteral$ of $\boxstar{\aformula}{\aformulabis}$, one can construct a derivation for 
$\prove_{\starsys}\aformula \separate \aformulabis\implies \aliteral$.
The main difficulty in the proof rests on showing that 
$\prove_{\starsys}\boxstar{\aformula}{\aformulabis} \implies \aformula \separate \aformulabis$.
To do so, we build a sequence of 
formulae $\aformula^{(1)}\!\separate \aformulabis^{(1)}$,
   $\aformula^{(2)} \!\separate \aformulabis^{(2)}$, $\dots$, 
$\aformula^{(k)} \!\separate \aformulabis^{(k)}$ satisfying the following conditions:
\begin{itemize}
\item  for all $i \in \interval{1}{k}$,
  $\prove_{\starsys} \boxstar{\aformula}{\aformulabis} \implies \aformula^{(i)} \separate \aformulabis^{(i)}$, 
  the formulae~$\aformula^{(i)}$ and $\aformulabis^{(i)}$ are conjunctions of core formulae, and
\item for all $j \in \interval{1}{i}$, $\aformula^{(j)} \inside \aformula^{(i)}$ and 
      $\aformulabis^{(j)} \inside \aformulabis^{(i)}$.
\item $\aformula = \aformula^{(k)}$ and $\aformulabis = \aformulabis^{(k)}$
   (modulo associativity/commutativity of the classical conjunction).
\end{itemize}
In order to build $\aformula^{i+1}$ (resp.  $\aformulabis^{i+1}$), we identify a literal $\aliteral$
in $\aformula$ (resp. in $\aformulabis$)  that does not occur yet in $\aformula^{i}$
and we show that $\prove_{\starsys}  \boxstar{\aformula}{\aformulabis} \implies
\aformula^{(i+1)} \separate \aformulabis^{(i+1)}$ with $\aformula^{i+1} \egdef \aformula^{(i)} \wedge \aliteral$
(resp. $\aformulabis^{(i+1)} \egdef \aformulabis^{(i)} \wedge \aliteral$)
and $\aformulabis^{(i+1)} \egdef \aformulabis^{(i)}$ (resp. $\aformula^{(i+1)} \egdef \aformula^{(i)}$).
The case analysis on the shape of the literal $\aliteral$ is rather mechanical but it remains to specify
how the first formulae $\aformula^{(1)}$ and $\aformulabis^{(1)}$ are designed. 
In short, $\aformula^{(1)}$ (resp. $\aformulabis^{(1)}$) is dedicated to the 
part of $\aformula$  (resp. $\aformulabis$) 
related to the size of the heap domain and to the allocated variables. Details
will follow.

To construct these above-mentioned derivations, some additional derivations are instrumental
in particular to establish that the formulae below are derivable in $\starsys$:
$$
\size \geq \inbound_1 + \inbound_2 \implies \size = \inbound_1 \separate \size \geq \inbound_2,
\ \ \ \ \ \ 
 \size = \inbound_1 + \inbound_2 \implies \size = \inbound_1 \separate \size = \inbound_2.
$$
Such derivations can be found in Appendix~\ref{appendix-derivation-one}. 
We now develop the proof of~Lemma~\ref{lemma:starPSLelim-sat}.
\end{proof}

\begin{proof}[Proof of~Lemma~\ref{lemma:starPSLelim-sat}]
First of all, let us briefly explain what is the rationale for having literals of the form $\avariable \neq \avariable$
in the definition of $\boxstar{\aformula}{\aformulabis}$.
Recall that ${\alloc{\avariable}\inside\andliterals{\aformula}{\aformulabis}}$ is a shortcut
to state that $\alloc{\avariable}$ occurs in both the core types $\aformula$ and
$\aformulabis$. Since ${(\alloc{\avariable} \wedge \aformula')} \separate
(\alloc{\avariable} \wedge \aformulabis')$ is unsatisfiable,  $\alloc{\avariable}\inside\andliterals{\aformula}{\aformulabis}$
entails that $\boxstar{\aformula}{\aformulabis}$ should be unsatisfiable. That is why, if $\alloc{\avariable}\inside\andliterals{\aformula}{\aformulabis}$, then~$\avariable \neq \avariable$ is part of
$\boxstar{\aformula}{\aformulabis}$.

\paragraph{($\Rightarrow$):}
Let us show that~$\prove_{\starsys}\aformula \separate \aformulabis\implies \boxstar{\aformula}{\aformulabis}$.
      We establish that $\prove_{\starsys}\aformula \separate
      \aformulabis\implies \aliteral$ holds for every literal $\aliteral$ of $\boxstar{\aformula}{\aformulabis}$. We reason by a case analysis on
      $\aliteral \inside \boxstar{\aformula}{\aformulabis}$.

  \begin{description}

  \item[case: $\aliteral$ is an (in)equality or $ \aliteral \, = \, \avariable \Ipto \avariablebis$]
    For all the equalities and inequalities in $\aformula$ or $\aformulabis$,
    as well as all the literals of the
  form $\avariable\Ipto\avariablebis$,
  $\prove_{\starsys}\aformula \separate
  \aformulabis\implies \aliteral$
  follows from the rule~\ref{rule:starinference} and
  the axiom~\ref{starAx:MonoCore}. Let us provide below the proper derivation
when $\aliteral$ is a literal in $\aformula$ that is an equality, an inequality or of the form $\avariable\Ipto\avariablebis$.
  \end{description}

\vspace{-13pt}

\noindent%
\begin{minipage}[t]{0.46\linewidth}
\begin{syntproof}
1 & \aformula \implies \aliteral
& \mbox{PC} \\
2 & \aformulabis \implies \top
& \mbox{PC} \\
3 & \aformula \separate \aformulabis \implies \aliteral \separate \top
& \mbox{\ref{rule:starintroLR}, 1, 2}
\end{syntproof}
\end{minipage}\quad
\begin{minipage}[t]{0.5\linewidth}
\begin{syntproof}
4 &  \aliteral \separate \top \implies \aliteral
& \mbox{\ref{starAx:MonoCore}} \\
5 & \aformula \separate \aformulabis \implies \aliteral
& \mbox{\ref{rule:imptr}, 3, 4}
\end{syntproof}
\end{minipage}
\vspace{0.2cm}

\begin{itemize}
  \item[] Assume there is a literal $\avariable\neq\avariable$ that occurs in $\boxstar{\aformula}{\aformulabis}$.
     As both $\aformula$ and $\aformulabis$ are satisfiable, and thanks to~\ref{coreAx:EqRef},
     this is necessarily due to $\alloc{\avariable}$ occurring both in $\aformula$ and
     $\aformulabis$.
\end{itemize}

    \vspace{-3pt}

  \noindent
  \begin{minipage}{0.48\linewidth}
  \begin{syntproof}
    1 & \aformula \implies \alloc{\avariable}
    & \mbox{PC} \\
    2  & \aformulabis \implies \alloc{\avariable}
    & \mbox{PC} \\
    3 & \aformula \separate \aformulabis \implies  \alloc{\avariable } \separate  \alloc{\avariable}
    & \mbox{\ref{rule:starintroLR}, 1, 2}
  \end{syntproof}
  \end{minipage}\quad
  \begin{minipage}{0.52\linewidth}
  \begin{syntproof}
    4 &   \alloc{\avariable } \separate  \alloc{\avariable}  \implies \perp
    & \mbox{\ref{starAx:DoubleAlloc}} \\
    5 & \perp \implies \avariable\neq\avariable
    & \mbox{PC} \\
    6 &  \aformula \separate \aformulabis \implies \avariable\neq\avariable
    & \mbox{\ref{rule:imptr}, 4, 5}
  \end{syntproof}
  \end{minipage}

  \vspace{5pt}

  \begin{description}

  \item[case: $\aliteral \,=\, \alloc{\avariable}$]
    Follows from~\ref{starAx:StarAlloc} and~\ref{rule:starinference}.

  \item[case: $\aliteral \,=\, \lnot \alloc{\avariable}$]
    Follows from~\ref{starAx:AllocNeg} and~\ref{rule:starinference}.

  \item[case: $\aliteral \,=\, \neg\avariable\Ipto\avariablebis$]
    Let~$\neg\avariable\Ipto\avariablebis$ be a literal occurring in $\boxstar{\aformula}{\aformulabis}$.
    So, ${\alloc{\avariable}\wedge\neg\avariable\Ipto\avariablebis}$ occurs in
    $\aformula$ or $\aformulabis$, say in $\aformula$
    (the other case is equivalent, due to~\ref{starAx:Commute}).

    \vspace{-2pt}

    \begin{syntproof}
      1 & \aformula \implies \alloc{\avariable}\wedge \neg\avariable\Ipto\avariablebis
      & \mbox{PC} \\
      2  & \aformulabis \implies \top
      & \mbox{PC} \\
      3 & \aformula \separate \aformulabis \implies  (\alloc{\avariable}\wedge \neg\avariable\Ipto\avariablebis) \separate \top
      & \mbox{\ref{rule:starintroLR}, 1, 2} \\
      4 &  (\alloc{\avariable}\wedge \neg\avariable\Ipto\avariablebis)  \separate  \top \implies \neg\avariable\Ipto\avariablebis
      & \mbox{\ref{starAx:PointsNeg}} \\
      5 &  \aformula \separate \aformulabis \implies \neg\avariable\Ipto\avariablebis
      & \mbox{\ref{rule:imptr}, 3, 4}
   \end{syntproof}

   \vspace{4pt}

  \item[case : $\aliteral \,=\, \size \geq \inbound_1 + \inbound_2$, where $\size \geq \inbound_1 \inside \aformula$ and $\size \geq \inbound_2 \inside \aformulabis$]
  \end{description}

  \vspace{-13pt}
  \noindent
  \begin{minipage}[t]{0.46\linewidth}
    \begin{syntproof}
      1 & \aformula \implies \size\geq\inbound_1
      & \mbox{PC} \\
      2  & \aformulabis \implies \size\geq\inbound_2
      & \mbox{PC} \\
    \end{syntproof}
  \end{minipage} \quad
  \begin{minipage}[t]{0.50\linewidth}
    \begin{syntproof}
      3 & \aformula \separate \aformulabis \implies  \size\geq\inbound_1 \separate   \size\geq\inbound_2
      & \mbox{\ref{rule:starintroLR}, 1, 2} \\
      4 &  \aformula \separate \aformulabis \implies  \size\geq (\inbound_1+\inbound_2)
      & \mbox{Def. $\size$}
   \end{syntproof}
  \end{minipage}

  \vspace{7pt}

  \begin{itemize}
   \item[] \noindent Notice that, as $\aformula$ and $\aformulabis$ are satisfiable core types, $\size \geq 0$ appears positively in both these formulae, and thus appears in $\boxstar{\aformula}{\aformulabis}$.
  \end{itemize}
  \begin{description}
  \item[case: $\aliteral \,=\, \lnot \size \geq \inbound_1 + \inbound_2 \dotminus 1$, where $\lnot \size \geq \inbound_1 \inside \aformula$ and $\lnot \size \geq \inbound_2 \inside \aformulabis$]~

    \begin{syntproof}
      1 & \aformula \implies \lnot \size\geq\inbound_1
      & \mbox{PC} \\
      2  & \aformulabis \implies \lnot \size\geq\inbound_2
      & \mbox{PC} \\
      3 & \aformula \separate \aformulabis \implies  \lnot \size\geq\inbound_1 \separate  \lnot \size\geq\inbound_2
      & \mbox{\ref{rule:starintroLR}, 1, 2} \\
      4 & \lnot \size\geq\inbound_1 \separate  \lnot \size\geq\inbound_2
          \implies \lnot \size \geq \inbound_1 + \inbound_2 \dotminus 1
        & \mbox{\ref{starAx:SizeNeg}}\\
      5 &  \aformula \separate \aformulabis \implies \lnot \size \geq \inbound_1 + \inbound_2 \dotminus 1
      & \mbox{\ref{rule:imptr}, 3, 4}
  \end{syntproof}
  \end{description}

\paragraph{($\Leftarrow$):}
   Let us show that $\prove_{\starsys}\boxstar{\aformula}{\aformulabis}\implies\aformula \separate \aformulabis$.
   If $\boxstar{\aformula}{\aformulabis}$ is unsatisfiable, then by completeness of $\coresys$ (Theorem~\ref{theo:corePSLcompl}), $\vdash_{\coresys} \boxstar{\aformula}{\aformulabis} \implies \false$, and thus $\prove_{\coresys}\boxstar{\aformula}{\aformulabis}\implies\aformula \separate \aformulabis$.
   Since $\starsys$ includes $\coresys$, we conclude that
   $\prove_{\starsys}\boxstar{\aformula}{\aformulabis}\implies\aformula \separate \aformulabis$.
   Otherwise, below, we assume $\boxstar{\aformula}{\aformulabis}$ to be satisfiable.
   In particular, this implies that no literals of the form ${\avariable \neq \avariable}$ or ${\neg \size \geq 0}$
   appear in $\boxstar{\aformula}{\aformulabis}$.
   Moreover, by definition of $\boxstar{\aformula}{\aformulabis}$,
   this implies that $\aformula$, $\aformulabis$ and $\boxstar{\aformula}{\aformulabis}$
   agree on the satisfaction of the core formulae
   ${\avariable=\avariablebis}$,
   i.e.~$\aformula$, $\aformulabis$ and $\boxstar{\aformula}{\aformulabis}$ contain exactly
   the same (in)equalities.
   Since $\aformula$ is satisfiable, these equalities
   define an equivalence relation.
   Let $\avariable_1,\ldots\avariable_n$ be a maximal enumeration of representatives of the
   equivalence classes (one per equivalence class) such that $\alloc{\avariable_i}$ occurs in $\boxstar{\aformula}{\aformulabis}$.
   As it is maximal, for every $\alloc{\avariable} \inside \boxstar{\aformula}{\aformulabis}$ there is $i \in \interval{1}{n}$ such that $\avariable_i$ is syntactically equal to $\avariable$.
   Consequently, from the definition of $\boxstar{\aformula}{\aformulabis}$,
   if $\alloc{\avariable}$ occurs in $\aformula$ or in $\aformulabis$, then
   there is some $\avariable_i$ such that $\avariable = \avariable_i$ occurs in $\aformula$ (and therefore also in
   $\aformulabis$ and in $\boxstar{\aformula}{\aformulabis}$).
   Let us define the formula $\ALLOC$ below:
   $$
    \ALLOC \egdef \big(\alloc{\avariable_1}\wedge\size=1\big) \separate \cdots \separate \big(\alloc{\avariable_n}\wedge\size=1\big).
   $$
   We have,
   \begin{syntproof}
      1   & \boxstar{\aformula}{\aformulabis} \implies \bigwedge_{i \in \interval{1}{n}} (\alloc{\avariable_i} \land \bigwedge_{j \in \interval{1}{n} \setminus \{i\}} \avariable_i \neq \avariable_j)
          & \mbox{PC, def.~of~$\avariable_1,\dots,\avariable_n$}\\
      2   &  \bigwedge_{i \in \interval{1}{n}} (\alloc{\avariable_i} \land \bigwedge_{j \in \interval{1}{n} \setminus \{i\}} \avariable_i \neq \avariable_j) \implies \ALLOC \separate \true
          & \mbox{\ref{starAx:DualView}}\\
      3   & \boxstar{\aformula}{\aformulabis} \implies \ALLOC \separate \top
          & \mbox{\ref{rule:imptr}, 1, 2}
   \end{syntproof}
   \noindent Moreover, we show that $\prove_{\starsys} \ALLOC \implies \size \geq n$ and $\prove_{\starsys} {\ALLOC \implies \lnot \size \geq n{+}1}$
   (theorems 4 and 7 below), and so $\prove_{\starsys} \ALLOC \implies \size = n$.
   \begin{syntproof}
    1 & \aformulater \land \size = 1 \implies \size \geq 1
      & \mbox{PC, def.~of~$\size = 1$}\\
    2 & \aformulater \land \size = 1 \implies \lnot \size \geq 2
      & \mbox{PC, def.~of~$\size = 1$}\\
    3 & \ALLOC \implies \bigseparate_{i \in \interval{1}{n}} \size \geq 1
      & \mbox{multiple applications of}\\[-3pt]
      && \mbox{\ref{rule:starinference}, 1, \ref{starAx:Commute} and \ref{rule:imptr}}\\
    4 & \ALLOC \implies \size \geq n
      & \mbox{3, def.~of~$\size \geq n$}\\
    5 & \ALLOC \implies \bigseparate_{i \in \interval{1}{n}} \lnot \size \geq 2
      & \mbox{multiple applications of}\\[-3pt]
      && \mbox{\ref{rule:starinference}, 2, \ref{starAx:Commute} and \ref{rule:imptr}}\\
    6 & \bigseparate_{i \in \interval{1}{n}} \lnot \size \geq 2 \implies \lnot \size \geq n+1
      & \mbox{$n$ applications of~\ref{starAx:SizeNeg} and~\ref{rule:starinference}}\\
    7 & \ALLOC \implies \lnot \size \geq n+1
      & \mbox{\ref{rule:imptr}, 5, 6}\\
    8 & \ALLOC \implies \size = n
      & \mbox{PC, 4, 7, def.~of $\size = n$}
   \end{syntproof}
   \noindent After deriving
   $\prove_{\starsys}
    \boxstar{\aformula}{\aformulabis} \implies \ALLOC \separate \top$ and
    $\prove_{\starsys} \ALLOC \implies \size = n$,
   the proof is divided in three steps: (1) we  isolate the allocated cells and the garbage, (2) we distribute the alloc and size literals according to the goal $\aformula \separate \aformulabis$ and (3) we add the missing literals.

   \paragraph{\textbf{\em Step 1, isolating allocated cells and garbage}}
   Since $\boxstar{\aformula}{\aformulabis}$ is a conjunction of literals built from core formulae,
   we can rely on~$\maxsize{\boxstar{\aformula}{\aformulabis}}$,
   i.e.~the maximum $\inbound$ among the formulae $\size \geq \inbound$ appearing positively in $\boxstar{\aformula}{\aformulabis}$. First, we show some important properties of~$\boxstar{\aformula}{\aformulabis}$, related to~$\maxsize{\boxstar{\aformula}{\aformulabis}}$.
   \begin{itemize}[align=left]
    \setlength{\itemsep}{3pt}
    \item[\itemlabel{\textbf{A}}{csl:star-proof-item-A}]
      $\maxsize{\boxstar{\aformula}{\aformulabis}} = \maxsize{\aformula} + \maxsize{\aformulabis}$,
    \item[\itemlabel{\textbf{B}}{csl:star-proof-item-B}]
      If there is $\inbound \in \Nat$ such that
      $\lnot \size \geq \inbound \inside \boxstar{\aformula}{\aformulabis}$, then
      \begin{center}
        \hfill
        $\lnot \size \geq \maxsize{\aformula}+1 \inside {\aformula}$,
        \hfill $\lnot \size \geq \maxsize{\aformulabis}+1 \inside {\aformulabis}$.
        \hfill\,
      \end{center}
    \item[\itemlabel{\textbf{C}}{csl:star-proof-item-C}]
      If there is $\inbound \in \Nat$ such that
      $\lnot \size \geq \inbound \inside \boxstar{\aformula}{\aformulabis}$, then
      \begin{center}
      $\lnot \size \geq \maxsize{\boxstar{\aformula}{\aformulabis}} + 1 \inside \boxstar{\aformula}{\aformulabis}$.
      \end{center}
   \end{itemize}
   \begin{proof}[Proof of~{\rm\ref{csl:star-proof-item-A}}]
      By definition of $\maxsize{.}$, we know that $\size \geq \maxsize{\aformula} \inside \aformula$ and ${\size \geq \maxsize{\aformulabis}} \inside \aformulabis$.
      By definition of $\boxstar{\aformula}{\aformulabis}$,
      $\size \geq \maxsize{\aformula} + \maxsize{\aformulabis} \inside \boxstar{\aformula}{\aformulabis}$.
      \emph{Ad absurdum}, suppose that $\maxsize{\aformula} + \maxsize{\aformulabis} \neq \maxsize{\boxstar{\aformula}{\aformulabis}}$ and thus, by definition of $\maxsize{.}$, there is $\inbound > \maxsize{\aformula} + \maxsize{\aformulabis}$ such that ${\size \geq \inbound} \inside \boxstar{\aformula}{\aformulabis}$.
      By definition of $\boxstar{\aformula}{\aformulabis}$,
      we conclude that there are $\inbound_1$ and $\inbound_2$ such that
      $\inbound_1 + \inbound_2 = \inbound$,
      $\size \geq \inbound_1 \inside \aformula$ and $\size \geq \inbound_2 \inside \aformulabis$.
      As $\inbound_1 + \inbound_2 > \maxsize{\aformula} + \maxsize{\aformulabis}$,
      either $\inbound_1 > \maxsize{\aformula}$ or $\inbound_2 > \maxsize{\aformulabis}$.
      Let us assume $\inbound_1 > \maxsize{\aformula}$
      (the other case is analogous).
      We have $\size \geq \inbound_1 \inside \aformula$.
      However, this is contradictory,
      since by definition of $\maxsize{.}$ for all $\inbound' > \maxsize{\aformula}$,
      $\size \geq \inbound' \not\inside \aformula$.
      Thus, $\maxsize{\aformula} + \maxsize{\aformulabis} = \maxsize{\boxstar{\aformula}{\aformulabis}}$.
   \end{proof}
  \begin{proof}[Proof of~{\rm\ref{csl:star-proof-item-B}}]
      Let $\inbound \in \Nat$ such that
      $\lnot \size \geq \inbound \inside {\boxstar{\aformula}{\aformulabis}}$.
      By definition of $\boxstar{\aformula}{\aformulabis}$, this implies
      that there are $\inbound_1,\inbound_2 \in \interval{0}{\bound}$
      such that $\inbound = \inbound_1 + \inbound_2 \dotminus 1$, $\lnot \size \geq \inbound_1 \inside \aformula$ and $\lnot \size \geq \inbound_2 \inside \aformulabis$.
      Since $\aformula$ and $\aformulabis$ are satisfiable,
      by definition of~$\maxsize{.}$, we derive that
      $\inbound_1 > \maxsize{\aformula}$ and $\inbound_2 > \maxsize{\aformulabis}$.
      This implies that the core formula $\size \geq \maxsize{\aformula}+1$ belongs to $\coreformulae{\asetvar}{\bound}$ and, analogously,
      that the core formula  $\size \geq \maxsize{\aformulabis}+1$ belongs to $\coreformulae{\asetvar}{\bound}$.
      Since $\aformula$ is in $\coretype{\asetvar}{\bound}$,
      this implies that $\size \geq \maxsize{\aformula}+1$ is an atomic formula appearing in $\aformula$. By definition of $\maxsize{\aformula}$, the formula cannot appear positively,
      i.e.~$\lnot \size \geq \maxsize{\aformula}+1 \inside \aformula$.
      Analogously, $\aformulabis$ is in $\coretype{\asetvar}{\bound}$,
      which leads to~$\lnot \size \geq \maxsize{\aformulabis}+1 \inside \aformulabis$.
  \end{proof}
  \begin{proof}[Proof of~{\rm\ref{csl:star-proof-item-C}}]
      Directly from~\ref{csl:star-proof-item-A} and~\ref{csl:star-proof-item-B}.
      Indeed, by definition of $\boxstar{\aformula}{\aformulabis}$,
      we know that
      for every ${\lnot \size \geq \inbound \inside \aformula}$ and every $\lnot \size \geq \inbound' \inside \aformulabis$,
      $\lnot \size \geq \inbound+\inbound'\dotminus1 \inside \boxstar{\aformula}{\aformulabis}$.
  \end{proof}

   \noindent
   Now, let us consider ${\inbound_{g} = \maxsize{\boxstar{\aformula}{\aformulabis}} \dotminus n}$.
   We define the formula $\GARB$ below:
    $$
    \GARB \egdef
    \begin{cases}
      \size = \inbound_{g}
      &\text{if}~\lnot \size \geq \inbound \inside {\boxstar{\aformula}{\aformulabis}},~\text{for some}~\inbound\\
      \size \geq \inbound_{g}
      &\text{otherwise},
    \end{cases}
    $$
   where we recall that $\size = \inbound_{g}$
   stands for $\size \geq \inbound_{g} \land \lnot (\size \geq \inbound_{g} + 1)$.
   Notice that $\GARB$ is a conjunction of literals where at least one $\size\geq \inbound$ occurs positively (i.e.~$\size \geq 0$).
   The objective of this step of the proof is to show that $\prove_{\starsys}\boxstar{\aformula}{\aformulabis}\implies\ALLOC \separate \GARB$.
   First, we focus on the positive part of $\GARB$, and prove
   ${\prove_{\starsys}\boxstar{\aformula}{\aformulabis}\implies\ALLOC \separate \size \geq \inbound_{g}}$.
   If ${\inbound_{g} = 0}$ then $\size \geq \inbound_{g} = \true$ and we have already shown
   $\prove_{\starsys}
   \boxstar{\aformula}{\aformulabis} \implies {\ALLOC \separate \top}$.
   So, let us assume that $\inbound_{g} > 1$. Notice that then
   $\maxsize{\boxstar{\aformula}{\aformulabis}} \dotminus n  = \maxsize{\boxstar{\aformula}{\aformulabis}} - n$.
   We have
   \begin{syntproof}
      1   & \true \implies \size \geq \inbound_{g} \lor \lnot \size \geq \inbound_{g}
          & \mbox{PC}\\
      2   & \ALLOC \separate \true \implies \ALLOC \separate (\size \geq \inbound_{g} \lor \lnot \size \geq \inbound_{g})
          & \mbox{\ref{rule:starinference}, \ref{starAx:Commute}, 1}\\
      3   & \ALLOC \separate (\size \geq \inbound_{g} \lor \lnot \size \geq \inbound_{g}) \implies\\[-2pt]
          &(\ALLOC \separate \size \geq \inbound_{g}) \lor (\ALLOC \separate \lnot \size \geq \inbound_{g})
          & \mbox{\ref{starAx:DistrOr}, \ref{starAx:Commute}}\\
      4   & \ALLOC \implies \lnot \size \geq n+1
          & \mbox{Previously derived}\\
      5   & \ALLOC \separate \lnot \size \geq \inbound_{g} \implies (\lnot \size \geq n+1) \separate \lnot \size \geq \inbound_{g}
          & \mbox{\ref{rule:starinference}, 4}\\
      6   & (\lnot \size \geq n+1) \separate \lnot \size \geq \inbound_{g} \implies \lnot \size \geq \maxsize{\boxstar{\aformula}{\aformulabis}}
          & \mbox{\ref{starAx:SizeNeg}, def.~of~$\inbound_{g}$}\\
      7   & \ALLOC \separate \true \implies (\ALLOC \separate \size \geq \inbound_{g})
            \lor \lnot \size \geq \maxsize{\boxstar{\aformula}{\aformulabis}}
          & \mbox{PC, 2, 3, 5, 6}\\
      8   & \boxstar{\aformula}{\aformulabis} \implies \size \geq \maxsize{\boxstar{\aformula}{\aformulabis}}
          & \mbox{PC, def.~of~$\maxsize{.}$}\\
      9   & \boxstar{\aformula}{\aformulabis} \implies \ALLOC \separate \true
          & \mbox{Previously derived}\\
      10  & \boxstar{\aformula}{\aformulabis} \implies (\ALLOC \separate \size \geq \inbound_{g})
      \lor \lnot \size \geq \maxsize{\boxstar{\aformula}{\aformulabis}}
          & \mbox{\ref{rule:imptr}, 7, 9}\\
      11  & \boxstar{\aformula}{\aformulabis} \implies \ALLOC \separate \size \geq \inbound_{g}
          & \mbox{PC, 8, 10}
   \end{syntproof}
   \noindent If for every $\inbound$, $\lnot \size \geq \inbound \not\inside {\boxstar{\aformula}{\aformulabis}}$, then by definition of $\GARB$ we
   conclude that
  $$
  \vdash_{\coresys(\separate)} \boxstar{\aformula}{\aformulabis} \implies \ALLOC \separate \GARB.
  $$
   Otherwise,
   suppose that there is $\inbound$ such that ${\lnot \size \geq \inbound} \inside {\boxstar{\aformula}{\aformulabis}}$.
   So, $\GARB$ is defined as $\size \geq \inbound_{g} \land \lnot (\size \geq \inbound_{g} + 1)$.
   Directly from~\ref{csl:star-proof-item-C}, we know that
   $\lnot \size \geq \maxsize{\boxstar{\aformula}{\aformulabis}} + 1 \inside {\boxstar{\aformula}{\aformulabis}}$.
    By propositional reasoning,
    $$
    \vdash_{\coresys(\separate)} \textstyle\boxstar{\aformula}{\aformulabis} \implies \lnot \size \geq \maxsize{\boxstar{\aformula}{\aformulabis}} + 1.
    $$
    Then, $\boxstar{\aformula}{\aformulabis} \implies \ALLOC \separate \GARB$ is derived as follows:

  \begin{syntproof}
    1 & \size \geq \inbound_{g} \implies
        (\size \geq \inbound_{g} \land \size \geq \inbound_{g}+1)
        \lor \size = \inbound_{g}
      & \mbox{PC, def. of~$\size = \inbound_{g}$}\\
    2 & \ALLOC \separate \size \geq \inbound_{g} \implies\\[-2pt]
      & \ALLOC {\separate} \big( (\size \geq \inbound_{g} \land \size \geq \inbound_{g}{+}1)
      \lor \size = \inbound_{g} \big)
      & \mbox{\ref{rule:starinference}, \ref{starAx:Commute}, 1}\\
    3 & \ALLOC {\separate} \big( (\size \geq \inbound_{g} \land \size \geq \inbound_{g}{+}1)
    \lor \size = \inbound_{g} \big)\\[-2pt]
     & \implies \big(\ALLOC \separate (\size \geq \inbound_{g} \land \size \geq \inbound_{g}{+}1)\big) {\lor} \big(\ALLOC \separate \size = \inbound_{g}\big)
      & \mbox{\ref{starAx:DistrOr}, \ref{starAx:Commute}}\\
    4 & \size \geq \inbound_{g} \land \size \geq \inbound_{g}+1 \implies \size \geq \inbound_{g}+1
      & \mbox{PC}\\
    5 & \ALLOC \implies \size \geq n
      & \mbox{Previously derived}\\
    6 & \ALLOC \separate (\size \geq \inbound_{g} {\land} \size \geq \inbound_{g}{+}1) \implies \size \geq n \separate \size \geq \inbound_{g}+1
      & \mbox{\ref{rule:starintroLR}, 4, 5}\\
    7 & \size \geq n \separate \size \geq \inbound_{g}+1
        \implies \size \geq \maxsize{\boxstar{\aformula}{\aformulabis}}+1
      & \mbox{\ref{starAx:Assoc}, def.~of~$\size \geq \inbound$}\\
    8 & \ALLOC \separate \size \geq \inbound_{g} \implies
      \size \geq \maxsize{\boxstar{\aformula}{\aformulabis}}+1\\
      & \quad \lor \big(\ALLOC \separate \size = \inbound_{g} \big)
      & \mbox{PC, 2, 3, 6, 7}\\
    9 & \boxstar{\aformula}{\aformulabis} \implies \ALLOC \separate \size \geq \inbound_{g}
      & \mbox{Previously derived}\\
    10  & \boxstar{\aformula}{\aformulabis} \implies \size \geq \maxsize{\boxstar{\aformula}{\aformulabis}}+1 {\lor} \big(\ALLOC \separate \size = \inbound_{g} \big)
        & \mbox{\ref{rule:imptr}, 8, 9}\\
    11 & \boxstar{\aformula}{\aformulabis} \implies \lnot \size \geq \maxsize{\boxstar{\aformula}{\aformulabis}}+1
        & \mbox{PC, see above}\\
    12 & \boxstar{\aformula}{\aformulabis} \implies \big(\ALLOC \separate \underbrace{\size = \inbound_{g}}_{\GARB}\big)
        & \mbox{PC, 10, 11}
  \end{syntproof}

   \paragraph{\textbf{\em Step 2, distributing alloc and size literals}}
   In this step, we aim at showing that
   $$\prove_{\starsys} \ALLOC \separate \GARB \implies \aformula^{(1)} \separate \aformulabis^{(1)}$$
   where $\aformula^{(1)}$ and $\aformulabis^{(1)}$ are two formulae defined as follows:
   $$
      \aformula^{(1)} \egdef
      \begin{cases}
        \size  =  \maxsize{\aformula} \land
        \bigwedge \{\alloc{\avariable_i} \inside \aformula \mid i \in \interval{1}{n}\}
        & \text{if}~\maxsize{\aformula} < \bound\\[3pt]
        \size  \geq  \maxsize{\aformula} \land
        \bigwedge \{\alloc{\avariable_i} \inside \aformula \mid i \in \interval{1}{n}\}
        & \text{otherwise}
      \end{cases}
  $$

  $$
      \aformulabis^{(1)} \egdef
      \begin{cases}
        \size  =  \maxsize{\aformulabis} \land \bigwedge \{\alloc{\avariable_i} \inside \aformulabis \mid i \in \interval{1}{n}\}
        & \text{if}~\maxsize{\aformulabis} < \bound\\[3pt]
        \size  \geq  \maxsize{\aformulabis} \land \bigwedge \{\alloc{\avariable_i} \inside \aformulabis \mid i \in \interval{1}{n}\}
        & \text{otherwise}
      \end{cases}
   $$
   We use the notations $\aformula^{(1)}$ and $\aformulabis^{(1)}$  since later in the proof,
   we shall consider sequences of formulae  $\aformula^{(1)}, \ldots, \aformula^{(k)}$
   and $\aformulabis^{(1)}, \ldots, \aformulabis^{(k)}$ with increasing amount of literals. 
   That is why, using  $\aformula^{(1)}$ and $\aformulabis^{(1)}$ at this early stage is meaningful. 
   Before tackling this derivation, a few more steps are required.
   First of all, notice that,
   if there is a formula $\alloc{\avariable}$ occurring both in $\aformula$ and $\aformulabis$,
   then, by definition of $\boxstar{\aformula}{\aformulabis}$,
   $\avariable\neq\avariable$ occurs in $\boxstar{\aformula}{\aformulabis}$.
   This contradicts the fact that $\boxstar{\aformula}{\aformulabis}$ is satisfiable.
   Therefore, we derive that the set of variables $\avariable_1,\dots,\avariable_n$ can be split into
   two  disjoint subsets, the ones ``allocated''
   in $\aformula$, and the others in $\aformulabis$. Let $n_\aformula$ (resp.~$n_\aformulabis$) denote the number
   of equivalence classes of variables allocated in $\aformula$ (resp.~$\aformulabis$).
   Clearly, $n = n_\aformula + n_\aformulabis$.
   Moreover, since $\aformula$ and $\aformulabis$ are satisfiable core types in $\coretype{\asetvar}{\bound}$, where $\bound \geq \card{\asetvar}$,
   we must have $n_\aformula \leq \maxsize{\aformula}$ and $n_\aformulabis \leq \maxsize{\aformulabis}$ (see 
   the axiom~\ref{coreAx:AllocSize}).
   By~\ref{csl:star-proof-item-A}, we conclude that $n \leq \maxsize{\boxstar{\aformula}{\aformulabis}}$.
   We define the following formulae
   $$
   \begin{aligned}
    &\ALLOC(\aformula)  \egdef  \bigseparate \{\alloc{\avariable_i}\wedge\size=1\mid\alloc{\avariable_i}
     \inside \aformula,\, i \in \interval{1}{n}\}\\
    &\GARB(\aformula) \egdef
    \begin{cases}
      \size = \maxsize{\aformula} - n_\aformula
      &\text{if}~\maxsize{\aformula} \,{<}\, \bound\\[2pt]
      \size \geq \maxsize{\aformula} - n_\aformula
      &\text{otherwise}
    \end{cases}
    \end{aligned}
    $$
    Notice that, since $\maxsize{\aformula} \geq n_\aformula$, the formula $\GARB(\aformula)$ is well-defined.
    The formulae~$\ALLOC(\aformulabis)$ and~$\GARB(\aformulabis)$ are defined accordingly.
    Obviously, $\ALLOC$ is equal to
    $\ALLOC(\aformula)  \separate \ALLOC(\aformulabis)$ modulo associativity and commutativity
    of the separating conjunction $\separate$.
    Hence, by taking advantage of the axioms~\ref{starAx:Commute} and ~\ref{starAx:Assoc}, we have
    $$
    \prove_{\starsys} \ALLOC \iff \ALLOC(\aformula)  \separate \ALLOC(\aformulabis).
    $$

    \noindent Let us now look at $\GARB(\aformula)$ and $\GARB(\aformulabis)$.
    We aim at deriving
    $$\prove_{\starsys}
    \GARB\implies\GARB(\aformula) \separate \GARB(\aformulabis).$$
    Since $\aformula$ is a core type, we know that if $\maxsize{\aformula} < \bound$ then,
    by definition of $\maxsize{\aformula}$,
    $\lnot \size \geq \maxsize{\aformula} + 1 \inside \aformula$.
    A similar analysis can be done for $\aformulabis$, which leads to the two following equivalences, by definition of $\GARB(\aformula)$ and $\GARB(\aformulabis)$:
    \begin{itemize}
      \item $\lnot \size \geq \maxsize{\aformula} + 1 \inside \aformula$ if and only if
            $\GARB(\aformula) = (\size = \maxsize{\aformula} {-} n_\aformula)$,
      \item $\lnot \size \geq \maxsize{\aformulabis} + 1 \inside \aformulabis$ if and only if
            $\GARB(\aformulabis) = (\size = \maxsize{\aformulabis} {-} n_\aformulabis)$.
    \end{itemize}
    By definition of $\GARB$,~\ref{csl:star-proof-item-B}
    and~\ref{csl:star-proof-item-C},
    we know that $\GARB = (\size = \maxsize{\boxstar{\aformula}{\aformulabis}} \dotminus n)$
    holds if and only if
    $\lnot \size \geq \maxsize{\aformula} + 1 \inside \aformula$ and
    ${\lnot \size \geq \maxsize{\aformula} + 1 \inside \aformulabis}$.
    From $n \leq \maxsize{\boxstar{\aformula}{\aformulabis}}$ and
    by relying on the previous two equivalences, this allows us to conclude that:
    \begin{itemize}[align=left]
      \item[\itemlabel{\textbf{D}}{csl:star-proof-item-D}]
      $\GARB(\aformula) = (\size = \maxsize{\aformula} {-} n_\aformula)$ and  $\GARB(\aformulabis) = (\size = \maxsize{\aformulabis} {-} n_\aformulabis)$
      if and only if
      $\GARB = (\size = \maxsize{\boxstar{\aformula}{\aformulabis}} - n)$.
    \end{itemize}
    To show $\vdash_{\coresys(\separate)} \GARB \implies (\GARB(\aformula) \separate \GARB(\aformulabis))$,
    we split the proof depending on whether
    $\GARB(\aformula) = (\size = \maxsize{\aformula} {-} n_\aformula)$
    and $\GARB(\aformulabis) = (\size = \maxsize{\aformulabis} {-} n_\aformulabis)$ hold.
    \begin{description}
      \item[case: {\normalfont$\GARB(\aformula) \neq (\size = \maxsize{\aformula} {-} n_\aformula)$}
      and {\normalfont$\GARB(\aformulabis) \neq (\size = \maxsize{\aformulabis} {-} n_\aformulabis)$}]~

      \noindent
      We have $\GARB(\aformula) = (\size {\geq} \maxsize{\aformula} {-} n_\aformula)$
      and $\GARB(\aformulabis) = (\size {\geq} \maxsize{\aformulabis} {-} n_\aformulabis)$.
      By definition of $\GARB$ and \ref{csl:star-proof-item-D},
      $\GARB = (\size \geq \maxsize{\boxstar{\aformula}{\aformulabis}} - n)$.
      By $n = n_\aformula + n_\aformulabis$ and \ref{csl:star-proof-item-A},
      $\maxsize{\boxstar{\aformula}{\aformulabis}} - n = (\maxsize{\aformula} {-} n_\aformula) + (\maxsize{\aformulabis} {-} n_\aformulabis)$.
      By definition of the core formula $\size \geq \inbound$,
      $\GARB$ is already equivalent to $\GARB(\aformula) \separate \GARB(\aformulabis)$,
      modulo associativity and commutativity
      of the separating conjunction $\separate$.
      Hence, by taking advantage of the axioms~\ref{starAx:Commute} and~\ref{starAx:Assoc},
      we have
      $
      \prove_{\starsys} \GARB \implies \GARB(\aformula)  \separate \GARB(\aformulabis).
      $

      \item[case: {\normalfont$\GARB(\aformula) = (\size = \maxsize{\aformula} {-} n_\aformula)$}
      and {\normalfont$\GARB(\aformulabis) \neq (\size = \maxsize{\aformulabis} {-} n_\aformulabis)$}]~

      \noindent We have $\GARB(\aformulabis) = (\size {\geq} \maxsize{\aformulabis} {-} n_\aformulabis)$ and,
      by definition of $\GARB$ and \ref{csl:star-proof-item-D},
      along with $n = n_\aformula + n_\aformulabis$ and \ref{csl:star-proof-item-A},
      $\GARB = (\size \geq (\maxsize{\aformula} {-} n_\aformula) + (\maxsize{\aformulabis} {-} n_\aformulabis))$.
      In this case,
      $
      \GARB \implies \GARB(\aformula)  \separate \GARB(\aformulabis)
      $
      is an instantiation of the following valid formula with $\inbound_1 = \maxsize{\aformula} {-} n_\aformula$ and $\inbound_2 = \maxsize{\aformulabis} - n_\aformulabis$:
      $$
        \size \geq \inbound_1 + \inbound_2 \implies \size = \inbound_1 \separate \size \geq \inbound_2.
      $$
      The derivability of this formula in $\coresys(\separate)$ is proven by induction on $\inbound_1$ 
      (see Appendix~\ref{appendix-derivation-one}). 
      
      \item[case: {\normalfont$\GARB(\aformula) \neq (\size = \maxsize{\aformula} {-} n_\aformula)$}
      and {\normalfont$\GARB(\aformulabis) = (\size = \maxsize{\aformulabis} {-} n_\aformulabis)$}]~

      \noindent
      Analogously to the previous case, we have $\GARB(\aformula) = (\size {\geq} \maxsize{\aformula} {-} n_\aformula)$
      and
      $\GARB = (\size \geq (\maxsize{\aformula} {-} n_\aformula) + (\maxsize{\aformulabis} {-} n_\aformulabis))$.
      We instantiate the theorem
      $$
        \size \geq \inbound_1 + \inbound_2 \implies \size = \inbound_1 \separate \size \geq \inbound_2,
      $$
      shown derivable in the previous case of the proof,
      with $\inbound_1 = \maxsize{\aformulabis} {-} n_\aformulabis$ and $\inbound_2 = \maxsize{\aformula} - n_\aformula$.
      This corresponds to $\GARB \implies \GARB(\aformulabis)  \separate \GARB(\aformula)$.
      Afterwards, by commutativity of the separating conjunction (axiom~\ref{starAx:Commute}) and propositional reasoning, we conclude that $ \prove_{\starsys} \GARB \implies \GARB(\aformula) \separate \GARB(\aformulabis)$.

      \item[case: {\normalfont$\GARB(\aformula) = (\size = \maxsize{\aformula} {-} n_\aformula)$}
      and {\normalfont$\GARB(\aformulabis) = (\size = \maxsize{\aformulabis} {-} n_\aformulabis)$}]~

      \noindent
      By \ref{csl:star-proof-item-D}, $n = n_\aformula + n_\aformulabis$ and \ref{csl:star-proof-item-A},
      $\GARB = (\size = (\maxsize{\aformula} {-} n_\aformula) + (\maxsize{\aformulabis} {-} n_\aformulabis))$.
      In this case,
      $
      \GARB \implies \GARB(\aformula)  \separate \GARB(\aformulabis)
      $
      is an instantiation of the following valid formula, with $\inbound_1 = \maxsize{\aformula} {-} n_\aformula$ and $\inbound_2 = \maxsize{\aformulabis} - n_\aformulabis$:
      $$
        \size = \inbound_1 + \inbound_2 \implies \size = \inbound_1 \separate \size = \inbound_2.
      $$
      Its derivation in  $\coresys(\separate)$ can be found in Appendix~\ref{appendix-derivation-two}.
    \end{description}
    Thanks to the case analysis above, we conclude that $\prove_{\starsys} \GARB \implies \GARB(\aformula) \separate \GARB(\aformulabis)$.
    Thus, $\prove_{\starsys} \ALLOC \separate \GARB \implies (\ALLOC(\aformula) \separate \GARB(\aformula)) \separate (\ALLOC(\aformulabis) \separate \GARB(\aformulabis))$.
    Indeed,
    \begin{syntproof}
      1 & \ALLOC \implies \ALLOC(\aformula) \separate \ALLOC(\aformulabis)
        & \mbox{Previously derived}\\
      2 & \GARB \implies \GARB(\aformula) \separate \GARB(\aformulabis)
        & \mbox{Previously derived}\\
      3 & \ALLOC \separate \GARB \implies (\ALLOC(\aformula) \separate \ALLOC(\aformulabis)) \separate (\GARB(\aformula) \separate \GARB(\aformulabis))
        & \mbox{\ref{rule:starintroLR}, 1, 2}\\
      4 &  (\ALLOC(\aformula) \separate \ALLOC(\aformulabis)) \separate (\GARB(\aformula) \separate \GARB(\aformulabis)) \implies\\[-2pt]
        & (\ALLOC(\aformula) \separate \GARB(\aformula)) \separate (\ALLOC(\aformulabis) \separate \GARB(\aformulabis))
        & \mbox{\ref{starAx:Commute}, \ref{starAx:Assoc}}\\
      5 & \ALLOC \separate \GARB \implies
       (\ALLOC(\aformula) \separate \GARB(\aformula)) \separate
      (\ALLOC(\aformulabis) \separate \GARB(\aformulabis))
        & \mbox{\ref{rule:imptr}, 3, 4}
    \end{syntproof}

    \noindent
    To conclude this step of the proof,
    it is sufficient to show ${\prove_{\starsys} \ALLOC(\aformula) \separate \GARB(\aformula) \implies \aformula^{(1)}}$
    and $\prove_{\starsys} \ALLOC(\aformulabis) \separate \GARB(\aformulabis) \implies \aformulabis^{(1)}$.
    Indeed, by relying on the rule \ref{rule:starintroLR},
    we then obtain $\prove_{\starsys} \ALLOC \separate \GARB \implies \aformula^{(1)} \separate \aformulabis^{(1)}$.
    Below, we show $\prove_{\starsys} \ALLOC(\aformula) \separate \GARB(\aformula) \implies \aformula^{(1)}$.
    The developments of $\prove_{\starsys} \ALLOC(\aformulabis) \separate \GARB(\aformulabis) \implies \aformulabis^{(1)}$ are analogous.
    We recall that the formula $\ALLOC(\aformula)$ is defined as
    $$
      \ALLOC(\aformula) \, = \, \bigseparate \{\alloc{\avariable_i}\wedge\size=1\mid\alloc{\avariable_i} \inside \aformula\}.
    $$
    First of all, let us show that $ \prove_{\starsys} \ALLOC(\aformula) \separate \true \implies \bigwedge \{\alloc{\avariable_i} \inside \aformula \mid i \in \interval{1}{n}\}$.
    The proof is divided in three cases:
    \begin{description}
      \item[case: $\{\alloc{\avariable_i}\wedge\size=1\mid\alloc{\avariable_i} \inside \aformula\} = \emptyset$]
        In this case, the formula
        we want to derive is syntactically equal to
        $\true \separate \true \implies \true$,
        which is derivable by propositional reasoning.
      \item[case: $\card{\{\alloc{\avariable_i}\wedge\size=1\mid\alloc{\avariable_i} \inside \aformula\}} = 1$]
        In this case,
        the formula we want to derive is syntactically equal to
        $(\alloc{\avariable} \land \size = 1) \separate \true \implies \alloc{\avariable}$.
        Therefore, it is derivable in $\starsys$ by~\ref{starAx:StarAlloc}
        and~\ref{rule:starinference}.
      \item[case: $\card{\{\alloc{\avariable_i}\wedge\size=1\mid\alloc{\avariable_i} \inside \aformula\}} {\geq} 2$]
        In the derivation below, we write $\ALLOC(\aformula)^{-i}$ for
        $\bigseparate \{ \alloc{\avariable_j} \land \size = 1 \mid j \in \interval{1}{n} \setminus \{i\}, \alloc{\avariable_j} \inside \aformula\}$.
        Roughly speaking, $\ALLOC(\aformula)^{-i}$ is obtained from $\ALLOC(\aformula)$ by removing the
        subformula $\alloc{\avariable_i} \land \size = 1$.
        Since $\card{\{\alloc{\avariable_i}\wedge\size=1\mid\alloc{\avariable_i} \inside \aformula\}} {\geq} 2$,
        the formula
        $\ALLOC(\aformula)^{-i}$ is different from $\true$.
        We have
        \begin{syntproof}
          1 & \ALLOC(\aformula) \separate \true \implies\\[-2pt]
            & \quad (\alloc{\avariable_i} \land \size = 1) \separate (\ALLOC(\aformula)^{-i} \separate \true)
            & \mbox{\ref{starAx:Commute}, \ref{starAx:Assoc}, def. of $\ALLOC(\aformula)$}\\[-3pt]
            && \mbox{where $\alloc{\avariable_i} \inside \aformula$ and $i \in \interval{1}{n}$}\\
          2 & \ALLOC(\aformula)^{-i} \separate \true \implies \true
            & \mbox{PC}\\
          3 & \alloc{\avariable_i} \land \size = 1 \implies \alloc{\avariable_i}
            & \mbox{PC}\\
          4 & (\alloc{\avariable_i} \land \size = 1) \separate (\ALLOC(\aformula)^{-i} \separate \true)
              \implies\\[-2pt]
            & \quad \alloc{\avariable_i} \separate \true
            & \mbox{\ref{rule:starintroLR}, 2, 3}\\
          5 & \alloc{\avariable_i} \separate \true \implies \alloc{\avariable_i}
            & \mbox{\ref{starAx:StarAlloc}}\\
          6 & \ALLOC(\aformula) \separate \true  \implies \alloc{\avariable_i}
            & \mbox{\ref{rule:imptr}, 1, 4, 5}\\
          7 & \ALLOC(\aformula) \separate \true \implies \bigwedge \{\alloc{\avariable_i} \inside \aformula \mid i \in \interval{1}{n}\}
            & \mbox{PC, repeating 6}\\[-3pt]
            && \mbox{for all $i \in \interval{1}{n}$ such that~$\alloc{\avariable_i} \inside \aformula$}
        \end{syntproof}
    \end{description}
    So, we have $ \prove_{\starsys} \ALLOC(\aformula) \separate \true \implies \bigwedge \{\alloc{\avariable_i} \inside \aformula \mid i \in \interval{1}{n}\}$.

    \noindent Now, recall that $\card{\{ i \in \interval{1}{n} \mid \alloc{\avariable_i} \inside \aformula\}} = n_\aformula$.
    At the beginning of the proof, we have shown a derivation of $\vdash_{\starsys} \ALLOC \implies \size = n$,
    where
    $\ALLOC$ is defined as $\bigseparate \{\alloc{\avariable_i}\wedge\size=1\mid
    i \in \interval{1}{n}\}$.
    Replacing $\ALLOC$ by $\ALLOC(\aformula)$ and $n$ by $n_\aformula$
    in the derivation of
    $\ALLOC \implies \size = n$
    leads to a derivation in $\starsys$ of $\ALLOC(\aformula) \implies \size = n_{\aformula}$.

    To show $\prove_{\starsys} \ALLOC(\aformula) \separate \GARB(\aformula) \implies \aformula^{(1)}$, we split the proof in two cases:
    \begin{description}
      \item[case: $\maxsize{\aformula} = \bound$]
        By definition of $\aformula^{(1)}$ and $\GARB(\aformula)$, we have:
        \vspace{3pt}
        \begin{itemize}
          \setlength{\itemsep}{3pt}
          \item $\aformula^{(1)} \, = \, \size \geq \maxsize{\aformula} \land
          \bigwedge \{\alloc{\avariable_i} \inside \aformula \mid i \in \interval{1}{n}\}$,
          \item $\GARB(\aformula) \, = \, \size \geq \maxsize{\aformula} - n_\aformula$,
        \end{itemize}
        \vspace{3pt}
        Then,
        \begin{syntproof}
          1 & \ALLOC(\aformula) \separate \true \implies \bigwedge \{\alloc{\avariable_i} \inside \aformula \mid i \in \interval{1}{n}\}
            & \mbox{Previously derived}\\
          2 & \GARB(\aformula) \implies \true
            & \mbox{PC}\\
          3 & \ALLOC(\aformula) \separate \GARB(\aformula) \implies \ALLOC(\aformula) \separate \true
            & \mbox{\ref{rule:starinference}, \ref{starAx:Commute}, 2}\\
          4 & \ALLOC(\aformula) \separate \GARB(\aformula) \implies \bigwedge \{\alloc{\avariable_i} \inside \aformula \mid i \in \interval{1}{n}\}
            & \mbox{\ref{rule:imptr}, 1, 3}\\
          5 & \ALLOC(\aformula) \implies \size = n_{\aformula}
            & \mbox{See above}\\
          6 & \size = n_{\aformula} \implies \size \geq n_{\aformula}
            & \mbox{PC, def.~of~$\size = n_{\aformula}$}\\
          7 & \ALLOC(\aformula) \implies \size \geq n_{\aformula}\\
          8 & \GARB(\aformula) \implies \size \geq \maxsize{\aformula}-n_{\aformula}
            & \mbox{PC, def.~of~$\GARB(\aformula)$}\\
          9  & \ALLOC(\aformula) {\separate} \GARB(\aformula)
              \implies \size \geq n_{\aformula} \separate \size \geq \maxsize{\aformula}{-}n_{\aformula}
            & \mbox{\ref{rule:starintroLR}, 7, 8}\\
          10 & \size \geq n_{\aformula} \separate \size \geq \maxsize{\aformula}-n_{\aformula} \implies \size \geq \maxsize{\aformula}
            & \mbox{\ref{starAx:Assoc}, \ref{starAx:Commute},
            def.~of~$\size \geq \inbound$}\\
          11 &  \ALLOC(\aformula) \separate \GARB(\aformula)
              \implies  \size \geq \maxsize{\aformula}
            & \mbox{\ref{rule:imptr}, 9, 10}\\
          12 & \ALLOC(\aformula) \separate \GARB(\aformula) \implies \aformula^{(1)}
            & \mbox{PC, 4, 11, def.~of~$\aformula^{(1)}$}
        \end{syntproof}
      \item[case: $\maxsize{\aformula} \neq \bound$]
        In this case, $\maxsize{\aformula} < \bound$ and so we have:
        \vspace{3pt}
        \begin{itemize}
          \setlength{\itemsep}{3pt}
          \item $\aformula^{(1)} \, = \, \size = \maxsize{\aformula} \land
          \bigwedge \{\alloc{\avariable_i} \inside \aformula \mid i \in \interval{1}{n}\}$,
          \item $\GARB(\aformula) \, = \, \size = \maxsize{\aformula} - n_\aformula$,
        \end{itemize}
        We can rely on the previous case of the proof in order to show that
        $$
          \prove_{\starsys} \ALLOC(\aformula) \separate \GARB(\aformula) \implies \size \geq \maxsize{\aformula} \land
          \bigwedge \{\alloc{\avariable_i} \inside \aformula \mid i \in \interval{1}{n}\}.
          $$
        By propositional reasoning, we can derive
        $\prove_{\starsys} \ALLOC(\aformula) \separate \GARB(\aformula) \implies \aformula^{(1)}$
        as soon as we show that
        $\prove_{\starsys} \ALLOC(\aformula) \separate \GARB(\aformula) \implies \lnot \size \geq \maxsize{\aformula}+1$,
        as we do now:

        \begin{syntproof}
          1 &   \ALLOC(\aformula) \implies \size = n_{\aformula}
            &   \mbox{Already discussed above}\\
          2 &   \size = n_{\aformula} \implies \lnot \size \geq n_{\aformula} + 1
            &   \mbox{PC, def.~of $\size = n_{\aformula}$}\\
          3 &   \ALLOC(\aformula) \implies \lnot \size \geq n_{\aformula}+1
            &   \mbox{PC, \ref{rule:imptr}, 1, 2}\\
          4 &   \GARB(\aformula) \implies \lnot \size \geq \maxsize{\aformula} - n_{\aformula} + 1
            &   \mbox{PC, def. of $\size = \inbound$}\\
          5 &   \ALLOC(\aformula) \separate \GARB(\aformula) \implies\\[-2pt]
            &   \qquad
            \lnot \size \geq n_{\aformula} + 1
            \separate \lnot \size \geq \maxsize{\aformula} - n_{\aformula} + 1
            &   \mbox{\ref{rule:starintroLR}, 3, 4}\\
          6 &   \lnot \size \geq n_{\aformula} + 1
                \separate
                \lnot \size \geq \maxsize{\aformula} - n_{\aformula} + 1
                \implies\\[-2pt]
            &   \qquad \lnot \size \geq \maxsize{\aformula} + 1
            & \mbox{\ref{starAx:SizeNeg}}\\
          7 &  \ALLOC(\aformula) \separate \GARB(\aformula) \implies \lnot \size \geq \maxsize{\aformula} + 1
            & \mbox{\ref{rule:imptr}, 5, 6}
        \end{syntproof}
    \end{description}
    This concludes the proof of $\prove_{\starsys} \ALLOC(\aformula) \separate \GARB(\aformula) \implies \aformula^{(1)}$.
    As already stated, one can analogously show that $\prove_{\starsys} \ALLOC(\aformulabis) \separate \GARB(\aformulabis) \implies \aformulabis^{(1)}$.
    Afterwards,  by~\ref{rule:starintroLR} and from $\prove_{\starsys} \ALLOC \separate \GARB \implies (\ALLOC(\aformula) \separate \GARB(\aformula)) \separate (\ALLOC(\aformulabis) \separate \GARB(\aformulabis))$, we conclude that
    $$\prove_{\starsys} \ALLOC \separate \GARB
    \implies \aformula^{(1)} \separate \aformulabis^{(1)}.$$

   \paragraph{\textbf{\em Step 3, add the missing literals}}
   From the first and second step of the proof, and by propositional reasoning,
   $\prove_{\starsys} \boxstar{\aformula}{\aformulabis} \implies \aformula^{(1)} \separate \aformulabis^{(1)}$.
   We now rely on $\boxstar{\aformula}{\aformulabis}$ to add to $\aformula^{(1)}$ and $\aformulabis^{(1)}$ missing literals from $\aformula$ and $\aformulabis$, respectively.
   We add the literals progressively,
   building a sequence of formulae $\aformula^{(1)}\!\separate \aformulabis^{(1)}$,
   $\aformula^{(2)} \!\separate \aformulabis^{(2)}$, $\dots$, $\aformula^{(k)} \!\separate \aformulabis^{(k)}$, where for all $i \in \interval{1}{k}$,
   $\aformula^{(i)}$ and $\aformulabis^{(i)}$ are conjunctions of core formulae such that
   $\prove_{\starsys} \boxstar{\aformula}{\aformulabis} \implies \aformula^{(i)} \separate \aformulabis^{(i)}$, and for all $j \in \interval{1}{i}$, $\aformula^{(j)} \inside \aformula^{(i)}$ and $\aformulabis^{(j)} \inside \aformulabis^{(i)}$.
   Fundamentally, we obtain $\aformula = \aformula^{(k)}$
   and $\aformulabis = \aformulabis^{(k)}$
   (modulo associativity and commutativity of the classical conjunction),
   which allows us to derive $\prove_{\starsys} \boxstar{\aformula}{\aformulabis} \implies \aformula \separate \aformulabis$, ending the proof.
   Below, we focus on the formula $\aformula^{(i)}$ and $\aformula$.
   Since $\boxstar{\aformula}{\aformulabis}$ is equal to $\boxstar{\aformulabis}{\aformula}$
   (by a quick inspection of the definition)
   and the separating conjunction is commutative (axiom~\ref{starAx:Commute}), a similar analysis can be done for $\aformulabis^{(i)}$ and $\aformulabis$.
   Thus, we assume that $\prove_{\starsys} \boxstar{\aformula}{\aformulabis} \implies \aformula^{(i)} \separate \aformulabis^{(i)}$ holds,
   where in particular $\aformula^{(1)} \inside \aformula^{(i)}$ and $\aformulabis^{(1)} \inside \aformulabis^{(i)}$,
   and that there is a literal $\aliteral \inside \aformula$ that does not appear in $\aformula^{(i)}$.
   By relying on the theorems in Lemma~\ref{lemma:separate-auxiliary-stuff},
   we show that $\prove_{\starsys} \boxstar{\aformula}{\aformulabis} \implies (\aformula^{(i)} \land \aliteral) \separate \aformulabis^{(i)}$
   by a case analysis on $\aliteral$.

  \begin{description}
    \item[case: $\aliteral \, = \, \avariable \sim \avariablebis$, where $\sim \in \{=,\neq\}$]
      By definition of $\boxstar{\aformula}{\aformulabis}$, $\avariable \sim \avariablebis \inside \boxstar{\aformula}{\aformulabis}$.
      \begin{syntproof}
        1 & \boxstar{\aformula}{\aformulabis} \implies \aformula^{(i)} \separate \aformulabis^{(i)}
          & \mbox{Hypothesis}\\
        2 & \boxstar{\aformula}{\aformulabis} \implies \avariable \sim \avariablebis
          & \mbox{PC, def.~of~$\boxstar{\aformula}{\aformulabis}$, see above}\\
        3 & \boxstar{\aformula}{\aformulabis} \implies \avariable \sim \avariablebis \land (\aformula^{(i)} \separate \aformulabis^{(i)})
          & \mbox{PC, 1, 2}\\
        4 & \avariable \sim \avariablebis \land (\aformula^{(i)} \separate \aformulabis^{(i)})
            \implies
            (\aformula^{(i)} \land \avariable \sim \avariablebis) \separate \aformulabis^{(i)}
          & \mbox{\ref{starAx:auxilary-1}}\\
        5 & \boxstar{\aformula}{\aformulabis} \implies (\aformula^{(i)} \land \avariable \sim \avariablebis) \separate \aformulabis^{(i)}
          & \mbox{\ref{rule:imptr}, 3, 4}
      \end{syntproof}
    \item[case: $\aliteral \, = \, \alloc{\avariable}$]
      Since $\alloc{\avariable} \inside \aformula$, by definition,
      $\alloc{\avariable} \inside \boxstar{\aformula}{\aformulabis}$.
      By definition of $\avariable_1,\dots,\avariable_n$,
      there is $j \in \interval{1}{n}$ such that $\avariable_j = \avariable \inside \boxstar{\aformula}{\aformulabis}$.
      Since $\aformula$ is a core type, $\alloc{\avariable_j} \inside \aformula$.
      By definition of $\aformula^{(1)}$, $\alloc{\avariable_j} \inside \aformula^{(1)}$.
      From $\aformula^{(1)} \inside \aformula^{(i)}$, we have
      $\alloc{\avariable_j} \inside \aformula^{(i)}$. Afterwards,

      \begin{syntproof}
        1 & \aformula^{(i)} \implies \aformula^{(i)} \land \alloc{\avariable_j}
          & \mbox{PC, see above}\\
        2 & \boxstar{\aformula}{\aformulabis} \implies \aformula^{(i)} \separate \aformulabis^{(i)}
          & \mbox{Hypothesis}\\
        3 &  \aformula^{(i)} \separate \aformulabis^{(i)} \implies
            ( \aformula^{(i)} \land \alloc{\avariable_j}) \separate \aformulabis^{(i)}
          & \mbox{\ref{rule:starinference}, 1}\\
        4 & \boxstar{\aformula}{\aformulabis} \implies \avariable_j = \avariable
          & \mbox{PC, see above}\\
        5 & \boxstar{\aformula}{\aformulabis} \implies \avariable_j = \avariable \land
              (( \aformula^{(i)} \land \alloc{\avariable_j}) \separate \aformulabis^{(i)})
          & \mbox{PC, 2, 3, 4}\\
        6 & \avariable_j = \avariable \land
        (( \aformula^{(i)} \land \alloc{\avariable_j}) \separate \aformulabis^{(i)})
            \implies ( \aformula^{(i)} \land \alloc{\avariable}) \separate \aformulabis^{(i)}
          & \mbox{\ref{starAx:auxilary-2}}\\
        7 & \boxstar{\aformula}{\aformulabis} \implies
        (( \aformula^{(i)} \land \alloc{\avariable}) \separate \aformulabis^{(i)})
          & \mbox{\ref{rule:imptr}, 5, 6}
      \end{syntproof}
    \end{description}
    Without loss of generality,
    thanks to the derivation above dealing with $\alloc{\avariable}$ literals,
    we now assume that for all $\alloc{\avariable} \inside \aformula$ and all $\alloc{\avariablebis} \inside \aformulabis$,
    we have $\alloc{\avariable} \inside \aformula^{(i)}$ and
    $\alloc{\avariablebis} \inside \aformulabis^{(i)}$.
    \begin{description}
    \item[case: $\aliteral \, = \, \lnot \alloc{\avariable}$]
        We distinguish two main subcases.
        \begin{itemize}
          \item First, assume
            $\lnot \alloc{\avariable} \inside \aformulabis$.
            By definition of $\boxstar{\aformula}{\aformulabis}$, $\lnot \alloc{\avariable} \inside \boxstar{\aformula}{\aformulabis}$.
            \begin{syntproof}
              1 &  \boxstar{\aformula}{\aformulabis} \implies \aformula^{(i)} \separate \aformulabis^{(i)}
              & \mbox{Hypothesis}\\
              2 & \boxstar{\aformula}{\aformulabis} \implies \lnot \alloc{\avariable}
                & \mbox{PC, def.~of~$\boxstar{\aformula}{\aformulabis}$, see above}\\
              3 & \boxstar{\aformula}{\aformulabis} \implies \lnot \alloc{\avariable}
                  \land (\aformula^{(i)} \separate \aformulabis^{(i)})
                & \mbox{PC, 1, 2}\\
              4 & \lnot \alloc{\avariable}
              \land (\aformula^{(i)} \separate \aformulabis^{(i)})
              \implies (\aformula^{(i)}\land \lnot \alloc{\avariable}) \separate \aformulabis^{(i)}
                & \mbox{\ref{starAx:auxilary-4}}\\
              5 & \boxstar{\aformula}{\aformulabis} \implies (\aformula^{(i)}
                  \land \lnot \alloc {\avariable}) \separate \aformulabis^{(i)}
                & \mbox{\ref{rule:imptr}, 3, 4}
            \end{syntproof}

          \item Otherwise, $\alloc{\avariable} \inside \aformulabis$.
                By assumption, $\alloc{\avariable} \inside \aformulabis^{(i)}$.
                \begin{syntproof}
                  1 & \aformulabis^{(i)} \implies \aformulabis^{(i)} \land \alloc{\avariable}
                    & \mbox{PC, see above}\\
                  2 & \boxstar{\aformula}{\aformulabis} \implies \aformula^{(i)} \separate \aformulabis^{(i)}
                    & \mbox{Hypothesis}\\
                  3 & \aformula^{(i)} \separate \aformulabis^{(i)} \implies (\aformulabis^{(i)} \land \alloc{\avariable}) \separate \aformula^{(i)}
                    & \mbox{\ref{starAx:Commute}, \ref{rule:starinference}, 1}\\
                  4 & (\aformulabis^{(i)} \land \alloc{\avariable}) \separate \aformula^{(i)}
                      \implies \aformulabis^{(i)} \separate (\aformula^{(i)} \land \lnot \alloc{\avariable})
                    & \mbox{\ref{starAx:auxilary-3}}\\
                  5 & \aformulabis^{(i)} \separate (\aformula^{(i)} \land \lnot \alloc{\avariable}) \implies (\aformula^{(i)} \land \lnot \alloc{\avariable}) \separate \aformulabis^{(i)}
                    & \mbox{\ref{starAx:Commute}}\\
                  6 & \boxstar{\aformula}{\aformulabis} \implies (\aformula^{(i)} \land \lnot \alloc{\avariable}) \separate \aformulabis^{(i)}
                    & \mbox{\ref{rule:imptr}, 2, 3, 4, 5}
                \end{syntproof}
        \end{itemize}
        \item[case: $\aliteral \, = \, \avariable \Ipto \avariablebis$]
        Similar to the case $\aliteral = \alloc{\avariable}$.
        Since $\aformula$ is a satisfiable core type, we have $\alloc{\avariable} \inside \aformula$ (see axiom~\ref{coreAx:PointAlloc}).
        By assumption, $\alloc{\avariable} \inside \aformula^{(i)}$.
        By definition of~$\boxstar{\aformula}{\aformulabis}$, we have
        ${\avariable \Ipto \avariablebis} \inside \boxstar{\aformula}{\aformulabis}$.
        \begin{syntproof}
          1 & \aformula^{(i)} \implies \aformula^{(i)} \land \alloc{\avariable}
            & \mbox{PC, see above}\\
          2 & \boxstar{\aformula}{\aformulabis} \implies \aformula^{(i)} \separate \aformulabis^{(i)}
            & \mbox{Hypothesis}\\
          3 & \boxstar{\aformula}{\aformulabis} \implies \avariable \Ipto \avariablebis
            & \mbox{PC, see above}\\
          4 & \aformula^{(i)} \separate \aformulabis^{(i)} \implies
              (\aformula^{(i)} \land \alloc{\avariable}) \separate \aformulabis^{(i)}
            & \mbox{\ref{rule:starinference}, 1}\\
          5 & \boxstar{\aformula}{\aformulabis} \implies
              \avariable \Ipto \avariablebis \land ((\aformula^{(i)} \land \alloc{\avariable}) \separate \aformulabis^{(i)})
            & \mbox{PC, 3, 4}\\
          6 & \avariable \Ipto \avariablebis \land ((\aformula^{(i)} \land \alloc{\avariable}) \separate \aformulabis^{(i)})
            \implies (\aformula^{(i)} \land \avariable \Ipto \avariablebis) \separate \aformulabis^{(i)}
            & \mbox{\ref{starAx:auxilary-5}}\\
          7 & \boxstar{\aformula}{\aformulabis} \implies (\aformula^{(i)} \land \avariable \Ipto \avariablebis) \separate \aformulabis^{(i)}
            & \mbox{\ref{rule:starinference}, 5, 6}
        \end{syntproof}
      \end{description}
      Without loss of generality,
      thanks to the previous cases dealing with $\lnot \alloc{\avariable}$ literals,
      below we assume that for every $\lnot \alloc{\avariable} \inside \aformula$ and every $\lnot \alloc{\avariablebis} \inside \aformulabis$,
      we have $\lnot \alloc{\avariable} \inside \aformula^{(i)}$ and
      $\lnot \alloc{\avariablebis} \inside \aformulabis^{(i)}$.
      \begin{description}
      \item[case: $\aliteral \, = \, \lnot \avariable \Ipto \avariablebis$]
        We distinguish two main subcases
        \begin{itemize}
          \item First, suppose $\alloc{\avariable} \inside \aformula$. In this case, by definition of $\boxstar{\aformula}{\aformulabis}$, we have ${\lnot \avariable \Ipto \avariablebis} \inside \boxstar{\aformula}{\aformulabis}$. Therefore,
          \begin{syntproof}
            1 & \boxstar{\aformula}{\aformulabis} \implies \lnot \avariable \Ipto \avariablebis
              & \mbox{PC, see above}\\
            2 & \boxstar{\aformula}{\aformulabis} \implies \aformula^{(i)} \separate \aformulabis^{(i)}
            & \mbox{Hypothesis}\\
            3 & \boxstar{\aformula}{\aformulabis} \implies
            \lnot \avariable \Ipto \avariablebis \land (\aformula^{(i)} \separate \aformulabis^{(i)})
            & \mbox{PC, 1, 2}\\
            4 & \lnot \avariable \Ipto \avariablebis \land (\aformula^{(i)} \separate \aformulabis^{(i)}) \implies (\aformula^{(i)} \land \lnot \avariable \Ipto \avariablebis) \separate \aformulabis^{(i)}
              & \mbox{\ref{starAx:auxilary-6}}
          \end{syntproof}

          \item Otherwise, we have $\lnot \alloc{\avariable} \inside \aformula$. By assumption,
          $\lnot \alloc{\avariable} \inside \aformula^{(i)}$, and thus
          \begin{syntproof}
            1 &  \aformula^{(i)} \implies \lnot \alloc{\avariable}
              & \mbox{PC, see above}\\
            2 & \lnot \alloc{\avariable} \implies \lnot \avariable \Ipto \avariablebis
              & \mbox{\ref{coreAx:PointAlloc}, PC}\\
            3 & \aformula^{(i)} \implies \lnot \avariable \Ipto \avariablebis
              & \mbox{\ref{rule:imptr}, 1, 2}\\
            4 & \aformula^{(i)} \implies \aformula^{(i)} \land \lnot \avariable \Ipto \avariablebis
              & \mbox{PC, 3}\\
            5 & \boxstar{\aformula}{\aformulabis} \implies \aformula^{(i)} \separate \aformulabis^{(i)}
              & \mbox{Hypothesis}\\
            6 & \aformula^{(i)} \separate \aformulabis^{(i)} \implies
                (\aformula^{(i)} \land \lnot \avariable \Ipto \avariablebis) \separate \aformulabis^{(i)}
              & \mbox{\ref{rule:starinference}, 4}\\
            7 & \boxstar{\aformula}{\aformulabis} \implies
                (\aformula^{(i)} \land \lnot \avariable \Ipto \avariablebis) \separate \aformulabis^{(i)}
              & \mbox{\ref{rule:imptr}, 5, 6}
          \end{syntproof}
        \end{itemize}
      \item[case: $\aliteral = \size \geq \inbound$]
          By definition of $\maxsize{.}$, $\inbound \leq \maxsize{\aformula}$.
          By definition of $\aformula^{(1)}$, $\size \geq \maxsize{\aformula} \inside \aformula^{(1)}$. From $\aformula^{(1)} \inside \aformula^{(i)}$, we get $\size \geq \maxsize{\aformula} \inside \aformula^{(i)}$.
          \begin{syntproof}
            1 & \aformula^{(i)} \implies \size \geq \maxsize{\aformula}
              & \mbox{PC, see above}\\
            2 & \size \geq \maxsize{\aformula} \implies \size \geq \inbound
              & \mbox{repeated~\ref{coreAx:Size}, PC, as $\inbound \leq \maxsize{\aformula}$}\\
            3 & \aformula^{(i)} \implies \aformula^{(i)} \land \size \geq \inbound
              & \mbox{PC, 1, 2}\\
            4 & \boxstar{\aformula}{\aformulabis} \implies \aformula^{(i)} \separate \aformulabis^{(i)}
              & \mbox{Hypothesis}\\
            5 & \aformula^{(i)} \separate \aformulabis^{(i)} \implies
                (\aformula^{(i)} \land \size \geq \inbound) \separate \aformulabis^{(i)}
              & \mbox{\ref{rule:starinference}, 3}\\
            6 & \boxstar{\aformula}{\aformulabis} \implies
                (\aformula^{(i)} \land \size \geq \inbound) \separate \aformulabis^{(i)}
              & \mbox{\ref{rule:imptr}, 4, 5}
          \end{syntproof}

      \item[case: $\aliteral = \lnot \size \geq \inbound$]
          In this case, $\maxsize{\aformula} < \bound$.
          Since $\aformula$ is a satisfiable core type, we have $\inbound > \maxsize{\aformula}$.
          Moreover, by definition of $\aformula^{(1)}$,
          $\lnot \size \geq \maxsize{\aformula}+1 \inside \aformula^{(1)}$.
          From $\aformula^{(1)} \inside \aformula^{(i)}$,
          we have
          $\lnot \size \geq \maxsize{\aformula}+1 \inside \aformula^{(i)}$.
          \begin{syntproof}
            1 & \aformula^{(i)} \implies \lnot \size \geq \maxsize{\aformula}{+}1
              & \mbox{PC, see above}\\
            2 & \lnot \size \geq \maxsize{\aformula}{+}1 \implies \lnot \size \geq \inbound
              & \mbox{repeated~\ref{coreAx:Size}, PC, as $\inbound > \maxsize{\aformula}$}\\[-3pt]
              && \mbox{by PC, the contrapositive of~\ref{coreAx:Size} is derivable}\\
            3 & \aformula^{(i)} \implies \aformula^{(i)} \land \lnot \size \geq \inbound
              & \mbox{PC, 1, 2}\\
            4 & \boxstar{\aformula}{\aformulabis} \implies \aformula^{(i)} \separate \aformulabis^{(i)}
              & \mbox{Hypothesis}\\
            5 & \aformula^{(i)} \separate \aformulabis^{(i)} \implies
                (\aformula^{(i)} \land \lnot \size \geq \inbound) \separate \aformulabis^{(i)}
              & \mbox{\ref{rule:starinference}, 3}\\
            6 & \boxstar{\aformula}{\aformulabis} \implies
                (\aformula^{(i)} \land \lnot \size \geq \inbound) \separate \aformulabis^{(i)}
              & \mbox{\ref{rule:imptr}, 4, 5} \hfill\qedhere
          \end{syntproof}
  \end{description}

\end{proof}

\begin{cor}[Star elimination]\label{lemma:starPSLelim}
    Let $\asetvar \subseteq_{\fin} \PVAR$ and $\bound \geq \card{\asetvar}$. 
    Let $\aformula$ and $\aformulabis$ 
    in $\coretype{\asetvar}{\bound}$.
   There is
   $\aformulater$ in $\conjcomb{\coreformulae{\asetvar}{2\bound}}$ such that
   $\prove_{\starsys} \aformula \separate \aformulabis \iff \aformulater$.
\end{cor}

\begin{proof}
    If both $\aformula$ and $\aformulabis$ are satisfiable, the results holds directly 
    by~Lemma~\ref{lemma:starPSLelim-sat}, 
    as $\boxstar{\aformula}{\aformulabis}$ is 
    in $\conjcomb{\coreformulae{\asetvar}{\bound+\bound}}$.
    Otherwise, 
    let us treat the case where one of the two formulas is unsatisfiable.
    For instance, assume that~$\aformula$ is unsatisfiable. 
    Then~$\prove_{\coresys} \aformula\implies \bot$ by completeness of $\coresys$ 
    (Lemma~\ref{lemma:corePSLtwo}) and, $\coresys(\separate)$ includes $\coresys$,~$\prove_{\coresys(\separate)} \aformula\implies \bot$.  
    By 
    the rule~\ref{rule:starinference} and by the axiom~\ref{starAx:False}, we get $\prove_{\starsys}\aformula*\aformulabis\implies\bot$.
    Thus $\aformulater$ can take the value $\lnot (\avariable = \avariable)$. 
    The case where $\aformulabis$ is not satisfiable is analogous, thanks to~\ref{starAx:Commute}.
\end{proof}

By the distributivity axiom~\ref{starAx:DistrOr}, Corollary~\ref{lemma:starPSLelim}
is extended from core types to arbitrary Boolean combinations of core formulae.
$\starsys$ is therefore complete for \slSA.
In order to derive  a valid formula $\aformula\in\slSA$, we repeatedly apply
the elimination of $\separate$  in a bottom-up fashion, starting from the leaves of $\aformula$ (which
are Boolean combinations of core formulae) and obtaining a Boolean combination of core formulae $\aformulabis$ that is equivalent to $\aformula$.
Then, we rely on the completeness of $\coresys$ (Theorem~\ref{theo:corePSLcompl}) to prove that $\aformulabis$ is derivable.

\begin{thm}\label{theo:starCompleteness}
A formula $\aformula$ in \slSA{} is valid iff $\prove_{\starsys} \aformula$.
\end{thm}

\begin{proof}
Soundness of the proof system $\starsys$ has been already established earlier.

As far as the completeness proof is concerned, we need to show that for every formula $\aformula$
in \slSA{}, there is a Boolean combination of core formulae $\aformulabis$ such that
$\prove_{\starsys} \aformula \Leftrightarrow \aformulabis$. In order to conclude the proof,
when $\aformula$ is valid for \slSA{}, by soundness of $\starsys$, we obtain that
$\aformulabis$ is valid too and therefore  $\prove_{\starsys} \aformulabis$
as $\coresys$ is a subsystem of $\starsys$ and $\coresys$ is complete by
Theorem~\ref{theo:corePSLcompl}. By propositional reasoning, we get that
$\prove_{\starsys} \aformula$.

To show that every formula $\aformula$ has a provably equivalent Boolean combination
of core formulae, we heavily rely on Corollary~\ref{lemma:starPSLelim}. The proof is by simple induction
on the number of occurrences of $\separate$ in $\aformula$ that are not involved in the definition of some core
formula of the form $\size \geq \inbound$. For the base case, when $\aformula$
has no occurrence of the separating conjunction, $\avariable = \avariablebis$ and $\avariable \Ipto
\avariablebis$ are already core formulae, and $\emp$ is logically equivalent to $\neg \size \geq 1$.

Before performing the induction step, let us observe that in $\starsys$, the replacement of provably equivalent
formulae holds true, which is stated as follows:

\begin{enumerate}[label=\textbf{R\arabic*}]
\setcounter{enumi}{-1}
\item\label{SC-auxlemmaR0} Let $\aformula, \aformula'$ and $\aformulabis$ be
formulae of \slSA{} such that
$\prove_{\starsys} \aformula \Leftrightarrow \aformula'$. Then,
$$\prove_{\starsys}  \aformulabis[\aformula]_{\rho} \Rightarrow
\aformulabis[\aformula']_{\rho}$$
\end{enumerate}

\noindent 
Above, $\aformulabis[\aformula]_{\rho}$ refers to the formula $\aformulabis$ in which the subformula
at the occurrence $\rho$ (in the standard sense) is replaced by $\aformula$.
($\aformula$ and $\aformula'$ are therefore placed at the same occurrence.)

To prove~\ref{SC-auxlemmaR0}, we first note that the following rules can be shown admissible in $\starsys$:
$$
\inference{\aformula \iff \aformula'}{\neg \aformula \iff \neg \aformula'} \ \ \
\inference{\aformula \iff \aformula'}{\aformula \vee \aformulabis \iff \aformula' \vee \aformulabis} \ \ \
\inference{\aformula \iff \aformula'}{\aformula \wedge \aformulabis \iff \aformula' \wedge \aformulabis}
$$
Admissibility of such rules is a direct consequence of the presence of axioms and modus ponens for the propositional calculus.
As a consequence of the presence of the rule~\ref{rule:starinference} in $\starsys$,  the rule below is also admissible:
$$\inference{\aformula \Leftrightarrow \aformula'}{\aformula \separate \aformulabis \Leftrightarrow \aformula'
\separate \aformulabis}
$$
Consequently, by structural induction on $\aformulabis$, one can conclude that
$\prove_{\starsys} \aformula \Leftrightarrow \aformula'$
 implies $\prove_{\starsys}  \aformulabis[\aformula]_{\rho} \Rightarrow
\aformulabis[\aformula']_{\rho}$
(the axiom~\ref{starAx:Commute} needs to be used here).

Assume that $\aformula$ is a formula in \slSA{} with $n+1$ occurrences of the separating conjunction not involved in the definition of some
$\size \geq \inbound$ ($n \geq 0$).
Let $\aformulabis$ be a subformula of $\aformula$ (at the occurrence $\rho$) of the form $\aformulabis_1 \separate \aformulabis_2$ such that
$\aformulabis_1$ and $\aformulabis_2$ are  Boolean combinations of
core formulae, in $\boolcomb{\coreformulae{\asetvar}{\bound_1}}$ and $\boolcomb{\coreformulae{\asetvar}{\bound_2}}$.
By pure propositional reasoning, one can show that there are formulae in disjunctive normal form
$\aformulabis_1^1 \vee \cdots \vee \aformulabis_1^{n_1}$ and ${\aformulabis_2^1 \vee \cdots \vee \aformulabis_2^{n_2}}$
such that $\vdash_{\coresys} \aformulabis_i \Leftrightarrow \aformulabis_i^1 \vee \cdots \vee \aformulabis_i^{n_i}$
for $i \in \set{1,2}$ and moreover, all the $\aformulabis_i^j$'s are core types in $\coretype{\asetvar}{\max(\card{\asetvar},\bound_1,\bound_2)}$.
Again, by using propositional reasoning but this time  using also the axiom~\ref{starAx:DistrOr} for distributivity, we have
$$
\vdash_{\starsys} \aformulabis_1 \separate \aformulabis_2 \Leftrightarrow
\bigvee_{j_1 \in \interval{1}{n_1}, j_2 \in \interval{1}{n_2}} \aformulabis_1^{j_1} \separate \aformulabis_2^{j_2}.
$$
\noindent
We now rely on Corollary~\ref{lemma:starPSLelim} and derive that there is a 
conjunction of core formulae $\aformulabis^{j_1,j_2}$ in $\conjcomb{\coreformulae{\asetvar}{2\max(\card{\asetvar},\bound_1,\bound_2)}}$ such that 
$\vdash_{\starsys} \aformulabis_1^{j_1} \separate \aformulabis_2^{j_2} \Leftrightarrow \aformulabis^{j_1,j_2}$.
By propositional reasoning, we get
$$
\vdash_{\starsys} \aformulabis_1 \separate \aformulabis_2 \Leftrightarrow
\bigvee_{j_1 \in \interval{1}{n_1}, j_2 \in \interval{1}{n_2}} \aformulabis^{j_1,j_2}.
$$
Consequently (thanks to the property~\ref{SC-auxlemmaR0}), we obtain
$$
\vdash_{\starsys} \aformula \Leftrightarrow \aformula[\bigvee_{j_1 \in \interval{1}{n_1}, j_2 \in \interval{1}{n_2}} \aformulabis^{j_1,j_2}]_{\rho}
$$
Note that the right-hand side formula has  $n$ occurrences of the separating
conjunction that are not involved in the definition of some core
formula of the form $\size \geq \inbound$. The induction hypothesis applies, which concludes the proof.

\end{proof}

\section{A constructive elimination of $\magicwand$ leading to full completeness}\label{section:magicwandelimination}
\begin{figure}
  \begin{footnotesize}

  \fbox{
    \begin{minipage}{0.95\linewidth}
    \vspace{2pt}
    \begin{enumerate}[align=left,leftmargin=*]
    \setlength\itemsep{4pt}
    \addtocounter{enumi}{20}
    \item[\axlab{A^{\magicwand}}{wandAx:Size}]$(\size = 1 \land \bigwedge_{\avariable \in \asetvar}\lnot \alloc{\avariable}) \septraction \true\!\assuming{\asetvar \subseteq_{\fin} \PVAR}$
    \item[\axlab{A^{\magicwand}}{wandAx:PointsTo}] $\lnot \alloc{\avariable} \implies ((\avariable \Ipto \avariablebis \land \size = 1) \septraction \true)$
    \item[\axlab{A^{\magicwand}}{wandAx:Alloc}] $\lnot \alloc{\avariable} \implies ((\alloc{\avariable} \land \size = 1 \land \bigwedge_{\avariablebis \in \asetvar}\lnot \avariable \Ipto \avariablebis ) \septraction \true) \assuming{\asetvar \subseteq_{\fin} \PVAR}$\\ 
    \end{enumerate}
    \hfill
    \rulelab{\textbf{$\separate$-Adj}}{rule:staradj}
      $\inference{\aformula \separate \aformulabis \implies \aformulater}{\aformula \implies (\aformulabis \magicwand \aformulater)}{}$
    \hfill
    \rulelab{\textbf{$\magicwand$-Adj}}{rule:magicwandadj}
      $\inference{\aformula \implies (\aformulabis \magicwand \aformulater)}{\aformula \separate \aformulabis \implies \aformulater}{}$
    \hfill\,
    \end{minipage}
  }
  \end{footnotesize}
  \caption{Additional axioms and rules for handling the separating implication.}
  \label{figure-proof-system-magicwand}
\end{figure}

In order to obtain the final proof system $\magicwandsys$, we add the axioms and rules
from Figure~\ref{figure-proof-system-magicwand} to the proof system $\starsys$. These new axioms
and rules are dedicated to the separating implication. 
The axioms involving~$\septraction$ (kind of dual of $\magicwand$, introduced in Section~\ref{section:preliminaries}) express that it is always
possible to extend a given heap with an extra cell, and that the address and the content of this cell can be fixed arbitrarily
(provided it is not already allocated). The adjunction rules~\ref{rule:staradj} and~\ref{rule:magicwandadj} are from the Hilbert-style axiomatisation
of Boolean \mbox{BI~\cite[Section 2]{Galmiche&Larchey06}}.
One can observe that, in $\magicwandsys$, 
the axioms~\ref{starAx:DistrOr},~\ref{starAx:False} and~\ref{starAx:StarAlloc} 
of~$\coresys(\separate)$
are derivable.

\begin{lem} \label{lemma:admissible-axioms-2}
The axioms~\ref{starAx:DistrOr},~\ref{starAx:False} and~\ref{starAx:StarAlloc} 
are derivable
in $\coresys(\separate,\magicwand)$.
\end{lem}

The derivations of~\ref{starAx:DistrOr},~\ref{starAx:False} 
and~\ref{starAx:StarAlloc} that lead to Lemma~\ref{lemma:admissible-axioms-2} are given in~Appendix~\ref{appendix-admissible-axioms-2}.


Fundamentally, $\magicwandsys$ enjoys the $\magicwand$ elimination property, as shown
below. Actually, we state the property with the help of $\septraction$ as we find the related statements and developments more intuitive. 
\begin{lem}\label{lemma:magicwandPSLelim}
  Let $\asetvar \subseteq_{\fin} \PVAR$ and $\bound \geq \card{\asetvar}$.
  Let $\aformula$ and $\aformulabis$ in $\coretype{\asetvar}{\bound}$. There is a conjunction 
  $\aformulater \in \conjcomb{\coreformulae{\asetvar}{\bound}}$
  such that $\prove_{\magicwandsys} (\aformula\septraction\aformulabis) \iff \aformulater$.
\end{lem}

\begin{proof}[Structure of the proof of~Lemma~\ref{lemma:magicwandPSLelim}]
  \renewcommand{\qedsymbol}{}
In the proof of Lemma~\ref{lemma:magicwandPSLelim}, the formula $\aformulater$
is explicitly constructed from $\aformula$ and $\aformulabis$, following a pattern analogous
to the construction of $\boxstar{.\,}{.}$ in~Figure~\ref{figure:boxstar} (see forthcoming 
Figure~\ref{figure:csl:boxseptra-formula}).
The derivation of the equivalence  $(\aformula\septraction\aformulabis) \iff \aformulater$
is shown as follows. First, the formulae 
$\aformulater \separate \aformula \implies \aformulabis$ and 
$\neg \aformulater \separate \aformula \implies \neg \aformulabis$ are shown valid 
(by using semantical means). As $\starsys$ is complete for $\slSA$, it is a subsystem of
$\magicwandsys$, and the formulae $\aformula$, $\aformulabis$ and $\aformulater$ are Boolean combinations
of core formulae, we get $\prove_{\magicwandsys} \aformulater \separate \aformula \implies \aformulabis$
and $\prove_{\magicwandsys} \neg \aformulater \separate \aformula \implies \neg \aformulabis$.
The latter theorem leads 
to  $\prove_{\magicwandsys} (\aformula\septraction\aformulabis) \implies \aformulater$
by using the definition of $\septraction$ and the rule~\ref{rule:staradj}. 
For the other direction, in order to show that $\prove_{\magicwandsys} \aformulater \implies (\aformula\septraction\aformulabis)$
holds, we take advantage of the admissibility of the theorem~\ref{mwAx:Mix}
(see Lemma~\ref{lemma:septractionadmissible} below)
for which an instance is 
 $(\aformula\septraction \top)\wedge (\aformula\magicwand\aformulabis) \implies 
(\aformula\septraction(\top\wedge\aformulabis))$.
From $\prove_{\magicwandsys} 
\aformulater \separate \aformula \implies \aformulabis$ 
and by~\ref{rule:staradj} we have
 $\prove_{\magicwandsys} 
\aformulater \implies (\aformula \magicwand \aformulabis)$.
Therefore, the main technical development lies in the proof of  
$\prove_{\magicwandsys} \aformulater \implies (\aformula \septraction \top)$,
which allows us to take advantage of~\ref{mwAx:Mix},
and leads to $\prove_{\magicwandsys} \aformulater \implies (\aformula \septraction \aformulabis)$ by propositional reasoning.
\end{proof}

In order to formalise the proof of Lemma~\ref{lemma:magicwandPSLelim} sketched above, we start by establishing 
several admissible axioms and rules (Lemma~\ref{lemma:septractionadmissible}).
Afterwards, we define the formula $\aformulater$ and show
the validity of $\aformulater \separate \aformula \implies \aformulabis$ and 
$\neg \aformulater \separate \aformula \implies \neg \aformulabis$ (Lemma~\ref{lemma:coreformulaseptraction}).
Then, come the final bits of the proof of Lemma~\ref{lemma:magicwandPSLelim} (see page~\pageref{proof-lemma-magicwandPSLelim}). 
 
\cut{
Actually, the definition of the formula $\aformulater$ is provided in the statement of Lemma~\ref{lemma:coreformulaseptraction}. 
In order to prove  Lemma~\ref{lemma:magicwandPSLelim}, we use the auxiliary lemma below.
}

\begin{lem}\label{lemma:septractionadmissible}
The following rules and axioms are admissible in $\magicwandsys$:

\vspace{2pt}

\begin{enumerate}[align=left,leftmargin=*,before=\vspace{3pt},after=\vspace{3pt}]
  \setlength{\itemsep}{4pt}
  \begin{minipage}[t]{0.47\linewidth}
    \item[\indexedlab{I^{\magicwand}}{\ref{lemma:septractionadmissible}}{mwAx:BotL}] $\bot\septraction\aformula\implies\bot$
    \item[\indexedlab{I^{\magicwand}}{\ref{lemma:septractionadmissible}}{mwAx:BotR}] $\aformula\septraction\bot\implies\bot$
    \item[\indexedlab{I^{\magicwand}}{\ref{lemma:septractionadmissible}}{mwAx:Cut}] $\aformula \separate (\aformula\magicwand \aformulabis)\implies \aformulabis$
    \item[\indexedlab{I^{\magicwand}}{\ref{lemma:septractionadmissible}}{mwAx:ImpL}] $\inference{\aformula \implies \aformulabis}{\aformula \septraction \aformulater \implies \aformulabis\septraction \aformulater}{}$
    \item[\indexedlab{I^{\magicwand}}{\ref{lemma:septractionadmissible}}{mwAx:ImpR}] $\inference{\aformula \implies \aformulabis}{\aformulater \septraction \aformula \implies \aformulater\septraction \aformulabis}{}$
    \item[\indexedlab{I^{\magicwand}}{\ref{lemma:septractionadmissible}}{mwAx:Curry}] $\aformula\septraction(\aformulabis\septraction\aformulater) \ \iff\  (\aformula*\aformulabis)\septraction\aformulater$
  \end{minipage}
  \begin{minipage}[t]{0.55\linewidth}
    \item[\indexedlab{I^{\magicwand}}{\ref{lemma:septractionadmissible}}{mwAx:OrL}] $(\aformula\vee\aformulabis)\septraction\aformulater \ \iff\  (\aformula\septraction\aformulater) \vee (\aformulabis\septraction\aformulater)$
    \item[\indexedlab{I^{\magicwand}}{\ref{lemma:septractionadmissible}}{mwAx:OrR}] $\aformulater\septraction(\aformula\vee\aformulabis) \ \iff\  (\aformulater\septraction\aformula) \vee (\aformulater\septraction\aformulabis)$
    \item[\indexedlab{I^{\magicwand}}{\ref{lemma:septractionadmissible}}{mwAx:Mix}] $(\aformula\septraction\aformulabis)\wedge (\aformula\magicwand\aformulater) \implies  (\aformula\septraction\aformulabis\wedge\aformulater)$
    \item[\indexedlab{I^{\magicwand}}{\ref{lemma:septractionadmissible}}{mwAx:SeptEq}] 
    $\avariable = \avariablebis \land (\aformula \septraction \aformulabis)
    \implies (\aformula \land  \avariable = \avariablebis  \septraction \aformulabis)$
    \item[\indexedlab{I^{\magicwand}}{\ref{lemma:septractionadmissible}}{mwAx:SeptIneq}] 
    $\avariable \neq \avariablebis \land (\aformula \septraction \aformulabis)
    \implies (\aformula \land  \avariable \neq \avariablebis \septraction \aformulabis)$
    \item[\indexedlab{I^{\magicwand}}{\ref{lemma:septractionadmissible}}{mwAx:SizeLiterals}] $(\aformula_{\size} 
      \land \bigwedge_{\avariable \in \asetvar}\lnot \alloc{\avariable}) 
  \septraction \true$,
  \end{minipage}

  \vspace{7pt}

  \noindent
  where, in~axiom~{\rm\ref{mwAx:SizeLiterals}},
  $\asetvar \subseteq_\fin \PVAR$ and $\aformula_{\size}$ is a satisfiable conjunction of literals of the form $\size \geq \inbound_1$ or $\lnot \size \geq \inbound_2$.





\end{enumerate}
\end{lem}

The proof of Lemma~\ref{lemma:septractionadmissible} can be found in Appendix~\ref{appendix-septractionadmissible}.

Let  $\aformula$ and $\aformulabis$ be two satisfiable core types
in $\conjcomb{\coreformulae{\asetvar}{\bound}}$.
Following the developments of Section~\ref{section:starelimination}, 
we define a formula
$\boxseptra{\aformula}{\aformulabis}$ in $\conjcomb{\coreformulae{\asetvar}{\bound}}$, 
for which we show that
$\aformula \septraction \aformulabis \iff \boxseptra{\aformula}{\aformulabis}$
is provable in $\coresys(\separate,\magicwand)$.
The formula $\boxseptra{\aformula}{\aformulabis}$ is 
defined in~Figure~\ref{figure:csl:boxseptra-formula}.

\begin{lem}\label{lemma:coreformulaseptraction}
Let $\asetvar \subseteq_{\fin} \PVAR$, $\bound \geq \card{\asetvar}$ and
$\aformula$, $\aformulabis$ be satisfiable core types in $\coretype{\asetvar}{\bound}$.
The formulae $\boxseptra{\aformula}{\aformulabis} \separate \aformula \implies \aformulabis$ and 
$(\lnot \boxseptra{\aformula}{\aformulabis}) \separate \aformula \implies \lnot \aformulabis$ are valid.
\end{lem}

\noindent 
Before presenting the proof for Lemma~\ref{lemma:coreformulaseptraction}, 
let us observe that since we aim at proving the derivability of 
$\aformula \septraction \aformulabis \iff \boxseptra{\aformula}{\aformulabis}$ in $\coresys(\separate,\magicwand)$, 
the validity of~${(\lnot \boxseptra{\aformula}{\aformulabis}) \separate \aformula \implies \lnot \aformulabis}$ should not surprise the reader. 
Indeed, by replacing $\boxseptra{\aformula}{\aformulabis}$ with $\aformula \septraction \aformulabis$ we obtain $(\lnot (\aformula \septraction \aformulabis)) \separate \aformula \implies \lnot \aformulabis$ which, unfolding the definition of $\septraction$, is equivalent to the valid formula $(\aformula \magicwand \lnot \aformulabis) \separate \aformula \implies \lnot \aformulabis$ (see~\ref{mwAx:Cut} in Lemma~\ref{lemma:septractionadmissible}).
On the other hand, the fact that $\boxseptra{\aformula}{\aformulabis} \separate \aformula \implies \aformulabis$ is valid can be puzzling at first, as the formula $(\aformula \septraction \aformulabis) \separate \aformula \implies \aformulabis$ is not valid (in general). In its essence, Lemma~\ref{lemma:coreformulaseptraction} shows that $(\aformula \septraction \aformulabis) \separate \aformula \implies \aformulabis$ is valid whenever $\aformula$ and $\aformulabis$ are restricted to core types.

\begin{figure}
  \arraycolsep=1.4pt
  {\Large
  $$
  \begin{array}{rll}
  &\bigwedge\formulasubset{\avariable \sim \avariablebis \inside \orliterals{\aformula}{\aformulabis}}{\bmat[\sim \in \{=,\neq\}]}
  & \land
  \bigwedge \formulasubset{\alloc{\avariable}}{\bmat[\lnot\alloc{\avariable}\inside\aformula\\ \alloc{\avariable}\inside\aformulabis]}
  \\ 
  \land&
  \bigwedge \aformulasubset{\lnot\alloc{\avariable} \inside \aformulabis}
  &\land
  \bigwedge \formulasubset{\lnot\alloc{\avariable}}{\bmat[\alloc{\avariable}\inside\aformula]}
  \\ 
  \land&
  \bigwedge \aformulasubset{\lnot \avariable {\Ipto} \avariablebis \inside \aformulabis}
   & \land
   \bigwedge \formulasubset{\avariable \Ipto \avariablebis}{\bmat[\lnot\alloc{\avariable}\inside\aformula\\ \avariable \Ipto \avariablebis \inside \aformulabis]}
  \\ 
  \land&
  \bigwedge \formulasubset{\avariable \neq \avariable}{\bmat[\alloc{\avariable}\land\lnot\avariable\Ipto\avariablebis\inside\aformula\\ \avariable \Ipto \avariablebis\inside\aformulabis]}
  & \land
  \bigwedge \formulasubset{\size\geq\inbound_2{+}1{\dotminus}\inbound_1}{\bmat[\lnot\size\geq\inbound_1\inside\aformula\\\size\geq\inbound_2\inside\aformulabis]}
  \\ 
  \land&
  \bigwedge \formulasubset{\avariable \neq \avariable}{\bmat[\avariable\Ipto\avariablebis\inside\aformula\\ \lnot\avariable \Ipto \avariablebis\inside\aformulabis]}
  & \land
  \bigwedge \formulasubset{\lnot\size\geq\inbound_2{\dotminus}\inbound_1}{\bmat[\size\geq\inbound_1\inside\aformula\\\lnot\size\geq\inbound_2\inside\aformulabis]}
  \\ 
  \land&
  \bigwedge \formulasubset{\avariable \neq \avariable}{\bmat[\alloc{\avariable}\inside\aformula\\\lnot\alloc{\avariable}\inside\aformulabis]}
  \end{array}
  $$
}
\caption{The formula $\boxseptra{\aformula}{\aformulabis}$.}
\label{figure:csl:boxseptra-formula}
\end{figure}

Below, we prove that $\boxseptra{\aformula}{\aformulabis} \separate \aformula \implies \aformulabis$ and 
$(\lnot \boxseptra{\aformula}{\aformulabis}) \separate \aformula \implies \lnot \aformulabis$ are valid, thus establishing~Lemma~\ref{lemma:coreformulaseptraction}.
Notice that the proof is carried out through semantical arguments.
Since $\aformula$, $\aformulabis$ and $\boxseptra{\aformula}{\aformulabis}$ are 
conjunctions of literals built from core formulae, derivability of these two tautologies in $\coresys(\separate,\magicwand)$ follows from the completeness of $\coresys(\separate)$ (Theorem~\ref{theo:starCompleteness}). 

\begin{proof}[Validity of $\boxseptra{\aformula}{\aformulabis} \separate \aformula \implies \aformulabis$] 
If $\boxseptra{\aformula}{\aformulabis} \separate \aformula$ is inconsistent, then $\boxseptra{\aformula}{\aformulabis} \separate \aformula \implies \aformulabis$ is straightforwardly valid.
Below, we assume that $\boxseptra{\aformula}{\aformulabis} \separate \aformula$ is
satisfiable. In particular, none of the conditions depicted in~Figure~\ref{figure:csl:boxseptra-formula} that result in $\boxseptra{\aformula}{\aformulabis}$ having a literal $\avariable \neq \avariable$ applies.
Let $\pair{\astore}{\aheap} \models \boxseptra{\aformula}{\aformulabis} \separate \aformula$.
Therefore, there are two disjoint heaps $\aheap_1$ and $\aheap_2$ such that $\aheap = \aheap_1 \heapsum \aheap_2$, $\pair{\astore}{\aheap_1} \models \boxseptra{\aformula}{\aformulabis}$
and $\pair{\astore}{\aheap_2} \models \aformula$.
We show that $\pair{\astore}{\aheap}$ satisfies each literal $\aliteral$ in $\aformulabis$.
We perform a simple case analysis on the shape of $\aliteral$.
Notice that, below, we have $\avariable,\avariablebis \in \asetvar$ and $\inbound_2 \in \interval{0}{\bound}$, as $\aformulabis$ is a core type in $\coretype{\asetvar}{\bound}$.
\begin{description}
      \item[case: $\aliteral \ = \ \avariable \sim \avariablebis$, where $\sim \in \{=,\neq\}$]
            By definition of $\boxseptra{\aformula}{\aformulabis}$, $\avariable \sim \avariablebis \inside \boxseptra{\aformula}{\aformulabis}$ and so $\pair{\astore}{\aheap_1} \models \avariable \sim \avariablebis$. We conclude that $\astore(\avariable) \sim \astore(\avariablebis)$, and thus $\pair{\astore}{\aheap} \models \avariable \sim \avariablebis$.
      \item[case: $\aliteral \ = \ \alloc{\avariable}$]
            If $\alloc{\avariable} \inside \aformula$, then $\pair{\astore}{\aheap_2} \models \alloc{\avariable}$, which implies $\astore(\avariable) \in \domain{\aheap}$ directly from $\aheap_2 \subheap \aheap$. Thus, $\pair{\astore}{\aheap} \models \alloc{\avariable}$.
            Otherwise, if  $\alloc{\avariable} \not\inside \aformula$ then, since $\aformula$ is a core type in $\coretype{\asetvar}{\bound}$, 
            we have $\lnot \alloc{\avariable} \inside \aformula$.
            By definition of $\boxseptra{\aformula}{\aformulabis}$, we derive that $\alloc{\avariable} \inside \boxseptra{\aformula}{\aformulabis}$. 
            So, $\pair{\astore}{\aheap_1} \models \alloc{\avariable}$ and thus, by $\aheap_1 \subheap \aheap$, $\astore(\avariable) \in \domain{\aheap}$.
            We conclude that $\pair{\astore}{\aheap} \models \alloc{\avariable}$.
      \item[case: $\aliteral \ = \ \lnot \alloc{\avariable}$]
            In this case, by definition of $\boxseptra{\aformula}{\aformulabis}$, we have 
            $\lnot \alloc{\avariable} \inside \boxseptra{\aformula}{\aformulabis}$, 
            which implies $\pair{\astore}{\aheap_1} \models \lnot \alloc{\avariable}$.
            \emph{Ad absurdum}, suppose $\pair{\astore}{\aheap_2} \models \alloc{\avariable}$. 
            Since $\aformula$ is a core type in $\coretype{\asetvar}{\bound}$, 
            we conclude that $\alloc{\avariable} \inside \aformula$.
            However, by definition of $\boxseptra{\aformula}{\aformulabis}$, 
            this implies $\avariable \neq \avariable \inside \boxseptra{\aformula}{\aformulabis}$, 
            which contradicts the fact that $\boxseptra{\aformula}{\aformulabis}$ is satisfiable.
            Thus, $\pair{\astore}{\aheap_2} \models  \lnot \alloc{\avariable}$, 
            which implies $\astore(\avariable) \not \in \domain{\aheap_2}$.
            From $\aheap = \aheap_1 \heapsum \aheap_2$ and $\astore(\avariable) \not \in \domain{\aheap_1}$ we conclude that $\astore(\avariable) \not \in \domain{\aheap}$.
            So, $\pair{\astore}{\aheap} \models \lnot \alloc{\avariable}$.
      \item[case: $\aliteral \ = \ \avariable \Ipto \avariablebis$]
            If $\lnot \alloc{\avariable} \inside \aformula$, then ${\avariable \Ipto \avariablebis \inside} \boxseptra{\aformula}{\aformulabis}$ holds by definition of $\boxseptra{\aformula}{\aformulabis}$.
            So, $\aheap_1(\astore(\avariable)) = \astore(\avariablebis)$ and,
            from $\aheap_1 \subheap \aheap$ we conclude that $\pair{\astore}{\aheap} \models \avariable \Ipto \avariablebis$.
            Otherwise, let us assume that $\alloc{\avariable} \inside \aformula$. 
            \emph{Ad absurdum}, suppose $\lnot \avariable \Ipto \avariablebis \inside \aformula$. 
            Then, by definition of $\boxseptra{\aformula}{\aformulabis}$, 
            we derive $\avariable \neq \avariable \inside \boxseptra{\aformula}{\aformulabis}$. 
            However, this contradicts the satisfiability of $\boxseptra{\aformula}{\aformulabis}$.
            Therefore, $\lnot \avariable \Ipto \avariablebis \not\inside \aformula$. 
            Since $\aformula$ is a core type, this implies 
            $\avariable \Ipto \avariablebis \inside \aformula$, 
            and therefore~$\aheap_2(\astore(\avariable)) = \astore(\avariablebis)$.
            From $\aheap_2 \subheap \aheap$ 
            we conclude that $\pair{\astore}{\aheap} \models \avariable \Ipto \avariablebis$.
      \item[case: $\aliteral \ = \ \lnot \avariable \Ipto \avariablebis$]
            By definition of $\boxseptra{\avariable}{\avariablebis}$, 
            we have $\lnot \avariable \Ipto \avariablebis \inside \boxseptra{\avariable}{\avariablebis}$, 
            which implies that if $\astore(\avariable) \in \domain{\aheap_1}$ then 
            $\aheap_1(\astore(\avariable)) \neq \astore(\avariablebis)$.
            \emph{Ad absurdum}, suppose $\avariable \Ipto \avariablebis \inside \aformula$.
            Then, by definition of $\boxseptra{\aformula}{\aformulabis}$,
            we derive $\avariable \neq \avariable \inside \boxseptra{\aformula}{\aformulabis}$. 
            However, this contradicts the satisfiability of $\boxseptra{\aformula}{\aformulabis}$.
            Therefore $\avariable \Ipto \avariablebis \not \inside \aformula$ and, 
            since $\aformula$ is a core type, $\lnot \avariable \Ipto \avariablebis \inside \aformula$. 
            So, if $\astore(\avariable) \in \domain{\aheap_2}$ then 
            $\aheap_2(\astore(\avariable)) \neq \astore(\avariablebis)$.
            By $\aheap = \aheap_1 + \aheap_2$ and the fact that $\aheap_1(\astore(\avariable)) \neq \astore(\avariablebis)$, we conclude that $\pair{\astore}{\aheap} \models \lnot \avariable \Ipto \avariablebis$.
      \item[case: $\aliteral \ = \ \size \geq \inbound_2$]  
            If $\size \geq \bound \inside \aformula$, then $\card{\domain{\aheap}} \geq \card{\domain{\aheap_2}} \geq \bound$, by $\aheap_2 \subheap \aheap$.
            As $\inbound_2 \in \interval{0}{\bound}$, this implies $\pair{\astore}{\aheap} \models \size \geq \inbound_2$.
            Otherwise, assume $\size \geq \bound \not\inside \aformula$.
            In particular, since $\aformula$ is in $\coretype{\asetvar}{\bound}$, 
            this implies that $\maxsize{\aformula} < \bound$ and 
            $$
                  \textstyle\size \geq \maxsize{\aformula} \land \lnot \size \geq \maxsize{\aformula}+1 \ \inside \ \aformula.
            $$
            We have $\card{\domain{\aheap_2}} = \maxsize{\aformula}$.
            If $\maxsize{\aformula} \geq \inbound_2$, then from $\aheap_2 \subheap \aheap$ 
            we conclude that $\pair{\astore}{\aheap} \models \size \geq \inbound_2$.
            Otherwise, let us assume $\inbound_2 > \maxsize{\aformula}$.
            By definition of $\boxseptra{\aformula}{\aformulabis}$, 
            we conclude that
            $\size \geq \inbound_2 + 1 \dotminus (\maxsize{\aformula}+1) \inside \boxseptra{\aformula}{\aformulabis}$.
            Together with $\inbound_2 > \maxsize{\aformula}$, this implies
            $\card{\domain{\aheap_1}} \geq \inbound_2 - \maxsize{\aformula}$.
            With $\card{\domain{\aheap_2}} = \maxsize{\aformula}$ and $\aheap  = \aheap_1 \heapsum \aheap_2$, 
            this implies $\pair{\astore}{\aheap} \models \size \geq \inbound_2$.
      \item[case: $\aliteral \ = \ \lnot \size \geq \inbound_2$]
            \emph{Ad absurdum}, suppose 
            that $\size \geq \bound \inside \aformula$. Then, by definition of 
            $\boxseptra{\aformula}{\aformulabis}$ we have $\lnot \size \geq \inbound_2 \dotminus \bound \inside \boxseptra{\aformula}{\aformulabis}$.
            However, since $\inbound_2 \in \interval{0}{\bound}$, 
            this means that $\lnot \size \geq 0 \inside \boxseptra{\aformula}{\aformulabis}$, 
            which contradicts the satisfiability of $\boxseptra{\aformula}{\aformulabis}$.
            Therefore, $\size \geq \bound \not\inside \aformula$.
            As $\aformula$ is in $\coretype{\asetvar}{\bound}$, 
            we derive $\maxsize{\aformula} < \bound$ and 
            $$
                  \textstyle\size \geq \maxsize{\aformula} \land \lnot \size \geq \maxsize{\aformula}+1 \ \inside \ \aformula.
            $$
            We conclude that $\card{\domain{\aheap_2}} \leq \maxsize{\aformula}$.
            From $\size \geq \maxsize{\aformula} \inside \aformula$
            and by definition of $\boxseptra{\aformula}{\aformulabis}$,
            we conclude that 
            $$ \textstyle\lnot \size \geq \inbound_2 \dotminus \maxsize{\aformula} \inside \boxseptra{\aformula}{\aformulabis}.$$
            If $\inbound_2 \leq \maxsize{\aformula}$, then $\lnot \size \geq 0 \inside \boxseptra{\aformula}{\aformulabis}$, which contradicts the satisfiability of $\boxseptra{\aformula}{\aformulabis}$. 
            Therefore, $\inbound_2 > \maxsize{\aformula}$.
            So, $\card{\domain{\aheap_1}} < \inbound_2 - \maxsize{\aformula}$.
            Together with $\card{\domain{\aheap_2}} \leq \maxsize{\aformula}$ and
            $\aheap = \aheap_1 \heapsum \aheap$, we conclude that 
            $\card{\domain{\aheap}} < \inbound_2$, and thus 
            $\pair{\astore}{\aheap} \models \lnot \size \geq \inbound_2$.
            \qedhere
\end{description}
\end{proof}

\begin{proof}[Validity of $(\lnot \boxseptra{\aformula}{\aformulabis}) \separate \aformula \implies \lnot \aformulabis$]
Let us assume
$ \pair{\astore}{\aheap} \models (\neg \boxseptra{\aformula}{\aformulabis}) \separate \aformula$. Consequently, there is a literal $\aliteral$ of $\boxseptra{\aformula}{\aformulabis}$ such that
$ \pair{\astore}{\aheap} \models  (\neg\aliteral) \separate \aformula$ holds. 
We show that $\pair{\astore}{\aheap} \models \lnot \aformulabis$.
Let $\aheap_1$ and $\aheap_2$ be two disjoint heaps such that $\aheap = \aheap_1 \heapsum \aheap_2$, 
$\pair{\astore}{\aheap_1} \models \lnot \aliteral$ and $\pair{\astore}{\aheap_2} \models \aformula$.
We perform a case analysis on the shape of $\aliteral$. As in the previous part of the proof, recall that $\avariable,\avariablebis \in \asetvar$ and $\inbound_1,\inbound_2 \in \interval{0}{\bound}$.
\begin{description}
\item[case: $\aliteral \ = \ \avariable\neq\avariable$]
Since $\aformula$ and $\aformulabis$ are satisfiable, by definition of $\boxseptra{\aformula}{\aformulabis}$, the fact that 
$\avariable \neq \avariable \inside \boxseptra{\aformula}{\aformulabis}$ implies that one of the following three cases holds: 
\begin{itemize}[align = left]
\item[1:] $\alloc{\avariable}\wedge \neg \avariable\Ipto\avariablebis\inside\aformula$ and $\avariable\Ipto\avariablebis\inside\aformulabis$.
\item[] From  $\alloc{\avariable}\wedge \neg \avariable\Ipto\avariablebis\inside\aformula$ and $\aheap_2 \subheap \aheap$, we have $\astore(\avariable) \in \domain{\aheap}$ 
and $\aheap(\astore(\avariable)) \neq \astore(\avariablebis)$.
Thus $\pair{\astore}{\aheap} \not \models \avariable \Ipto \avariablebis$, and so, by $\avariable\Ipto\avariablebis\inside\aformulabis$, 
$\pair{\astore}{\aheap} \models \lnot \aformulabis$.
\item[2:] $\avariable\Ipto\avariablebis\inside\aformula$ and $\neg\avariable\Ipto\avariablebis\inside\aformulabis$. 
\item[] From $\avariable \Ipto \avariablebis \inside \aformula$ and $\aheap_2 \subheap \aheap$, $\aheap(\astore(\avariable)) = \astore(\avariablebis)$. 
Thus $\pair{\astore}{\aheap} \models \avariable \Ipto \avariablebis$ and so, by  
$\neg\avariable\Ipto\avariablebis\inside\aformulabis$, $\pair{\astore}{\aheap} \models \lnot \aformulabis$.
\item[3:] $\alloc{\avariable}\inside\aformula$ and $\neg\alloc{\avariable}\inside\aformulabis$.
\item[] From $\alloc{\avariable}\inside\aformula$ and $\aheap_2 \subheap \aheap$, 
      $\astore(\avariable) \in \domain{\aheap}$.
      Thus $\pair{\astore}{\aheap} \models \alloc{\avariable}$ and so, 
      by $\neg\alloc{\avariable}\inside\aformulabis$,
      $\pair{\astore}{\aheap} \models \lnot \aformulabis$.
\end{itemize}

\item[case: $\aliteral \  = \ \avariable \sim \avariablebis$, where $\sim \in \{=,\neq\}$] 
      In this case, since $\pair{\astore}{\aheap_1} \models \lnot \aliteral$, 
      then we have $\pair{\astore}{\aheap} \models \lnot \aliteral$.
      Now, it cannot be that $\aliteral \inside \aformula$, 
      as it would imply $\pair{\astore}{\aheap} \models \aliteral$, which is contradictory.
      Therefore, by definition of $\boxseptra{\aformula}{\aformulabis}$,
      we must have $\aliteral\inside\aformulabis$.
      This implies $\pair{\astore}{\aheap} \models \lnot \aformulabis$. 

\item[case: $\aliteral \  = \  \alloc{\avariable}$]
      By definition of $\boxseptra{\aformula}{\aformulabis}$, 
      $\neg \alloc{\avariable}\inside\aformula$ and
      $\alloc{\avariable}\inside\aformulabis$. 
      From $\pair{\astore}{\aheap_1} \models \lnot \alloc{\avariable}$ 
      we conclude that $\astore(\avariable) \not \in \domain{\aheap_1}$.
      By $\neg \alloc{\avariable}\inside\aformula$, $\astore(\avariable) \not \in \domain{\aheap_2}$.
      By $\aheap = \aheap_1 \heapsum \aheap_2$, $\astore(\avariable) \not \in \domain{\aheap}$.
      As $\alloc{\avariable}\inside\aformulabis$, 
      ${\pair{\astore}{\aheap} \models \lnot \aformulabis}$. 
      
\item[case: $\aliteral \  = \  \neg\alloc{\avariable}$] 
      As $\pair{\astore}{\aheap_1} \models \lnot \aliteral$, we have $\astore(\avariable) \in \domain{\aheap_1}$.
      According to the definition of $\boxseptra{\aformula}{\aformulabis}$, 
      either $\alloc{\avariable} \inside \aformula$ or $\lnot \alloc{\avariable} \inside \aformulabis$. The first case cannot hold, as it implies  $\astore(\avariable) \in \domain{\aheap_2}$
      which contradicts the fact that $\aheap_1$ and $\aheap_2$ are disjoint.
      In the second case, from~$\astore(\avariable) \in \domain{\aheap_1}$ and $\aheap_1 \subheap \aheap$, we have $\pair{\astore}{\aheap} \models \alloc{\avariable}$. So, $\pair{\astore}{\aheap} \models \lnot \aformulabis$.

\item[case: $\aliteral \  = \ \avariable\Ipto\avariablebis$]
Then by definition of $\boxseptra{\aformula}{\aformulabis}$, 
$\neg\alloc{\avariable}\inside\aformula$ and $\avariable \Ipto \avariablebis \inside \aformulabis$.
From $\pair{\astore}{\aheap_1} \models \lnot \aliteral$, 
if $\astore(\avariable) \in \domain{\aheap_1}$ then $\aheap_1(\astore(\avariable)) \neq \astore(\avariablebis)$.
As $\neg\alloc{\avariable}\inside\aformula$, $\astore(\avariable) \not\in \domain{\aheap_2}$ and therefore, by $\aheap = \aheap_1 \heapsum \aheap_2$, $\aheap(\astore(\avariable)) \neq \astore(\avariablebis)$.
From $\avariable \Ipto \avariablebis \inside \aformulabis$, 
we conclude that $\pair{\astore}{\aheap} \models \lnot \aformulabis$.

\item[case: $\aliteral \  = \ \neg\avariable\Ipto\avariablebis$]
Then, by definition of $\boxseptra{\aformula}{\aformulabis}$, $\neg\avariable\Ipto\avariablebis\inside\aformulabis$.
From $\pair{\astore}{\aheap_1} \models \lnot \aliteral$ and $\aheap_1 \subheap \aheap$, 
we derive $\aheap(\astore(\avariable)) = \astore(\avariablebis)$.
From $\neg\avariable\Ipto\avariablebis\inside\aformulabis$, 
we derive $\pair{\astore}{\aheap} \models \lnot \aformulabis$.

\item[case: $\aliteral \  = \ \size\geq\inbound_2+1\dotminus\inbound_1$, where
$\size\geq\inbound_2\inside\aformulabis$ and $\neg\size\geq\inbound_1\inside\aformula$]
Since it holds that $\pair{\astore}{\aheap_1} \models \lnot \aliteral$
and $\pair{\astore}{\aheap_2} \models \aformula$, we derive (respectively)
$\card{\domain{\aheap_1}} \leq \inbound_2 \dotminus \inbound_1$ 
and ${\card{\domain{\aheap_2}} < \inbound_1}$. 
From $\aheap = \aheap_1 \heapsum \aheap_2$, we conclude that $\card{\domain{\aheap}} < \inbound_2$.
From $\size\geq\inbound_2\inside\aformulabis$, 
we derive $\pair{\astore}{\aheap} \models \lnot \aformulabis$.

\item[case: $\aliteral \ = \ \neg \size\geq\inbound_2\dotminus\inbound_1$, where
$\neg\size\geq\inbound_2\inside\aformulabis$ and $\size\geq\inbound_1\inside\aformula$]
Since we have 
$\pair{\astore}{\aheap_1} \models \lnot \aliteral$ and 
$\pair{\astore}{\aheap_2} \models \aformula$, we conclude that
$\card{\domain{\aheap_1}} \geq \inbound_2 \dotminus \inbound_1$
and $\card{\domain{\aheap_2}} \geq \inbound_1$.
So, $\aheap = \aheap_1 \heapsum \aheap_2$
implies  $\card{\domain{\aheap}} \geq \inbound_2$.
By $\neg\size\geq\inbound_2\inside\aformulabis$, 
we derive $\pair{\astore}{\aheap} \models \lnot \aformulabis$.
\qedhere
\end{description}
\end{proof}

We are now ready to tackle the proof of~Lemma~\ref{lemma:magicwandPSLelim}.

\begin{proof}[Proof~of Lemma~\ref{lemma:magicwandPSLelim}]
  \label{proof-lemma-magicwandPSLelim}
  As in the statement of the lemma, let us consider 
  $\asetvar \subseteq_{\fin} \PVAR$ and $\bound \geq \card{\asetvar}$, 
  and two core types
  $\aformula$ and $\aformulabis$ in $\coretype{\asetvar}{\bound}$. 
  We want to show that there is a conjunction 
  $\aformulater \in \conjcomb{\coreformulae{\asetvar}{\bound}}$
  such that $\prove_{\magicwandsys} (\aformula\septraction\aformulabis) \iff \aformulater$.

First of all, 
if $\aformula$ or $\aformulabis$ is unsatisfiable, then $\prove_{\magicwandsys}
\aformula\septraction\aformulabis \implies \bot$ by using Lemma~\ref{lemma:corePSLtwo}
and the admissible axioms~\ref{mwAx:ImpL} and~\ref{mwAx:ImpR} from Lemma~\ref{lemma:septractionadmissible}.
Therefore, in this case, it is enough to take $\aformulater$ equal to $\lnot \avariable = \avariable$ to complete the proof. 
Otherwise, let us assume that $\aformula$ and $\aformulabis$ are satisfiable. 
We consider $\aformulater \egdef \boxseptra{\aformula}{\aformulabis}$ (see~Figure~\ref{figure:csl:boxseptra-formula}), and show that 
$\prove_{\magicwandsys} (\aformula\septraction\aformulabis) \iff \boxseptra{\aformula}{\aformulabis}$. We derive each implication separately.

\paragraph{($\Rightarrow$):} Given Lemma~\ref{lemma:coreformulaseptraction}, the proof of $\prove_{\magicwandsys} \aformula\septraction\aformulabis\implies\boxseptra{\aformula}{\aformulabis}$ is straightforward:
    \begin{syntproof}
    1 & \neg \boxseptra{\aformula}{\aformulabis} \separate \aformula \implies \neg \aformulabis 
    & \mbox{Lemma~\ref{lemma:coreformulaseptraction}, Theorem~\ref{theo:starCompleteness}} \\
    2 & \neg \boxseptra{\aformula}{\aformulabis} \implies (\aformula \magicwand \neg \aformulabis)
    & \mbox{\ref{rule:staradj}, 1} \\
    3 & \neg (\aformula \magicwand \neg \aformulabis) \implies \boxseptra{\aformula}{\aformulabis}
    & \mbox{PC, 2} \\
    4 & (\aformula \septraction \aformulabis) \implies \boxseptra{\aformula}{\aformulabis}
    & \mbox{Def. of~$\septraction$, 3} 
    \end{syntproof}


\paragraph{($\Leftarrow$):} 
  Let us now show that $\prove_{\magicwandsys} \boxseptra{\aformula}{\aformulabis}\implies \aformula\septraction\aformulabis$.
    First, let us  note that, since $\boxseptra{\aformula}{\aformulabis} \separate \aformula \implies \aformulabis$
    is valid (Lemma~\ref{lemma:coreformulaseptraction}),
    it is derivable in $\starsys{}$ (Theorem~\ref{theo:starCompleteness}), and therefore, by the rule~\ref{rule:staradj},
    $\prove_{\magicwandsys}\boxseptra{\aformula}{\aformulabis}\implies\aformula\magicwand\aformulabis$. From that, it follows that it is enough to show that
    $\boxseptra{\aformula}{\aformulabis}\implies\aformula\septraction\top$ is derivable in $\magicwandsys$. Indeed,
    from $\boxseptra{\aformula}{\aformulabis}\implies \aformula\septraction\top$ and $\boxseptra{\aformula}{\aformulabis}\implies\aformula\magicwand\aformulabis$, we get, by \ref{mwAx:Mix},
    that $\boxseptra{\aformula}{\aformulabis}\implies\aformula\septraction\aformulabis$ is derivable too.

    Thus, let us prove that $\boxseptra{\aformula}{\aformulabis}\implies\aformula\septraction\top$ is derivable.
    If $\boxseptra{\aformula}{\aformulabis}$ is unsatisfiable, then
    from the completeness of $\coresys$ with respect to Boolean combinations of core formulae (Theorem~\ref{theo:corePSLcompl}),
    we conclude that $\vdash_{\coresys} \boxseptra{\aformula}{\aformulabis} \implies \false$. 
    Since $\coresys(\separate,\magicwand)$ extends $\coresys$, we have
    $\vdash_{\coresys(\separate,\magicwand)} \boxseptra{\aformula}{\aformulabis} \implies \false$.
    By propositional reasoning,  $\vdash_{\coresys(\separate,\magicwand)} \boxseptra{\aformula}{\aformulabis}\implies \aformula\septraction\top$.
    Otherwise, let us assume that $\boxseptra{\aformula}{\aformulabis}$ is satisfiable.

\begin{proof}[Structure of the remaining part of the proof]
\renewcommand{\qedsymbol}{}
Before presenting the technical arguments 
for the derivation of $\vdash_{\coresys(\separate,\magicwand)} \boxseptra{\aformula}{\aformulabis}\implies \aformula\septraction\top$ when $\boxseptra{\aformula}{\aformulabis}$ is  satisfiable, 
let us explain what are the main
ingredients. The proof establishing that 
$\vdash_{\coresys(\separate,\magicwand)} \boxseptra{\aformula}{\aformulabis} \implies \aformula \septraction \true$ 
is by induction on the number $j$  of variables $\avariable \in \asetvar$ for which 
$\alloc{\avariable} \inside \aformula$ holds.
As $\aformula$, $\aformulabis$ and $\boxseptra{\aformula}{\aformulabis}$  are currently assumed to be
 satisfiable, they
have exactly the same equalities and inequalities  and this is used in the proof. 
The base case $j = 0$ can be handled using several derivations taking advantage
of Lemma~\ref{lemma:septractionadmissible}. 
For the induction step $j > 0$, some more substantial work is needed and this is briefly described below.
We distinguish the case $\maxsize{\aformula} < \bound$ from the case  $\maxsize{\aformula} = \bound$.
Both cases, we introduce the formula $\ATOM{\avariable_i}$ where 
 $\alloc{\avariable_i} \inside \aformula$. 
    $$
    \ATOM{\avariable_i} \egdef 
      \begin{cases}
        \avariable_i \Ipto \avariablebis \land \size = 1 &\text{if}~\avariable_i \Ipto \avariablebis \inside \aformula, \text{ for some } \avariablebis \in \asetvar\\
        \alloc{\avariable_i} \land  \size  = 1 \land \bigwedge_{\avariablebis \in \asetvar} \lnot \avariable_i \Ipto \avariablebis &\text{otherwise}
      \end{cases}
    $$
In the case $\maxsize{\aformula} < \bound$, we introduce a formula $\aformula'$ as a very slight variant
of $\aformula$ such that $\aformula'$ enjoys the following essential properties. 
    \begin{itemize}[align = left]
      \item[\ref{csl:lemma62:prime-formula:prop1}] 
        $\aformula'$ is a satisfiable core type in $\coretype{\asetvar}{\bound}$.
      \item[\ref{csl:lemma62:prime-formula:prop2}] 
        $( \ATOM{\avariable_i} \separate \aformula') \implies \aformula$ is valid.
       \item[\ref{csl:lemma62:prime-formula:prop3}] 
        $(\boxseptra{\aformula}{\aformulabis} \separate \ATOM{\avariable_i}) \implies \boxseptra{\aformula'}{\aformulabis}$ is valid.
    \end{itemize}
In order to conclude  $\vdash_{\coresys(\separate,\magicwand)} 
\boxseptra{\aformula}{\aformulabis} \implies (\aformula \septraction \true)$, 
we take advantage of the  completeness of $\coresys(\separate)$ 
to derive the tautologies in~\ref{csl:lemma62:prime-formula:prop2}
    and~\ref{csl:lemma62:prime-formula:prop3}.
    Moreover, as by construction of $\aformula'$, we have  $\lnot \alloc{\avariable_i} \inside \aformula'$ and, 
    for every $\avariablebis \in \asetvar$, $\lnot \alloc{\avariablebis} \inside \aformula$ implies 
    $\lnot \alloc{\avariablebis} \inside \aformula'$, we shall be able to apply the 
     induction hypothesis on $\aformula'$ to
    get 
    $\vdash_{\coresys(\separate,\magicwand)} \boxseptra{\aformula'}{\aformulabis} \implies 
    (\aformula' \septraction \true)$, which will be essential in the final derivation for 
   $\boxseptra{\aformula}{\aformulabis} \implies (\aformula \septraction \true)$.

In the remaining case $\maxsize{\aformula} = \bound$,
we are still looking for some formula $\aformula'$
such that  $\aformula' \separate \ATOM{\avariable_i} \Rightarrow \aformula$ is valid
but we cannot hope for $\aformula'$ to be a core type in $\coretype{\asetvar}{\bound}$. 
Instead, we introduce two core types $\aformula_\bound'$ and $\aformula_{\bound-1}'$, 
      and define $\aformula'$ as $\aformula_\bound' \lor \aformula_{\bound-1}'$. 
The only difference between $\aformula_\bound'$ and $\aformula_{\bound-1}'$ rests on the fact
that $\size \geq \bound \inside \aformula_\bound'$ whereas $\neg \size \geq \bound \inside \aformula_{\bound-1}'$
(both formulae contain $\size \geq \bound -1$). 
Similarly to the previous case, the  properties below shall be shown. 
\begin{itemize}[align=left]
        \item[\ref{csl:lemma62:prime-formula:prop4}]
          $\aformula_\bound'$ and $\aformula_{\bound-1}'$ are satisfiable core types in $\coretype{\asetvar}{\bound}$
        \item[\ref{csl:lemma62:prime-formula:prop5}]
          $( \ATOM{\avariable_i} \separate (\aformula_\bound' \lor \aformula_{\bound-1}')) \implies \aformula$ is valid.       \item[\ref{csl:lemma62:prime-formula:prop6}]
          $(\boxseptra{\aformula}{\aformulabis} \separate \ATOM{\avariable_i}) \implies \boxseptra{\aformula_{\bound}'}{\aformulabis} \lor \boxseptra{\aformula_{\bound-1}'}{\aformulabis}$ is valid.
      \end{itemize}
The derivation of $\vdash_{\coresys(\separate,\magicwand)} 
\boxseptra{\aformula}{\aformulabis} \implies (\aformula \septraction \true)$ follows then
a principle similar to one for the case $\maxsize{\aformula} < \bound$.
\end{proof}


Now, let us present the technical developments. 
Directly from the definition of $\boxseptra{\aformula}{\aformulabis}$, the following simple facts hold.
    \begin{itemize}[align = left]
    \item[\itemlabel{1}{csl:lemma62:prop1}]  $\aformula$, $\aformulabis$ and $\boxseptra{\aformula}{\aformulabis}$ have exactly the same equalities and inequalities. 
    \item[\itemlabel{2}{csl:lemma62:prop2}]   $\neg \size \ge 0$ is not part of $\boxseptra{\aformula}{\aformulabis}$, 
    and therefore, following the definition of $\boxseptra{\aformula}{\aformulabis}$, 
      there are no $\size \geq \inbound_1 \inside \aformula$ and $\neg \size \geq \inbound_2 \inside \aformulabis$
           with $\inbound_1 \geq \inbound_2$. 
    \item[\itemlabel{3}{csl:lemma62:prop3}]   $\avariable \neq \avariable$ does not belong to $\boxseptra{\aformula}{\aformulabis}$. 
    In particular, by definition of $\boxseptra{\aformula}{\aformulabis}$, none of the following conditions apply: 
    \begin{itemize}[align = left]
      \item there is $\avariable \in \asetvar$ such that $\alloc{\avariable} \inside \aformula$ and $\lnot \alloc{\avariable} \inside \aformulabis$,
      \item there are $\avariable,\avariablebis \in \asetvar$ such that $\avariable \Ipto \avariablebis \in \aformula$ and $\lnot \avariable \Ipto \avariablebis \inside \aformulabis$,
      \item there are $\avariable,\avariablebis \in \asetvar$ such that $\alloc{\avariable} \land \lnot \avariable \Ipto \avariablebis \inside \aformula$ and $\avariable \Ipto \avariablebis \inside \aformulabis$.
    \end{itemize} 
    \end{itemize}

    From~\ref{csl:lemma62:prop1}, we know that $\boxseptra{\aformula}{\aformulabis}$ and $\aformula$ satisfy the same (in)equalities. 
    Similarly to the proof of Lemma~\ref{lemma:starPSLelim-sat},
    let $\avariable_1,\ldots\avariable_n$ be a maximal enumeration of representatives of the
    equivalence classes (one per equivalence class) such that $\alloc{\avariable_i}$ occurs in $\aformula$.
    As it is maximal, for every $\alloc{\avariable}$ in $\literals{\aformula}$ there is $i \in \interval{1}{n}$ such that $\avariable_i$ is syntactically equal to~$\avariable$. 
    Moreover, by definition of $\boxseptra{\aformula}{\aformulabis}$, for every $i \in \interval{1}{n}$, 
    $\lnot \alloc{\avariable_i} \inside \boxseptra{\aformula}{\aformulabis}$.
    The proof of $\vdash_{\coresys(\separate,\magicwand)} \boxseptra{\aformula}{\aformulabis} \implies \aformula \septraction \true$ is by induction on the number $j$ of variables $\avariable \in \asetvar$ for which $\alloc{\avariable} \inside \aformula$ holds.
    \begin{description}
      \item[base case: $j = 0$] In the base case, no formula $\alloc{\avariable}$ occurs positively in $\aformula$. Since $\aformula$ is a core type, this implies that for every $\avariable \in \asetvar$, $\lnot \alloc{\avariable} \inside \aformula$. Moreover, 
      since $\aformula$ is satisfiable, for every $\avariable,\avariablebis \in \asetvar$, 
      $\lnot \avariable \Ipto \avariablebis \inside \aformula$ (see the axiom~\ref{coreAx:PointAlloc}).
      Therefore, the core type $\aformula$ is syntactically equivalent (up to associativity and commutativity of conjunction) to the formula \,
      $
        \aformula_{\size} \land \aforUNALLOC \land \aformula_{\not\Ipto} \land \aforEQ,
      $
      \,
      where 
      \begin{itemize}
        \item $\aformula_{\size} \egdef \bigwedge \big( \{\size \geq \inbound \inside \aformula\} \cup \{ \lnot \size \geq \inbound \inside \aformula\}\big)$,
        \item $\aformula_{\lnot\mathtt{alloc}} \egdef \bigwedge_{\avariable \in \asetvar} \lnot \alloc{\avariable}$,
        \item $\aformula_{\not\Ipto} \egdef \bigwedge_{\avariable,\avariablebis \in \asetvar}\lnot \avariable \Ipto \avariablebis$,
        \item $\aforEQ \egdef \bigwedge \{ \avariable \sim \avariablebis \inside \aformula \mid \sim \in \{=,\neq\}\}$.
      \end{itemize}
      \vspace{3pt}
      Since $\aformula$ is satisfiable, so is $\aformula_{\size}$. We show that 
      $\vdash_{\coresys(\separate,\magicwand)} (\aformula_{\size} \land \aforUNALLOC \land  \aformula_{\not\Ipto}) \septraction \true$:
      \begin{syntproof}
        1   & \aformula_{\size} \land \aforUNALLOC \septraction \true 
            & \mbox{\ref{mwAx:SizeLiterals}}\\
        2   & \lnot \alloc{\avariable} \implies \lnot \avariable \Ipto \avariablebis 
            & \mbox{\ref{coreAx:PointAlloc}, PC}\\
        3   & \aforUNALLOC \implies \aformula_{\not\Ipto}
            & \mbox{PC, repeated 2}\\
        4   & \aformula_{\size} \land \aforUNALLOC \implies \aformula_{\size} \land \aforUNALLOC \land \aformula_{\not\Ipto}
            & \mbox{PC, 3}\\
        5   & (\aformula_{\size} \land \aforUNALLOC \septraction \true )
              \implies (\aformula_{\size} \land \aforUNALLOC \land \aformula_{\not\Ipto} \septraction \true) 
            & \mbox{\ref{mwAx:ImpL}, 4}\\
        6   & \aformula_{\size} \land \aforUNALLOC \land \aformula_{\not\Ipto} \septraction \true
            & \mbox{Modus Ponens, 1, 5}
      \end{syntproof}
      \noindent Now, let us treat the formula $\aforEQ$. From the definition of $\boxseptra{\aformula}{\aformulabis}$, we have $\aforEQ \inside \boxseptra{\aformula}{\aformulabis}$, and so
      by propositional reasoning, $\vdash_{\coresys(\separate,\magicwand)} \boxseptra{\aformula}{\aformulabis} \implies \aforEQ$.
      This allows us to conclude that 
      \begin{equation}
      \vdash_{\coresys(\separate,\magicwand)} \boxseptra{\aformula}{\aformulabis} \implies 
      \big((\aformula_{\size} \land \aforUNALLOC \land  \aformula_{\not\Ipto} \land \aforEQ) \septraction \true\big),
        \tag{$\dagger$}\label{csl:lemma62:dagger}
      \end{equation}
      by induction on the number of literals $\avariable \sim \avariablebis$  appearing in $\aforEQ$, 
      and by relying on the two theorems \ref{mwAx:SeptEq} and~\ref{mwAx:SeptIneq}. In the base case, $\aforEQ = \true$, and so 
      \begin{syntproof}
        7 & \aformula_{\size} \land \aforUNALLOC \land \aformula_{\not\Ipto} \implies \aformula_{\size} \land \aforUNALLOC \land \aformula_{\not\Ipto} \land \aforEQ 
          & \mbox{PC}\\
        8 & (\aformula_{\size} \land \aforUNALLOC \land \aformula_{\not\Ipto} \septraction \true ) 
        \implies (\aformula_{\size} \land \aforUNALLOC \land \aformula_{\not\Ipto} \land \aforEQ \septraction \true) 
          & \mbox{\ref{mwAx:ImpL}, 7}\\
        9 & \aformula_{\size} \land \aforUNALLOC \land \aformula_{\not\Ipto} \land \aforEQ \septraction \true
          & \mbox{Modus Ponens, 6, 8}\\
        10 & \boxseptra{\aformula}{\aformulabis} \implies 
        (\aformula_{\size} \land \aforUNALLOC \land  \aformula_{\not\Ipto} \land \aforEQ \septraction \true) 
          & \mbox{PC, 9}
      \end{syntproof}
      \noindent In the induction step, let $\aforEQ = \aforEQ' \land \avariable \sim \avariablebis$, where $\avariable \sim \avariablebis \not\inside \aforEQ'$.
      We have,
      \begin{syntproof}
        1 & \boxseptra{\aformula}{\aformulabis} \implies 
        (\aformula_{\size} \land \aforUNALLOC \land  \aformula_{\not\Ipto} \land \aforEQ' \septraction \true )
          & \mbox{Induction Hypothesis}\\
        2 & \boxseptra{\aformula}{\aformulabis} \implies \avariable \sim \avariablebis
          & \mbox{PC, as $\aforEQ \inside \boxseptra{\aformula}{\aformulabis}$}\\
        3 & \avariable \sim \avariablebis \land (\aformula_{\size} \land \aforUNALLOC \land  \aformula_{\not\Ipto} \land \aforEQ' \septraction \true) \implies\\[-2pt]
          & \qquad (\aformula_{\size} \land \aforUNALLOC \land  \aformula_{\not\Ipto} \land \aforEQ' \land \avariable \sim \avariablebis \septraction \true)
          & \mbox{\ref{mwAx:SeptEq}/\ref{mwAx:SeptIneq}}\\
        4 & \boxseptra{\aformula}{\aformulabis} \implies (\aformula_{\size} \land \aforUNALLOC \land  \aformula_{\not\Ipto} \land \aforEQ' \land \avariable \sim \avariablebis \septraction \true)
          & \mbox{PC, 1, 2, 3}\\ 
        5 & \boxseptra{\aformula}{\aformulabis} \implies (\aformula_{\size} \land 
            \aforUNALLOC \land  \aformula_{\not\Ipto} \land \aforEQ \septraction \true)
          & \mbox{Def. of $\aforEQ'$, 4}
      \end{syntproof} 
      \noindent Since $\aformula_{\size} \land 
      \aforUNALLOC \land  \aformula_{\not\Ipto} \land \aforEQ$ is equivalent to $\aformula$, 
      from~(\ref{csl:lemma62:dagger}) and by~\ref{mwAx:ImpL}, we conclude
      that $\vdash_{\coresys(\separate,\magicwand)} \boxseptra{\aformula}{\aformulabis} \implies 
      \aformula \septraction \true$.
      
    \item[induction step: $j \geq 1$] In this case, let $i \in \interval{1}{n}$ such that $\alloc{\avariable_i} \inside \aformula$ and thus,
    by definition of $\boxseptra{\aformula}{\aformulabis}$, 
    $\lnot \alloc{\avariable_i} \inside \boxseptra{\aformula}{\aformulabis}$.
    As announced earlier, we define the formula:
    $$
    \ATOM{\avariable_i} \egdef 
      \begin{cases}
        \avariable_i \Ipto \avariablebis \land \size = 1 &\text{if}~\avariable_i \Ipto \avariablebis \inside \aformula, \text{ for some } \avariablebis \in \asetvar\\
        \alloc{\avariable_i} \land  \size  = 1 \land \bigwedge_{\avariablebis \in \asetvar} \lnot \avariable_i \Ipto \avariablebis &\text{otherwise}
      \end{cases}
    $$
    Notice that, if there is $\avariablebis \in \asetvar$ such that $\avariable_i \Ipto \avariablebis \inside \aformula$, then the axiom schema \ref{wandAx:PointsTo} can be instantiated to
    $\lnot \alloc{\avariable_i} \implies (\ATOM{\avariable_i} \septraction \true)$.
    Otherwise (for all $\avariablebis \in \asetvar$, $\avariable_i \Ipto \avariablebis \not\inside \aformula$) this formula is an instantiation of the axiom schema \ref{wandAx:Alloc}. This allows us to show the following theorem:
    \begin{equation}
      \boxseptra{\aformula}{\aformulabis} \implies \big(\ATOM{\avariable_i}  \septraction
      (\boxseptra{\aformula}{\aformulabis} \separate \ATOM{\avariable_i})\big)
      \tag{$\ddagger$}\label{csl:lemma62:ddagger}
    \end{equation}
    \begin{syntproof}
      1 & \neg \alloc{\avariable_i} \implies (\ATOM{\avariable_i}  \septraction \top)
      & \mbox{\ref{wandAx:PointsTo}/\ref{wandAx:Alloc}} \\
      2 & \boxseptra{\aformula}{\aformulabis} \implies  \neg \alloc{\avariable_i} 
      & \mbox{Def. of $\boxseptra{\aformula}{\aformulabis}$, PC} \\
      3 & \boxseptra{\aformula}{\aformulabis}  \implies  (\ATOM{\avariable_i}  \septraction \top)
      & \mbox{\ref{rule:imptr}, 1, 2} \\ 
      4 & \boxseptra{\aformula}{\aformulabis} \separate \ATOM{\avariable_i} \implies
          \boxseptra{\aformula}{\aformulabis} \separate \ATOM{\avariable_i} 
      & \mbox{PC} \\
      5 & \boxseptra{\aformula}{\aformulabis}
        \implies (\ATOM{\avariable_i} \magicwand 
        \boxseptra{\aformula}{\aformulabis} \separate \ATOM{\avariable_i})
      & \mbox{\ref{rule:staradj}, 4} \\
      6 & \boxseptra{\aformula}{\aformulabis} \implies (\ATOM{\avariable_i}  \septraction
      \boxseptra{\aformula}{\aformulabis} \separate \ATOM{\avariable_i})
      & \mbox{\ref{mwAx:Mix}, 3, 5, PC} 
    \end{syntproof}
    \noindent 
    From the hypothesis $\card{\asetvar} \leq \bound$, together with $\alloc{\avariable_i} \inside \aformula$ and the fact that $\aformula$ is satisfiable, 
    we have $\maxsize{\aformula} \geq 1$ (see~\ref{coreAx:AllocSize}, instantiated with~$\asetvar = \{\avariable_i\}$).
    In order to show that $\vdash_{\coresys{\separate,\magicwand}} \boxseptra{\aformula}{\aformulabis} \implies (\aformula \septraction \true)$,
    we split the proof depending on whether $\maxsize{\aformula} < \bound$ holds.
    \begin{description}
      \item[case: $\maxsize{\aformula} < \bound$] Since $\aformula$ is a satisfiable core type in $\coretype{\asetvar}{\bound}$, by definition of $\maxsize{.}$, we have $\size \geq \maxsize{\aformula} \land \lnot \size \geq \maxsize{\aformula}+1 \inside \aformula$.
      Below, we consider the formula $\aformula'$ obtained from $\aformula$ by: 
      \begin{itemize}
        \item replacing $\size \geq \maxsize{\aformula} \inside \aformula$ with $\lnot \size \geq \maxsize{\aformula}$,
        \item for every $\avariable \in \asetvar$ such that $\avariable = \avariable_i \inside \aformula$, replacing every literal $\alloc{\avariable} \inside \aformula$ with $\lnot \alloc{\avariable}$, and every literal $\avariable \Ipto \avariablebis \inside \aformula$ with 
        $\lnot \avariable \Ipto \avariablebis$, where $\avariablebis \in \asetvar$.
      \end{itemize}
      Explicitly, 
      {\small$$\aformula' \egdef 
      \begin{aligned}[t]
        &\bigwedge \{ \avariable \sim \avariablebis \inside \aformula \mid \sim \in \{=,\neq\} \} 
        \land \bigwedge \{ \alloc{\avariable} \inside \aformula \mid \avariable \neq \avariable_i \inside \aformula\}
        \land \\ 
        & \bigwedge \{ \lnot \alloc{\avariable} \inside \aformula \} \land
        \bigwedge \{ \lnot \alloc{\avariable} \mid \avariable = \avariable_i \inside \aformula \}
        \land \bigwedge \{ \avariable \Ipto \avariablebis \inside \aformula \mid \avariable \neq \avariable_i \inside \aformula \}
        \land \\ 
        &
        \bigwedge \{ \lnot \avariable \Ipto \avariablebis \inside \aformula \}
        \land \bigwedge \{ \lnot \avariable \Ipto \avariablebis \mid \avariable = \avariable_i \land \avariable \Ipto \avariablebis \inside \aformula \} \land \lnot \size \geq {\textstyle\maxsize{\aformula}} \land \\
        & \bigwedge \{ \size \geq \inbound \inside \aformula \mid \inbound < {\textstyle\maxsize{\aformula}}\} 
        \land \bigwedge \{\lnot \size \geq \inbound \inside \aformula\}.
      \end{aligned}$$}

    The formula $\aformula'$ enjoys the following properties (to be shown below):

    \vspace{3pt}

    \begin{itemize}[align=left]
      \setlength{\itemsep}{3pt}
      \item[\itemlabel{\textbf{A}}{csl:lemma62:prime-formula:prop1}] 
        $\aformula'$ is a satisfiable core type in $\coretype{\asetvar}{\bound}$.
      \item[\itemlabel{\textbf{B}}{csl:lemma62:prime-formula:prop2}] 
        $( \ATOM{\avariable_i} \separate \aformula') \implies \aformula$ is valid.
      \item[\itemlabel{\textbf{C}}{csl:lemma62:prime-formula:prop3}] 
          $(\boxseptra{\aformula}{\aformulabis} \separate \ATOM{\avariable_i}) \implies \boxseptra{\aformula'}{\aformulabis}$ is valid.
    \end{itemize}

    \vspace{5pt}

    \noindent
    Fundamentally, $\aformula'$ enjoys the induction hypothesis, which reveals to be useful
    later~on. 

    \begin{proof}[Proof of~\rm\ref{csl:lemma62:prime-formula:prop1}]
        Since $\aformula'$ is obtained from $\aformula$ simply by changing the polarity of some of the literals in $\literals{\aformula}$, clearly $\aformula'$ is in $\coretype{\asetvar}{\bound}$. To show that $\aformula'$ is satisfiable, we rely on the fact that $\aformula$ is satisfiable. Let $\pair{\astore}{\aheap}$ be a memory state satisfying~$\aformula$. 
        Since $\alloc{\avariable_i} \inside \aformula$, we conclude that $\astore(\avariable_i) \in \domain{\aheap}$. 
        Let us consider the disjoint heaps $\aheap_1$ and $\aheap_2$ such that 
        $\aheap = \aheap_1 + \aheap_2$ and $\domain{\aheap_1} = \{\astore(\avariable_i)\}$.
        We show that $\pair{\astore}{\aheap_2} \models \aformula'$ by 
        considering every $\aliteral \in \literals{\aformula'}$ and showing that 
        $\pair{\astore}{\aheap_2} \models \aliteral$.
        \begin{description}
          \item[case: $\aliteral \, = \, \avariable \sim \avariablebis$, where $\sim \in \{=,\neq\}$] 
          By definition of $\aformula'$, $\pair{\astore}{\aheap} \models \aliteral$ and therefore $\astore(\avariable) \sim \astore(\avariablebis)$. 
          Thus, $\pair{\astore}{\aheap_2} \models \aliteral$.
          \item[case: $\aliteral \, = \, \lnot \alloc{\avariable}$]
            If $\avariable = \avariable_i \inside \aformula$ then $\astore(\avariable) \in \domain{\aheap_1}$, and therefore, by $\aheap_1 \disjoint \aheap_2$, $\astore(\avariable) \not \in \domain{\aheap_2}$. So, $\pair{\astore}{\aheap_2} \models \lnot \alloc{\avariable}$.
            Otherwise ($\avariable \neq \avariable_i \inside \aformula$), by definition of $\aformula'$, we have $\lnot \alloc{\avariable} \inside \aformula$. 
            So $\astore(\avariable) \not\in \domain{\aheap}$ and, from $\aheap_2 \subheap \aheap$, we conclude that $\pair{\astore}{\aheap_2} \models \lnot \alloc{\avariable}$.
          \item[case: $\aliteral \, = \, \lnot \avariable \Ipto \avariablebis$] 
            Similar to the previous case. Briefly, if $\avariable = \avariable_i \inside \aformula$ then, by definition of $\ATOM{\avariable_i}$, $\pair{\astore}{\aheap_2} \not \models \alloc{\avariable}$, which implies $\pair{\astore}{\aheap_2} \models \lnot \avariable \Ipto \avariablebis$.
            Otherwise, by definition of $\aformula'$, $\lnot \avariable \Ipto \avariablebis \inside \aformula$ and thus $\pair{\astore}{\aheap} \models \lnot \avariable \Ipto \avariablebis$.
            From $\aheap_2 \subheap \aheap$, we conclude that $\pair{\astore}{\aheap_2} \models \lnot \avariable \Ipto \avariablebis$. 
          \item[case: $\aliteral \, = \, \alloc{\avariable}$]
            By definition of $\aformula'$, $\alloc{\avariable} \land \avariable \neq \avariable_i \inside \aformula$. 
            Therefore $\astore(\avariable) \in \domain{\aheap}$ and, by definition of $\ATOM{\avariable_i}$, $\astore(\avariable) \not \in \domain{\aheap_1}$.
            Since $\aheap = \aheap_1 \heapsum \aheap_2$, we conclude that $\pair{\astore}{\aheap_2} \models \alloc{\avariable}$.
          \item[case: $\aliteral \, = \, \avariable \Ipto \avariablebis$]
            Similar to the previous case. By definition of $\aformula'$, we have
            ${\avariable \Ipto \avariablebis} \land \avariable \neq \avariable_i \inside \aformula$.
            Thus, $\aheap(\astore(\avariable)) = \astore(\avariablebis)$.
            By definition of $\ATOM{\avariable_i}$, 
            $\astore(\avariable) \in \domain{\aheap_2}$ and thus 
            $\aheap_2(\astore(\avariable)) = \astore(\avariablebis)$.
            So, $\pair{\astore}{\aheap_2} \models \avariable \Ipto \avariablebis$. 
          \item[case: $\aliteral \, = \, \size \geq \inbound$]
            By definition of $\aformula'$, $\inbound < \maxsize{\aformula}$.
            Since $\pair{\astore}{\aheap} \models \aformula$, 
            we have $\card{\domain{\aheap}} \geq \maxsize{\aformula}$. 
            By definition of $\ATOM{\avariable_i}$ and from $\aheap = \aheap_1 \heapsum \aheap_2$, 
            we have $\card{\domain{\aheap_2}} = \card{\domain{\aheap}} {-} 1 \geq \maxsize{\aformula} {-} 1 \geq \inbound$. 
            So, ${\pair{\astore}{\aheap_2} \models \size \geq \inbound}$.
          \item[case: $\aliteral \, = \, \lnot \size \geq \inbound$]
            By definition of $\aformula'$,
            ${\lnot \size \geq \inbound \inside \aformula}$ 
            or
            ${\inbound = \maxsize{\aformula}}$. 
            In the former case, since $\aformula$ is satisfiable, we know that $\inbound > \maxsize{\aformula}$. Hence, in both cases we have $\inbound \geq \maxsize{\aformula}$. 
            As $\pair{\astore}{\aheap} \models \aformula$ and $\lnot \size \geq \maxsize{\aformula}+1 \inside \aformula$, we have 
            $\card{\domain{\aheap}} \leq \maxsize{\aformula}$.
            Since $\card{\domain{\aheap_1}} = 1$, by $\aheap = \aheap_1 \heapsum \aheap_2$ we derive $\card{\domain{\aheap_2}} < \maxsize{\aformula} \leq \inbound$. 
            Therefore, $\pair{\astore}{\aheap_2} \models \lnot \size \geq \inbound$.
            \qedhere
        \end{description}
      \end{proof}

      \begin{proof}[Proof of~\rm\ref{csl:lemma62:prime-formula:prop2}]
        Let $\pair{\astore}{\aheap} \models \ATOM{\avariable_i} \separate \aformula'$.
        So, there are $\aheap_1$ and $\aheap_2$ such that $\aheap = \aheap_1 \heapsum \aheap_2$, 
        $\pair{\astore}{\aheap_1} \models \ATOM{\avariable_i}$ and $\pair{\astore}{\aheap_2} \models \aformula'$.
        By definition of $\ATOM{\avariable_i}$, $\domain{\aheap_1} = \{\astore(\avariable_i)\}$.
        In order to prove~\ref{csl:lemma62:prime-formula:prop2},
        we show that $\pair{\astore}{\aheap} \models \aliteral$, for every literal $\aliteral \in \literals{\aformula}$.
        \begin{description}
          \item[case: $\aliteral \, = \, \avariable \sim \avariablebis$, where $\sim \in \{=,\neq\}$] 
          By definition of $\aformula'$, $\pair{\astore}{\aheap_2} \models \aliteral$ and therefore $\astore(\avariable) \sim \astore(\avariablebis)$. 
          Hence, $\pair{\astore}{\aheap} \models \aliteral$. 

          \item[case: $\aliteral \, = \, \lnot \alloc{\avariable}$] By definition of $\ATOM{\avariable_i}$, $\alloc{\avariable_i} \inside \aformula$ and therefore $\astore(\avariable) \not\in \domain{\aheap_1}$.  
          By definition of $\aformula'$, for every $\avariablebis \in \asetvar$, $\alloc{\avariablebis} \inside \aformula'$ implies $\alloc{\avariablebis} \inside \aformula$. 
          Therefore, $\astore(\avariable) \not \in \domain{\aheap_2}$. 
          We conclude that $\astore(\avariable) \not \in \domain{\aheap}$, and so 
          $\pair{\astore}{\aheap} \models \lnot \alloc{\avariable}$.

          \item[case: $\aliteral \, = \, \lnot \avariable \Ipto \avariablebis$] Similar to the previous case. Briefly, by definition of $\ATOM{\avariable_i}$, $\pair{\astore}{\aheap_1} \models \lnot \avariable \Ipto \avariablebis$. 
          By definition of $\aformula'$, $\pair{\astore}{\aheap_2} \models \lnot \avariable \Ipto \avariablebis$. So, $\pair{\astore}{\aheap} \models \lnot \avariable \Ipto \avariablebis$.  

          \item[case: $\aliteral \, = \, \alloc{\avariable}$] If $\avariable = \avariable_i \inside \aformula$, then 
          $\astore(\avariable) = \astore(\avariable_i)$ (first case of the proof),
          and by definition of $\ATOM{\avariable_i}$, $\astore(\avariable) \in \domain{\aheap_1}$. 
          As $\aheap_1 \subheap \aheap$, we conclude that $\pair{\astore}{\aheap} \models \alloc{\avariable}$.
          Otherwise, if $\avariable \neq \avariable_i \inside \aformula$, then by definition of $\aformula'$ we have $\alloc{\avariable} \inside \aformula'$. 
          This implies that $\astore(\avariable) \in \domain{\aheap_2}$ and so, from $\aheap_2 \subheap \aheap$, we conclude that $\pair{\astore}{\aheap} \models \alloc{\avariable}$.

          \item[case: $\aliteral \, = \, \avariable \Ipto \avariablebis$] 
            Similar to the previous case. Briefly, if $\avariable = \avariable_i \inside \aformula$ then, by definition of $\ATOM{\avariable_i}$, 
            $\pair{\astore}{\aheap_1} \models \avariable \Ipto \avariablebis$ and so 
            $\pair{\astore}{\aheap} \models \avariable \Ipto \avariablebis$.
            Else (${\avariable \neq \avariable_i} \inside \aformula$), 
            $\avariable \Ipto \avariablebis \inside \aformula'$ and therefore 
            $\pair{\astore}{\aheap_2} \models \avariable \Ipto \avariablebis$. 
            So, $\pair{\astore}{\aheap} \models \avariable \Ipto \avariablebis$.

          \item[case: $\aliteral \, = \, \size \geq \inbound$] 
            If $\inbound < \maxsize{\aformula}$, then directly by definition of $\aformula'$, we have $\pair{\astore}{\aheap_2} \models \size \geq \inbound$. From $\aheap_2 \subheap \aheap$, we conclude that  
            $\pair{\astore}{\aheap} \models \size \geq \inbound$.
            Otherwise, $\inbound = \maxsize{\aformula}$. Recall that $\maxsize{\aformula} \geq 1$ and so,
            by definition of~$\aformula'$, $\size \geq \maxsize{\aformula}-1 \inside \aformula'$.
            Thus, $\card{\domain{\aheap_2}} \geq \maxsize{\aformula}-1$.
            By definition of $\ATOM{\avariable_i}$ we have $\card{\domain{\aheap_1}} = 1$.
            As $\aheap = \aheap_1 \heapsum \aheap_2$, we conclude that 
            $\pair{\astore}{\aheap} \models \size \geq \maxsize{\aformula}$.

          \item[case: $\aliteral \, = \, \lnot \size \geq \inbound$]
            As~$\aformula$ is satisfiable, $\inbound > \maxsize{\aformula}$.
            By definition of the formula~$\aformula'$,
            ${\lnot \size \geq \maxsize{\aformula} \inside \aformula'}$
            and thus $\card{\domain{\aheap_2}} < \maxsize{\aformula}$. 
            From ${\card{\domain{\aheap_1}} = 1}$ we can derive $\card{\domain{\aheap}} \leq \maxsize{\aformula} < \inbound$,
            which allows us to conclude that
            ${\pair{\astore}{\aheap} \models \lnot \size \geq \inbound}$.
            \qedhere
        \end{description}
    \end{proof}

    \vspace{3pt}

    \begin{figure}
      \arraycolsep=1.4pt
      {\Large
      $$
      \begin{array}{rll}
      &\bigwedge\formulasubset{\avariable \sim \avariablebis \inside \orliterals{\aformula'}{\aformulabis}}{\bmat[\sim \in \{=,\neq\}]}
      & \land
      \bigwedge \formulasubset{\alloc{\avariable}}{\bmat[\lnot\alloc{\avariable}\inside\aformula'\\ \alloc{\avariable}\inside\aformulabis]}
      \\ 
      \land&
      \bigwedge \aformulasubset{\lnot\alloc{\avariable} \inside \aformulabis}
      &\land
      \bigwedge \formulasubset{\lnot\alloc{\avariable}}{\bmat[\alloc{\avariable}\inside\aformula']}
      \\ 
      \land&
      \bigwedge \aformulasubset{\lnot \avariable {\Ipto} \avariablebis \inside \aformulabis}
       & \land
       \bigwedge \formulasubset{\avariable \Ipto \avariablebis}{\bmat[\lnot\alloc{\avariable}\inside\aformula'\\ \avariable \Ipto \avariablebis \inside \aformulabis]}
      \\ 
      \land&
      \bigwedge \formulasubset{\avariable \neq \avariable}{\bmat[\alloc{\avariable}\land\lnot\avariable\Ipto\avariablebis\inside\aformula'\\ \avariable \Ipto \avariablebis\inside\aformulabis]}
      & \land
      \bigwedge \formulasubset{\size\geq\inbound_2{+}1{\dotminus}\inbound_1}{\bmat[\lnot\size\geq\inbound_1\inside\aformula'\\\size\geq\inbound_2\inside\aformulabis]}
      \\ 
      \land&
      \bigwedge \formulasubset{\avariable \neq \avariable}{\bmat[\avariable\Ipto\avariablebis\inside\aformula'\\ \lnot\avariable \Ipto \avariablebis\inside\aformulabis]}
      & \land
      \bigwedge \formulasubset{\lnot\size\geq\inbound_2{\dotminus}\inbound_1}{\bmat[\size\geq\inbound_1\inside\aformula'\\\lnot\size\geq\inbound_2\inside\aformulabis]}
      \\ 
      \land&
      \bigwedge \formulasubset{\avariable \neq \avariable}{\bmat[\alloc{\avariable}\inside\aformula'\\\lnot\alloc{\avariable}\inside\aformulabis]}
      \end{array}
      $$
    }
    \caption{The formula $\boxseptra{\aformula'}{\aformulabis}$.}
    \label{figure:csl:boxseptra-formula-again}
    \end{figure}

    \begin{itemize}[align=left]
    \item[\textit{Proof of~\ref{csl:lemma62:prime-formula:prop3}}.]
        Figure~\ref{figure:csl:boxseptra-formula-again} recalls the definition of $\boxseptra{\aformula'}{\aformulabis}$.
        First of all, notice that
        it cannot be that there is $\avariable \in \asetvar$ such that 
        $\avariable \neq \avariable \inside \boxseptra{\aformula'}{\aformulabis}$. 
        Indeed, \emph{ad absurdum}, suppose the opposite. 
        By definition of $\boxseptra{\aformula'}{\aformulabis}$, this implies that 
        (1) $\alloc{\avariable} \land \lnot \avariable \Ipto \avariablebis \inside \aformula'$ 
        and $\avariable \Ipto \avariablebis \inside \aformulabis$, 
        (2) $\avariable \Ipto \avariablebis \inside \aformula'$ and $\lnot \avariable \Ipto \avariablebis \inside \aformulabis$, or (3) $\alloc{\avariable} \inside \aformula'$ and $\lnot \alloc{\avariable} \inside \aformulabis$. 
        By definition of $\aformula'$, this implies that 
        (1) $\alloc{\avariable} \land \lnot \avariable \Ipto \avariablebis \inside \aformula$, 
        (2) $\avariable \Ipto \avariablebis \inside \aformula$ 
        or (3) $\alloc{\avariable} \inside \aformula$.
        However, by definition of $\boxseptra{\aformula}{\aformulabis}$, 
        this implies that $\avariable \neq \avariable \inside \boxseptra{\aformula}{\aformulabis}$, 
        in contradiction with the satisfiability of $\boxseptra{\aformula}{\aformulabis}$. 
        Therefore, below we assume that for all $\avariable \in \asetvar$, 
        $\avariable \neq \avariable \not\inside \boxseptra{\aformula'}{\aformulabis}$.

        Let $\pair{\astore}{\aheap} \models \boxseptra{\aformula}{\aformulabis} \separate \ATOM{\avariable_i}$.
        There are $\aheap_1$ and $\aheap_2$ such that $\aheap = \aheap_1 \heapsum \aheap_2$, 
        $\pair{\astore}{\aheap_1} \models \boxseptra{\aformula}{\aformulabis}$ and 
        $\pair{\astore}{\aheap_2} \models \ATOM{\avariable_i}$.
        By definition of $\ATOM{\avariable_i}$, $\domain{\aheap_2} = \{\astore(\avariable_i)\}$.
        To prove~\ref{csl:lemma62:prime-formula:prop3},
        we show that $\pair{\astore}{\aheap} \models \aliteral$,
        for every literal $\aliteral \in \literals{\boxseptra{\aformula'}{\aformulabis}}$.

        \begin{description}
          \item[case: $\aliteral \, = \, \avariable \sim \avariablebis$, where $\sim \in \{=,\neq\}$] 
          By definition of $\boxseptra{\aformula'}{\aformulabis}$,
          $\aliteral \inside \orliterals{\aformula'}{\aformulabis}$ and so, 
          by definition of $\aformula'$, $\aliteral \inside \orliterals{\aformula}{\aformulabis}$.
          By definition of $\boxseptra{\aformula}{\aformulabis}$,
          $\aliteral \inside \boxseptra{\aformula}{\aformulabis}$. 
          From $\pair{\astore}{\aheap_1} \models \boxseptra{\aformula}{\aformulabis}$ we derive $\astore(\avariable) \sim \astore(\avariablebis)$.
          So, $\pair{\astore}{\aheap} \models \aliteral$.

          \item[case: $\aliteral \, = \, \lnot \alloc{\avariable}$] 
            From the definition of $\boxseptra{\aformula'}{\aformulabis}$, 
            either we have $\lnot \alloc{\avariable} \inside \aformulabis$ or we
            have $\alloc{\avariable} \inside \aformula'$. 
            In the first case, by definition of $\boxseptra{\aformula}{\aformulabis}$, 
            $\lnot \alloc{\avariable} \inside \boxseptra{\aformula}{\aformulabis}$, and 
            therefore $\astore(\avariable) \not\in \domain{\aheap_1}$.
            Moreover, since 
            $\boxseptra{\aformula}{\aformulabis}$ is satisfiable, $\alloc{\avariable} \not \inside \aformula$ (otherwise we would have $\avariable \neq \avariable \inside \boxseptra{\aformula}{\aformulabis}$). 
            Therefore, by definition of $\ATOM{\avariable_i}$,
            we conclude that $\astore(\avariable) \not \in \domain{\aheap_2}$.
            From $\aheap = \aheap_1 \heapsum \aheap_2$, we derive 
            $\astore(\avariable) \not \in \domain{\aheap}$, and thus $\pair{\astore}{\aheap} \models \lnot \alloc{\avariable}$.

            In the second case, ($\alloc{\avariable} \inside \aformula'$), by definition of $\aformula'$ we have $\alloc{\avariable} \inside \aformula$ and $\avariable \neq \avariable_i \inside \aformula$. 
            By definition of $\ATOM{\avariable_i}$, $\astore(\avariable) \not\in \domain{\aheap_2}$.
            By definition of $\boxseptra{\aformula}{\aformulabis}$, 
            $\lnot \alloc{\avariable} \inside \boxseptra{\aformula}{\aformulabis}$, 
            and therefore $\astore(\avariable) \not \in \domain{\aheap_1}$.
            Again, by $\aheap = \aheap_1 \heapsum \aheap_2$, 
            we have $\pair{\astore}{\aheap} \models \lnot \alloc{\avariable}$.

          \item[case: $\aliteral \, = \, \lnot \avariable \Ipto \avariablebis$] 
            Following the definition of $\boxseptra{\aformula'}{\aformulabis}$,
            $\lnot \avariable \Ipto \avariablebis \inside \aformulabis$ and therefore 
            $\lnot \avariable \Ipto \avariablebis \inside \boxseptra{\aformula}{\aformulabis}$.
            Therefore, $\pair{\astore}{\aheap_1} \models \lnot \avariable \Ipto \avariablebis$.
            Since $\boxseptra{\aformula}{\aformulabis}$ is satisfiable, $\lnot \avariable \Ipto \avariablebis \inside \aformula$.
            By definition of $\ATOM{\avariable_i}$, we derive $\pair{\astore}{\aheap_2} \models \lnot \avariable \Ipto \avariablebis$. 
            From $\aheap = \aheap_1 \heapsum \aheap_2$, 
            $\pair{\astore}{\aheap} \models \lnot \avariable \Ipto \avariablebis$.

          \item[case: $\aliteral \, = \, \alloc{\avariable}$] 
            By definition of $\boxseptra{\aformula'}{\aformulabis}$, 
            we have
            $\lnot \alloc{\avariable} \inside \aformula'$ and $\alloc{\avariable} \inside \aformulabis$.
            First, let us suppose $\alloc{\avariable} \inside \aformula$. 
            By definition of $\aformula'$, 
            $\avariable = \avariable_i \inside \aformula$ and so, by definition of 
            $\ATOM{\avariable_i}$, $\astore(\avariable) \in \domain{\aheap_2}$. 
            From~$\aheap_2 \subheap \aheap$, $\pair{\astore}{\aheap} \models \alloc{\avariable}$. 
            Otherwise ($\lnot \alloc{\avariable} \inside \aformula$), 
            by definition of $\boxseptra{\aformula}{\aformulabis}$, 
            $\alloc{\avariable} \inside \boxseptra{\aformula}{\aformulabis}$.
            So, $\astore(\avariable) \in \domain{\aheap_1}$,
            and by $\aheap_1 \subheap \aheap$, $\pair{\astore}{\aheap} \models \alloc{\avariable}$.

          \item[case: $\aliteral \, = \, \avariable \Ipto \avariablebis$]
            By definition of $\boxseptra{\aformula'}{\aformulabis}$,  we have
            $\lnot \alloc{\avariable} \inside \aformula'$ and 
            ${\avariable \Ipto \avariablebis} \inside \aformulabis$.
            First, suppose ${\alloc{\avariable} \inside \aformula}$.
            By definition of $\aformula'$, 
            ${\avariable = \avariable_i} \inside \aformula$. By definition of 
            $\ATOM{\avariable_i}$, $\astore(\avariable) \in \domain{\aheap_2}$. 
            \emph{Ad absurdum}, suppose $\aheap(\astore(\avariable)) \neq \astore(\avariablebis)$.
            By definition of $\ATOM{\avariable_i}$, we have that
            ${\alloc{\avariable} \land \lnot \avariable \Ipto \avariablebis \inside \aformula}$.
            However, from $\avariable \Ipto \avariablebis \inside \aformulabis$, 
            this implies $\avariable \neq \avariable \inside \boxseptra{\aformula}{\aformulabis}$, which contradicts the satisfiability of $\boxseptra{\aformula}{\aformulabis}$. 
            Therefore, $\aheap(\astore(\avariable)) = \astore(\avariablebis)$ and, from $\aheap_2 \subheap \aheap$, we conclude that $\pair{\astore}{\aheap} \models \avariable \Ipto \avariablebis$.
            Otherwise ($\lnot \alloc{\avariable} \inside \aformula$), 
            by definition of $\boxseptra{\aformula}{\aformulabis}$, 
            $\avariable \Ipto \avariablebis \inside \boxseptra{\aformula}{\aformulabis}$.
            So, $\aheap_1(\astore(\avariable)) = \astore(\avariablebis)$,
            and by $\aheap_1 \subheap \aheap$, we derive $\pair{\astore}{\aheap} \models \avariable \Ipto \avariablebis$.

          \item[{\parbox[b]{1.07\linewidth}{case: $\aliteral \, = \, \size \geq
            \inbound_2 {+} 1 \dotminus \inbound_1$, where   $\lnot \size \geq
            \inbound_1 \inside \aformula'$\\
            \null \hfill and 
          $\size \geq \inbound_2 \inside \aformulabis$}}]
            By definition of $\aformula'$, 
            $\lnot \size \geq \inbound_1 \inside \aformula$, and so $\inbound_1 > \maxsize{\aformula}$, since $\aformula$ is satisfiable.
            By definition of $\boxseptra{\aformula}{\aformulabis}$ and as $\lnot \size \geq \maxsize{\aformula}+1 \inside \aformula$, we have
            $\size \geq \inbound_2 + 1 \dotminus (\maxsize{\aformula}+1)\inside \boxseptra{\aformula}{\aformulabis}$, which in turn implies
            $\card{\domain{\aheap_1}} \geq \inbound_2 \dotminus \maxsize{\aformula}$.
            By definition of $\ATOM{\avariable_i}$, ${\card{\domain{\aheap_2}} \geq 1}$.
            By ${\aheap = \aheap_1 \heapsum \aheap_2}$, 
            $\card{\domain{\aheap}} \geq (\inbound_2 \dotminus \maxsize{\aformula}) + 1 \geq (\inbound_2 + 1) \dotminus \maxsize{\aformula}$.
            As $\inbound_1 > \maxsize{\aformula}$,
            $\pair{\astore}{\aheap} \models \size \geq \inbound_2 {+} 1 \dotminus \inbound_1$.

          \item[{\parbox[b]{1.07\linewidth}{case: $\aliteral \, = \, \lnot \size
            \geq \inbound_2 \dotminus \inbound_1$, where $\size \geq \inbound_1
            \inside \aformula'$\\\null\hfill and $\lnot \size \geq \inbound_2 \inside
            \aformulabis$}}]
            By definition of $\aformula'$, $\inbound_1 < \maxsize{\aformula}$.
            By definition of $\boxseptra{\aformula}{\aformulabis}$, 
            we have
            $\lnot \size \geq \inbound_2 \dotminus \maxsize{\aformula} \inside \boxseptra{\aformula}{\aformulabis}$. 
            Notice that, since $\boxseptra{\aformula}{\aformulabis}$ is satisfiable, 
            $\inbound_2 > \maxsize{\aformula}$.
            Thus, $\card{\domain{\aheap_1}} < \inbound_2 - \maxsize{\aformula}$.
            By definition of $\ATOM{\avariable_i}$, $\card{\domain{\aheap_2}} \leq 1$. 
            From $\aheap = \aheap_1 \heapsum \aheap_2$, 
            we conclude that
            $\card{\domain{\aheap}} < (\inbound_2 - \maxsize{\aformula}) + 1$.
            As $\inbound_1 < \maxsize{\aformula}$, we have ${\inbound_2 - \maxsize{\aformula} + 1} \leq \inbound_2 \dotminus \inbound_1$. 
            Therefore, $\pair{\astore}{\aheap} \models \lnot \size \geq \inbound_2 \dotminus \inbound_1$.
        \end{description}
    \end{itemize}

    Continuing the proof of~Lemma~\ref{lemma:magicwandPSLelim},
    we prove ${\vdash_{\coresys(\separate,\magicwand)} \boxseptra{\aformula}{\aformulabis} \implies (\aformula \septraction \true)}$.
    Notice that, by the completeness of $\coresys(\separate)$ (Theorem~\ref{theo:starCompleteness}),
    we conclude that the tautologies in~\ref{csl:lemma62:prime-formula:prop2}
    and~\ref{csl:lemma62:prime-formula:prop3} are derivable in $\coresys(\separate,\magicwand)$.
    Moreover, notice that $\lnot \alloc{\avariable_i} \inside \aformula'$ and, for every $\avariablebis \in \asetvar$, $\lnot \alloc{\avariablebis} \inside \aformula$ implies $\lnot \alloc{\avariablebis} \inside \aformula'$. 
    This allows us to 
    rely on the induction hypothesis, and conclude that $\vdash_{\coresys(\separate,\magicwand)} \boxseptra{\aformula'}{\aformulabis} \implies (\aformula' \septraction \true)$.
    The derivation of $\boxseptra{\aformula}{\aformulabis} \implies (\aformula \septraction \true)$
    is given below:

    \begin{syntproof}
      1 & \boxseptra{\aformula'}{\aformulabis} \implies (\aformula' \septraction \true)
      & \mbox{Induction hypothesis} \\
      2 &  \ATOM{\avariable_i} \separate \aformula' \implies \aformula
      & \mbox{\ref{csl:lemma62:prime-formula:prop2}, Theorem~\ref{theo:starCompleteness}} \\
      3 & \boxseptra{\aformula}{\aformulabis} \separate \ATOM{\avariable_i}\implies \boxseptra{\aformula'}{\aformulabis}
      & \mbox{\ref{csl:lemma62:prime-formula:prop3}, Theorem~\ref{theo:starCompleteness}} \\
      4 & (\boxseptra{\aformula}{\aformulabis} \separate \ATOM{\avariable_i}) \implies (\aformula' \septraction \top)
      & \mbox{\ref{rule:imptr}, 1, 3} \\ 
      5 & \boxseptra{\aformula}{\aformulabis} \implies (\ATOM{\avariable_i}  \septraction
      \boxseptra{\aformula}{\aformulabis} \separate \ATOM{\avariable_i})
      & \mbox{(\ref{csl:lemma62:ddagger})} \\ 
      6 & (\ATOM{\avariable_i}  \septraction
      \boxseptra{\aformula}{\aformulabis} \separate \ATOM{\avariable_i})
      \implies \big(\ATOM{\avariable_i}  \septraction
      (\aformula' \septraction \true )\big)
         & \mbox{\ref{mwAx:ImpR}, 4}\\
      7 & \big(\ATOM{\avariable_i}  \septraction
      (\aformula' \septraction \true )\big) \implies 
      (\ATOM{\avariable_i} \separate \aformula' \septraction \true)
      & \mbox{\ref{mwAx:Curry}} \\
      8 & (\ATOM{\avariable_i} \separate \aformula' \septraction \true) \implies (\aformula \septraction \true)
         & \mbox{\ref{mwAx:ImpL}, 2} \\
      9 & \boxseptra{\aformula}{\aformulabis} \implies (\aformula \septraction \true)
      & \mbox{\ref{rule:imptr}, 5, 6, 7, 8}
    \end{syntproof}

    \item[case: $\maxsize{\aformula} = \bound$]
      In this case, we have $\size \geq \bound \inside \aformula$, where we recall that $\bound = \maxsize{\aformula} \geq 1$.
      Following the developments of the previous case, we would like to 
      define a formula $\aformula'$ for which the formula $\aformula' \separate \ATOM{\avariable_i} \Rightarrow 
      \aformula$ is valid. 
      However, since $\aformula$ is in $\coretype{\asetvar}{\bound}$, we cannot hope for $\aformula'$ to be a core type in $\coretype{\asetvar}{\bound}$. 
      Indeed, because of $\size \geq \bound \inside \aformula$, in order to achieve the valid formula above we must differentiate between the case where $\aformula$ is satisfied by a memory state $\pair{\astore}{\aheap}$ such that $\card{\domain{\aheap}} > \bound$, to the case where $\card{\domain{\aheap}} = \bound$. 
      Therefore, below we introduce two core types $\aformula_\bound'$ and $\aformula_{\bound-1}'$, 
      and define $\aformula'$ as $\aformula_\bound' \lor \aformula_{\bound-1}'$. 
      Since the separating conjunction distributes over disjunctions,
      after defining these two core types, we can easily adapt the arguments of the previous case to prove that $\boxseptra{\aformula}{\aformulabis} \implies (\aformula \septraction \true)$.

      The formula $\aformula_\bound'$ is obtained from $\aformula$ by replacing,  
      for every $\avariable \in \asetvar$ such that $\avariable = \avariable_i \inside \aformula$,
      every literal $\alloc{\avariable} \inside \aformula$ with $\lnot \alloc{\avariable}$, and every $\avariable \Ipto \avariablebis \inside \aformula$ with 
        $\lnot \avariable \Ipto \avariablebis$, where $\avariablebis \in \asetvar$.
      Notice that $\aformula_{\bound}'$ is defined similarly to $\aformula'$ (in the previous case of the proof), with the exception that we do not modify the polarity of size literals.
      Explicitly, $\aformula_{\bound}'$ is defined as follows.
      {\small$$\aformula_{\bound}' \egdef 
      \begin{aligned}[t]
        &\bigwedge \{ \avariable \sim \avariablebis \inside \aformula \mid \sim \in \{=,\neq\} \} 
        \land \bigwedge \{ \alloc{\avariable} \inside \aformula \mid \avariable \neq \avariable_i \inside \aformula\}
        \land \bigwedge \{ \lnot \alloc{\avariable} \inside \aformula \} \land\\ 
        & \bigwedge \{ \lnot \alloc{\avariable} \mid \avariable = \avariable_i \inside \aformula \}
        \land \bigwedge \{ \avariable \Ipto \avariablebis \inside \aformula \mid \avariable \neq \avariable_i \inside \aformula \}
        \land \bigwedge \{ \lnot \avariable \Ipto \avariablebis \inside \aformula \}
        \land\\ 
        &
         \bigwedge \{ \lnot \avariable \Ipto \avariablebis \mid \avariable = \avariable_i \land \avariable \Ipto \avariablebis \inside \aformula \} \land 
        \bigwedge \{ \size \geq \inbound \mid \inbound \in \interval{0}{\bound-1} \} \land \underline{\size \geq \bound}.
      \end{aligned}$$}
      The formula $\aformula_{\bound-1}'$ is obtained from $\aformula_{\bound}'$ by replacing $\size \geq \bound$ (highlighted in the definition of $\aformula_{\bound}'$ above), by $\lnot \size \geq \bound$.
      The  following properties are satisfied:
      \vspace{3pt}
      \begin{itemize}[align=left]
        \setlength{\itemsep}{3pt}
        \item[\itemlabel{\textbf{D}}{csl:lemma62:prime-formula:prop4}]
          $\aformula_\bound'$ and $\aformula_{\bound-1}'$ are satisfiable core types in $\coretype{\asetvar}{\bound}$,
        \item[\itemlabel{\textbf{E}}{csl:lemma62:prime-formula:prop5}]
          $( \ATOM{\avariable_i} \separate (\aformula_\bound' \lor \aformula_{\bound-1}')) \implies \aformula$ is valid.
        \item[\itemlabel{\textbf{F}}{csl:lemma62:prime-formula:prop6}]
            $(\boxseptra{\aformula}{\aformulabis} \separate \ATOM{\avariable_i}) \implies \boxseptra{\aformula_{\bound}'}{\aformulabis} \lor \boxseptra{\aformula_{\bound-1}'}{\aformulabis}$ is valid.
      \end{itemize}

      \vspace{5pt}

      \begin{proof}[Proof of~\rm\ref{csl:lemma62:prime-formula:prop4}]
          The proof is very similar to the one of the property~\ref{csl:lemma62:prime-formula:prop1}. Here, we pinpoint the main differences.
          First of all, 
          since both $\aformula_{\bound}'$ and $\aformula_{\bound-1}'$ are obtained from $\aformula$ by changing the polarity of some of the literals in $\literals{\aformula}$, they are both in $\coretype{\asetvar}{\bound}$. 
          To show that $\aformula_{\bound}'$ and $\aformula_{\bound-1}'$ are satisfiable,
          we rely on the fact that $\aformula$ is satisfiable. Let $\pair{\astore}{\aheap}$ be a memory state satisfying~$\aformula$. 
          Since $\size \geq \bound \inside \aformula$, $\card{\domain{\aheap}} \geq \bound$.
          Without loss of generality, we can assume $\card{\domain{\aheap}} > \bound$. Indeed, if 
          $\card{\domain{\aheap}} = \bound$ it is sufficient to add a memory cell $\pair{\alocation}{\alocation}$ to $\aheap$, such that $\alocation$ does not correspond to a program variable $\avariable \in \asetvar$. It is straightforward to check that the resulting memory state still satisfies $\aformula$.
          We introduce a second heap $\aheap'$.
          Let $\asetoflocs = \domain{\aheap} \cap \{\astore(\avariable) \mid \avariable \in \asetvar\}$ be the set of locations in $\domain{\aheap}$ that corresponds to variables in $\asetvar$.
          Since $\card{\asetvar} \leq \bound$, $\card{\asetoflocs} \leq \bound$.
          Let $\aheap' \subheap \aheap$ such that $\asetoflocs \subseteq \domain{\aheap'}$ 
          and $\card{\domain{\aheap'}} = \bound$.
          Again, it is straightforward to see that $\pair{\astore}{\aheap'}$ satisfies~$\aformula$.
          Intuitively, we rely on $\pair{\astore}{\aheap}$ to show that $\aformula_{\bound}'$ is satisfiable, and on $\pair{\astore}{\aheap'}$ to show that $\aformula_{\bound-1}'$ is satisfiable. 
          As $\alloc{\avariable_i} \inside \aformula$, we have $\astore(\avariable)\in \domain{\aheap}$ and $\astore(\avariable) \in \domain{\aheap'}$.
          We consider heaps $\aheap_1$ and $\aheap_2$ such that $\aheap = \aheap_1 \heapsum \aheap_2$ and $\domain{\aheap_1} = \{\astore(\avariable_i)\}$. 
          Similarly, we consider heaps $\aheap_1'$ and $\aheap_2'$ such that $\aheap' = \aheap_1' \heapsum \aheap_2'$ and $\domain{\aheap_1'} = \{\astore(\avariable_i)\}$.
          We show that $\pair{\astore}{\aheap_2} \models \aformula_{\bound}'$ and 
          $\pair{\astore}{\aheap_2'} \models \aformula_{\bound-1}'$. 
          Let us first discuss the former result. Let 
          $\aliteral \in \literals{\aformula_{\bound}'}$. 
          If $\aliteral$ is not of the form $\size \geq \inbound$ or $\lnot \size \geq \inbound$, 
          then $\pair{\astore}{\aheap_2} \models \aliteral$ follows exactly as in the proof 
          of~\ref{csl:lemma62:prime-formula:prop1}. Otherwise,
          \begin{description}
            \item[case: $\aliteral \, = \, \size \geq \inbound$]
              By definition of $\aheap_2$, $\card{\domain{\aheap_2}} = \card{\domain{\aheap}} - 1 \geq \bound$. Since $\inbound \leq \bound$ (as $\aformula_{\bound}'$ is in $\coretype{\asetvar}{\bound}$), we conclude that $\pair{\astore}{\aheap_2} \models \size \geq \inbound$.
            \item[case: $\aliteral \, = \, \lnot \size \geq \inbound$]  
              By definition of $\aformula_{\bound}'$, no literals of the form ${\lnot \size \geq \inbound}$ belongs to $\literals{\aformula_{\bound}'}$. Therefore, 
              this case does not occur.
          \end{description}
          This concludes the proof of $\pair{\astore}{\aheap_2} \models \aformula_{\bound}'$.
          For the proof of $\pair{\astore}{\aheap_2'} \models \aformula_{\bound-1}'$, 
          let us consider $\aliteral \in \literals{\aformula_{\bound-1}'}$.
          Again, if $\aliteral$ is not of the form $\size \geq \inbound$ or $\lnot \size \geq \bound$, 
          then $\pair{\astore}{\aheap_2'} \models \aliteral$ follows exactly as in the proof 
          of~\ref{csl:lemma62:prime-formula:prop1} (replacing $\aheap$ by $\aheap'$ and $\aheap_2$ by $\aheap_2'$). Otherwise, 
          \begin{description}
            \item[case: $\aliteral \, = \, \size \geq \inbound$]
              By definition of $\aformula_{\bound-1}'$, we have $\inbound < \bound$. 
              By definition of $\aheap_2'$, $\card{\domain{\aheap_2'}} = \card{\domain{\aheap'}} - 1 = \bound - 1$. Therefore, $\pair{\astore}{\aheap_2'} \models \size \geq \inbound$.
            \item[case: $\aliteral \, = \, \lnot \size \geq \inbound$]  
              By definition of $\aformula_{\bound-1}'$, $\inbound = \bound$. 
              Since $\card{\domain{\aheap_2'}} = \bound - 1$, 
              we conclude that $\pair{\astore}{\aheap_2'} \models \lnot \size \geq \inbound$.
              \qedhere
          \end{description}
        \end{proof}
  
        \begin{proof}[Proof of~\rm\ref{csl:lemma62:prime-formula:prop5}]
          The proof is very similar to the one of the 
          property~\ref{csl:lemma62:prime-formula:prop2}.
          We show that $(\ATOM{\avariable_i} \separate \aformula_\bound') \implies \aformula$ and 
          $(\ATOM{\avariable_i} \separate \aformula_{\bound-1}') \implies \aformula$. 
          Then, \ref{csl:lemma62:prime-formula:prop5} follows as the separating conjunction distributes over disjunction.
          First, let us consider $(\ATOM{\avariable_i} \separate \aformula_\bound') \implies \aformula$, and a memory state $\pair{\astore}{\aheap}$ satisfying $\ATOM{\avariable_i} \separate \aformula_\bound'$. 
          There are $\aheap_1$ and $\aheap_2$ such that $\aheap = \aheap_1 \heapsum \aheap_2$, $\pair{\astore}{\aheap_1} \models \ATOM{\avariable_i}$ and $\pair{\astore}{\aheap_2} \models \aformula_{\bound}'$. 
          Let $\aliteral \in \literals{\aformula}$. 
          Notice that $\aformula$ does not contain negated $\size \geq \inbound$ literals.
          If $\aliteral$ is not $\size \geq \inbound$, for some $\inbound \in \interval{0}{\bound}$, 
          then $\pair{\astore}{\aheap} \models \aliteral$ follows exactly as it is shown in the proof of~\ref{csl:lemma62:prime-formula:prop2}. 
          Otherwise, suppose $\aliteral = \size \geq \inbound$, where $\inbound \in \interval{0}{\bound}$. By definition 
          of~$\aformula_{\bound}'$, 
          $\size \geq \bound \inside \aformula_{\bound}'$. 
          Hence, $\card{\domain{\aheap_2}} \geq \bound$ and, from $\aheap_2 \subheap \aheap$, 
          we derive $\pair{\astore}{\aheap} \models \size \geq \inbound$.
          So, $\pair{\astore}{\aheap} \models \aformula$.

          Let us now consider $(\ATOM{\avariable_i} \separate \aformula_{\bound-1}') \implies \aformula$ and a memory state $\pair{\astore}{\aheap}$ satisfying $\ATOM{\avariable_i} \separate \aformula_{\bound-1}'$.
          There are $\aheap_1$ and $\aheap_2$ such that $\aheap = \aheap_1 \heapsum \aheap_2$, $\pair{\astore}{\aheap_1} \models \ATOM{\avariable_i}$ and $\pair{\astore}{\aheap_2} \models \aformula_{\bound-1}'$. 
          Let $\aliteral \in \literals{\aformula}$. 
          Again, $\aformula$ does not contain negated $\size \geq \inbound$ literals, and if 
          $\aliteral$ is not $\size \geq \inbound$, for some $\inbound \in \interval{0}{\bound}$, 
          then $\pair{\astore}{\aheap} \models \aliteral$ follows exactly as is shown in the proof of~\ref{csl:lemma62:prime-formula:prop2}. 
          Otherwise, suppose $\aliteral = \size \geq \inbound$, where $\inbound \in \interval{0}{\bound}$. 
          By definition of $\aformula_{\bound-1}'$, 
          $\size \geq \bound \dotminus 1 \inside \aformula_{\bound-1}'$.
          Therefore, $\card{\domain{\aheap_2}} \geq \bound-1$.
          By definition of $\ATOM{\avariable_i}$, $\card{\domain{\aheap_1}} = 1$. 
          From $\aheap = \aheap_1 \heapsum \aheap_2$,  
          we conclude that $\card{\domain{\aheap}} \geq \bound$ and thus $\pair{\astore}{\aheap} \models \size \geq \inbound$.
          Therefore, $\pair{\astore}{\aheap} \models \aformula$.
      \end{proof}

      \vspace{5pt}

      \begin{proof}[Proof of~\rm\ref{csl:lemma62:prime-formula:prop6}]
          Recall that $\boxseptra{\aformula}{\aformulabis}$ is satisfiable. 
          In particular, from its definition together with $\size \geq \bound \inside \aformula$, this implies that ${\size \geq \bound} \inside \aformulabis$, as otherwise we would have $\lnot \size \geq 0 \inside \boxseptra{\aformula}{\aformulabis}$. So, as $\aformulabis$ is a satisfiable core type in $\coretype{\asetvar}{\bound}$, 
          for all $\inbound \in \interval{0}{\bound}$,
          $\size \geq \inbound \inside \aformulabis$.
          Alternatively, $\aformulabis$ does not contain $\lnot \size \geq \inbound$ literals.
          We look at the definitions of $\boxseptra{\aformula_{\bound}'}{\aformulabis}$ 
          and~$\boxseptra{\aformula_{\bound-1}'}{\aformulabis}$.
          \begin{itemize}[align=left]
          \item[\itemlabel{a}{csl:lemma62:a-prop}] Since for all 
          $\inbound \in \interval{0}{\bound}$, 
          $\size \geq \inbound \inside \aformula_{\bound}'$ and $\size \geq \inbound \inside \aformulabis$, 
          we derive that $\boxseptra{\aformula_{\bound}'}{\aformulabis}$ 
          does not contain $\size \geq \inbound$ nor $\lnot \size \geq \inbound$ literals (for all~${\inbound \in \interval{0}{\bound}}$). 
          This holds directly by definition of 
          $\boxseptra{\aformula_{\bound}'}{\aformulabis}$, which can be retrieved by substituting $\aformula'$ by $\aformula_{\bound}'$ in Figure~\ref{figure:csl:boxseptra-formula-again}.

          \item[\itemlabel{b}{csl:lemma62:b-prop}] Analogously, we know that
          $\lnot \size \geq \bound \inside \aformula_{\bound-1}'$ whereas for every $\inbound \in \interval{0}{\bound-1}$, $\size \geq \inbound \inside \aformula_{\bound-1}'$, 
          and therefore among all the literals $\size \geq \inbound$ or 
          $\lnot \size \geq \inbound$ ($\inbound \in \interval{0}{\bound}$), 
          $\boxseptra{\aformula_{\bound-1}'}{\aformulabis}$
          only contains $\size \geq 1$ (occurring positively).
          \end{itemize}
          By definition and with the sole exception of the polarity of the formula ${\size \geq \bound}$ (occurring positively in $\aformula_{\bound}'$ and negatively in $\aformula_{\bound-1}'$), the two core types 
          $\aformula_{\bound-1}'$ and $\aformula_{\bound}'$ are equal.
          Directly by definition of 
          $\boxseptra{\aformula_{\bound}'}{\aformulabis}$
          and 
          $\boxseptra{\aformula_{\bound-1}'}{\aformulabis}$, 
          together with~\ref{csl:lemma62:a-prop} and~\ref{csl:lemma62:b-prop}, 
          this implies that $\boxseptra{\aformula_{\bound-1}'}{\aformulabis}$
          is syntactically equal to ${\boxseptra{\aformula_{\bound}'}{\aformulabis} \land \size \geq 1}$ (up to commutativity and associativity of conjunction).
          This means that the formula $\boxseptra{\aformula_{\bound-1}'}{\aformulabis} \implies \boxseptra{\aformula_{\bound}'}{\aformulabis}$ is valid, 
          and suggests us that,
          in order to show~\ref{csl:lemma62:prime-formula:prop6}, 
          we can simply establish that $(\boxseptra{\aformula}{\aformulabis} \separate \ATOM{\avariable_i}) \implies \boxseptra{\aformula_{\bound}'}{\aformulabis}$ is valid.
          As we already stated, $\aformula_{\bound}'$ is defined as $\aformula'$ (in the previous step of the proof), with the exception that we do not modify the polarity of $\size \geq \inbound$ literals. 
          Because of this, we can rely on the proof of~\ref{csl:lemma62:prime-formula:prop3}.
          Briefly, 
          we consider a memory state $\pair{\astore}{\aheap}$ satisfying $\boxseptra{\aformula}{\aformulabis} \separate \ATOM{\avariable_i}$.
          There are $\aheap_1$ and $\aheap_2$ such that $\aheap = \aheap_1 \heapsum \aheap_2$, 
          $\pair{\astore}{\aheap_1} \models \boxseptra{\aformula}{\aformulabis}$ and 
          $\pair{\astore}{\aheap_2} \models \ATOM{\avariable_i}$.
          Let $\aliteral \in \literals{\boxseptra{\aformula_{\bound-1}'}{\aformulabis}}$. 
          By \ref{csl:lemma62:a-prop},  $\aliteral$ is neither of the form $\size \geq \inbound$ nor of the form $\lnot \size \geq \inbound$.
          Therefore, $\pair{\astore}{\aheap} \models \aliteral$ follows exactly as shown in the proof of~\ref{csl:lemma62:prime-formula:prop3}.
      \end{proof}

      We are now ready to prove that $\boxseptra{\aformula}{\aformulabis} \implies (\aformula \septraction \true)$.
      By Theorem~\ref{theo:starCompleteness}, the tautologies 
      in~\ref{csl:lemma62:prime-formula:prop4} and~\ref{csl:lemma62:prime-formula:prop6} 
      are derivable in $\coresys(\separate,\magicwand)$. Moreover, 
      since $\lnot \alloc{\avariable_i} \inside \andliterals{\aformula_{\bound}'}{\aformula_{\bound-1}'}$ and, for every $\avariablebis \in \asetvar$, $\lnot \alloc{\avariablebis} \inside \aformula$ implies $\lnot \alloc{\avariablebis} \inside \andliterals{\aformula_{\bound}'}{\aformula_{\bound-1}'}$, 
      we rely on the induction hypothesis to derive 
      $$
        \vdash_{\coresys(\separate,\magicwand)} \boxseptra{\aformula_{\bound}'}{\aformulabis} \implies (\aformula_{\bound}' \septraction \true),
        \qquad
        \vdash_{\coresys(\separate,\magicwand)} \boxseptra{\aformula_{\bound-1}'}{\aformulabis} \implies (\aformula_{\bound-1}' \septraction \true).
      $$
      We derive $\boxseptra{\aformula}{\aformulabis} \implies (\aformula \septraction \true)$ (see Figure~\ref{figure-final-derivation}) concluding the proof of~Lemma~\ref{lemma:magicwandPSLelim}
      \qedhere
\begin{figure}
      \begin{syntproof}
        1 & \boxseptra{\aformula_{\bound}'}{\aformulabis} \implies (\aformula_{\bound}' \septraction \true)
        & \mbox{Induction hypothesis} \\
        2 & \boxseptra{\aformula_{\bound-1}'}{\aformulabis} \implies (\aformula_{\bound-1}' \septraction \true)
        & \mbox{Induction hypothesis} \\
        3 &  \ATOM{\avariable_i} \separate (\aformula_{\bound}' \lor \aformula_{\bound-1}') \implies \aformula
        & \mbox{\ref{csl:lemma62:prime-formula:prop5}, Theorem~\ref{theo:starCompleteness}} \\
        4 & \boxseptra{\aformula}{\aformulabis} \separate \ATOM{\avariable_i} \implies \boxseptra{\aformula_{\bound}'}{\aformulabis} \lor \boxseptra{\aformula_{\bound-1}'}{\aformulabis}
        & \mbox{\ref{csl:lemma62:prime-formula:prop6}, Theorem~\ref{theo:starCompleteness}} \\
        5 & \boxseptra{\aformula_{\bound}'}{\aformulabis} \lor \boxseptra{\aformula_{\bound-1}'}{\aformulabis}
            \implies (\aformula_{\bound}' \septraction \true) \lor (\aformula_{\bound-1}' \septraction \true)
          & \mbox{PC, 1, 2}\\
        6 & (\aformula_{\bound}' \septraction \true) \lor (\aformula_{\bound-1}' \septraction \true)
        \implies (\aformula_{\bound}' \lor \aformula_{\bound-1}' \septraction \true)
          & \mbox{\ref{mwAx:OrL}}\\
        7 & \boxseptra{\aformula}{\aformulabis} \separate \ATOM{\avariable_i} \implies (\aformula_{\bound}' \lor \aformula_{\bound-1}' \septraction \true)
          & \mbox{\ref{rule:imptr}, 4, 5, 6} \\ 
        8 & \boxseptra{\aformula}{\aformulabis} \implies (\ATOM{\avariable_i}  \septraction
        \boxseptra{\aformula}{\aformulabis} \separate \ATOM{\avariable_i})
        & \mbox{(\ref{csl:lemma62:ddagger})} \\ 
        9 & (\ATOM{\avariable_i}  \septraction
        \boxseptra{\aformula}{\aformulabis} \separate \ATOM{\avariable_i})
        \implies\\[-2pt] 
          & \quad (\ATOM{\avariable_i}  \septraction
        (\aformula_{\bound}' \lor \aformula_{\bound-1}' \septraction \true))
           & \mbox{\ref{mwAx:ImpR}, 7}\\
        10 & \big(\ATOM{\avariable_i}  \septraction
        (\aformula_{\bound}' \lor \aformula_{\bound-1}' \septraction \true )\big) \implies\\[-2pt]
        & \quad (\ATOM{\avariable_i} \separate (\aformula_{\bound}' \lor \aformula_{\bound-1}') \septraction \true)
        & \mbox{\ref{mwAx:Curry}} \\
        11 & (\ATOM{\avariable_i} \separate (\aformula_{\bound}' \lor \aformula_{\bound-1}') \septraction \true) \implies (\aformula \septraction \true)
           & \mbox{\ref{mwAx:ImpL}, 3} \\
        12 & \boxseptra{\aformula}{\aformulabis} \implies (\aformula \septraction \true)
        & \mbox{\ref{rule:imptr}, 8, 9, 10, 11}
      \end{syntproof}
\caption{Proof of~Lemma~\ref{lemma:magicwandPSLelim}: the final derivation.}
\label{figure-final-derivation}
\end{figure} 
    \end{description}
  \end{description}

\end{proof}

Lemma~\ref{lemma:magicwandPSLelim} in which $\aformula$ and $\aformulabis$ 
are core types can be extended to arbitrary Boolean combinations of core formulae, as we show that the distributivity of $\septraction$ over disjunctions is provable in $\magicwandsys$.
As a consequence of this development, we achieve the main result of the paper.

\begin{thm}\label{theo:PSLcompleteAx}
$\magicwandsys$ is sound and complete for \slSW.
\end{thm}

\begin{proof}
Soundness of the proof system $\magicwandsys$ has been already established earlier, see Lemma~\ref{lemma:magicwandPSLvalid}.
As far as the completeness proof is concerned, its structure is very similar to the proof
of Theorem~\ref{theo:starCompleteness} except that we have to be able to handle the separating
implication. In order to be self-contained, we reproduce some of its arguments albeit adapted to
$\magicwandsys$.

We need to show that for every formula $\aformula$
in \slSW, there is a Boolean combination of core formulae $\aformulabis$ such that
$\prove_{\magicwandsys} \aformula \Leftrightarrow \aformulabis$. In order to conclude the proof,
when $\aformula$ is valid for \slSW, by soundness of $\magicwandsys$, we obtain that
$\aformulabis$ is valid too and therefore  $\prove_{\magicwandsys} \aformulabis$
as $\coresys$ is a subsystem of $\magicwandsys$ and $\coresys$ is complete by
Theorem~\ref{theo:corePSLcompl}. By propositional reasoning, we get that
$\prove_{\magicwandsys} \aformula$.

In order to show that every formula $\aformula$ has a provably equivalent Boolean combination
of core formulae, we heavily rely on Corollary~\ref{lemma:starPSLelim} and on Lemma~\ref{lemma:magicwandPSLelim}. The proof is by simple induction
on the number of occurrences of $\separate$ or $\magicwand$ in $\aformula$ that are not involved in the definition of some core
formula of the form $\size \geq \inbound$ or $\alloc{\avariable}$. For the base case, when $\aformula$
has no occurrence of the separating connectives, $\avariable = \avariablebis$ and $\avariable \Ipto
\avariablebis$ are already core formulae, whereas $\emp$ is logically equivalent to $\neg \size \geq 1$.

Before performing the induction step, let us observe that in $\magicwandsys$, the replacement of provably equivalent
formulae holds true, which is stated as follows:

\begin{enumerate}[label=\textbf{R\arabic*}]
\setcounter{enumi}{0}
\item\label{SC-auxlemmaR1} Let $\aformula, \aformula'$ and $\aformulabis$ be
formulae of \slSW such that
$\prove_{\magicwandsys} \aformula \Leftrightarrow \aformula'$. Then,
$$\prove_{\magicwandsys}  \aformulabis[\aformula]_{\rho} \Rightarrow
\aformulabis[\aformula']_{\rho}$$
\end{enumerate}

In order to prove~\ref{SC-auxlemmaR1}, we are almost done as we have already shown~\ref{SC-auxlemmaR0}
in the proof of Theorem~\ref{theo:starCompleteness} and the same properties hold for \slSW though
the language is richer.

As a direct consequence of the admissibility of the rules~\ref{mwAx:ImpL} and~\ref{mwAx:ImpR}
from Lemma~\ref{lemma:septractionadmissible},
the rules below are also admissible:
$$\inference{\aformula \Leftrightarrow \aformula'}{\aformula \magicwand \aformulabis \Leftrightarrow \aformula'
\magicwand \aformulabis} \ \ \ \ \ \inference{\aformula \Leftrightarrow \aformula'}{\aformulabis \magicwand \aformula \Leftrightarrow \aformulabis
\magicwand \aformula'}
$$

We need  the two rules as $\magicwand$ is not commutative.
Consequently, by structural induction on $\aformulabis$, one can conclude that
$\prove_{\magicwandsys} \aformula \Leftrightarrow \aformula'$
 implies $\prove_{\magicwandsys}  \aformulabis[\aformula]_{\rho} \Rightarrow
\aformulabis[\aformula']_{\rho}$.

Now, assume $\aformula$ is a formula in \slSW. Without loss of generality, we can assume that
the separating connectives in $\aformula$ are restricted to $\separate$ and $\septraction$
for the occurrences that are not related to abbreviations for core formulae. Indeed,
$\aformulabis'\septraction\aformulabis$ is a shortcut for $\neg(\aformulabis'\magicwand\neg\aformulabis)$
and therefore one can replace every occurrence of $\aformulabis'\magicwand\aformulabis$
by $\neg (\aformulabis'\septraction\neg\aformulabis)$ assuming that $\aformulabis'$ and $\aformulabis$ are already of the appropriate shape.
Such a replacement is possible thanks to~\ref{SC-auxlemmaR1}.

Assume that $\aformula$ is a formula in $\seplogic{\separate,\septraction}$ with $n+1$ occurrences of $\separate$
or  $\septraction$ not involved in the definition of core formulae.

\cut{
Let $\aformulabis$ be a subformula of $\aformula$ (at the occurrence $\rho$) of the form $\aformulabis_1 \separate \aformulabis_2$ such that
$\aformulabis_1$ and $\aformulabis_2$ are  Boolean combinations of
core formulae. As in the proof of Theorem~\ref{theo:starCompleteness}, we can derive a Boolean combination of
core formulae such that
$$
\vdash_{\magicwandsys} \aformulabis_1 \separate \aformulabis_2 \Leftrightarrow 
\bigvee_{j_1 \in \interval{1}{n_1}, j_2 \in \interval{1}{n_2}} \aformulabis^{j_1,j_2}.
$$
Consequently (thanks to the property~\ref{SC-auxlemmaR1}), we obtain
$$
\vdash_{\magicwandsys} \aformula \Leftrightarrow \aformula[\bigvee_{j_1 \in \interval{1}{n_1}, j_2 \in \interval{1}{n_2}} \aformulabis^{j_1,j_2}]_{\rho}
$$
}

Let $\aformulabis$ be a subformula of $\aformula$ (at the occurrence $\rho$) of the form $\aformulabis_1 \septraction \aformulabis_2$ such that
$\aformulabis_1$ and $\aformulabis_2$ are in $\boolcomb{\coreformulae{\asetvar}{\bound_1}}$ and 
$\boolcomb{\coreformulae{\asetvar}{\bound_2}}$, respectively.
By propositional reasoning, one can show that there are formulae in disjunctive normal form
$\aformulabis_1^1 \vee \cdots \vee \aformulabis_1^{n_1}$ and $\aformulabis_2^1 \vee \cdots \vee \aformulabis_2^{n_2}$
such that $\vdash_{\coresys} \aformulabis_i \Leftrightarrow \aformulabis_i^1 \vee \cdots \vee \aformulabis_i^{n_i}$
for $i \in \set{1,2}$, and moreover every $\aformulabis_i^j$'s is a core type in $\coretype{\asetvar}{\max(\card{\asetvar},\bound_1,\bound_2)}$.
Again, by using propositional reasoning but this time establishing also distributivity of $\vee$ over $\septraction$, we have
$$
\vdash_{\magicwandsys} \aformulabis_1 \septraction \aformulabis_2 \Leftrightarrow
\bigvee_{j_1 \in \interval{1}{n_1}, j_2 \in \interval{1}{n_2}} \aformulabis_1^{j_1} \septraction \aformulabis_2^{j_2}.
$$
We rely on Lemma~\ref{lemma:magicwandPSLelim}, and conclude that there is 
a conjunction of core formulae $ \aformulabis^{j_1,j_2}$ in $\conjcomb{\coreformulae{\asetvar}{\max(\card{\asetvar},\bound_1,\bound_2)}}$
such that  $\vdash_{\magicwandsys} \aformulabis_1^{j_1} \septraction \aformulabis_2^{j_2} \Leftrightarrow \aformulabis^{j_1,j_2}$.
By propositional reasoning, we get
$$
\vdash_{\magicwandsys} \aformulabis_1 \septraction \aformulabis_2 \Leftrightarrow
\bigvee_{j_1 \in \interval{1}{n_1}, j_2 \in \interval{1}{n_2}} \aformulabis^{j_1,j_2}.
$$
Consequently (thanks to the property~\ref{SC-auxlemmaR1}), we obtain
$$
\vdash_{\magicwandsys} \aformula \Leftrightarrow \aformula[\bigvee_{j_1 \in \interval{1}{n_1}, j_2 \in \interval{1}{n_2}} \aformulabis^{j_1,j_2}]_{\rho}
$$
Note that the right-hand side formula has  $n$ occurrences of the separating
connnectives that are not involved in the definition of some core
formula. The induction hypothesis applies, which concludes the proof.

The case when $\aformulabis$ is a subformula of $\aformula$ (at the occurrence $\rho$) 
of the form $\aformulabis_1 \separate \aformulabis_2$ is treated as in the proof of Theorem~\ref{theo:starCompleteness}
and therefore is omitted herein.

\end{proof}

\section{Related work}\label{section:related-work}
In this section, we briefly compare our Hilbert-style proof system $\magicwandsys$ 
with existing proof systems for \slSW, fragments or extensions and we recall a few landmark works
proposing proof systems for abstract separation logics or for logics that are variants of Boolean BI. 
Those
latter proof systems are not necessarily Hilbert-style and may contain labels or other similar machineries.
So, this section completes the presentation of the context from Section~\ref{section:introduction}
while pinpointing the main original features of our calculus.
Finally, we also evoke several works that use the idea of axiomatising a fragment of a logic
and to provide in the proof system means to transform any formula into an equivalent formula from that
fragment. This is clearly similar to the approach we have followed, but we aim at picking examples from
outside the realm of spatial and resource logics. In order to keep the length of this section reasonable, 
we limit ourselves to the main bibliographical entries but  additional relevant works can be found in the cited materials. 

\vspace{15pt}
\noindent\textbf{Proof systems for quantifier-free separation logic.}  
Surprisingly, as far as we know, sound and complete proof systems for \slSW are very rare
and the only system we are aware of is a tableaux-based calculus from~\cite{Galmiche&Mery10}
with labelled formulae (each formula is enriched with a label to be interpreted by 
some heap)  and with resource graphs to encode symbolically constraints between heap expressions (i.e. labels). Of course, 
 translations from separation logics into
logics or theories have been designed, see e.g.~\cite{Calcagno&Gardner&Hague05,Reynoldsetal16},
but the finding of proof systems for \slSW with all Boolean connectives and the separating
connectives $\separate$ and $\magicwand$ has been quite challenging. 
Unlike~\cite{Galmiche&Mery10}, $\magicwandsys$ uses only \slSW formulae and therefore
can be viewed as a quite orthodox Hilbert-style calculus with no extra syntactic objects. 
In particular,  $\magicwandsys$ has no syntactic machinery to refer to  heaps or to other
semantical objects related to \slSW.
In~\cite{Galmiche&Mery10}, the resource graphs attached to the tableaux are designed
to reason about heap constraints, and to provide control for designing strategies that lead 
to termination. Interestingly, the calculus in~\cite{Galmiche&Mery10} is 
intended to be helpful 
to synthesize countermodels (which is a standard feature for labelled deduction systems~\cite{Gabbay96d}) 
or to be extended to the first-order case, which is partly done 
in~\cite{Galmiche&Mery10} but we know that completeness is theoretically impossible.
Besides, a sound labelled sequent calculus for the first-order extension of \slSW
is presented in~\cite{Hou&Gore&Tiu15} but completeness for the sublogic \slSW is not established.
The calculus in~\cite{Hou&Gore&Tiu15} has also labels, which differs from our puristic approach. 
A complete sequent-style calculus for the symbolic heap fragment has been designed quite early 
in~\cite{Berdine&Calcagno&OHearn04} but does not deal with full \slSW (in particular it is
not closed under Boolean connectives and does not contain the separating implication). 
A complexity-wise optimal decision procedure for the symbolic heap fragment is designed in~\cite{Cooketal11}
based on a characterisation in terms of homomorphisms.

\subsection*{Frameworks for abstract separation logics.} 
Bunched logics, such as the bunched logic BI introduced in~\cite{OHearn&Pym99},  
are known to be closely related to separation logics that can be
viewed as concretisation of (Boolean) BI with models made of memory states, 
see e.g.~\cite{Pym02,Reynolds02,Galmiche&Mery05,Pym&Spring&OHearn18}.
Actually, bunched logics come with different flavours, Boolean BI being considered as  the genuine abstract 
version of \slSW. Though Boolean BI has been shown undecidable in~\cite{LarcheyWendling&Galmiche13,Brotherston&Kanovich14},
a Hilbert-style axiomatisation can be found in~\cite{Galmiche&Larchey06}. Our proof system $\magicwandsys$ inherits
all the axiom schemas and inference rules for Boolean BI from~\cite{Galmiche&Larchey06}, which is expected
as 
\slSW can be viewed as Boolean BI on concrete heaps but with the notable difference of having built-in atomic formulae 
$\avariable = \avariablebis$
and $\avariable \Ipto \avariablebis$. 
Bunched logics, such as Boolean BI, can be defined in several ways, for instance assuming classical or intuitionistic
connectives, and in~\cite{Brotherston12}, a unified proof theory based on display calculi~\cite{Belnap82}
is designed for a variety of four bunched logics, including Boolean BI (see also the nested
sequent calculus for Boolean BI in~\cite{Park&Seo&Park13}). In display calculi, structural connectives
enrich the sequent-style structures, providing a family of structural connectives accompanying the standard
comma from sequent-style calculi. The main results in~\cite{Brotherston12} include cut-elimination, soundness
and completeness. So, compared to our calculus $\magicwandsys$, the calculi in~\cite{Brotherston12} are designed
for logics with more abstract semantical structures and owns a proof-theoretical machinery that does not
include labels but instead  complex structured sequents.

The quest for designing frameworks dedicated to classes of abstract separation logics have been pursued in several directions. 
For instance, models for Boolean BI are typically relational commutative monoids but  properties can be
added leading to a separation theory. In~\cite{Brotherston&Villard14}, 
a hybrid version of Boolean BI is introduced, called HyBBI, in which nominals (in the sense of hybrid modal logics, see 
e.g.~\cite{Areces&Blackburn&Marx01}) are added in order to be able to express rich standard properties in
separation theory, such as cancellativity. Not only an Hilbert-style proof system is provided for HyBBI~\cite{Brotherston&Villard14}
but also a parametric completeness result is shown. More precisely, any extension of the proof system for HyBBI
with a set of specific axioms is actually complete with respect to the class of models that satisfy the axioms,
which is analogous to Sahlqvist's Theorem for modal logics~\cite{Sahlqvist75,Blackburn&deRijke&Venema01}.
This provides a very general means to axiomatise variants of Boolean BI but at the cost of having the extra machinery
for nominals. Moreover,  as HyBBI and its extensions are abstract separation logics with no atomic formulae of
the form $\avariable = \avariablebis$ or $\avariable \Ipto \avariablebis$, the tools developed in~\cite{Brotherston&Villard14}
are of no help to design an Hilbert-style proof system for \slSW (except that its part dealing with Boolean BI 
is precisely borrowed from~\cite{Galmiche&Larchey06} too).

Besides, in~\cite{Houetal18} labelled sequent calculi are designed for several abstract separation logics by considering different sets of
properties. The sequents contain labelled formulae (a formula prefixed by a label to be interpreted
as an abstract heap)  as well as relational atoms to express relationships between abstract heaps. Though the framework in~\cite{Houetal18}
is modular and very general to handle abstract separation logics, it is not tailored to separation logics with concrete semantics,
see~\cite[Section 7]{Houetal18} for possible future directions. In contrast, as explained already, the paper~\cite{Hou&Gore&Tiu15} deals with first-order
separation logic with concrete semantics and presents a sound labelled sequent calculus for it. Of course, the calculus 
cannot be complete but more importantly in the context of the current paper, completeness is not established for the quantifier-free fragment. 
In~\cite{Hou&Gore&Tiu15}, the sequents contain labelled formulae and relational atoms, similarly to~\cite{Houetal18} (see also~\cite{Hou15}).
Hence, this does not meet our requirements to have a pure axiomatisation in which only logical formulae
from quantifier-free separation logic are allowed.

Modularity of the approaches from~\cite{Brotherston12,Brotherston&Villard14,Houetal18} is further developed in the recent 
work~\cite{Docherty&Pym18,Docherty19} by proposing a framework for labelled tableaux systems parametrised by the choice of separation theories
(in the very sense of~\cite{Brotherston&Villard14}). It is remarkable that the developments in~\cite{Docherty&Pym18,Docherty19} 
are very general as it can handle separation theories that can be expressed in the rich class of so-called coherent first-order
formulae, included in the first-order fragment $\Pi_2$. The first-order axioms are directly translated into inference rules.
The calculi use labelled formulae (every formula is decorated by a sign and by a label) as well as constraints enforcing
properties between worlds/resources. Unlike~\cite{Galmiche&Mery10}, the reasoning about labels is not outsourced 
but handled directly by the calculus. As several works mentioned above, the framework in~\cite{Docherty&Pym18,Docherty19} does not
provide for free a proof system for \slSW (which might have been a close cousin of the one in~\cite{Galmiche&Mery10}).
More importantly, similarly to the works~\cite{Galmiche&Mery10,Brotherston&Villard14,Houetal18}, the labelled tableaux systems handle
syntactic objects referring to semantical concepts related to the abstract separation logics that go beyond the only presence
of formulae. In a way, modularity of the approach prevents from having a puristic calculus for \slSW, apart from the fact that
\slSW is not part of the logics handled in~\cite{Docherty&Pym18}.

\subsection*{Axiomatising knowledge logics with reduction axioms.} 
In order to conclude this section, let us recall that the derivations in 
$\magicwandsys$ are able to simulate the bottom-up elimination of
separating connectives, leading to Boolean combinations of core formulae
for which the system $\magicwandsys$ is also complete. As the core formulae
are (simple) formulae in \slSW, the axiomatisation provided by $\magicwandsys$
uses only \slSW formulae and is complete for the full logic \slSW (and not only
for Boolean combinations of core formulae).  
Note that as a by-product of our completeness proof for \slSW, we get expressive
completeness of \slSW with respect to Boolean combinations of core formulae,
with a proof different from the developments in~\cite{Lozes04bis,Brochenin&Demri&Lozes09,Echenim&Iosif&Peltier19}.

This general principle described above is  familiar for axiomatising dynamic epistemic 
logics in which dynamic connectives might be eliminated with the help of so-called 
 \emph{reduction axioms}, see e.g. standard examples 
in~\cite{vanDitmarsch&vanderHoek&Kooi08,vanBenthem2011ldii,WangC13,Fervari&VelazquezQuesada19}.
In a nutshell, every formula  containing a dynamic operator is provably reduced to a formula without such an operator.
Completeness is then established thanks to the completeness of the underlying `basic' language,
A similar approach for the linear $\mu$-calculus is recently presented in~\cite{Doumane17} 
for which a form of constructive completeness is advocated, see also~\cite{Luck18}. 
Hilbert-style axiomatisations following similar high-level principles for the
modal separation logics MSL($\separate$,$\Diamond$) and MSL($\separate$,$\langle \neq \rangle$)
introduced in~\cite{Demri&Fervari19}, have been designed in~\cite{Demri&Fervari&Mansutti19}.

\section{Conclusion}\label{section:conclusion}
We  presented a method to axiomatise internally quantifier-free separation logic \slSW 
based on the axiomatisation of Boolean combinations of core formulae
(and even more precisely, based on the  restricted fragment of core types).
We designed the first proof system for \slSW
that is completely internal and highlights the essential ingredients
of the heaplet semantics.
The fact that the calculus is internal simply means that the 
axioms and inference rules
involve schemas instantiated by formulae in \slSW (no use of nominals, labels or other
syntactic objects that are not \slSW formulae).
Obviously, the Hilbert-style proof system presented in the paper 
is of theoretical interest, at least to grasp what are the essential features of \slSW . Still, it remains
to be seen whether applications are possible for designing decision procedures, for instance
to feed provers with appropriate axiom instances to accelerate the proof search. 
Furthermore, we have not investigated whether the proof system $\magicwandsys$ (see Figure~\ref{figure-full-proof-system})
can be simplified without loosing completeness. This might be rewarding for using
the calculus for other logics or for other applications. Most probably the most obvious 
part to study in that respect would be $\coresys(\separate)$. 

To provide further evidence that our method is robust, it is desirable to
apply it to axiomatise other separation logics, for instance
by  adding the list segment predicate $\ls$~\cite{Berdine&Calcagno&OHearn04}
(or more generally user-defined inductive predicates)
or by adding first-order quantification. A key step in our approach is first to show  that
the logic admits a characterisation in terms of core formulae and such formulae
need to be designed adequately. Of course, it is required that the set of valid formulae is recursively
enumerable, which discards any attempt with $\seplogic{\separate,\magicwand,\ls}$ or with
the first-order version of \slSW~\cite{DemriLM18,Brochenin&Demri&Lozes12}.
The second part of the paper~\cite{Demri&Lozes&Mansutti20} introduces  an extension of  $\seplogic{\separate,\ls}$
 and  presents an axiomatisation  with our method. More separation logics could be axiomatised that way,
other good candidates are the version of separation logic with one individual variable
studied in~\cite{Demrietal17} as well as the quantifier-free separation logic with general universes
from~\cite{Echenim&Iosif&Peltier19}. 

\noindent
{\bf Acknowledgements.} We would like to thank the anonymous reviewers for their numerous 
remarks and suggestions that help us to improve the quality of the document. 

\bibliographystyle{alpha}
\bibliography{paper-bibliography}

\newcommand{\etalchar}[1]{$^{#1}$}
\begin{thebibliography}{DGLWM17}

\bibitem[ABM01]{Areces&Blackburn&Marx01}
C.~Areces, P.~Blackburn, and M.~Marx.
\newblock Hybrid logics: characterization, interpolation and complexity.
\newblock {\em The Journal of Symbolic Logic}, 66(3):977--1010, 2001.

\bibitem[BCO04]{Berdine&Calcagno&OHearn04}
J.~Berdine, C.~Calcagno, and P.W. O'Hearn.
\newblock A decidable fragment of separation logic.
\newblock In {\em FST\&TCS'04}, volume 3328 of {\em LNCS}, pages 97--109.
  Springer, 2004.

\bibitem[BDL09]{Brochenin&Demri&Lozes09}
R.~Brochenin, S.~Demri, and {\'{E}}.~Lozes.
\newblock Reasoning about sequences of memory states.
\newblock {\em Annals of Pure and Applied Logic}, 161(3):305--323, 2009.

\bibitem[BDL12]{Brochenin&Demri&Lozes12}
R.~Brochenin, S.~Demri, and {\'{E}}.~Lozes.
\newblock On the almighty wand.
\newblock {\em Information and Computation}, 211:106--137, 2012.

\bibitem[BdRV01]{Blackburn&deRijke&Venema01}
P.~Blackburn, M.~de~Rijke, and Y.~Venema.
\newblock {\em Modal Logic}.
\newblock Cambridge University Press, 2001.

\bibitem[Bel82]{Belnap82}
N.~Belnap.
\newblock Display logic.
\newblock {\em Journal of Philosophical Logic}, 11:375--417, 1982.

\bibitem[BIP10]{Bozga&Iosif&Perarnau10}
M.~Bozga, R.~Iosif, and S.~Perarnau.
\newblock Quantitative separation logic and programs with lists.
\newblock {\em Journal of Automated Reasoning}, 45(2):131--156, 2010.

\bibitem[BK14]{Brotherston&Kanovich14}
J.~Brotherston and M.~Kanovich.
\newblock Undecidability of propositional separation logic and its neighbours.
\newblock {\em Journal of the Association for Computing Machinery}, 61(2),
  2014.

\bibitem[BK18]{Brotherston&Kanovich18}
J.~Brotherston and M.~Kanovich.
\newblock On the complexity of pointer arithmetic in separation logic.
\newblock In {\em APLAS'18}, volume 11275 of {\em LNCS}, pages 329--349.
  Springer, 2018.

\bibitem[Bro12]{Brotherston12}
J.~Brotherston.
\newblock Bunched logics displayed.
\newblock {\em Studia Logica}, 100(6):1223--1254, 2012.

\bibitem[BV14]{Brotherston&Villard14}
J.~Brotherston and J.~Villard.
\newblock Parametric completeness for separation theories.
\newblock In {\em POPL'14}, pages 453--464. ACM, 2014.

\bibitem[CGH05]{Calcagno&Gardner&Hague05}
C.~Calcagno, Ph. Gardner, and M.~Hague.
\newblock From separation logic to first-order logic.
\newblock In {\em FoSSaCS'05}, volume 3441 of {\em LNCS}, pages 395--409.
  Springer, 2005.

\bibitem[CHO{\etalchar{+}}11]{Cooketal11}
B.~Cook, C.~Haase, J.~Ouaknine, M.~Parkinson, and J.~Worrell.
\newblock Tractable reasoning in a fragment of separation logic.
\newblock In {\em CONCUR'11}, volume 6901 of {\em LNCS}, pages 235--249.
  Springer, 2011.

\bibitem[COY01]{Calcagno&Yang&OHearn01}
C.~Calcagno, P.W. O'Hearn, and H.~Yang.
\newblock Computability and complexity results for a spatial assertion language
  for data structures.
\newblock In {\em FST\&TCS'01}, volume 2245 of {\em LNCS}, pages 108--119.
  Springer, 2001.

\bibitem[DD15]{DemriDeters15bis}
S.~Demri and M.~Deters.
\newblock Separation logics and modalities: A survey.
\newblock {\em Journal of Applied Non-Classical Logics}, 25(1):50--99, 2015.

\bibitem[DF19]{Demri&Fervari19}
S.~Demri and R.~Fervari.
\newblock The power of modal separation logics.
\newblock {\em Journal of Logic and Computation}, 29(8):1139--1184, 2019.

\bibitem[DFM19]{Demri&Fervari&Mansutti19}
S.~Demri, R.~Fervari, and A.~Mansutti.
\newblock Axiomatising logics with separating conjunction and modalities.
\newblock In {\em JELIA'19}, volume 11468 of {\em LNAI}, pages 692--708.
  Springer, 2019.

\bibitem[DGLWM17]{Demrietal17}
S.~Demri, D.~Galmiche, D.~Larchey-Wendling, and D.~Mery.
\newblock Separation logic with one quantified variable.
\newblock {\em Theory of Computing Systems}, 61:371--461, 2017.

\bibitem[DLM18a]{DemriLM18}
S.~Demri, {\'{E}}.~Lozes, and A.~Mansutti.
\newblock The effects of adding reachability predicates in propositional
  separation logic.
\newblock In {\em FoSSaCS'18}, volume 10803 of {\em LNCS}, pages 476--493.
  Springer, 2018.

\bibitem[DLM18b]{Demri&Lozes&Mansutti18bis}
S.~Demri, {\'{E}}.~Lozes, and A.~Mansutti.
\newblock The effects of adding reachability predicates in propositional
  separation logic.
\newblock arXiv:1810.05410, October 2018.
\newblock 44 pages. Long version of~\cite{DemriLM18}.

\bibitem[DLM20]{Demri&Lozes&Mansutti20}
S.~Demri, {\'{E}}.~Lozes, and A.~Mansutti.
\newblock Internal calculi for separation logics.
\newblock In {\em CSL'20}, Leibniz International Proceedings in Informatics,
  pages 19:1--19:18. Leibniz-Zentrum f{\"u}r Informatik, 2020.

\bibitem[Doc19]{Docherty19}
S.~Docherty.
\newblock {\em Bunched logics: a uniform approach}.
\newblock PhD thesis, University College London, 2019.

\bibitem[Dou17]{Doumane17}
A.~Doumane.
\newblock Constructive completeness for the linear-time {\(\mu\)}-calculus.
\newblock In {\em LiCS'17}, pages 1--12. {IEEE} Computer Society, 2017.

\bibitem[DP18]{Docherty&Pym18}
S.~Docherty and D.~Pym.
\newblock Modular tableaux calculi for separation theories.
\newblock In {\em FoSSaCS'18}, volume 10803 of {\em LNCS}, pages 441--458.
  Springer, 2018.

\bibitem[EIP19]{Echenim&Iosif&Peltier19}
M.~Echenim, R.~Iosif, and N.~Peltier.
\newblock The {B}ernays-{S}ch{\"{o}}nfinkel-{R}amsey class of separation logic
  on arbitrary domains.
\newblock In {\em FoSSaCS'19}, volume 11425 of {\em LNCS}, pages 242--259.
  Springer, 2019.

\bibitem[FVQ19]{Fervari&VelazquezQuesada19}
R.~Fervari and F.~R. Vel\'azquez-Quesada.
\newblock Introspection as an action in relational models.
\newblock {\em Journal of Logical and Algebraic Methods in Programming},
  108:1--23, 2019.

\bibitem[Gab96]{Gabbay96d}
D.~Gabbay.
\newblock {\em Labelled Deductive Systems}.
\newblock Oxford University Press, 1996.

\bibitem[GLW06]{Galmiche&Larchey06}
D.~Galmiche and D.~Larchey-Wending.
\newblock Expressivity properties of boolean {BI} through relational models.
\newblock In {\em FST\&TCS'06}, volume 4337 of {\em LNCS}, pages 358--369.
  Springer, 2006.

\bibitem[GM05]{Galmiche&Mery05}
D.~Galmiche and D.~Mery.
\newblock Characterizing provability in {BI}'s pointer logic through resource
  graphs.
\newblock In {\em LPAR'05}, volume 3835 of {\em LNCS}, pages 459--473.
  Springer, 2005.

\bibitem[GM10]{Galmiche&Mery10}
D.~Galmiche and D.~M\'ery.
\newblock Tableaux and resource graphs for separation logic.
\newblock {\em Journal of Logic and Computation}, 20(1):189--231, 2010.

\bibitem[GvD06]{Goranko&vanDrimmelen06}
V.~Goranko and G.~van Drimmelen.
\newblock Complete axiomatization and decidability of alternating-time temporal
  logic.
\newblock {\em Theoretical Computer Science}, 353(1-3):93--117, 2006.

\bibitem[HCGT18]{Houetal18}
Z.~H{\'o}u, R.~Clouston, R.~Gor\'e, and A.~Tiu.
\newblock Modular labelled sequent calculi for abstract separation logics.
\newblock {\em ACM Transactions on Computational Logic}, 19(2):13:1--13:35,
  2018.

\bibitem[HGT15]{Hou&Gore&Tiu15}
Z.~H{\'o}u, R.~Gor{\'{e}}, and A.~Tiu.
\newblock Automated theorem proving for assertions in separation logic with all
  connectives.
\newblock In {\em CADE'15}, volume 9195 of {\em LNCS}, pages 501--516.
  Springer, 2015.

\bibitem[H{\'o}u15]{Hou15}
Z.~H{\'o}u.
\newblock {\em Labelled sequent calculi and automated reasoning for assertions
  in separation logic}.
\newblock PhD thesis, Australian National University, November 2015.

\bibitem[IO01]{Ishtiaq&OHearn01}
S.~Ishtiaq and P.W. O'Hearn.
\newblock {BI} as an assertion language for mutable data structures.
\newblock In {\em POPL'01}, pages 14--26. ACM, 2001.

\bibitem[Kai95]{Kaivola95}
R.~Kaivola.
\newblock Axiomatising linear time mu-calculus.
\newblock In {\em CONCUR'95}, volume 962 of {\em LNCS}, pages 423--437.
  Springer, 1995.

\bibitem[LG13]{LarcheyWendling&Galmiche13}
D.~Larchey{-}Wendling and D.~Galmiche.
\newblock Nondeterministic phase semantics and the undecidability of {B}oolean
  {BI}.
\newblock {\em ACM Transactions on Computational Logic}, 14(1), 2013.

\bibitem[LMX16]{Larsen&Mardare&Xue16}
K.G. Larsen, R.~Mardare, and B.~Xue.
\newblock Probabilistic mu-calculus: Decidability and complete axiomatization.
\newblock In {\em FST\&TCS'16}, volume~65 of {\em LIPIcs}, pages 25:1--25:18.
  Schloss Dagstuhl - Leibniz-Zentrum fuer Informatik, 2016.

\bibitem[Loz04a]{Lozes04bis}
{\'{E}}.~Lozes.
\newblock {\em Expressivit{\'e} des Logiques Spatiales}.
\newblock PhD thesis, ENS Lyon, 2004.

\bibitem[Loz04b]{Lozes04}
{\'{E}}.~Lozes.
\newblock Separation logic preserves the expressive power of classical logic.
\newblock In {\em SPACE'04}, 2004.

\bibitem[L{\"{u}}c18]{Luck18}
M.~L{\"{u}}ck.
\newblock Axiomatizations of team logics.
\newblock {\em Annals of Pure and Applied Logic}, 169(9):928--969, 2018.

\bibitem[Man18]{Mansutti18}
A.~Mansutti.
\newblock Extending propositional separation logic for robustness properties.
\newblock In {\em FST\&TCS'18}, volume 122 of {\em LIPIcs}, pages 42:1--42:23.
  Schloss Dagstuhl - Leibniz-Zentrum fuer Informatik, 2018.

\bibitem[Man20]{Mansutti20}
A.~Mansutti.
\newblock {\em Reasoning with {S}eparation {L}ogics: {C}omplexity, {E}xpressive
  {P}ower, {P}roof {S}ystems}.
\newblock PhD thesis, Universit\'e Paris-Saclay, December 2020.

\bibitem[O'H12]{OHearn12}
P.W. O'Hearn.
\newblock A primer on separation logic.
\newblock In {\em Software Safety and Security: Tools for Analysis and
  Verification}, volume~33 of {\em NATO Science for Peace and Security Series},
  pages 286--318, 2012.

\bibitem[OP99]{OHearn&Pym99}
P.W. O'Hearn and D.~Pym.
\newblock The logic of bunched implications.
\newblock {\em Bulletin of Symbolic Logic}, 5(2):215--244, 1999.

\bibitem[PSO18]{Pym&Spring&OHearn18}
D.~Pym, J.~Spring, and P.W. O'Hearn.
\newblock Why separation logic works.
\newblock {\em Philosophy \& Technology}, pages 1--34, 2018.

\bibitem[PSP13]{Park&Seo&Park13}
J.~Park, J.~Seo, and S.~Park.
\newblock A theorem prover for {B}oolean {BI}.
\newblock In {\em POPL'13}, pages 219--232. {ACM}, 2013.

\bibitem[PWZ13]{Piskac&Wies&Zufferey13}
R.~Piska{\'c}, Th. Wies, and D.~Zufferey.
\newblock Automating separation logic using {SMT}.
\newblock In {\em CAV'13}, volume 8044 of {\em LNCS}, pages 773--789. Springer,
  2013.

\bibitem[Pym02]{Pym02}
D.~Pym.
\newblock {\em The Semantics and Proof Theory of the Logic of Bunched
  Implications}, volume~26 of {\em Applied Logic}.
\newblock Kluwer Academic Publishers, 2002.

\bibitem[Rey01]{Reynolds01}
M.~Reynolds.
\newblock An axiomatization of full computation tree logic.
\newblock {\em The Journal of Symbolic Logic}, 66(3):1011--1057, 2001.

\bibitem[Rey02]{Reynolds02}
J.C. Reynolds.
\newblock Separation logic: a logic for shared mutable data structures.
\newblock In {\em LiCS'02}, pages 55--74. IEEE, 2002.

\bibitem[RISK16]{Reynoldsetal16}
A.~Reynolds, R.~Iosif, C.~Serban, and T.~King.
\newblock A decision procedure for separation logic in {SMT}.
\newblock In {\em ATVA'16}, volume 9938 of {\em LNCS}, pages 244--261, 2016.

\bibitem[Sah75]{Sahlqvist75}
H.~Sahlqvist.
\newblock Completeness and correspondence in the first and second order
  semantics for modal logics.
\newblock In S.~Kanger, editor, {\em 3rd Scandinavian Logic Symposium, Uppsala,
  Sweden, 1973}, pages 110--143. North Holland, 1975.

\bibitem[SV18]{Schroder&Venema18}
L.~Schr{\"{o}}der and Y.~Venema.
\newblock Completeness of flat coalgebraic fixpoint logics.
\newblock {\em ACM Transactions on Computational Logic}, 19(1):4:1--4:34, 2018.

\bibitem[vB11]{vanBenthem2011ldii}
J.~van Benthem.
\newblock {\em Logical Dynamics of Information and Interaction}.
\newblock Cambridge University Press, 2011.

\bibitem[vDvdHK08]{vanDitmarsch&vanderHoek&Kooi08}
H.~van Ditmarsch, W.~van~der Hoek, and B.~Kooi.
\newblock {\em Dynamic Epistemic Logic}, volume 337 of {\em Synthese Library
  Series}.
\newblock Springer, Dordrecht, 2008.

\bibitem[Wal00]{Walukiewicz00}
I.~Walukiewicz.
\newblock Completeness of {K}ozen's axiomatisation of the propositional
  $\mu$-calculus.
\newblock {\em Information and Computation}, 157(1--2):142--182, 2000.

\bibitem[WC13]{WangC13}
Y.~Wang and Q.~Cao.
\newblock On axiomatizations of public announcement logic.
\newblock {\em Synthese}, 190(Supplement-1):103--134, 2013.

\bibitem[Yan01]{Yang01}
H.~Yang.
\newblock {\em Local Reasoning for Stateful Programs}.
\newblock PhD thesis, University of Illinois, Urbana-Champaign, 2001.

\end{thebibliography}

\newpage
\appendix
\noindent As in the rest of the paper, in the derivations below we use the following precedence between the various connectives of~$\slSW$:
$\set{\lnot} > \{\land,\lor,\separate\} > \set{\implies,\magicwand,\septraction} > \{\iff\}$.

\section{Proof of Lemma~\ref{lemma:separate-auxiliary-stuff}}
\label{appendix-separate-auxiliary-stuff}

\begin{proof}[Proof of~\rm\ref{starAx:auxilary-1}.]~
    \begin{syntproof}
        1   & \aformula \implies (\aformula \land \avariable \sim \avariablebis) \lor (\aformula \land\! \lnot \avariable \sim \avariablebis)
            & \mbox{PC}\\
        2   & \aformula \separate \aformulabis \implies ((\aformula \land \avariable \sim \avariablebis) \lor (\aformula \land\! \lnot \avariable \sim \avariablebis)) \separate \aformulabis 
            & \mbox{\ref{rule:starinference}, 1}\\
        3   & ((\aformula \,{\land}\, \avariable \hspace{1pt}{\sim}\hspace{1pt} \avariablebis) \hspace{1pt}{\lor}\hspace{1pt} (\aformula {\land} \lnot \avariable \hspace{1pt}{\sim}\hspace{1pt} \avariablebis)) \separate \aformulabis \implies ((\aformula {\land} \avariable \hspace{1pt}{\sim}\hspace{1pt} \avariablebis) \separate \aformulabis) 
            {\lor} ((\aformula {\land} \lnot \avariable \hspace{1pt}{\sim}\hspace{1pt} \avariablebis) \separate \aformulabis)
            & \mbox{\ref{starAx:DistrOr}}\\
        4   & \aformula \land\! \lnot \avariable \sim \avariablebis \implies \lnot \avariable \sim \avariablebis 
            & \mbox{PC}\\
        5   & \aformulabis \implies \true 
            & \mbox{PC}\\
        6   & (\aformula \,{\land} \lnot \avariable \hspace{1pt}{\sim}\hspace{1pt} \avariablebis) \separate \aformulabis \implies (\lnot \avariable \sim \avariablebis) \separate \true
            & \mbox{\ref{rule:starintroLR}, 4, 5}\\
        7   & (\lnot \avariable \sim \avariablebis) 
                \separate \true \implies \lnot \avariable \sim \avariablebis
            & \mbox{\ref{starAx:MonoCore}}\\ 
        8   & (\aformula \,{\land} \lnot \avariable \hspace{1pt}{\sim}\hspace{1pt} \avariablebis) \separate \aformulabis \implies \lnot \avariable \sim \avariablebis
            & \mbox{\ref{rule:imptr}, 6, 7}\\
        9   & ((\aformula \,{\land}\, \avariable \hspace{1pt}{\sim}\hspace{1pt} \avariablebis) \separate \aformulabis) 
        \lor ((\aformula \,{\land} \lnot \avariable \hspace{1pt}{\sim}\hspace{1pt} \avariablebis) \separate \aformulabis) \implies 
        ((\aformula \,{\land}\, \avariable \hspace{1pt}{\sim}\hspace{1pt} \avariablebis) \separate \aformulabis) \lor \lnot \avariable \sim \avariablebis 
            & \mbox{8, PC}\\
        10   & \aformula \separate \aformulabis \implies ((\aformula \,{\land}\, \avariable \hspace{1pt}{\sim}\hspace{1pt} \avariablebis) \separate \aformulabis) \lor \lnot \avariable \sim \avariablebis 
            & \mbox{\ref{rule:imptr}, 2, 3, 9}\\
        11  & \avariable \sim \avariablebis \land (\aformula \separate \aformulabis) \implies (\aformula \land \avariable \sim \avariablebis) \separate \aformulabis 
            & \mbox{10, PC}\hfill\qedhere
    \end{syntproof}
\end{proof}
\vspace{0.1cm}
\begin{proof}[Proof of~\rm\ref{starAx:auxilary-2}.]~
    \begin{syntproof}
        1   & \alloc{\avariable} \land \avariable = \avariablebis \implies \alloc{\avariablebis}
            & \mbox{\ref{coreAx:EqSub}}\\
        2   & \avariable = \avariablebis \land ((\aformula \land \alloc{\avariable}) \separate \aformulabis) 
        \implies ((\aformula \land \alloc{\avariable} \land \avariable = \avariablebis) \separate \aformulabis) 
            & \mbox{\ref{starAx:auxilary-1}}\\ 
        3   & (\aformula \land \alloc{\avariable} \land \avariable = \avariablebis) \separate \aformulabis \implies (\aformula \land \alloc{\avariablebis}) \separate \aformulabis 
            & \mbox{PC, \ref{rule:starinference}, 1}\\
        4   & \avariable = \avariablebis \land ((\aformula \land \alloc{\avariable}) \separate \aformulabis) \implies (\aformula \land \alloc{\avariablebis}) \separate \aformulabis 
            & \mbox{\ref{rule:imptr}, 2, 3} \hfill\qedhere
    \end{syntproof}
\end{proof}
\vspace{0.1cm}
\begin{proof}[Proof of~\rm\ref{starAx:auxilary-3}.]~
    \begin{syntproof}
        1   &   \aformulabis \implies (\aformulabis \land \alloc{\avariable}) \lor (\aformulabis \land\! \lnot \alloc{\avariable})
            & \mbox{PC}\\
        2   &   (\aformula \land \alloc{\avariable}) \separate \aformulabis 
                \implies\\[-2pt] 
            & \quad (\aformula \land \alloc{\avariable}) \separate ((\aformulabis \land \alloc{\avariable}) \lor (\aformulabis \land\! \lnot \alloc{\avariable}))
            & \mbox{\ref{starAx:Commute}, \ref{rule:starinference}, 1}\\
        3   & (\aformula \land \alloc{\avariable}) \separate ((\aformulabis \land \alloc{\avariable}) \lor (\aformulabis \land\! \lnot \alloc{\avariable}))
            \implies\\[-2pt]
            & \quad ((\aformula \land \alloc{\avariable}) \separate (\aformulabis \land \alloc{\avariable})) \lor 
            ((\aformula \land \alloc{\avariable}) \separate (\aformulabis \land\! \lnot \alloc{\avariable}))
            & \quad \mbox{\ref{starAx:Commute}, \ref{starAx:DistrOr}, 2}\\
        4   & \aformulater \land \alloc{\avariable} \implies \alloc{\avariable}
            & (\aformulater \in \{\aformula,\aformulabis\}),\mbox{ PC}\\
        5   & (\aformula \land \alloc{\avariable}) \separate (\aformulabis \land \alloc{\avariable})
            \implies \alloc{\avariable} \separate \alloc{\avariable}
            & \mbox{\ref{rule:starintroLR}, 4}\\
        6   & \alloc{\avariable} \separate \alloc{\avariable} \implies \false 
            & \mbox{\ref{starAx:DoubleAlloc}}\\
        7   & (\aformula \land \alloc{\avariable}) \separate (\aformulabis \land \alloc{\avariable})
        \implies \false 
            & \mbox{\ref{rule:imptr}, 5, 6}\\
        8   & (\aformula \land \alloc{\avariable}) \separate \aformulabis 
            \implies 
            {\false \lor 
            ((\aformula \land \alloc{\avariable}) \separate (\aformulabis \land\! \lnot \alloc{\avariable}))}
            & \mbox{PC, 2, 3, 7}\\
        9 & (\aformula \land \alloc{\avariable}) \separate \aformulabis 
        \implies (\aformula \land \alloc{\avariable}) \separate (\aformulabis \land\! \lnot \alloc{\avariable})
            & \mbox{PC, 8}\\
        10 & \aformula \land \alloc{\avariable} \implies \aformula
            &\mbox{PC}\\
        11  &  (\aformula \land \alloc{\avariable}) \separate (\aformulabis \land\! \lnot \alloc{\avariable}) \implies \aformula \separate (\aformulabis \land\! \lnot \alloc{\avariable}) 
            &\mbox{\ref{rule:starinference}, 10}\\
        12 &     (\aformula \land \alloc{\avariable}) \separate \aformulabis  \implies 
                 \aformula \separate (\aformulabis \land\! \lnot \alloc{\avariable}) 
           &\mbox{\ref{rule:imptr}, 9, 11} \hfill\qedhere
    \end{syntproof}
\end{proof}
\vspace{0.1cm}
\begin{proof}[Proof of~\rm\ref{starAx:auxilary-4}.]~
    \begin{syntproof}
        1   &   \aformula \implies (\aformula \land \alloc{\avariable}) \lor (\aformula \land\! \lnot \alloc{\avariable})
            & \mbox{PC}\\   
        2   & \aformula \separate \aformulabis  
            \implies
            \big((\aformula \land \alloc{\avariable}) \lor (\aformula \land\! \lnot \alloc{\avariable})\big) \separate \aformulabis
            & \mbox{\ref{rule:starinference}, 1}\\
        3   & \big((\aformula \land \alloc{\avariable}) \lor (\aformula \land\! \lnot \alloc{\avariable})\big) \separate \aformulabis \implies\\[-2pt]
            &   \quad ((\aformula \land \alloc{\avariable}) \separate \aformulabis) \lor 
                ((\aformula \land \lnot \alloc{\avariable}) \separate \aformulabis)
            &\mbox{\ref{starAx:DistrOr}}\\ 
        4   & \aformula \land \alloc{\avariable} \implies \alloc{\avariable}
            & \mbox{PC}\\
        5   & \aformulabis \implies \true 
            & \mbox{PC}\\
        6   & (\aformula \land \alloc{\avariable}) \separate \aformulabis \implies (\alloc{\avariable} \separate \true)
            & \mbox{\ref{rule:starintroLR}, 4, 5}\\
        7  & \alloc{\avariable} \separate \true \implies \alloc{\avariable} 
            & \mbox{\ref{starAx:StarAlloc}}\\
        8  & \aformula \separate \aformulabis  \implies
            \alloc{\avariable} \lor ((\aformula \land \lnot\alloc{\avariable}) \separate \aformulabis)
            & \mbox{PC, 2, 3, 6, 7}\\
        9  & \lnot \alloc{\avariable} \land (\aformula \separate \aformula)
        \implies (\aformula \land\! \lnot\alloc{\avariable}) \separate \aformulabis
            & \mbox{PC, 8}\hfill\qedhere
    \end{syntproof}
\end{proof}
\vspace{0.1cm}
\begin{proof}[Proof of~\rm\ref{starAx:auxilary-4bis}.]~
    \begin{syntproof}
        1   & \aformula \implies (\aformula \land \alloc{\avariable}) \lor (\aformula \land\! \lnot \alloc{\avariable})
            & \mbox{PC}\\
        2   & \aformula \separate (\lnot \alloc{\avariable} \land \aformulabis) 
        \implies\\[-2pt]
            &\quad ((\aformula {\land} \alloc{\avariable}) \separate (\aformulabis {\land} \lnot \alloc{\avariable})) \lor 
            ((\aformula {\land} \lnot \alloc{\avariable}) \separate (\aformulabis {\land} \lnot \alloc{\avariable})) 
            & \mbox{\ref{rule:starinference}, 1, \ref{starAx:DistrOr}}\\
        3   & \aformulater \land\! \lnot \alloc{\avariable} \implies \lnot \alloc{\avariable}
            & (\aformulater \in \{\aformula,\aformulabis\}),\mbox{ PC}\\
        4   & (\aformula \land\! \lnot \alloc{\avariable}) \separate (\aformulabis \land\! \lnot \alloc{\avariable}) \implies \lnot \alloc{\avariable} \separate \lnot \alloc{\avariable}
            & \mbox{PC, \ref{rule:starintroLR}, 3}\\
        5   & \lnot \alloc{\avariable} \separate \lnot \alloc{\avariable} \implies \lnot \alloc{\avariable}
            & \mbox{\ref{starAx:AllocNeg}}\\
        6   & \aformula \separate (\lnot \alloc{\avariable} \land \aformulabis) 
        \implies ((\aformula \land \alloc{\avariable}) \separate (\aformulabis \land\! \lnot \alloc{\avariable})) \lor \lnot \alloc{\avariable}
            & \mbox{PC, 2, 4, 5}\\
        7   & \alloc{\avariable} \land (\aformula \separate (\lnot \alloc{\avariable} \land \aformulabis)) \implies (\aformula \land \alloc{\avariable}) \separate (\aformulabis \land\! \lnot \alloc{\avariable})
            & \mbox{PC, 6}\hfill\qedhere
    \end{syntproof}
\end{proof}
\vspace{0.1cm}
\begin{proof}[Proof of~\rm\ref{starAx:auxilary-5}.]~
    \begin{syntproof}
        1   & \aformula \land \alloc{\avariable} \implies (\aformula \land \alloc{\avariable} \land \avariable \Ipto \avariablebis) \lor (\aformula \land \alloc{\avariable} \land \lnot \avariable \Ipto \avariablebis)
            & \mbox{PC}\\
        2   & (\aformula \land \alloc{\avariable}) \separate \aformulabis \implies\\
        & \quad \big((\aformula \land \alloc{\avariable} \land \avariable \Ipto \avariablebis) \lor (\aformula \land \alloc{\avariable} \land \lnot \avariable \Ipto \avariablebis)\big) \separate \aformulabis
            & \mbox{\ref{rule:starinference}, 1}\\
        3   & (\aformula \land \alloc{\avariable}) \separate \aformulabis \implies\\
        & \quad ((\aformula \land \alloc{\avariable} \land \avariable \Ipto \avariablebis) \separate \aformulabis) \lor ((\aformula \land \alloc{\avariable} \land \lnot \avariable \Ipto \avariablebis) \separate \aformulabis)
            & \mbox{\ref{starAx:DistrOr}, \ref{rule:imptr}, 2}\\
        4   & \aformula \land \alloc{\avariable} \land \lnot \avariable \Ipto \avariablebis \implies \alloc{\avariable} \land\! \lnot \avariable \Ipto \avariablebis 
            & \mbox{PC}\\
        5   & \aformulabis \implies \true 
            & \mbox{PC}\\
        6   & (\aformula \land \alloc{\avariable} \land \lnot \avariable \Ipto \avariablebis) \separate \aformulabis \implies (\alloc{\avariable} \land\! \lnot \avariable \Ipto \avariablebis) \separate \true 
            & \mbox{\ref{rule:starintroLR}}\\
        7   & (\alloc{\avariable} \land\! \lnot \avariable \Ipto \avariablebis) \separate \true 
            \implies \lnot \avariable \Ipto \avariablebis 
            & \mbox{\ref{starAx:PointsNeg}}\\
        8   & (\aformula \land \alloc{\avariable}) \separate \aformulabis \implies
        ((\aformula \land \alloc{\avariable} \land \avariable \Ipto \avariablebis) \separate \aformulabis) \lor \lnot \avariable \Ipto \avariablebis 
            & \mbox{PC, 3, 6, 7}\\
        9   & \avariable \Ipto \avariablebis \land ((\alloc{\avariable} \land \aformula) \separate \aformulabis) \implies (\aformula \land \alloc{\avariable} \land \avariable \Ipto \avariablebis) \separate \aformulabis
            & \mbox{PC, 8}\\ 
        10  & \aformula \land \alloc{\avariable} \land \avariable \Ipto \avariablebis
                \implies \aformula \land \avariable \Ipto \avariablebis
            & \mbox{PC}\\
        11  &  (\aformula \land \alloc{\avariable} \land \avariable \Ipto \avariablebis) \separate \aformulabis \implies 
        (\aformula \land \avariable \Ipto \avariablebis) \separate \aformulabis
            & \mbox{\ref{rule:starinference}, 10}\\
        12  & \avariable \Ipto \avariablebis \land ((\alloc{\avariable} \land \aformula) \separate \aformulabis) \implies 
        (\aformula \land \avariable \Ipto \avariablebis) \separate \aformulabis
            & \mbox{\ref{rule:imptr}, 9, 11}
        \hfill\qedhere
    \end{syntproof}
\end{proof}

\begin{proof}[Proof of~\rm\ref{starAx:auxilary-6}.]
    Similar to the proof of \ref{starAx:auxilary-4}, by replacing  $\alloc{\avariable}$ with $\avariable \Ipto \avariablebis$.
    \begin{syntproof}
        1   &   \aformula \implies (\aformula \land \avariable \Ipto \avariablebis) \lor (\aformula \land\! \lnot \avariable \Ipto \avariablebis)
            & \mbox{PC}\\   
        2   & \aformula \separate \aformulabis  
            \implies
            ((\aformula \land \avariable \Ipto \avariablebis) \separate \aformulabis) \lor 
                ((\aformula \land \lnot \avariable \Ipto \avariablebis) \separate \aformulabis)
            &\mbox{\ref{rule:starinference}, 1, \ref{starAx:DistrOr}}\\ 
        3   & \aformula \land \avariable \Ipto \avariablebis \implies \avariable \Ipto \avariablebis
            & \mbox{PC}\\
        4   & \aformulabis \implies \true 
            & \mbox{PC}\\
        5   & (\aformula \land \avariable \Ipto \avariablebis) \separate \aformulabis \implies (\avariable \Ipto \avariablebis \separate \true)
            & \mbox{\ref{rule:starintroLR}, 3, 4}\\
        6  & \avariable \Ipto \avariablebis \separate \true \implies \avariable \Ipto \avariablebis
            & \mbox{\ref{starAx:MonoCore}}\\
        7  & \aformula \separate \aformulabis  \implies
        \avariable \Ipto \avariablebis \lor ((\aformula \land \lnot\avariable \Ipto \avariablebis) \separate \aformulabis)
            & \mbox{PC, 2, 5, 6}\\
        8  & \lnot \avariable \Ipto \avariablebis \land (\aformula \separate \aformulabis)
        \implies (\aformula \land\! \lnot\avariable \Ipto \avariablebis) \separate \aformulabis
            & \mbox{PC, 7}\hfill\qedhere
    \end{syntproof}
\end{proof}

\section{Derivation of the~$\size$ formulae required for~Lemma~\ref{lemma:starPSLelim-sat}}
\label{appendix-derivation-one}
\label{appendix-derivation-two}

In this appendix, we show the derivations in $\coresys(\separate)$ of~$\size \geq \inbound_1 + \inbound_2 \implies \size = \inbound_1 \separate \size \geq \inbound_2$ and $\size = \inbound_1 + \inbound_2 \implies \size = \inbound_1 \separate \size = \inbound_2$, which are required for the proof of~Lemma~\ref{lemma:starPSLelim-sat}.

The derivation of
$
\size \geq \inbound_1 + \inbound_2 \implies \size = \inbound_1 \separate \size \geq \inbound_2
$
is proven by induction on~$\inbound_1$.
The derivation for the base case $\inbound_1 = 0$ is:
      \begin{syntproof}
      1 & \size \geq \inbound_2 \implies \emp \separate \size \geq \inbound_2
      & \mbox{\ref{starAx:Emp}} \\
      2 & \emp \implies \size \geq 0 \land \lnot \size \geq 1
      & \mbox{PC, def.~of $\size \geq 1$} \\
      3 & \emp \separate \size \geq \inbound_2
          \implies
          \size = 0 \separate  \size \geq \inbound_2
      & \mbox{\ref{rule:starinference}, 2, def.~of $\size = 0$} \\
      4 &  \size \geq \inbound_2 \implies  \size = 0 \separate  \size \geq \inbound_2
      & \mbox{\ref{rule:imptr}, 1, 3}
      \end{syntproof}
      \noindent For the induction step, let us suppose the formula to be derivable for a certain $\inbound_1$, and let us prove that it is also derivable for $\inbound_1+1$.
      \begin{syntproof}
      1 & \size \geq \inbound_1+1+\inbound_2 \implies \size \geq 1 \separate \size \geq \inbound_1 + \inbound_2
        & \mbox{def.~of $\size \geq \inbound$, \ref{starAx:Commute}, \ref{starAx:Assoc}}\\
      2 &  \size \geq 1 \implies \size = 1 \separate \top
        & \mbox{\ref{starAx:SizeOne}, def.~of~$\size \geq 1$} \\
      3 & \size \geq 1 \separate \size \geq \inbound_1 + \inbound_2 \implies\\[-2pt]
        & \quad (\size = 1 \separate \top) \separate \size \geq \inbound_1 + \inbound_2
        & \mbox{\ref{rule:starinference}, 2}\\
      4 & (\size = 1 \separate \top) \separate \size \geq \inbound_1 + \inbound_2
          \implies\\[-2pt]
        & \quad \size = 1 \separate \size \geq \inbound_1 + \inbound_2
        & \mbox{PC, \ref{starAx:Commute}, \ref{starAx:Assoc}, \ref{starAx:MonoCore}}
        \\
      5 & \size \geq \inbound_1 + \inbound_2 \implies \size = \inbound_1 \separate \size \geq \inbound_2
        & \mbox{Induction Hypothesis}\\
      6 & \size = 1 \separate \size \geq \inbound_1 + \inbound_2 \implies\\[-2pt]
        & \quad (\size = 1 \separate \size = \inbound_1) \separate \size \geq \inbound_2
        & \mbox{\ref{starAx:Commute}, \ref{rule:starinference}, \ref{starAx:Assoc}}\\
      7 & \size = \widetilde{\inbound} \implies \size \geq \widetilde{\inbound}
        & \mbox{PC, def.~of~$\size = \widetilde{\inbound}$}\\
      8 & \size = \widetilde{\inbound} \implies \lnot \size \geq \widetilde{\inbound}+1
      & \mbox{PC, def.~of~$\size = \widetilde{\inbound}$}\\
      9   & \size = 1 \separate \size = \inbound_1 \implies \size \geq 1 \separate \size \geq \inbound_1
          & \mbox{\ref{rule:starintroLR}, 7}\\
      10  & \size = 1 \separate \size = \inbound_1 \implies
            \lnot \size \geq 2 \separate \lnot \size \geq \inbound_1+1
      & \mbox{\ref{rule:starintroLR}, 8}\\
      11  & \size \geq 1 \separate \size \geq \inbound_1 \implies \size \geq \inbound_1+1
          & \mbox{def.~of $\size \geq \inbound$, \ref{starAx:Commute}, \ref{starAx:Assoc}}\\
      12  & \lnot \size \geq 2 \separate \lnot \size \geq \inbound_1+1  \implies \lnot \size \geq \inbound_1 + 2
          & \mbox{\ref{starAx:SizeNeg}}\\
      13  & \size = 1 \separate \size = \inbound_1 \implies \size = \inbound_1+1
          & \mbox{PC, 9--12, def.~of $\size = \inbound_1$}\\
      14 & (\size = 1 \separate \size = \inbound_1) \separate \size \geq \inbound_2 \implies\\[-2pt]
          & \quad \size = \inbound_1+1 \separate \size \geq \inbound_2
          & \mbox{\ref{rule:starinference}, 13}\\
      15  & \size \geq \inbound_1+1+\inbound_2 \implies
            \size = \inbound_1+1 \separate \size \geq \inbound_2
          & \ref{rule:imptr}, 1, 3, 4, 6, 14
      \end{syntproof}

\noindent The derivation of the formula
$\size = \inbound_1 + \inbound_2 \implies \size = \inbound_1 \separate \size = \inbound_2$
is provided below.

      \begin{syntproof}
        1 & \size = \inbound_1 + \inbound_2 \implies \size \geq \inbound_1 + \inbound_2
          & \mbox{PC, def.~of~$\size = \inbound$}\\
        2 & \size \geq \inbound_1 + \inbound_2 \implies \size = \inbound_1 \separate \size \geq \inbound_2
          & \mbox{Previously derived}\\
        3 & \size \geq \inbound_2 \implies (\size \geq \inbound_2 \land \size \geq \inbound_2+1) \lor \size = \inbound_2
          & \mbox{PC, def.~of~$\size = \inbound_2$}\\
        4 & \size = \inbound_1 \separate \size \geq \inbound_2
          \implies\\[-2pt]
          & \quad \size = \inbound_1 \separate ((\size \geq \inbound_2 \land \size \geq \inbound_2+1) \lor \size = \inbound_2)
          & \mbox{\ref{starAx:Commute}, \ref{rule:starinference}, 3}\\
        5 & \size = \inbound_1 \separate ((\size \geq \inbound_2 \land \size \geq \inbound_2+1) \lor \size = \inbound_2)
        \implies\\[-2pt]
          & \quad (\size = \inbound_1 \separate (\size \geq \inbound_2 {\land} \size \geq \inbound_2+1)) {\lor} (\size = \inbound_1 \separate \size = \inbound_2)
          &\quad  \mbox{\ref{starAx:Commute}, \ref{starAx:DistrOr}}\\
        6 & \size \geq \widetilde{\inbound} \land \aformulater \implies \size \geq \widetilde{\inbound}
          & \mbox{PC}\\
        7 & \size = \inbound_1 \separate (\size \geq \inbound_2 {\land} \size \geq \inbound_2+1)
            \implies\\[-2pt]
          & \size \geq \inbound_1 \separate \size \geq \inbound_2+1
          & \mbox{PC, \ref{rule:starintroLR}, 6}\\
        8 & \size \geq \inbound_1 \separate \size \geq \inbound_2+1 \implies
            \size \geq \inbound_1 + \inbound_2 + 1
          & \mbox{\ref{starAx:Commute}, \ref{starAx:Assoc}}\\
        9 & \size = \inbound_1 \separate (\size \geq \inbound_2 {\land} \size \geq \inbound_2+1)
        \implies \size \geq \inbound_1 + \inbound_2 + 1
          & \mbox{\ref{rule:imptr}, 7, 8}\\
        10  &  \size = \inbound_1 \separate ((\size \geq \inbound_2 \land \size \geq \inbound_2+1) \lor \size = \inbound_2)
        \implies\\[-2pt]
          & \quad \size \geq \inbound_1 + \inbound_2 + 1 {\lor} (\size = \inbound_1 \separate \size = \inbound_2)
          & \mbox{PC, 5, 9}\\
        11 & \size = \inbound_1 + \inbound_2 \implies \size \geq \inbound_1 + \inbound_2 + 1 {\lor} (\size = \inbound_1 \separate \size = \inbound_2)
          & \mbox{\ref{rule:imptr}, 1, 2, 4, 10}\\
        12 & \size = \inbound_1 + \inbound_2 \implies \lnot \size \geq \inbound_1 + \inbound_2 + 1
          & \mbox{PC, def. of~$\size = \inbound$}\\
        13 & \size = \inbound_1 + \inbound_2 \implies \size = \inbound_1 \separate \size = \inbound_2
          & \mbox{PC, 11, 12}
      \end{syntproof}

\section{Proof of Lemma~\ref{lemma:admissible-axioms-2}}
\label{appendix-admissible-axioms-2}

\vspace{3pt}

\noindent \textit{Proof of~\rm\ref{starAx:DistrOr}.}
\begin{syntproof}
  1 & (\aformula \separate \aformulater) \implies (\aformula \separate \aformulater) \lor (\aformulabis \separate \aformulater)
  & \mbox{PC} \\
  2 & (\aformulabis \separate \aformulater) \implies (\aformula \separate \aformulater) \lor (\aformulabis \separate \aformulater)
  & \mbox{PC} \\
  3 & \aformula \implies (\aformulater \magicwand (\aformula \separate \aformulater) \lor (\aformulabis \separate \aformulater))
  & \mbox{\ref{rule:staradj}, 1} \\
  4 & \aformulabis \implies (\aformulater \magicwand (\aformula \separate \aformulater) \lor (\aformulabis \separate \aformulater))
  & \mbox{\ref{rule:staradj}, 2} \\
  5 & \aformula \lor \aformulabis \implies (\aformulater \magicwand (\aformula \separate \aformulater) \lor (\aformulabis \separate \aformulater))
  & \mbox{PC, 3, 4} \\
  6 & (\aformula \lor \aformulabis) \separate \aformulater \implies (\aformula \separate \aformulater) \lor (\aformulabis \separate \aformulater)
  & \mbox{\ref{rule:magicwandadj}, 5} 
  \hfill\qed
\end{syntproof}


\noindent \textit{Proof of~\rm\ref{starAx:False}.}
The axiom~\ref{starAx:False} is provable by~\ref{rule:staradj}. Indeed, proving $(\false \separate \aformula) \implies \false$ 
  reduces to proving
  $\false \implies (\aformula \magicwand \false)$. The latter is a tautology by propositional reasoning. \qed

\vspace{3pt}

\noindent \textit{Proof of~\rm\ref{starAx:StarAlloc}.}
  \begin{syntproof}
  1 & \perp \separate \top \implies \perp
  & \mbox{\ref{starAx:False}} \\
  2 & (\avariable \Ipto \avariable \magicwand \perp) \implies (\avariable \Ipto \avariable \magicwand \perp)
  & \mbox{PC} \\
  3 & (\avariable \Ipto \avariable \magicwand \perp) \separate \avariable \Ipto \avariable \implies  \perp
  & \mbox{\ref{rule:magicwandadj}, 2} \\
  4 & \avariable \Ipto \avariable \separate (\avariable \Ipto \avariable\magicwand \perp) \implies (\avariable \Ipto \avariable 
  \magicwand \perp) \separate \avariable \Ipto \avariable
  & \mbox{\ref{starAx:Commute}} \\
  5 &  \avariable \Ipto \avariable \separate (\avariable \Ipto \avariable \magicwand \perp)\implies \perp
  & \mbox{\ref{rule:imptr}, 4, 3} \\
  6 &  (\avariable \Ipto \avariable \separate (\avariable \Ipto \avariable \magicwand \perp)) \separate \top \implies \perp \separate \top
  & \mbox{\ref{rule:starinference}, 5} \\
  7 & ((\avariable \Ipto \avariable \magicwand \perp) \separate \top) \separate (\avariable \Ipto \avariable)
      \implies (\avariable \Ipto \avariable \separate (\avariable \Ipto \avariable \magicwand \perp)) \separate \top
  & \mbox{\ref{starAx:Commute}, \ref{starAx:Assoc}}\\
  8 & ((\avariable \Ipto \avariable \magicwand \perp) \separate \top) \separate (\avariable \Ipto \avariable) \implies \perp
  & \mbox{\ref{rule:imptr}, 7, 6, 1} \\
  9 & (\avariable \Ipto \avariable \magicwand \perp) \separate \top \implies (\avariable \Ipto \avariable \magicwand \perp)
  & \mbox{\ref{rule:staradj}, 8} \\
  10 & \alloc{\avariable} \separate \top \implies \alloc{\avariable}
  & \mbox{Def. $\alloc{\avariable}$, 9} \hfill\qed
  \end{syntproof}


\section{Proof of Lemma~\ref{lemma:septractionadmissible}}
\label{appendix-septractionadmissible}


\begin{proof}[Proof~of~\rm\ref{mwAx:BotL}]~

\noindent \begin{minipage}{0.46\linewidth}
  \begin{syntproof}
  1 & \perp \separate \top \implies \perp
  & \mbox{\ref{starAx:False}} \\
  2 & \perp \implies \lnot \aformula
  & \mbox{PC} \\
  3 &  \perp \separate \top \implies \lnot \aformula
  & \mbox{\ref{rule:imptr}, 1, 2}
  \end{syntproof}
\end{minipage}
\hfill
\begin{minipage}{0.51\linewidth}
  \begin{syntproof}
    4 &  \top \implies (\perp \magicwand \lnot \aformula)
    & \mbox{\ref{starAx:Commute}, \ref{rule:staradj}} \\
    5 &  \top \implies \neg (\perp \septraction \  \aformula)
    & \mbox{Def. $\septraction$, PC} \\
    6 &  (\perp \septraction \ \aformula) \implies \perp
    & \mbox{5, PC} \hfill\qedhere
    \end{syntproof}
\end{minipage}
\end{proof}

\vspace{3pt}

\begin{proof}[Proof of~\rm\ref{mwAx:BotR}]~

\noindent
\begin{minipage}{0.46\linewidth}
  \begin{syntproof}
  1 & \top \separate \aformula \implies \top
  & \mbox{PC} \\
  2 & \top \implies (\aformula \magicwand \top)
  & \mbox{\ref{rule:staradj}}
  \end{syntproof}
\end{minipage}
\hfill
\begin{minipage}{0.51\linewidth}
  \begin{syntproof}
  3 & \lnot (\aformula \magicwand \top) \implies \perp
  & \mbox{PC, 2} \\
  4 &  (\aformula \septraction \perp) \implies \perp
  & \mbox{Def. $\septraction$, PC}
  \end{syntproof}
\end{minipage}

\vspace{5pt}

\noindent Note that implicitly, we have assumed that we can replace $\lnot \top$ by $\perp$ in the scope
of $\septraction$ or $\magicwand$, which is possible as the replacement of equivalents holds
in the calculus~$\coresys(\separate,\magicwand)$ (see e.g. the proof of Theorem~\ref{theo:PSLcompleteAx}).
\end{proof}

\vspace{3pt}

\begin{proof}[Proof of~\rm\ref{mwAx:Cut}]~
  \begin{syntproof}
    1 & (\aformula \magicwand \aformulabis) \implies (\aformula \magicwand \aformulabis)
    & \mbox{PC} \\
    2 & (\aformula \magicwand \aformulabis) \separate \aformula \implies  \aformulabis
    & \mbox{\ref{rule:magicwandadj}, 1} \\
    3 & \aformula \separate (\aformula \magicwand \aformulabis) \implies (\aformula \magicwand \aformulabis) \separate \aformula
    & \mbox{\ref{starAx:Commute}} \\
    4 &  \aformula \separate (\aformula\magicwand\aformulabis)\implies \aformulabis
    & \mbox{\ref{rule:imptr}, 3, 2}\hfill\qedhere
  \end{syntproof}
\end{proof}

\vspace{3pt}

\begin{proof}[Proof of~\rm\ref{mwAx:ImpL}]~
  \begin{syntproof}
    1 & \aformula \implies \aformulabis
    & \mbox{Hypothesis} \\
    2 & \aformulabis \separate (\aformulabis\magicwand\neg\aformulater)\implies \neg \aformulater
    & \mbox{\ref{mwAx:Cut}} \\
    3 & (\aformulabis\magicwand\neg\aformulater) \separate \aformula  \implies \aformula \separate (\aformulabis\magicwand\neg\aformulater)
    & \mbox{\ref{starAx:Commute}} \\
    4 &  \aformula \separate (\aformulabis\magicwand\neg\aformulater) \implies  \aformulabis
    \separate (\aformulabis\magicwand\neg\aformulater)
    & \mbox{\ref{rule:starinference}, 1} \\
    5 & \aformula \separate (\aformulabis\magicwand\neg\aformulater) \implies \neg \aformulater
    & \mbox{\ref{rule:imptr}, 2, 4} \\
    6 &  (\aformulabis\magicwand\neg\aformulater) \separate \aformula  \implies \neg \aformulater
    & \mbox{\ref{rule:imptr}, 3, 5} \\
    7 & \aformulabis \magicwand \neg \aformulater \implies \aformula \magicwand \neg \aformulater
    & \mbox{\ref{rule:staradj}, 6} \\
    8 & \lnot (\aformula \magicwand\neg \aformulater) \implies  \lnot (\aformulabis \magicwand \neg \aformulater)
    & \mbox{PC, 7} \\
    9 & (\aformula \septraction \aformulater) \implies (\aformulabis\septraction \aformulater)
    & \mbox{Def. $\septraction$, 8} \hfill \qedhere
  \end{syntproof}
\end{proof}

\vspace{3pt}

\begin{proof}[Proof of~\rm\ref{mwAx:ImpR}]~
\begin{syntproof}
  1 & \aformula \implies \aformulabis
  & \mbox{Hypothesis} \\
  2 & \lnot \aformulabis \implies \lnot \aformula
  & \mbox{PC, 1} \\
  3 & \aformulater \separate (\aformulater \magicwand \neg \aformulabis) \implies \neg \aformulabis
  & \mbox{\ref{mwAx:Cut}} \\
  4 &  \aformulater \separate (\aformulater \magicwand \neg \aformulabis) \implies  \lnot \aformula
  & \mbox{\ref{rule:imptr}, 3, 2} \\
  5 & (\aformulater \magicwand \neg \aformulabis) \separate \aformulater  \implies \aformulater \separate (\aformulater \magicwand \neg \aformulabis)
  & \mbox{\ref{starAx:Commute}} \\
  6 & (\aformulater \magicwand \neg \aformulabis) \separate \aformulater  \implies \neg \aformula
  & \mbox{\ref{rule:imptr}, 4, 5} \\
  7 & (\aformulater \magicwand \neg \aformulabis) \implies  (\aformulater \magicwand \neg \aformula)
  & \mbox{\ref{rule:staradj}, 6} \\
  8 & \neg (\aformulater \magicwand \neg \aformula) \implies \neg (\aformulater \magicwand \neg \aformulabis)
  & \mbox{PC, 7} \\
  9 &  (\aformulater \septraction \aformula) \implies (\aformulater\septraction \aformulabis)
  & \mbox{Def. $\septraction$}\hfill\qedhere
\end{syntproof}
\end{proof}

\vspace{3pt}

\begin{proof}[Proof of~\rm\ref{mwAx:Curry}]
  By definition of the septraction operator $\septraction$,~\ref{mwAx:Curry} is equivalent to
  $\aformula\magicwand (\aformulabis\magicwand \neg \aformulater)) \ \iff\ (\aformula*\aformulabis)\magicwand\neg\aformulater$. 
  This equivalence is provable in $\coresys(\separate,\magicwand)$, thanks to 
  the adjunction rules, as we now show.

\begin{syntproof}
  1 & (\aformula \separate \aformulabis) \separate (\aformula \separate \aformulabis \magicwand \neg \aformulater)\implies \neg\aformulater & 
  \mbox{\ref{mwAx:Cut}}
  \\
  2 & \aformulabis \separate (\aformula \separate (\aformula \separate \aformulabis \magicwand \neg \aformulater))\implies \neg\aformulater & \mbox{\ref{starAx:Commute}, \ref{starAx:Assoc}, 1}
  \\
  3 & \aformula \separate (\aformula \separate \aformulabis \magicwand \neg \aformulater)\implies (\aformulabis \magicwand
  \neg \aformulater) & \mbox{\ref{rule:staradj}, 2}
  \\
  4 & (\aformula \separate \aformulabis \magicwand \neg \aformulater) \implies (\aformula \magicwand (\aformulabis\magicwand \neg \aformulater))
  & \mbox{\ref{rule:staradj}, 3, \ref{starAx:Commute}}
  \\
  5 & \aformula \separate (\aformula \magicwand (\aformulabis \magicwand \neg \aformulater))\implies
  (\aformulabis \magicwand \neg \aformulater) &
  \ref{mwAx:Cut}
  \\
  6 & \aformulabis \separate \aformula \separate (\aformula \magicwand (\aformulabis \magicwand \neg \aformulater))\implies
  \neg \aformulater
  & \mbox{\ref{rule:magicwandadj}, 5, \ref{starAx:Commute}, \ref{starAx:Assoc}}
  \\
  7 & (\aformula \separate \aformulabis) \separate (\aformula \magicwand (\aformulabis \magicwand \neg \aformulater))
  \implies \neg \aformulater &
  \mbox{\ref{starAx:Commute}, \ref{starAx:Assoc}, 6}
  \\
  8 & (\aformula \magicwand (\aformulabis \magicwand \neg \aformulater))
  \implies (\aformula \separate \aformulabis\magicwand \neg \aformulater)
  & \mbox{\ref{rule:staradj}, 7}
  \\
  9 & \aformula \magicwand (\aformulabis \magicwand \neg \aformulater)
  \iff (\aformula * \aformulabis)\magicwand \neg \aformulater
  & \mbox{PC, 4, 8} \hfill\qedhere
\end{syntproof}
\end{proof}

\vspace{3pt}

\begin{proof}[Proof of~\rm\ref{mwAx:OrL}]
We derive each implication separately.
\begin{syntproof}
  1 & (\aformula \magicwand \neg \aformulater) \wedge  (\aformulabis \magicwand \neg \aformulater)
  \implies (\aformulabis \magicwand \neg \aformulater)
  & \mbox{PC} \\
  2 & \aformulabis \separate ((\aformula \magicwand \neg \aformulater) \wedge  (\aformulabis \magicwand \neg \aformulater))  \implies \aformulabis \separate
  (\aformulabis \magicwand \neg \aformulater)
  & \mbox{\ref{rule:starintroLR}, 1} \\
  3 & (\aformula \magicwand \neg \aformulater) \wedge (\aformulabis \magicwand \neg \aformulater)
  \implies (\aformula \magicwand \neg \aformulater)
  & \mbox{PC} \\
  4 & \aformula \separate ((\aformula \magicwand \neg \aformulater) \wedge  (\aformulabis \magicwand \neg \aformulater))
  \implies \aformula \separate
  (\aformula \magicwand \neg \aformulater)
  & \mbox{\ref{rule:starintroLR}, 3} \\
  5 & \aformula \separate  (\aformula \magicwand \neg \aformulater) \implies \neg \aformulater
  & \mbox{\ref{mwAx:Cut}} \\
  6 & \aformulabis\separate  (\aformulabis \magicwand \neg \aformulater) \implies \neg \aformulater
  & \mbox{\ref{mwAx:Cut}} \\
  7 & \aformulabis \separate  ((\aformula \magicwand \neg \aformulater) \wedge (\aformulabis \magicwand \neg \aformulater))
  \implies \neg \aformulater
  & \mbox{\ref{rule:imptr}, 2, 6} \\
  8 & \aformula \separate  ((\aformula \magicwand \neg \aformulater) \wedge (\aformulabis \magicwand \neg \aformulater))
  \implies \neg \aformulater
  & \mbox{\ref{rule:imptr}, 4, 5} \\
  9 & (\aformula \vee \aformulabis)  \separate  ((\aformula \magicwand \neg \aformulater) \wedge (\aformulabis \magicwand \neg \aformulater))  \implies \\
  &\quad (\aformula \separate  (\aformula \magicwand \neg \aformulater \wedge \aformulabis \magicwand \neg \aformulater)) \vee
  (\aformulabis \separate  (\aformula \magicwand \neg \aformulater \wedge \aformulabis \magicwand \neg \aformulater))
  & \mbox{\ref{starAx:DistrOr}} \\
  10 & (\aformula \vee \aformulabis)  \separate  ((\aformula \magicwand \neg \aformulater) \wedge (\aformulabis \magicwand \neg \aformulater))  \implies \neg \aformulater
  & \mbox{PC, 7, 8, 9}\\
  11 & ((\aformula \magicwand \neg \aformulater) \wedge( \aformulabis \magicwand \neg \aformulater))
  \separate  (\aformula \vee \aformulabis) \implies (\aformula \vee \aformulabis)  \separate  ((\aformula \magicwand \neg \aformulater) \wedge (\aformulabis \magicwand \neg \aformulater))
  & \mbox{\ref{starAx:Commute}} \\
  12 & ((\aformula \magicwand \neg \aformulater) \wedge (\aformulabis \magicwand \neg \aformulater))
  \separate  (\aformula \vee \aformulabis) \implies \neg \aformulater
  & \mbox{\ref{rule:imptr}, 12, 10} \\
  13 &  (\aformula \magicwand \neg \aformulater) \wedge (\aformulabis \magicwand \neg \aformulater)
  \implies  (\aformula \vee \aformulabis \magicwand \neg \aformulater)
  & \mbox{\ref{rule:staradj}, 12} \\
  14 & \neg (\aformula \vee \aformulabis \magicwand \neg \aformulater)
  \implies \neg (\aformula \magicwand \neg \aformulater) \vee \neg (\aformulabis \magicwand \neg \aformulater)
  & \mbox{PC, 13} \\
  15 & (\aformula \vee \aformulabis \septraction \aformulater)
  \implies (\aformula \septraction \aformulater) \vee (\aformulabis \septraction \aformulater)
  & \mbox{Def. $\septraction$, 14}
\end{syntproof}

\noindent The derivation of the other implication can be found below.

\begin{syntproof}
1 & \aformula \implies \aformula \vee \aformulabis
& \mbox{PC} \\
2 & \aformulabis \implies \aformula \vee \aformulabis
& \mbox{PC} \\
3 & (\aformula \septraction \aformulater) \implies (\aformula \vee \aformulabis \septraction \aformulater)
& \mbox{\ref{mwAx:ImpL}, 1} \\
4 & (\aformulabis \septraction \aformulater) \implies (\aformula \vee \aformulabis \septraction \aformulater)
& \mbox{\ref{mwAx:ImpL}, 2} \\
5 & (\aformulabis \septraction \aformulater)
\vee  (\aformula \septraction \aformulater) \implies (\aformula \vee \aformulabis \septraction \aformulater)
& \mbox{PC, 3, 4}
\hfill\qedhere
\end{syntproof}
\end{proof}


\begin{proof}[Proof of~\rm\ref{mwAx:OrR}]
We handle each implication separately,
and we follow a pattern similar to the one used in the proof of~\ref{mwAx:OrL}.

\begin{syntproof}
  1 & \aformulater \separate (\aformulater \magicwand \neg \aformula) \implies \neg \aformula
  & \mbox{\ref{mwAx:Cut}} \\
  2 & (\aformulater \magicwand \neg \aformula) \land (\aformulater \magicwand \neg \aformulabis)
  \implies \aformulater \magicwand \neg \aformula
  & \mbox{PC} \\
  3 & \aformulater \separate ((\aformulater \magicwand \neg \aformula) \land (\aformulater \magicwand \neg \aformulabis))
  \implies \aformulater \separate (\aformulater \magicwand \neg \aformula)
  & \mbox{\ref{rule:starintroLR},2} \\
  4 & \aformulater \separate ((\aformulater \magicwand \neg \aformula) \land (\aformulater \magicwand \neg \aformulabis))
  \implies  \neg \aformula
  & \mbox{\ref{rule:imptr}, 3, 1} \\
  5 & \aformulater \separate (\aformulater \magicwand \neg \aformulabis) \implies \neg \aformulabis
  & \mbox{\ref{mwAx:Cut}} \\
  6 & (\aformulater \magicwand \neg \aformula) \land (\aformulater \magicwand \neg \aformulabis)
  \implies (\aformulater \magicwand \neg \aformulabis)
  & \mbox{PC} \\
  7 & \aformulater \separate ((\aformulater \magicwand \neg \aformula) \land (\aformulater \magicwand \neg \aformulabis))
  \implies \aformulater \separate (\aformulater \magicwand \neg \aformulabis)
  & \mbox{\ref{rule:starintroLR},6} \\
  8 & \aformulater \separate ((\aformulater \magicwand \neg \aformula) \land (\aformulater \magicwand \neg \aformulabis))
  \implies  \neg \aformulabis
  & \mbox{\ref{rule:imptr}, 7, 5} \\
  9  & \aformulater \separate ((\aformulater \magicwand \neg \aformula) \land (\aformulater \magicwand \neg \aformulabis))
  \implies  \neg (\aformula \vee \aformulabis)
  & \mbox{PC, 4, 8} \\
  10 & ((\aformulater \magicwand \neg \aformula) \land (\aformulater \magicwand \neg \aformulabis))  \separate
  \aformulater \implies  \neg (\aformula \vee \aformulabis)
  & \mbox{\ref{starAx:Commute} + \ref{rule:imptr}, 9} \\
  11 &  (\aformulater \magicwand \neg \aformula) \land (\aformulater \magicwand \neg \aformulabis)  \implies
  (\aformulater \magicwand  \neg (\aformula \vee \aformulabis))
  & \mbox{\ref{rule:staradj}, 10} \\
  12 & \neg (\aformulater \magicwand  \neg (\aformula \vee \aformulabis)) \implies
  \neg (\aformulater \magicwand \neg \aformula) \vee \neg (\aformulater \magicwand \neg \aformulabis)
  & \mbox{PC, 11} \\
  13 & (\aformulater \septraction (\aformula \vee \aformulabis)) \implies
  (\aformulater \septraction \aformula) \vee (\aformulater \septraction \aformulabis)
  & \mbox{Def. $\septraction$, 12}
\end{syntproof}



\noindent The derivation of the other implication can be found below.

\begin{syntproof}
  1 & \aformula \implies \aformula \vee \aformulabis
  & \mbox{PC} \\
  2 & \aformulabis \implies \aformula \vee \aformulabis
  & \mbox{PC} \\
  3 & (\aformulater \septraction  \aformula)  \implies ( \aformulater  \septraction \aformula \vee \aformulabis)
  & \mbox{\ref{mwAx:ImpR}, 1} \\
  4 & (\aformulater \septraction  \aformulabis)  \implies ( \aformulater  \septraction \aformula \vee \aformulabis)
  & \mbox{\ref{mwAx:ImpR}, 2} \\
  5 & (\aformulater \septraction  \aformula)
  \vee  (\aformulater \septraction \aformulabis) \implies
  ( \aformulater  \septraction \aformula \vee \aformulabis)
  & \mbox{PC, 3, 4}\hfill\qedhere
\end{syntproof}
\end{proof}

\vspace{3pt}

\noindent \textit{Proof of~\rm\ref{mwAx:Mix}.}
\begin{syntproof}
  1 & \aformula \separate (\aformula \magicwand \aformulater) \implies \aformulater
  & \mbox{\ref{mwAx:Cut}} \\
  2 & \aformula \separate (\aformula \magicwand \neg (\aformulabis \wedge \aformulater)) \implies
  \neg (\aformulabis \wedge \aformulater)
  & \mbox{\ref{mwAx:Cut}} \\
  3 &   (\aformula \separate (\aformula \magicwand \aformulater)) \wedge
  (\aformula \separate (\aformula \magicwand \neg (\aformulabis \wedge \aformulater))) \implies
  \neg \aformulabis
  & \mbox{PC, 1, 2} \\
  4 & \aformula \separate ( (\aformula \magicwand \aformulater) \wedge
    (\aformula \magicwand \neg (\aformulabis \wedge \aformulater)))
    \implies \\
  & \quad (\aformula \separate (\aformula \magicwand \aformulater)) \wedge
      (\aformula \separate (\aformula \magicwand \neg (\aformulabis \wedge \aformulater)))
  & \mbox{\ref{rule:starintroLR}, PC} \\
  5 &  \aformula \separate ( (\aformula \magicwand \aformulater) \wedge
    (\aformula \magicwand \neg (\aformulabis \wedge \aformulater))) \implies
  \neg \aformulabis
  & \mbox{\ref{rule:imptr}, 4} \\
  6 & (\aformula \magicwand \aformulater) \wedge
    (\aformula \magicwand \neg (\aformulabis \wedge \aformulater)) \implies (\aformula \magicwand \neg \aformulabis)
  & \mbox{\ref{starAx:Commute}, \ref{rule:staradj}, 5} \\
  7 & (\aformula \magicwand \aformulater) \wedge  \neg (\aformula \magicwand \neg \aformulabis)
      \implies \neg (\aformula \magicwand \neg (\aformulabis \wedge \aformulater))
  & \mbox{PC}  \\
  8 & (\aformula\magicwand\aformulater)  \wedge (\aformula\septraction\aformulabis) \implies
  (\aformula\septraction \aformulabis\wedge\aformulater)
  & \mbox{Def. $\septraction$, 7} \hfill\qed
\end{syntproof}






\vspace{3pt}

\noindent \textit{Proof of~{\rm\ref{mwAx:SeptEq}} and~\rm\ref{mwAx:SeptIneq}.}
Below, we provide the derivation for the admissible axiom schema~\ref{mwAx:SeptEq} (the derivation 
for~\ref{mwAx:SeptIneq} is very similar and is thus omitted).

\begin{syntproof}
1 & \aformula \implies (\aformula \wedge \avariable = \avariablebis) \vee (\aformula \wedge \avariable \neq 
\avariablebis) 
& \mbox{PC} \\
2 & (\aformula \septraction \aformulabis) \implies ((\aformula \wedge \avariable = \avariablebis) \vee (\aformula \wedge \avariable \neq 
\avariablebis) \septraction \aformulabis)
& \mbox{\ref{mwAx:ImpL}, 1} \\
3 & (\aformula \septraction \aformulabis) \implies 
   (\aformula \wedge \avariable = \avariablebis \septraction \aformulabis) \vee 
  (\aformula \wedge \avariable \neq  \avariablebis \septraction \aformulabis)
& \mbox{\ref{mwAx:OrL}, \ref{rule:imptr}, 2} \\
4 & \avariable = \avariablebis \separate \avariable \neq \avariablebis \implies \avariable = \avariablebis
& \mbox{\ref{starAx:MonoCore}, \ref{rule:starintroLR}} \\
5 & \avariable \neq \avariablebis \separate \avariable = \avariablebis \implies \avariable \neq \avariablebis
& \mbox{\ref{starAx:MonoCore}, \ref{rule:starintroLR}} \\
6 & \avariable = \avariablebis \separate \avariable \neq \avariablebis \implies  \avariable = \avariablebis \wedge 
 \avariable \neq \avariablebis
& \mbox{\ref{starAx:Commute}, \ref{rule:imptr}, PC, 4, 5} \\
7 & \avariable = \avariablebis \separate \avariable \neq \avariablebis \implies \neg \top
& \mbox{PC, 6} \\
8 & \lnot \top \implies \lnot \aformulabis  & \mbox{PC }\\
9 &  \avariable = \avariablebis \separate \avariable \neq \avariablebis \implies \lnot \aformulabis &
\mbox{PC, 7, 8} \\ 
10 & \avariable = \avariablebis \implies (\avariable \neq \avariablebis  \magicwand \neg \aformulabis)
& \mbox{\ref{rule:staradj}, 9} \\
11 & \neg (\avariable \neq \avariablebis  \magicwand \neg \aformulabis) \implies \avariable \neq \avariablebis
& \mbox{PC, 10} \\ 
12 & (\avariable \neq \avariablebis  \septraction \aformulabis) \implies \avariable \neq \avariablebis
& \mbox{Def. $\septraction$, 11} \\ 
13 & \aformula \wedge \avariable \neq \avariablebis  \implies \avariable \neq \avariablebis
& \mbox{PC} \\
14 & (\aformula \wedge \avariable \neq \avariablebis \septraction \aformulabis) \implies 
(\avariable \neq \avariablebis
\septraction \aformulabis)
& \mbox{\ref{mwAx:ImpL}, 13} \\
15 &  (\aformula \wedge \avariable \neq \avariablebis \septraction \aformulabis) \implies 
\avariable \neq \avariablebis
& \mbox{\ref{rule:imptr}, 12, 14} \\
16 &  \avariable =  \avariablebis \wedge (\aformula \septraction \aformulabis) \implies  
(\aformula \wedge \avariable = \avariablebis \septraction \aformulabis) \vee \avariable \neq \avariablebis
& \mbox{PC, 3, 15} \\
17 & \avariable =  \avariablebis \wedge (\aformula \septraction \aformulabis) 
\implies  (\aformula \wedge \avariable = \avariablebis \septraction \aformulabis)
& \mbox{PC, 16}\hfill\qed
\end{syntproof}

\cut{
\begin{syntproof}
1 & \aformula \implies (\aformula \wedge \avariable = \avariablebis) \vee (\aformula \wedge \avariable \neq 
\avariablebis) 
& \mbox{PC} \\
2 & (\aformula \septraction \top) \implies ((\aformula \wedge \avariable = \avariablebis) \vee (\aformula \wedge \avariable \neq 
\avariablebis) \septraction \top)
& \mbox{\ref{mwAx:ImpL}, 1} \\
3 & (\aformula \septraction \top) \implies 
   (\aformula \wedge \avariable = \avariablebis \septraction \top) \vee 
  (\aformula \wedge \avariable \neq  \avariablebis \septraction \top)
& \mbox{\ref{mwAx:OrL}, \ref{rule:imptr}, 2} \\
4 & \avariable = \avariablebis \separate \avariable \neq \avariablebis \implies \avariable = \avariablebis
& \mbox{\ref{starAx:MonoCore}, \ref{rule:starintroLR}} \\
5 & \avariable \neq \avariablebis \separate \avariable = \avariablebis \implies \avariable \neq \avariablebis
& \mbox{\ref{starAx:MonoCore}, \ref{rule:starintroLR}} \\
6 & \avariable = \avariablebis \separate \avariable \neq \avariablebis \implies  \avariable = \avariablebis \wedge 
 \avariable \neq \avariablebis
& \mbox{\ref{starAx:Commute}, \ref{rule:imptr}, PC, 4, 5} \\
7 & \avariable = \avariablebis \separate \avariable \neq \avariablebis \implies \neg \top
& \mbox{PC, 6} \\
8 & \avariable = \avariablebis \implies (\avariable \neq \avariablebis  \magicwand \neg \top)
& \mbox{\ref{rule:staradj}, 7} \\
9 & \neg (\avariable \neq \avariablebis  \magicwand \neg \top) \implies \avariable \neq \avariablebis
& \mbox{PC, 8} \\ 
10 & (\avariable \neq \avariablebis  \septraction \top) \implies \avariable \neq \avariablebis
& \mbox{Def. $\septraction$, 9} \\ 
11 & \aformula \wedge \avariable \neq \avariablebis  \implies \avariable \neq \avariablebis
& \mbox{PC} \\
12 & (\aformula \wedge \avariable \neq \avariablebis \septraction \top) \implies (\avariable \neq \avariablebis
\septraction \top)
& \mbox{\ref{mwAx:OrL}, 11} \\
13 &  (\aformula \wedge \avariable \neq \avariablebis \septraction \top) \implies \avariable \neq \avariablebis
& \mbox{\ref{rule:imptr}, 10, 12} \\
14 &  \avariable =  \avariablebis \wedge (\aformula \septraction \top) \implies  (\aformula \wedge \avariable = \avariablebis \septraction \top) \vee \avariable \neq \avariablebis
& \mbox{PC, 3, 13} \\
15 & \avariable =  \avariablebis \wedge (\aformula \septraction \top) \implies  (\aformula \wedge \avariable = \avariablebis \septraction \top)
& \mbox{PC, 14}\hfill\qed
\end{syntproof}
}

\begin{proof}[Proof of~\rm\ref{mwAx:SizeLiterals}]
Notice that, since $\aformula_{\size}$ is satisfiable, 
for every $\inbound_1,\inbound_2 \in \Nat$ such that 
$\size \geq \inbound_1 \land \lnot \size \geq \inbound_2 \inside \aformula_{\size}$,
we must have $\inbound_1 < \inbound_2$.
Moreover, thanks to~\ref{coreAx:Size} and~\ref{mwAx:ImpL}, 
without loss of generality, we can restrict ourselves to 
$\aformula_{\size}$ of the form:
\begin{itemize}[align=left]
     \item[\parenthesislabel{1}{csl:septraction:add-size:shape-1}] 
          $\aformula_{\size} = \size \geq \inbound$  for some $\inbound \geq 0$,
     \item[\parenthesislabel{2}{csl:septraction:add-size:shape-2}]  
          $\aformula_{\size} = \neg (\size \geq \inbound)$  for some $\inbound > 0$,
     \item[\parenthesislabel{3}{csl:septraction:add-size:shape-3}] 
          $\aformula_{\size} = \size \geq \inbound_1 \wedge \neg (\size \geq \inbound_2)$ for some $\inbound_2 > \inbound_1$. 
\end{itemize}
Indeed, given an arbitrary $\aformula_{\size}$, 
every positive literal $\size \geq \inbound$ such that $\inbound < \maxsize{\aformula_{\size}}$ can be derived starting from $\size \geq \maxsize{\aformula_{\size}}$, by repeated applications of~\ref{coreAx:Size}. 
Similarly, let $\overline{\inbound}$ be the smallest natural number
such that $\lnot \size \geq \overline{\inbound} \inside \aformula_{\size}$, if any. 
Every literal $\lnot \size \geq \inbound' \inside \aformula_{\size}$ with $\inbound' \geq \overline{\inbound}$ 
can be derived from 
$\lnot \size \geq \overline{\inbound}$, by repeated applications of the axiom~\ref{coreAx:Size} (taken in contrapositive form i.e.~$\lnot \size \geq \inbound \implies \lnot \size \geq \inbound+1$, which is derivable in $\coresys$ by propositional reasoning).

We write $\UNALLOC(\asetvar)$ to denote the conjunction 
$\bigwedge_{\avariable \in \asetvar} \lnot \alloc{\avariable}$ (`$\UNALLOC$' stands for `unallocated'). 
Below, given $\inbound \in \Nat$, 
we aim at deriving the formula $(\size = \inbound \land \UNALLOC(\asetvar)) \septraction \true$
since this implies that \ref{mwAx:SizeLiterals} is derivable in 
its instances~\ref{csl:septraction:add-size:shape-1}--\ref{csl:septraction:add-size:shape-3}, as shown below.
\begin{description}
     \item[case \ref{csl:septraction:add-size:shape-1}] Let $\aformula_{\size} = \size \geq \inbound$.
          \begin{syntproof}
               1    & \size = \inbound \land \UNALLOC(\asetvar) \septraction \true
                    & \mbox{Hypothesis}\\
               2    & \size = \inbound  \land \UNALLOC(\asetvar) \implies \size \geq \inbound  \land \UNALLOC(\asetvar)
                    & \mbox{PC, def. of~$\size = \inbound$}\\
               3    & (\size = \inbound \land \UNALLOC(\asetvar) \septraction \true)
                         \implies 
                         (\size \geq \inbound \land \UNALLOC(\asetvar) \septraction \true)
                    & \mbox{\ref{mwAx:ImpL}, 2}\\
               4    & \size \geq \inbound \land \UNALLOC(\asetvar) \septraction \true
                    & \mbox{Modus Ponens, 1, 3}
          \end{syntproof}  
     \item[case \ref{csl:septraction:add-size:shape-2}] Let $\aformula_{\size} = \lnot \size \geq \inbound$. Since $\aformula_{\size}$ is satisfiable, we have $\inbound \geq 1$.
     \begin{syntproof}
          1    & \size = \inbound{-}1 \land \UNALLOC(\asetvar) \septraction \true
               & \mbox{Hypothesis}\\
          2    & \size = \inbound{-}1  \land \UNALLOC(\asetvar) \implies \lnot \size \geq \inbound  \land \UNALLOC(\asetvar)
               & \mbox{PC, def. of~$\size = \inbound{-}1$}\\
          3    & (\size = \inbound{-}1 \land \UNALLOC(\asetvar) \septraction \true)
                    \implies
                    (\lnot\size \geq \inbound \land \UNALLOC(\asetvar) \septraction \true)
               & \mbox{\ref{mwAx:ImpL}, 2}\\
          4    & \lnot\size \geq \inbound \land \UNALLOC(\asetvar) \septraction \true
               & \mbox{Modus Ponens, 1, 3}
     \end{syntproof} 
     \item[case \ref{csl:septraction:add-size:shape-3}] Let $\aformula_{\size} = \size \geq \inbound_1 \land \lnot \size \geq \inbound_2$. Since $\aformula_{\size}$ is satisfiable, $\inbound_2 > \inbound_1$.
     \begin{syntproof}
          1    & \size = \inbound_2{-}1 \land \UNALLOC(\asetvar) \septraction \true
               & \mbox{Hypothesis}\\
          2    & \size = \inbound_2{-}1 \implies \size \geq \inbound_1 
               & \mbox{repeated~\ref{coreAx:Size}, as $\inbound_2 > \inbound_1$}\\
          3    & \size = \inbound_2{-}1 \implies \lnot \size \geq \inbound_2  
               & \mbox{PC, def. of~$\size = \inbound{-}1$}\\
          4    & \size = \inbound_2{-}1  \land \UNALLOC(\asetvar) \implies
               \size \geq \inbound_1 \land  \lnot \size \geq \inbound_2  \land \UNALLOC(\asetvar)
               & \mbox{PC, 2, 3}\\
          5    & \big(\size = \inbound_2{-}1 \land \UNALLOC(\asetvar) \septraction \true\big)
                    \implies\\[-2pt]
               & \qquad (\size \geq \inbound_1 \land  \lnot \size \geq \inbound_2 \land \UNALLOC(\asetvar) \septraction \true)
               & \mbox{\ref{mwAx:ImpL}, 4}\\
          6    & \size \geq \inbound_1 \land  \lnot \size \geq \inbound_2 \land \UNALLOC(\asetvar) \septraction \true
               & \mbox{Modus Ponens, 1, 5}
     \end{syntproof}  
\end{description}
To conclude the proof, let us show that $(\size = \inbound \land \UNALLOC(\asetvar)) \septraction \true$ is derivable in $\coresys(\separate,\magicwand)$.
The proof is by induction on $\inbound$, with two base cases, for $\inbound = 0$ and $\inbound = 1$.
\begin{description}
     \item[base case: $\inbound = 0$]
          In this case, $\size = 0$ is equal to $\size \geq 0 \land \lnot \size \geq 1$.
          We have, 
          \begin{syntproof}
               1 & (\emp \magicwand \perp) \implies  \emp \separate (\emp \magicwand \perp)  
               & \mbox{\ref{starAx:Emp}} \\
               2 & \emp \separate (\emp \magicwand \perp) \implies \perp
               & \mbox{\ref{mwAx:Cut}} \\
               3 &  (\emp \magicwand \perp) \implies \perp
               & \mbox{\ref{rule:imptr}, 1, 2} \\
               4 & \emp \septraction \top
               & \mbox{PC, 3, def.~of $\septraction$} \\
               5 & \alloc{\avariable} \implies \size \geq 1 
               & \mbox{\ref{coreAx:AllocSize}} \\
               6 & \emp \implies \neg \alloc{\avariable}
               & \mbox{PC, 5, as $\size \geq 1 = \lnot \emp$} \\
               7 & \emp \implies \UNALLOC(\asetvar)
               & \mbox{PC, 6 used for all $\avariable \in \asetvar$} \\
               8 & \emp \implies  \size \geq 0 \wedge \neg (\size \geq 1)
               & \mbox{PC, def. of $\size \geq \inbound$} \\
               9    &\emp \implies \size \geq 0 \wedge \neg (\size \geq 1) \land \UNALLOC(\asetvar)
                    & \mbox{PC, 7, 8}\\
               10 & (\emp \septraction \top) \implies (\size \geq 0 \wedge \neg (\size \geq 1)  \wedge \UNALLOC(\asetvar) \septraction \top)
               & \mbox{\ref{mwAx:ImpL}, 9} \\
               11 & \size \geq 0 \wedge \neg (\size \geq 1)  \wedge \UNALLOC(\asetvar) \septraction \top
               & \mbox{Modus Ponens, 4, 10}
          \end{syntproof}
     \item[base case: $\inbound = 1$] This case corresponds exactly 
     to the axiom~\ref{wandAx:Size}.
     \item[induction step: $\inbound \geq 2$]
          First of all, we notice that the following formula is valid: 
          \begin{equation}
          (\size = 1 \land \UNALLOC(\asetvar)) \separate (\size = \inbound{-}1 \land \UNALLOC(\asetvar)) \implies \size = \inbound \land \UNALLOC(\asetvar).
          \tag{$\dagger$}\label{csl:septraction:auxiliary:dagger}
          \end{equation}
          Indeed, 
          let $\pair{\astore}{\aheap}$ be a memory state satisfying the antecedent of the implication above. 
          So, there are disjoint heaps $\aheap_1$ and $\aheap_2$ such that 
          $\aheap = \aheap_1 \heapsum \aheap_2$, $\card{\domain{\aheap_1}} = 1$, 
          $\card{\domain{\aheap_2}} = \inbound-1$, 
          and for every $\avariable \in \asetvar$, 
          $\astore(\avariable) \not \in \domain{\aheap_1}$ and $\astore(\avariable) \not \in \domain{\aheap_2}$.
          By $\aheap = \aheap_1 \heapsum \aheap_2$, 
          $\card{\domain{\aheap}} = \card{\domain{\aheap_1}} + \card{\domain{\aheap_2}} = \inbound$, and for every~$\avariable \in \asetvar$, 
          $\astore(\avariable) \not \in \domain{\aheap}$.
          Thus, $\pair{\astore}{\aheap} \models \size = \inbound \land \UNALLOC(\asetvar)$.

          As~(\ref{csl:septraction:auxiliary:dagger}) can be 
          seen as a formula in~$\slSA$, 
          by Theorem~\ref{theo:starCompleteness} it is derivable in $\coresys(\separate)$ and thus in $\coresys(\separate,\magicwand)$.
          Now, let us derive 
          $(\size = \inbound \land \UNALLOC(\asetvar)) \septraction \true$.
          Let us consider as induction hypothesis the derivability of 
          $(\size = \inbound{-}1 \land \UNALLOC(\asetvar)) \septraction \true$. 
          Therefore, 
          \begin{syntproof}
               1    & \size = \inbound{-}1 \land \UNALLOC(\asetvar) \septraction \true
                    & \mbox{Induction Hypothesis}\\
               2    & (\size = 1 \land \UNALLOC(\asetvar)) \separate (\size = \inbound{-}1 \land \UNALLOC(\asetvar)) \implies \size = \inbound \land \UNALLOC(\asetvar) 
                    & \mbox{(\ref{csl:septraction:auxiliary:dagger}), see above}\\
               3    & \size = 1 \land \UNALLOC(\asetvar) \septraction \true 
                    & \mbox{\ref{wandAx:Size}}\\
               4    & \true \implies (\size = \inbound{-}1 \land \UNALLOC(\asetvar) \septraction \true)
                    & \mbox{PC, 1}\\
               5    &  (\size = 1 \land \UNALLOC(\asetvar) \septraction \true)
                         \implies\\[-2pt]
                    &  \quad \big(\size = 1 \land \UNALLOC(\asetvar) \septraction 
                       (\size = \inbound{-}1 \land \UNALLOC(\asetvar) \septraction \true)\big)
                    & \mbox{\ref{mwAx:ImpR}, 4}\\
               6    & \big(\size = 1 \land \UNALLOC(\asetvar) \septraction 
                    (\size = \inbound{-}1 \land \UNALLOC(\asetvar) \septraction \true)\big) \implies\\[-2pt]
                    & \quad \big((\size = 1 \land \UNALLOC(\asetvar)) \separate 
                    (\size = \inbound{-}1 \land \UNALLOC(\asetvar)) \septraction \true \big)
                    & \mbox{\ref{mwAx:Curry}}\\
               7    & \big((\size = 1 \land \UNALLOC(\asetvar)) \separate 
               (\size = \inbound{-}1 \land \UNALLOC(\asetvar)) \septraction \true \big) 
               \implies\\[-2pt]
                    & \quad (\size = \inbound \land \UNALLOC(\asetvar) \septraction \true) 
                    & \mbox{\ref{mwAx:ImpL}, 2}\\
               8    & (\size = 1 \land \UNALLOC(\asetvar) \septraction \true) 
                         \implies (\size = \inbound \land \UNALLOC(\asetvar) \septraction \true) 
                    & \mbox{\ref{rule:imptr}, 5, 6, 7}\\
               9    & \size = \inbound \land \UNALLOC(\asetvar) \septraction \true
                    & \mbox{Modus Ponens, 3, 8}\hfill\qedhere
          \end{syntproof}
\end{description}
\end{proof}


\end{document}